\documentclass[12pt]{article}

\usepackage{amsmath, amssymb, amsthm, bm, graphicx, xcolor, physics, newtxtext, newtxmath, tikz, eso-pic, enumitem, float}
\usetikzlibrary {arrows.meta}

\usepackage{geometry}
\usepackage{hyperref}
\hypersetup{breaklinks, colorlinks=true, linkcolor=blue, urlcolor=cyan, citecolor=magenta}

\usepackage[lastpage, user]{zref}
\usepackage{fancyhdr}
\fancypagestyle{plain}{}

\usepackage[numbered, framed]{matlab-prettifier}
\usepackage[numbers]{natbib}

\newcommand*\CircleAroundChar[2][\small]{\tikz[baseline=(char.base)]{\node[shape=circle, draw, inner sep=1pt](char){#1#2};}}

\newcommand{\fis}{\raisebox{2pt}{$\chi$} }
\newcommand{\tdfis}{\raisebox{2pt}{$\widetilde\chi$} }

\theoremstyle{plain}
\newtheorem{theorem}{Theorem}[section]

\theoremstyle{definition}

\newtheorem{proposition}[theorem]{Proposition}
\newtheorem{remark}[theorem]{Remark}

\newtheorem{conjecture}[theorem]{Conjecture}

\renewenvironment{proof}{{\noindent \bf Proof.}}{\qed}

\usepackage [autostyle, english = american]{csquotes}
\MakeOuterQuote{"}

\usepackage{multirow,booktabs} % You can put additional packages here.

\begin{document}

\title{Bifurcation and fission in the liquid drop model:\\ a phase-field approach\thanks{This work is supported in part by the National Science Foundation DMS-2012562 and DMS-1937254.
}} 

\author{Zirui Xu\thanks{Department of Applied Physics and Applied
Mathematics, Columbia University, New York, NY
10027, USA. Email: zx2250@columbia.edu.}
 \quad and \quad Qiang Du\thanks{Department of Applied Physics and Applied
Mathematics, and Data Science Institute,
Columbia University, New York, NY
10027, USA. Email: qd2125@columbia.edu.}\quad(corresponding author)
}

\date{\vspace{-5ex}}

% [Author] {Title}
%\AuthorTitle[Zirui Xu\quad and\quad Qiang Du\\ \small Columbia University, Applied Mathematics]{Bifurcation and fission in the liquid drop model:\\ a phase-field approach\footnote{\;Suggest to think more about the title, to be somewhat more specific with respect to the main contributions and  less generic. }} % [Author] {Title}

\maketitle

\noindent
%\qd{Stuff that Prof. Du wants to add}\\
%\q{Stuff that Prof. Du wants to delete}\\
%\zx{Stuff that Zirui wants to add}\\
%\z{Stuff that Zirui wants to delete}
%\zx{Blue\, texts are the texts that the authors have \;added\; during the revision process}\\
%\z{Cyan texts are the texts that the authors have deleted during the revision process}

\section*{Abstract}

The liquid drop model, originally used to model atomic nuclei, describes the competition between surface tension and Coulomb force. To help understand how a ball loses stability and becomes prone to fission, we calculate the minimum energy path of the fission process, and study the bifurcation branch conjectured by Bohr and Wheeler. We then present the two-dimensional analogue for comparison. Our study is conducted with the help of numerical simulations via a phase-field approach.

%\tableofcontents

\section{Introduction}

The competition between attractive and repulsive interactions gives rise to a wide variety of pattern formation phenomena, ranging from polymer systems to ferroelectric/ferromagnetic systems to quantum systems to reaction-diffusion systems \cite{muratov1998theory,muratov2002theory}. It also accounts for the structure of nuclear matter in the crust of neutron stars \cite{knupfer2016low}. One of the earliest theoretical studies in the literature is the liquid drop model of the atomic nucleus conceived by Gamow in 1928, which captures the competition between the attractive short-range nuclear force and the repulsive long-range Coulomb force \cite{choksi2017old}. Another relevant study is the Ohta\textendash Kawasaki free energy introduced in 1986 to model the self-assembly of diblock copolymers \cite[Equation (5)]{choksi2017old}. Because of their abilities to reproduce various fine mesoscopic structures found in block copolymers and many other systems, the Ohta\textendash Kawasaki free energy and its variants have drawn wide mathematical interests (see, e.g., \cite{xu2022ternary} and the references therein). In a mathematical study by Choksi and Peletier \cite{choksi2010small}, the liquid drop model resurfaced as the leading-order term of the Ohta\textendash Kawasaki free energy in the vanishing volume limit \cite[Top-right of Page 1277]{choksi2017old}. Since then the liquid drop model started receiving much more attention from the mathematical community (see, e.g., \cite{bonacini2014local,frank2016nonexistence,emmert2020liquid} and many references therein), serving as an exemplary model of pattern formation driven by energetic competitions.

The liquid drop model was originally used to model the shape of the atomic nucleus. Although an atomic nucleus typically makes up more than 99.9\% of the atomic mass, it occupies much less than 0.01\% of the atomic volume. In fact, its diameter ranges from 2 to 12 fm (note that 1 fm equals $10^{-15}$ m), much smaller than the resolution of the best microscope \cite[Page 16]{cook2010models}. This means that in nuclear physics, the shape of the very object under study cannot be seen by any microscope. Nonetheless, the shape of a nucleus is a basic notion in theoretical models of nuclear structure and reactions \cite{ivanyuk2013scission}. A proper description of nuclear shapes is essential for explaining and theoretically studying nuclear activities such as nuclear fission \cite{ivanyuk2009optimal}. In the liquid drop model, the atomic nucleus is treated as a drop of incompressible and uniformly charged fluid \cite{choksi2017old}, therefore its energy is given by the surface energy plus the Coulomb potential energy. More precisely, for a nucleus occupying a domain $\Omega\subseteq\mathbb R^3$, up to some rescaling, the energy is given by \eqref{energy-liquid-drop-model}. The volume $V$ of $\Omega$ is related to the mass number, i.e, the number of nucleons. Nucleons are bound together by the attractive nuclear force, so the nucleons on the surface of the nucleus lack neighbors and have higher potential energy, thus giving rise to the surface energy term. Despite its simplicity, the liquid drop model offers satisfactory explanations for the gross properties of nuclei such as the overall trend in the binding energy.

For a measurable set $\Omega\subseteq\mathbb R^3$, we consider the following energy functional \cite[Equation (1)]{choksi2017old}
\begin{equation}
\label{energy-liquid-drop-model}
I(\Omega):=\text{Per}\,\Omega+\frac12\int_\Omega\int_\Omega G(\vec x,\vec y)\dd{\vec x}\dd{\vec y},\quad G(\vec x,\vec y)=\frac1{4\pi|\vec x\!-\!\vec y|},
\end{equation}
under the volume constraint $|\Omega|=V:=40\pi\fis$, where $\fis$ is the so-called fissility parameter, and Per$\,\Omega$ denotes the surface area of $\Omega$ (or the perimeter of $\Omega$ in the 2-D case). Note that in 2-D we define $G(\vec x,\vec y)=-\ln|\vec x\!-\!\vec y|/(2\pi)$ and $\fis=(V/\pi)^{3/2}\big/12$, where $V$ is the area of $\Omega$. In the physics literature \cite[Page 2]{krappe2012theory}, the fissility parameter is denoted as $x$. We use the notation $\fis$ to avoid confusion with $\vec x$. We hereinafter focus on the 3-D case which has relevance to nuclear physics. Under the volume constraint, a ball minimizes the first term in \eqref{energy-liquid-drop-model} due to the isoperimetric inequality, but maximizes the second term due to Riesz rearrangement inequality \cite[Page 755]{choksi2011small}. For $\fis\ll1$, the first term is dominant, and the global minimizer $\Omega$ of $I=I(\Omega)$ is expected to be a ball; for $\fis\gg1$, the second term is dominant, and $\Omega$ is expected to split into many fragments in order to prevent the second term from growing too large. In fact, as $\fis$ exceeds $\fis_2:=(\sqrt[3]2+\sqrt[3]4)^{-1}\approx0.351$, a single ball starts to have higher energy than two balls of equal radii infinitely far apart \cite[Page 755]{choksi2011small}.

According to \cite[Page 4437 and Theorem 3.4]{frank2015compactness} and \cite{frank2016nonexistence} (see also \cite{muratov2014isoperimetric,julin2014isoperimetric,bonacini2014local,lu2014nonexistence}), there exist $0<\fis_2^-\leqslant\fis_2^+\leqslant4/5$ such that:
\begin{enumerate}[label=\protect\CircleAroundChar{\arabic*}]
\item for $\fis<\fis_2^-$, a ball is the unique global minimizer (up to translation);
\item for $\fis=\fis_2^-$, a ball is a global minimizer;
\item for $\fis>\fis_2^-$, a ball is not a global minimizer;
\item for $\fis\in(0,\fis_2^+]\backslash O$ with some open set $O\subseteq(\fis_2^-,\fis_2^+)$, there is a global minimizer;
\item for $\fis\in(\fis_2^+,\infty)\cup O$, there is no global minimizer.
\end{enumerate}
Obviously $\fis_2^-\leqslant\fis_2$. It is widely believed but not proved that $\fis_2^-=\fis_2^+=\fis_2$ \cite[Page 755]{choksi2011small}. For $\fis$ not very small, the global minimizer is not rigorously known. However, several qualitative results are known. According to \cite[Bottom-right of Page 1280]{choksi2017old} and \cite[Theorem 2.7]{bonacini2014local}, any local minimizer has $C^\infty$ reduced boundaries, and is essentially bounded with a finite number of connected components; moreover, any global minimizer is connected. According to \cite[Theorem 1]{julin2017remark}, any $C^2$-regular critical set has analytic boundaries. Although there is no global minimizer for $\fis\in(\fis_2^+,\infty)\cup O$, a generalized global minimizer always exists, consisting of a finite collection of connected components which are infinitely far apart (with each component bounded, connected and $C^\infty$-regular after a zero Lebesgue measure modification) \cite[Theorem 4.5]{knupfer2016low}.

For any $\fis>0$, a ball is always a stationary point of $I$ (i.e., a regular critical set in the sense of \cite[Definition 2.5]{bonacini2014local}). Note that a stationary point is not necessarily stable. In fact, a ball is strictly stable for $\fis<\fis_*:=1$ and unstable for $\fis>\fis_*$ \cite[Theorem 2.9]{bonacini2014local}. By "strictly stable" we mean that for $\fis<\fis_*$, a ball is an isolated local minimizer in the sense of \cite[Definition 2.4]{bonacini2014local}. Thus, for $\fis\in(\fis_2, \fis_*)$, a ball remains an isolated local minimizer although it is no longer a global minimizer. Notice that the stability is often defined in the $L^1$-topology \cite[Equations (2.1) and (2.8)]{kohn1989local} or $L^2$-topology \cite[Bottom of Page 911]{ren2000multiplicity}. Furthermore, strict stability may be defined modulo translation (see \cite[Definition 2.4]{acerbi2013minimality} and \cite[Definition 2.4]{bonacini2014local}). It is known that strict stability can be deduced from strict positivity of the second variation (modulo translation) of a $C^1$-regular critical set \cite[Theorem 2.8]{bonacini2014local}. It should be analogous to define strict stability modulo both translation and rotation, although we have not seen such a definition in the literature.

For $\fis$ slightly larger than $\fis_*$, we are interested to see what will happen to a ball after it loses its stability and ceases to be a local minimizer. For $\fis\approx\fis_*$, a smooth family of stationary points has been constructed in \cite{frank2019non}. Every member of this family is axisymmetric and resembles a prolate \footnote{\;Throughout this paper, we regard "oblate" and "prolate" as nouns instead of adjectives, and refer to the interiors of oblate and prolate spheroids, respectively.} (for $\fis<\fis_*$), a ball (for $\fis=\fis_*$), or an oblate (for $\fis>\fis_*$), as shown on the left of Figure \ref{Spheroids}. This family is illustrated as the tilted line on the right of Figure \ref{Spheroids} (where the vertical axis qualitatively represents the shape, e.g., the height to width ratio). According to \cite[Page 071506-2]{frank2019non}, this family exchanges stability with the ball at $\fis=\fis_*$, where a transcritical bifurcation happens. More specifically, for $\fis<\fis_*$, the prolate-like stationary point is unstable and has higher energy than the ball; for $\fis>\fis_*$, the oblate-like stationary point is "stable" and has lower energy than the ball. Therefore it seems reasonable to speculate that as $\fis$ slightly exceeds $\fis_*$, the ball will deform slightly into a stable oblate-like equilibrium.

\begin{figure}[H]
\centering
$\overset{\includegraphics[width=47.328pt]{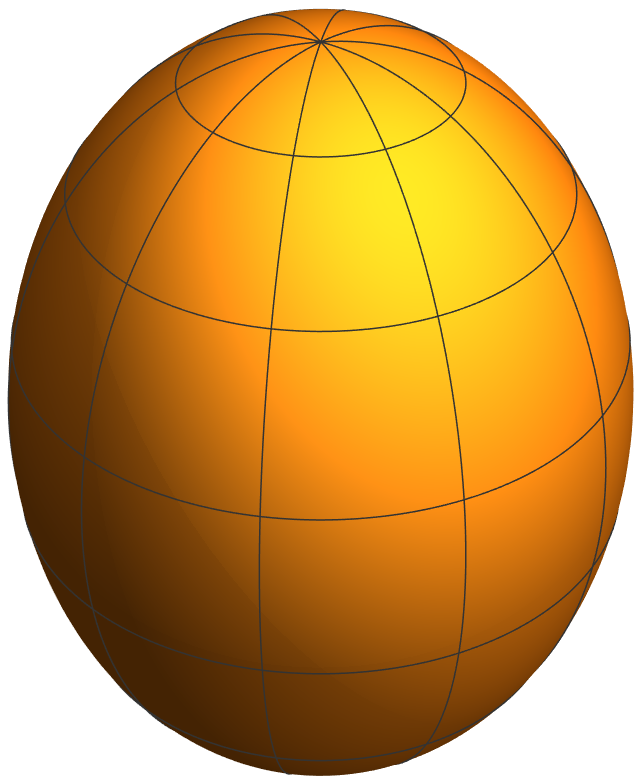}}{\text{prolate}}$
$\overset{\raisebox{3pt}{\includegraphics[width=52.2pt]{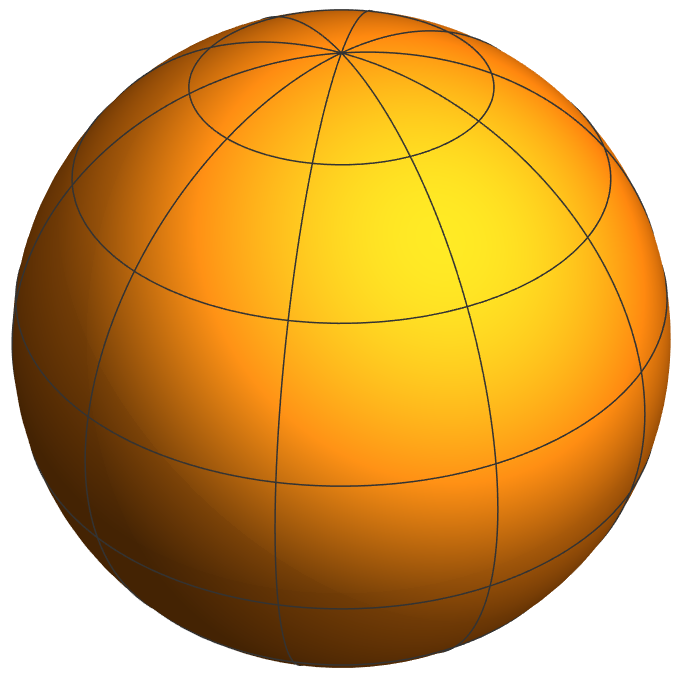}}}{\text{ball}}$
$\overset{\raisebox{4.5pt}{\includegraphics[width=57.42pt]{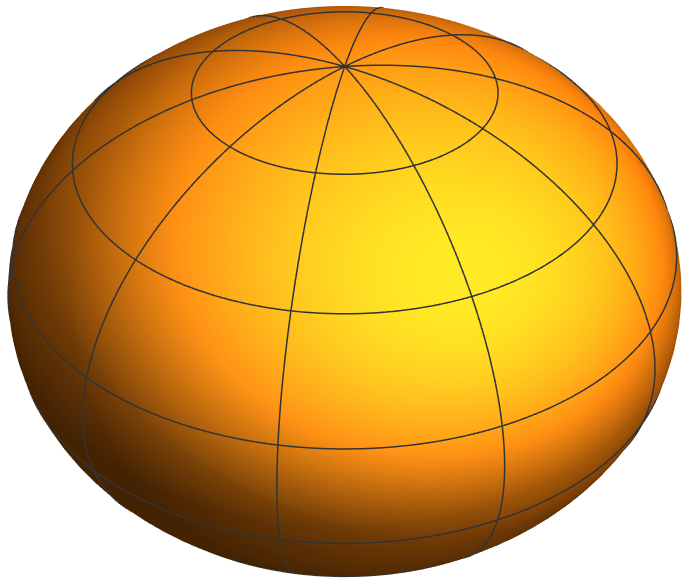}}}{\text{oblate}}$
\hspace{20pt}
\raisebox{-8pt}{\begin{tikzpicture}
\draw [dashed, line width=1pt](-60pt,30pt) -- (0pt,0pt);
\node[rotate=-26] at (-37pt,24pt) {prolate-like};
\filldraw[black] (0pt,0pt) circle (2pt);
\draw [line width=1pt](0pt,0pt) -- (62pt,-31pt);
\draw [line width=1pt](-65pt,0pt) -- (0pt,0pt);
\draw [dashed, line width=1pt](73pt,0pt) -- (0pt,0pt);
\node[rotate=-26] at (48pt,-17pt) {oblate-like};
\node at (21pt,6pt) {$\fis\!>\!\fis_*$};
\node at (63pt,-6pt) {ball};
\node at (-17pt,-6pt) {$\fis\!<\!\fis_*$};
\node at (-55pt,-6pt) {ball};
\end{tikzpicture}}
\\
\caption{Left: three types of spheroids. Right: transcritical bifurcation diagram, where the solid and dashed lines indicate stable and unstable branches, respectively.}
\label{Spheroids}
\end{figure}
Upon closer scrutiny, however, we note that those oblate-like stationary points have only been proved to be stable against certain axisymmetric perturbations \cite[Corollary 5]{frank2019non}. In fact, as shown in Section \ref{Symmetry breaking of oblate-like equilibria}, our numerical simulations indicate that they are unstable against non-axisymmetric perturbations and will eventually split into two disconnected components. In other words, symmetry breaking occurs. Usually in a transcritical bifurcation, the unstable branch exchanges stability with the stable branch and becomes stable when the parameter exceeds the critical value \cite[Pages 51 and 244]{strogatz2018nonlinear}. But as mentioned above, we believe that both the prolate-like and oblate-like branches are unstable, and we will explain this counter-intuitive phenomenon in Section \ref{explainations3D}.

For $\fis<\fis_*$, we are interested in the minimum energy path of the fission process. To the best of our knowledge, in the literature people have not obtained complete fission paths without severe restrictions on the admissible shapes of the nucleus. Using the string method \cite{weinan2007simplified}, we are able to obtain complete fission paths (see Section \ref{Minimum energy paths of fission}). With the help of the shrinking dimer method \cite{zhang2012shrinking}, we are able to obtain the saddle point along the fission path. Such a saddle point, also called a {\em transition state} \footnote{\;A transition state is a point of locally highest energy along the minimum energy path.}, is of particular interest, because it is related to the activation energy required for the nuclear fission to take place. In 1939, Bohr and Wheeler \cite{bohr1939mechanism} studied those saddle points and obtained the leading-order behaviors for $\fis\approx\fis_*$ using a Legendre polynomial expansion. Such asymptotic results were justified rigorously by Frank \cite{frank2019non}, who proved the existence (and uniqueness) of those saddle points for $\fis\approx\fis_*$ (see Figure \ref{Spheroids}). In addition, Bohr and Wheeler conjectured that this bifurcation branch can be extended from $\fis\approx\fis_*$ to arbitrarily small $\fis$, with the saddle point changing from spheroid-like shapes to dumbbell-like shapes and then converging to two touching balls of equal sizes \cite[Figure 2]{bohr1939mechanism}. This conjecture remains open, since the proof given by Frank is only valid in a small neighborhood of $\fis=\fis_*$ \cite[Page 071506-2]{frank2019non}. In the literature, numerical calculation of the Bohr\textendash Wheeler branch is scarce \cite[Top right of Page 1280]{choksi2017old}. Using a phase-field approach, we are able to numerically compute the entire Bohr\textendash Wheeler branch, as shown in Figure \ref{bifurcation branch}, thus providing numerical supports for their conjecture. Note that our computation covers the entire range of $\fis$, i.e., $\fis\in(0,\fis_5)$, with $\fis_5\approx 1.8$ corresponding to a discocyte-like stationary point whose two faces touch each other. In the literature we have not seen any numerical calculation of the portion $\fis\in(\fis_*,\fis_5)$.

\begin{figure}[H]
\centering
\hspace{-30pt}$\overset{\includegraphics[width=70pt,bb=0 0 284 291]{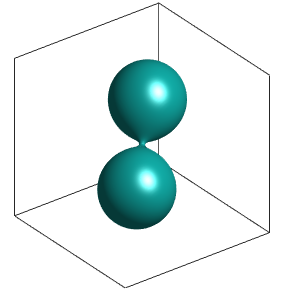}}{\fontsize{10pt}{0pt}\selectfont\text{0.124}\quad}$
\hspace{-12pt}
$\overset{\includegraphics[width=70pt,bb=0 0 284 291]{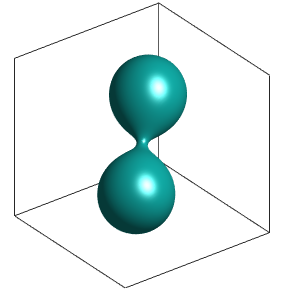}}{\fontsize{10pt}{0pt}\selectfont\text{0.248}\quad}$
\hspace{-12pt}
$\overset{\includegraphics[width=70pt,bb=0 0 284 291]{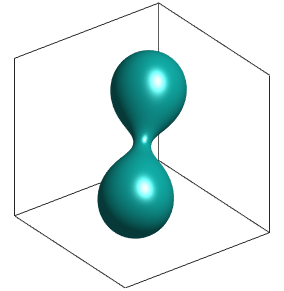}}{\fontsize{10pt}{0pt}\selectfont\text{0.397}\quad}$
\hspace{-12pt}
$\overset{\includegraphics[width=70pt,bb=0 0 284 291]{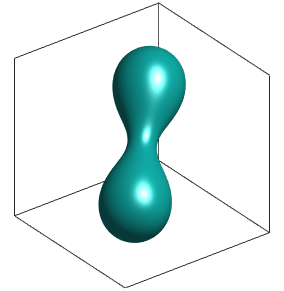}}{\fontsize{10pt}{0pt}\selectfont\text{0.583}\quad}$
\hspace{-12pt}
$\overset{\includegraphics[width=70pt,bb=0 0 284 291]{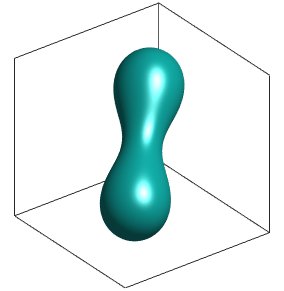}}{\fontsize{10pt}{0pt}\selectfont\text{0.657}\quad}$
\hspace{-12pt}
$\overset{\includegraphics[width=70pt,bb=0 0 284 291]{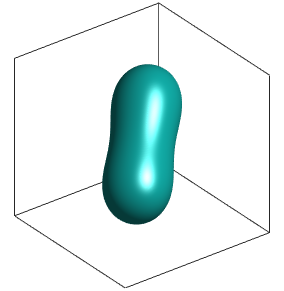}}{\fontsize{10pt}{0pt}\selectfont\text{0.744}\quad}$\\
\hspace{30pt}$\overset{\includegraphics[width=70pt,bb=0 0 284 291]{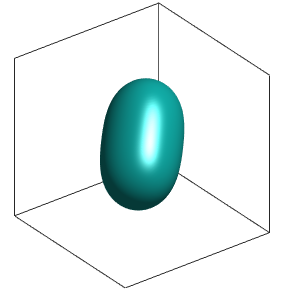}}{\fontsize{10pt}{0pt}\selectfont\text{0.844}\quad}$
\hspace{-12pt}
$\overset{\includegraphics[width=70pt,bb=0 0 284 291]{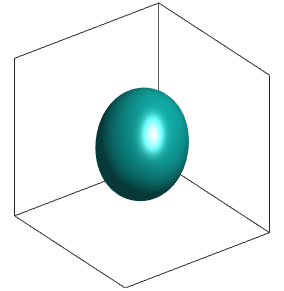}}{\fontsize{10pt}{0pt}\selectfont\text{0.930}\quad}$
\hspace{-12pt}
$\overset{\includegraphics[width=70pt,bb=0 0 284 291]{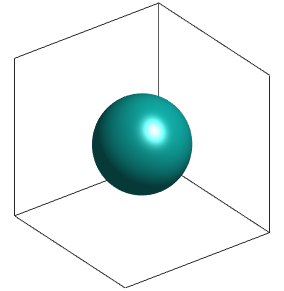}}{\fontsize{10pt}{0pt}\selectfont\text{0.992}\quad}$
\hspace{-12pt}
$\overset{\includegraphics[width=70pt,bb=0 0 284 291]{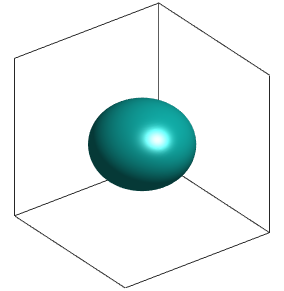}}{\fontsize{10pt}{0pt}\selectfont\text{1.054}\quad}$
\hspace{-12pt}
$\overset{\includegraphics[width=70pt,bb=0 0 284 291]{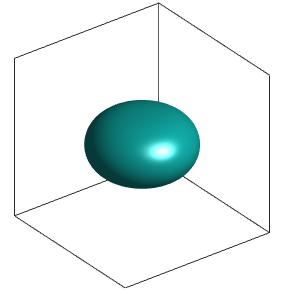}}{\fontsize{10pt}{0pt}\selectfont\text{1.116}\quad}$
\hspace{-12pt}
$\overset{\includegraphics[width=70pt,bb=0 0 284 291]{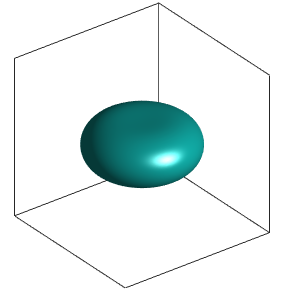}}{\fontsize{10pt}{0pt}\selectfont\text{1.178}\quad}$\\
\hspace{-30pt}$\overset{\includegraphics[width=70pt,bb=0 0 284 291]{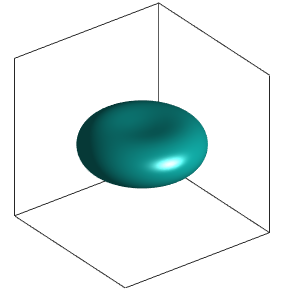}}{\fontsize{10pt}{0pt}\selectfont\text{1.240}\quad}$
\hspace{-12pt}
$\overset{\includegraphics[width=70pt,bb=0 0 284 291]{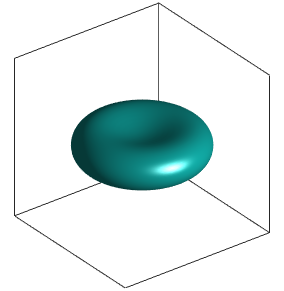}}{\fontsize{10pt}{0pt}\selectfont\text{1.365}\quad}$
\hspace{-12pt}
$\overset{\includegraphics[width=70pt,bb=0 0 284 291]{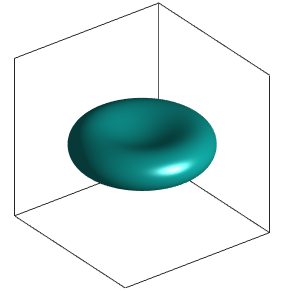}}{\fontsize{10pt}{0pt}\selectfont\text{1.489}\quad}$
\hspace{-12pt}
$\overset{\includegraphics[width=70pt,bb=0 0 284 291]{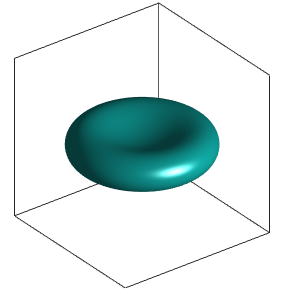}}{\fontsize{10pt}{0pt}\selectfont\text{1.613}\quad}$
\hspace{-12pt}
$\overset{\includegraphics[width=70pt,bb=0 0 284 291]{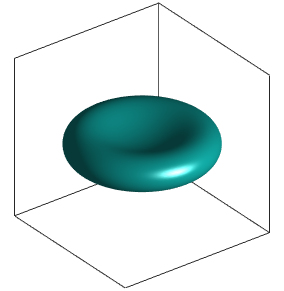}}{\fontsize{10pt}{0pt}\selectfont\text{1.737}\quad}$
\hspace{-12pt}
$\overset{\includegraphics[width=70pt,bb=0 0 284 291]{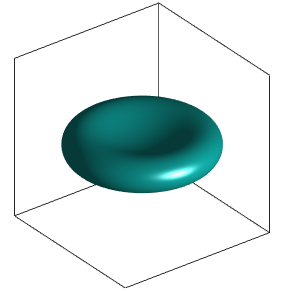}}{\fontsize{10pt}{0pt}\selectfont\text{1.799}\quad}$
\caption{Computed Bohr\textendash Wheeler bifurcation branch. Below each shape is $\tdfis$.}
\label{bifurcation branch}
\end{figure}

In Figure \ref{bifurcation branch}, the simulation box is a unit cube. The shapes are rescaled and their volumes are 0.09. They are axisymmetric, and the axes of revolution for the first nine shapes are chosen to be a body diagonal of the cube, while those for the last nine are chosen to be a vertical line passing through the center of the cube. As discussed in Section \ref{Phase-field approach: the diffuse interface energy}, $\tdfis$ is the phase-field approximation of $\fis$. Following the terminology in \cite{nix1969further}, if an axisymmetric object lacks mirror symmetry perpendicular to the axis of revolution, then we call it mass-asymmetric (e.g., a bowling pin or an European pear that has a heavier bottom and a lighter top). There is a critical value $\fis_3\approx0.395\pm0.001$ (the so-called Businaro\textendash Gallone point), such that for $\fis<\fis_3$, the saddle point becomes unstable to mass-asymmetric perturbations, that is, the two ends of the dumbbell-like shape tend to develop size inequality, until one of them vanishes and the dumbbell becomes a ball (see Section \ref{Businaro-Gallone point}). In other words, as $\fis$ drops below $\fis_3$, the saddle point gains one more unstable direction and its index increases from 1 to 2.

To understand the fate of the Bohr\textendash Wheeler branch after it disappears when $\fis$ exceeds $\fis_5$, we carry out numerical simulations for $\fis>\fis_5$. We observe that the discocyte-like shape on the Bohr\textendash Wheeler branch changes hysteretically to a torus-like shape on a separate branch (see Appendix \ref{torus appendix}). As $\fis$ increases, such a torus-like equilibrium expands and becomes thinner. For $\fis$ large enough, equilibria on this branch have been found by Ren and Wei \cite{ren2011toroidal}. As $\fis$ decreases, such a torus-like shape constricts and thickens, and disappears for $\fis<\fis_4\approx0.972$, through what we believe is a saddle-node bifurcation.

A summary of the above-mentioned qualitative changes is illustrated in Figure \ref{several qualitative changes}. For $\fis=0$, we recover the classical isoperimetric problem upon rescaling \cite[Equation (2.1)]{julin2014isoperimetric}, and any local minimizer must be a ball \cite{delgadino2019alexandrov}. In a tiny neighborhood of $\fis=0$, Julin proved a similar result: for any $M>0$, there is a $\varepsilon>0$ such that for any $\fis\in(0,\varepsilon)$, if a smooth local minimizer of \eqref{energy-liquid-drop-model} satisfies $\text{Per}\,\Omega\leqslant M|\Omega|^{2/3}$, then it must be a ball \cite[Top of Page 071506-2]{frank2019non}. As $\fis$ exceeds $\fis_1:=1/5$ and $\fis_2$, the shape consisting of two balls of equal radii infinitely far apart becomes a generalized local minimizer (see Appendix \ref{Local minimality of two balls}) and a generalized global minimizer (conjectured but unproved), respectively. Once $\fis$ exceeds $\fis_3$, the Bohr\textendash Wheeler branch changes the degree of instability from 2 to 1. In other words, $\fis_3$ is where the dumbbell-like saddle point changes stability against mass-asymmetric perturbations, whereas $\fis_1$ is where the shape consisting of two balls of equal radii infinitely far apart changes stability against mass-asymmetric perturbations. When $\fis$ exceeds $\fis_4$, the torus-like branch comes to existence. When $\fis$ exceeds $\fis_*$, the ball loses its local minimality and becomes unstable. When $\fis$ exceeds $\fis_5$, the Bohr\textendash Wheeler branch ceases to exist.
\begin{figure}[H]
\centering
\begin{tikzpicture}
\draw [-{Stealth[length=8*1.200pt]},line width = 0.6pt](0*1.200pt,0*1.200pt) -- (300*1.200pt,0*1.200pt);
\filldraw[black] (0*1.200pt,0*1.200pt) circle (0.9*1.200pt);
\node at (0*1.200pt,6*1.200pt) {$0$};
\draw [-{Stealth[length=5*1.200pt]},line width = 0.6*1.200pt,dashed](0*1.200pt,0*1.200pt) -- (0*1.200pt,-20.5*1.200pt);
\node[draw,align=center] at (-8*1.200pt,-32*1.200pt) {\fontsize{11pt}{0pt}\selectfont One ball is the only\\[-3*1.200pt]\fontsize{11pt}{0pt}\selectfont local minimizer};
\filldraw[black] (0.2*150*1.200pt,0*1.200pt) circle (0.9*1.200pt);
\node at (0.2*150*1.200pt-2*1.200pt,-6*1.200pt) {$\fis_1$};
\draw [-{Stealth[length=5*1.200pt]},line width = 0.6*1.200pt,dashed](0.2*150*1.200pt,0*1.200pt) -- (0.2*150*1.200pt,20.5*1.200pt);
\node[draw, align=center] at (0.18*150*1.200pt,32*1.200pt) {\fontsize{11pt}{0pt}\selectfont Two-ball becomes\\[-3*1.200pt]\fontsize{11pt}{0pt}\selectfont a local minimizer};
\filldraw[black] (0.3512*150*1.200pt,0*1.200pt) circle (0.9*1.200pt);
\node at (0.3512*150*1.200pt-5*1.200pt,-6*1.200pt) {$\fis_2$};
\draw [-{Stealth[length=5*1.200pt]},line width = 0.6*1.200pt,dashed](0.3512*150*1.200pt,0*1.200pt) -- (0.3512*150*1.200pt+5*1.200pt,-19.7*1.200pt);
\node[draw, align=center] at (0.3512*150*1.200pt+22.5*1.200pt,-32*1.200pt) {\fontsize{11pt}{0pt}\selectfont Two-ball becomes\\[-3*1.200pt]\fontsize{11pt}{0pt}\selectfont a global minimizer};
\filldraw[black] (0.395*150*1.200pt,0*1.200pt) circle (0.9*1.200pt);
\node at (0.395*150*1.200pt+13*1.200pt,6*1.200pt) {$\fis_3$};
\draw [line width = 0.6*1.200pt,dashed](0.395*150*1.200pt+0.44*13.1*1.200pt,0.44*19.8*1.200pt) -- (0.395*150*1.200pt,0*1.200pt);
\draw [-{Stealth[length=5*1.200pt]},line width = 0.6*1.200pt,dashed](0.395*150*1.200pt+0.56*13.1*1.200pt,0.56*19.8*1.200pt) -- (0.395*150*1.200pt+13.1*1.200pt,19.8*1.200pt);
\node[draw, align=center] at (0.395*150*1.200pt+50*1.200pt,32*1.200pt) {\fontsize{11pt}{0pt}\selectfont Businaro\textendash Gallone\\[-3*1.200pt]\fontsize{11pt}{0pt}\selectfont point};
\filldraw[black] (0.972*150*1.200pt,0*1.200pt) circle (0.9*1.200pt);
\node at (0.972*150*1.200pt-7*1.200pt,-6*1.200pt) {$\fis_4$};
\draw [-{Stealth[length=5*1.200pt]},line width = 0.6*1.200pt,dashed](0.972*150*1.200pt,0*1.200pt) -- (0.972*150*1.200pt,-20.5*1.200pt);
\node[draw, align=center] at (0.972*150*1.200pt,-32*1.200pt) {\fontsize{11pt}{0pt}\selectfont Torus-like\\[-3*1.200pt]\fontsize{11pt}{0pt}\selectfont branch starts};
\filldraw[black] (150*1.200pt,0*1.200pt) circle (0.9*1.200pt);
\node at (150*1.200pt+9*1.200pt,6*1.200pt) {$\fis_*$};
\draw [-{Stealth[length=5*1.200pt]},line width = 0.6*1.200pt,dashed](150*1.200pt,0*1.200pt) -- (150*1.200pt+5*1.200pt,19.5*1.200pt);
\node[draw, align=center] at (150*1.200pt+38*1.200pt,32*1.200pt) {\fontsize{11pt}{0pt}\selectfont One ball loses\\[-3*1.200pt]\fontsize{11pt}{0pt}\selectfont local minimality};
\filldraw[black] (1.8*150*1.200pt,0*1.200pt) circle (0.9*1.200pt);
\node at (1.8*150*1.200pt,6*1.200pt) {$\fis_5$};
\draw [-{Stealth[length=5*1.200pt]},line width = 0.6*1.200pt,dashed](1.8*150*1.200pt,0*1.200pt) -- (1.8*150*1.200pt,-20.5*1.200pt);
\node[draw, align=center] at (1.8*150*1.200pt-10*1.200pt,-32*1.200pt) {\fontsize{11pt}{0pt}\selectfont Bohr\textendash Wheeler\\[-3*1.200pt]\fontsize{11pt}{0pt}\selectfont branch ends};
\node at (2*150*1.200pt,-7*1.200pt) {$\fis$};
\end{tikzpicture}
\caption{Several qualitative changes as $\fis$ increases.}
\label{several qualitative changes}
\end{figure}

So far we have only considered the 3-D case. The 2-D case is also interesting, since it corresponds to the cylindrical phase in the diblock copolymer systems \cite[Page 880]{ren2007many}. In 2-D, when a disk loses stability at $\fis=\fis_*$, we believe that it corresponds to a subcritical pitchfork bifurcation (unlike the transcritical bifurcation in 3-D). For $\fis$ slightly smaller than $\fis_*$, there is an intermediate "local minimizer" resembling a smooth connected non-convex Cassini oval, a shape that we refer to as an eye mask. The transition state from the eye mask to two disks resembles a lemniscate (see Section \ref{Analogue in 2-D}).

The rest of this paper is organized as follows. In Section \ref{Physics background and related studies} we compare the present work to related studies in the literature. In Section \ref{Preliminaries} we present our phase-field reformulation and numerical methods. In Section \ref{Simulation results} we present our numerical results in 3-D and 2-D. In Section \ref{Resistance to rupture in the classical isoperimetric problem} we study the periodic isoperimetric problem involving only the perimeter term. In Section \ref{Two touching balls (or disks)} we present asymptotic analysis of two touching balls (or disks) as $\fis\rightarrow0$. In Section \ref{Discussion} we conclude with remarks about future directions.

\section{Comparison with existing literature}

\label{Physics background and related studies}

In the literature, shape parametrizations are often empirically chosen and widely used in the study of the fission process. The nucleus is then restricted to a certain class of shapes, and the variational problem is reduced to an optimization problem in a finite dimension, e.g., 18 dimensions \cite[Page 408]{cohen1963deformation}. With such a reduction, Cohen and Swiatecki \cite[Figure 1a]{cohen1963deformation} obtained the portion $0.3\leqslant\fis\leqslant1$ of the Bohr\textendash Wheeler branch. However, it is questionable whether the chosen class of shapes can accurately approximate the true solution, especially for $\fis<0.28$ \cite[Page 409]{cohen1963deformation}.

In 2009, Ivanyuk and Pomorski \cite{ivanyuk2009shapes, ivanyuk2009optimal} managed to numerically solve the variational problem without relying on shape parameterizations. They assume the nucleus to be axisymmetric with a sharp interface, so that its shape can be described by the meridian of the surface of revolution. They then use iterations to solve the Euler\textendash Lagrange equation subject to the volume constraint. In order to study fission, they introduce another constraint that fixes the elongation. The elongation is a functional which characterizes how elongated a shape is. For a sequence of different values of elongation, they obtain a sequence of shapes, which they think can depict the fission process (because during fission, the nucleus should become more and more elongated, intuitively). Their results therefore depend on the definition of elongation \cite[Figure 3]{ivanyuk2009optimal}, although the dependence may be weak for good choices of definitions. However, the dependence becomes more noticeable when the elongation approaches its maximum, with the nucleus now resembling a dumbbell. Beyond this maximum elongation, no solution is obtained, rendering it difficult to study the later stages of fission (i.e., the rupture of the neck of the dumbbell and the separation of the nucleus into two fragments). As a workaround, Ivanyuk \cite[Figure 6-(b)]{ivanyuk2013scission} carries out a similar calculation for two separate fragments and then pieces together the entire fission process. However, the resulting fission curve is non-smooth and possibly inaccurate. In other words, near the scission point, their method cannot handle topological changes effectively, due to its sharp interface description of the nuclear shape. Note that the scission point is where the nucleus is going through topological changes and about to be split into two fragments, a terminology from \cite[Page 427]{cohen1963deformation} and \cite{nix1969further}. Despite its importance, the scission point is not well understood \cite{ivanyuk2013scission}. Furthermore, it seems difficult to adapt Ivanyuk's method to the case where the axisymmetric assumption is dropped.

In contrast, axisymmetry is not assumed in our phase-field simulations. With the nuclear shape described by a thin layer of diffuse interface, topological changes can be easily simulated. A similar attempt to numerically study the liquid drop model using the phase-field approach has already been made in \cite[Section 5]{generau2018large} for the 2-D case. Inevitably, we introduce an approximation error by changing a sharp interface problem into a diffuse interface problem, but our results suggest that the error is very small. Moreover, the sharp interface model itself is also an approximate description of the actual physics. For example, nuclear densities are said to have exponential tails rather than a sharp interface \cite[Page 6]{krappe2012theory} (see also \cite[Section 6]{ivanyuk2009shapes}). Our method can conveniently take into account the short (but nonzero) range of the nuclear force, by adjusting the diffuseness of the interface if necessary. We can simulate the entire fission process by calculating the minimum energy path, without the need to introduce an artificial definition of elongation. Our phase-field approach will have even more advantages when the shape is non-axisymmetric, e.g., if we also consider the angular momentum of the nucleus.

Businaro and Gallone in 1955, as well as Cohen and Swiatecki in 1962, stated that for $\fis>\fis_*$, the oblate-like stationary points on the Bohr\textendash Wheeler branch are unstable against non-axisymmetric perturbations \cite[Paragraph 2 in Page 211]{natarajan1987role}, but they did not provide evidence \cite[Paragraph 1 in Page 69]{tsamopoulos1985dynamics}. The present study provides convincing numerical evidence for this claim. It remains an open question to analytically prove that the oblate-like stationary points constructed by Frank \cite{frank2019non} are indeed unstable against non-axisymmetric perturbations. As a side note, a similar symmetry breaking phenomenon occurs in a slightly different problem \cite[Bottom of Page 225, Top of Page 218]{natarajan1987role}, where the liquid drop is assumed to be conductive and thus the charge is distributed only on the surface (although such a problem was later found to be ill-posed and thus regularization is needed \cite{muratov2016well}).

Nix \cite[Bottom of Page 265]{nix1969further} suspected that for $\fis\geqslant0.78$, the minimum energy path of fission may involve three fragments, which was not allowed by the parametrization that was used (with only six degrees of freedom). Since our phase-field method does not make use of axisymmetric assumptions or any shape parametrizations, our numerical results confirm that the fission paths are indeed axisymmetric, and that there is no such three-fragments fission.

In the literature, it has been noticed that there is a rapid (but continuous) change in the shape of the saddle point on the Bohr\textendash Wheeler branch around $\fis\approx0.67$ \cite[Abstract]{cohen1963deformation}, and our results in Figure \ref{bifurcation branch} seem to support this observation.

The Businaro\textendash Gallone point is related to mass-asymmetric fission of light and medium nuclei. For $\fis>\fis_3$, the fission can proceed in a mass-symmetric manner, from a ball to a dumbbell and eventually to two separate balls. In other words, every point along this transition path is mass-symmetric, and the transition state is an index-1 saddle point whose energy defines a threshold energy (a notion used in \cite[Page 144]{cohen1962deformation}, equivalent to activation energy in transition state theory). For $\fis<\fis_3$, the saddle point is of index 2, so its energy does not define a threshold energy, and the process tends to pick a mass-asymmetric path. In extremely mass-asymmetric cases, where a very light fragment (consisting of only a few nucleons) is emitted from the parent nucleus, the process is called spallation \cite[Page 144]{cohen1962deformation}. Below the Businaro\textendash Gallone point, it is thought that there is no longer a distinction between fission and spallation \cite[Paragraph 1 of Page 145]{cohen1962deformation}, because the mass-symmetric fission path is not energetically preferable to mass-asymmetric paths. Our numerical simulations confirm the above qualitative statements (see Sections \ref{Minimum energy paths of fission} and \ref{Businaro-Gallone point} as well as Appendix \ref{Simulations in 3-D under periodic boundary conditions}). Businaro and Gallone \cite{businaro1957asymmetric} estimated $\fis_3$ to be 0.47. Later on Cohen and Swiatecki \cite{cohen1963deformation} determined $\fis_3$ to be approximately 0.394, Nix \cite{nix1969further} and Thomas et al. \cite{thomas1985conditional} determined $\fis_3$ to be 0.396, and then Ivanyuk \cite[Page 519]{ivanyuk2010fission} determined $\fis_3$ to be approximately 0.4. Our estimation $\fis_3\approx0.395\pm0.001$ is consistent with those existing results. In the physics literature, it is unknown if there is any instability against mass-asymmetric perturbations during the descent from the saddle point to the scission point \cite[Middle of Page 242]{nix1969further}. Nix suggested that there is likely no such instability \cite[Bottom of Page 265]{nix1969further}. Our results in Section \ref{Minimum energy paths of fission} are consistent with Nix's speculation.

As shown at the top-left of Figure \ref{bifurcation branch}, our results demonstrate that when $\fis$ is small, the saddle point on the Bohr\textendash Wheeler branch resembles two equally large balls connected by a thin and short neck. We believe that it will converge to two equally large balls touching at a single point as $\fis\rightarrow0$, as conjectured by Bohr and Wheeler \cite[Figure 2-(a)]{bohr1939mechanism}. We perform asymptotic analysis in Section \ref{Two touching balls (or disks)} and calculate the asymptotic neck circumference as $\fis\rightarrow0$.

\section{Phase-field models and numerical methods}
%Preliminaries
\label{Preliminaries}

%Throughout this article, by axisymmetry we refer to cylindrical symmetry, i.e., the symmetry obtained from rotation around an axis (called the axis of revolution).

\subsection{Phase-field approach}
\label{Phase-field approach: the diffuse interface energy}

Phase-field methods are useful tools to quantitatively simulate interfaces and geometric motions such as those related to the microstructures found in materials or biological specimens, see \cite{du2020phase} for a recent review.  The basic idea is to use a narrow but diffuse interface in place of the sharp interface. The diffuse interface is described by a phase-field variable $\eta$ (also called an order parameter), which takes different values in different phases and has a rapid but smooth transition across interfaces from one phase to another, with a diffuseness parameter controlling the thickness of the interfacial layer. In phase-field models, there is no need to explicitly track the interfaces, since the governing equations are in a unified form throughout the domain of simulation. We reformulate the liquid drop energy \eqref{energy-liquid-drop-model} into a diffuse interface version. For a fixed diffuseness parameter which is small enough, we study the diffuse interface problem numerically and then infer similar properties about the sharp interface problem. Note that our diffuse interface formulation is very similar to the Ohta\textendash Kawasaki free energy \cite[Equation (6.37)]{choksi2006periodic}.

%\zx{After re-scaling (dividing $\sqrt c$), $\tilde\gamma$ represents $\gamma/\sqrt c$, and the current $K$ represents $K/\sqrt c$. In other words, the $\tilde\gamma>\gamma$. The $K$ before is smaller, the $K$ now is larger.}

%\zx{After re-scaling (dividing $\sqrt{c/6}$), $\tilde\gamma$ represents $\gamma\sqrt{6/c}$, and the current $K$ represents $K\sqrt{6/c}$. In other words, the $\tilde\gamma>\gamma$. The $K$ before is smaller, the $K$ now is larger.}

We consider the following phase-field reformulation of \eqref{energy-liquid-drop-model}
\begin{equation}
\label{diffuse interface version of energy}
\begin{aligned}
\tilde I(\eta):=&\sqrt{\frac6c}\int_D\!W\big(\eta(\vec x)\big)\dd{\vec x}+\sqrt{\frac{3c}2}\!\int_D\big|\nabla\eta(\vec x)\big|^2\dd{\vec x}+\frac{\tilde\gamma}2\!\int_D\int_D f\big(\eta(\vec x)\big)G(\vec x,\vec y)f\big(\eta(\vec y)\big)\dd{\vec x}\dd{\vec y}\\
&+\frac K2\Big(\int_Df\big(\eta(\vec x)\big)\dd{\vec x}-\omega|D|\Big)^2,
\end{aligned}
\end{equation}
%\begin{equation}
%\label{diffuse interface version of energy}
%\begin{aligned}
%\tilde I(\eta):=&\frac1{\sqrt c}\int_D\!W\big(\eta(\vec x)\big)\dd{\vec x}+\frac{\sqrt c}2\!\int_D\big|\nabla\eta(\vec x)\big|^2\dd{\vec x}+\frac{\tilde\gamma}2\!\int_D\int_D f\big(\eta(\vec x)\big)G(\vec x,\vec y)f\big(\eta(\vec y)\big)\dd{\vec x}\dd{\vec y}\\
%&+\frac K2\Big(\int_Df\big(\eta(\vec x)\big)\dd{\vec x}-\omega|D|\Big)^2,
%\end{aligned}
%\end{equation}
% \begin{equation}
% \label{diffuse interface version of energy}
% \begin{aligned}
% \tilde I(\eta):=&\int_D\!W\big(\eta(\vec x)\big)\dd{\vec x}+\frac c2\!\int_D\big|\nabla\eta(\vec x)\big|^2\dd{\vec x}+\frac\gamma2\!\int_D\int_D f\big(\eta(\vec x)\big)G(\vec x,\vec y)f\big(\eta(\vec y)\big)\dd{\vec x}\dd{\vec y}\\
% &+\frac K2\bigg(\int_D\Big(f\big(\eta(\vec x)\big)-\omega\Big)\dd{\vec x}\bigg)^2,
% \end{aligned}
% \end{equation}
where $D\subseteq\mathbb R^n$ is the simulation box ($n\in\{2,3\}$, since the 1-D case has already been solved in \cite{ren2000multiplicity}), $W(z):=3(z\!-\!z^2)^2$, the diffuseness parameter $0<c\ll1$ controls the interfacial thickness, the nonlocal coefficient $\tilde\gamma>0$ controls the strength of Coulomb interactions, $G$ is the Green's function of the negative Laplacian $-\Delta$ on $D$ with suitable boundary conditions, $f(z):=3z^2-2z^3$, and $K\gg1$ is the penalty coefficient for the volume constraint with $\omega$ being the volume fraction. As we will explain below \eqref{rescaled energy-liquid-drop-model}, in order to make clear the physical meaning of the coefficient $\tilde\gamma$ in terms of the fissility parameter, we use the following reparametrization of $\tilde\gamma$,
\begin{equation}
\label{reparametrization of nonlocal coefficient}
\tilde\gamma=\frac{40\pi\tdfis}{\omega|D|}\;\text{in 3-D\quad and}\quad\tilde\gamma=12\tdfis\Big(\frac\pi{\omega|D|}\Big)^{3/2}\;\text{in 2-D}.
\end{equation}

The double well potential $W$ has two minimum points at $z=0$ and $z=1$. In order for the first term of $\tilde I(\eta)$ to be minimized, $\eta$ should be equal to 0 or 1 almost everywhere. However, the second term of $\tilde I(\eta)$ penalizes rapid spatial variations in $\eta$. Since the coefficient $c$ is small, we expect a thin and smooth transition layer to separate the domains where $\eta\approx0$ and $\eta\approx1$, respectively.

The nonlinear function $f$ is a technique introduced in \cite{wang2019bubble} to overcome some numerical difficulties. Since it satisfies $f(0)=0$ and $f(1)=1$, and we expect $\eta\approx0$ or $\eta\approx1$ at most places in $D$, the function $f$ has roughly the same effect on $\tilde I(\eta)$ as the identity function. In addition, since we have $f'(0)=f'(1)=0$, the solution $\eta$ can exhibit the desired two-phases profile, which would have required a much smaller $c$ and thus a much finer grid if we used the identity function instead of $f$. See \cite{wang2019bubble} and \cite[Section 4.1]{xu2019energy} for more details.

For clarity let us for a moment replace $f$ by the identity function, i.e., $f(z):=z$. In this way we can see that the third term of $\tilde I$ resembles the Coulomb potential energy term of $I$ in \eqref{energy-liquid-drop-model}, and that the fourth term of $\tilde I$ is to penalize the violation of the volume constraint $\int_D\eta=\omega|D|$.

\subsection{Sharp interface limit}

As $K\rightarrow\infty$ and $c\rightarrow0$, we expect that the minimizer $\eta$ converges to some indicator function $\bm1_\Omega$, and that the diffuse interface energy \eqref{diffuse interface version of energy} $\Gamma$-converges to the sharp interface energy \eqref{energy-liquid-drop-model} up to some rescaling,
\begin{equation*}
\tilde I(\eta)\overset{\Gamma}{\longrightarrow}\,\text{Per}\,\Omega+\frac{\tilde\gamma}2\big\langle\bm1_\Omega,(-\Delta)^{-1}\bm1_\Omega\big\rangle,\quad\text{subject to}\;|\Omega|=\omega|D|,
\end{equation*}
where $\langle\,\cdot\,,\,\cdot\,\rangle$ represents the $L^2$ inner product, and $(-\Delta)^{-1}$ is equipped with the same boundary conditions as \eqref{diffuse interface version of energy}. In fact, a related $\Gamma$-convergence result is proved in \cite[Section 2]{ren2000multiplicity} under Neumann boundary conditions.

In the variational problem \eqref{energy-liquid-drop-model}, the nonlocal coefficient is fixed, and the only adjustable parameter is the volume of $\Omega$. Note that according to the scaling property \cite[Equation (2.1)]{julin2014isoperimetric}, adjusting the volume and adjusting the nonlocal coefficient are interchangeable, so we can fix the volume and adjust the nonlocal coefficient. More specifically, minimization of \eqref{energy-liquid-drop-model} is equivalent (upon rescaling) to minimizing
\begin{equation}
\label{rescaled energy-liquid-drop-model}
\text{Per}\,\Omega+\frac\gamma2\big\langle\bm1_\Omega,(-\Delta)^{-1}\bm1_\Omega\big\rangle,\quad\text{subject to}\;|\Omega|=\omega|D|,
\end{equation}
where $\gamma=40\pi\fis/(\omega|D|)$ in 3-D, $\gamma=12\fis(\omega|D|/\pi)^{-3/2}$ in 2-D, and $(-\Delta)^{-1}$ is equipped with no boundary conditions since \eqref{energy-liquid-drop-model} is posed on the whole space.

In order to make it convenient to compare our results with the existing results in the physics literature, we define a new parameter $\tdfis$ in \eqref{reparametrization of nonlocal coefficient} so that $\gamma/\fis=\tilde\gamma/\tdfis$. In other words, $\tdfis$ is the counterpart of $\fis$ in the diffuse interface setting. As mentioned before, in the sharp interface problem, a ball loses stability when $\fis=\fis_*$, where $\fis_*$ equals 1. However, we do not expect its counterpart $\tdfis_*$ in the diffuse interface setting to be exactly 1, due to the error introduced by the diffuse interfaces. Nevertheless, we do expect $\tdfis_*$ to converge to 1 as $K\rightarrow\infty$ and $c\rightarrow0$.

\subsection{Energy minimization}

In order to find the local minimizers of $\tilde I$, we use the following $L^2$ gradient flow, which is also called the penalized Allen\textendash Cahn\textendash Ohta\textendash Kawasaki (pACOK) dynamics \cite[Equation (1.13)]{xu2019energy}:
\begin{equation}
\label{Allen-Cahn dynamics}
\frac{\partial\eta}{\partial t}=-\frac{\delta\tilde I(\eta)}{\delta \eta}=-\sqrt{6/c}\,W'(\eta)+\sqrt{6c}\,\Delta\hspace{0.5pt}\eta+\tilde\gamma f'(\eta)\,\Delta^{-1}f(\eta)-Kf'(\eta)\big\langle f(\eta)\!-\!\omega,\,1\big\rangle,
\end{equation}
with a given initial value of $\eta$ at $t=0$, where $t$ is the time variable, and $\Delta^{-1}$ is equipped with the same boundary conditions as \eqref{diffuse interface version of energy}. It is easy to see $W'(z) = 12z^3\!-\!18z^2\!+\!6z$ and $f'(z)=6z\!-\!6z^2$. Therefore any stationary point of $\tilde I$ should satisfy
\begin{equation}
\label{stationary point diffuse interface}
\sqrt{6/c}\,W'(\eta)-\sqrt{6c}\,\Delta\hspace{0.5pt}\eta-\tilde\gamma f'(\eta)\,\Delta^{-1}f(\eta)=-Kf'(\eta)\big\langle f(\eta)\!-\!\omega,\,1\big\rangle.
\end{equation}
Analogously, any stationary point of the sharp interface energy \eqref{rescaled energy-liquid-drop-model} satisfies the following equation \cite[Equation (1.1)]{ren2008spherical} (see also \cite[Equation (6)]{choksi2017old})
\begin{equation}
\label{stationary point sharp interface}
(n\!-\!1)H+\gamma\phi=\lambda\;\text{on}\;\partial\Omega,\quad\text{with}\;\phi:=(-\Delta)^{-1}\bm1_\Omega,
\end{equation}
where $n$ is the spatial dimension, $H$ is the mean curvature of $\partial\Omega$ (mean of principal curvatures, nonnegative if $\Omega$ is convex), $\lambda$ is the Lagrange multiplier, and $(-\Delta)^{-1}$ is equipped with no boundary conditions. In Appendix \ref{Formal derivation of Euler-Lagrange equation from diffuse to sharp interface settings}, we will formally derive \eqref{stationary point sharp interface} from \eqref{stationary point diffuse interface}.

\subsection{Numerical discretization methods}
\label{Numerical methods}

For the time discretization of \eqref{Allen-Cahn dynamics}, we use a semi-implicit time-marching scheme (see, e.g., \cite[Section 2.2]{chen1998applications}) to improve the numerical stability for larger time steps without significantly increasing the computational cost during each time step. We adopt the following scheme that treats linear terms implicitly and nonlinear terms explicitly
\begin{equation}
\label{semi-implicit scheme}
\frac{\eta_{i+1}-\eta_i}{\Delta t}=-\sqrt{\frac6c}\big(W'_1(\eta_{i+1})+W'_2(\eta_i)\big)+\sqrt{6c}\,\Delta\hspace{0.5pt}\eta_{i+1}+\tilde\gamma\big(g_1(\eta_{i+1})+g_2(\eta_i)\big)-K\big(h_1(\eta_{i+1})+h_2(\eta_i)\big),
\end{equation}
where\;\;$g_1(\eta):=\Delta^{-1}\eta$,\;\;$g_2(\eta):=f'(\eta)\,\Delta^{-1}f(\eta)-\Delta^{-1}\eta$,\;\;$h_1(\eta):=\langle\eta\!-\!\omega,\,1\rangle$,\;\;$h_2(\eta):=f'(\eta)\big\langle f(\eta)\!-\!\omega,\,1\big\rangle-\langle\eta\!-\!\omega,\,1\rangle$,\;\;$W_1(z) := 5z^2$,\;\;and\;\;$W_2(z) := 3(z\!-\!z^2)^2 - 5z^2$.

We choose $W_1+W_2$ to be a convex splitting of $W$ (see, e.g., \cite[Page 476]{shen2019new}). It is easy to verify that $W_1$ is convex, and that $W_2$ is concave on the interval $z\in[-0.1,1.1]$. During the time-marching, $W_1$ and $W_2$ will be treated implicitly and explicitly, respectively. Since $W$ is a double well potential whose minimizers are 0 and 1, as long as the initial value satisfies $0\leqslant\eta(\,\cdot\,,0)\leqslant1$, in most of our numerical experiments we observe $-0.1\leqslant\eta(\,\cdot\,,t)\leqslant1.1$ for any $t>0$. Therefore, although $W_2$ is not concave outside the interval $[-0.1,1.1]$, the stability should not be compromised.

The above convex splitting scheme should be unconditionally stable if $f$ was linear. However, because $f$ is chosen to be a nonlinear function, we try another technique to improve the stability, inspired by \cite[Section 2.2]{xu2019energy}. In the third term of $\tilde I(\eta)$, we consider the following way of splitting
\begin{equation*}
\big\langle f(\eta),(-\Delta)^{-1}f(\eta)\big\rangle=\big\langle\eta,(-\Delta)^{-1}\eta\big\rangle+\Big(\big\langle f(\eta),(-\Delta)^{-1}f(\eta)\big\rangle-\big\langle\eta,(-\Delta)^{-1}\eta\big\rangle\Big),
\end{equation*}
where the first and second summands will be treated implicitly and explicitly, giving rise to $g_1$ and $g_2$, respectively. Analogously, in the fourth term of $\tilde I(\eta)$, we have
\begin{equation*}
\big\langle f(\eta)\!-\!\omega,\,1\big\rangle^2=\langle\eta\!-\!\omega,\,1\rangle^2+\Big(\big\langle f(\eta)\!-\!\omega,\,1\big\rangle^2-\langle\eta\!-\!\omega,\,1\rangle^2\Big),
\end{equation*}
where the first and second summands will be treated implicitly and explicitly, giving rise to $h_1$ and $h_2$, respectively.

In order to study the relevant transition path, we use the string method \cite{weinan2007simplified} to find the  minimum energy path. More specifically, given two local minimizers (denoted by $\eta_1$ and $\eta_m$, respectively), we construct an initial string of nodes $\{\eta_i\}_{i=1}^m$ by using, for example, linear interpolation between $\eta_1$ and $\eta_m$. Then we apply one step of gradient descent \eqref{Allen-Cahn dynamics} to each $\eta_i$ and obtain $\eta_i^\star$. After each gradient descent step, the nodes are no longer uniformly distributed along the string, so we perform linear interpolation/reparametrization of the string by equal arc length \cite[Section IV-B]{weinan2007simplified} (we can also locally refine the string by using non-uniform arc length). Iterating the above procedure until convergence, we obtain the minimum energy path. The node of the highest energy among $\{\eta_i\}_{i=1}^m$ can be used to approximate the saddle point. To obtain the saddle point more accurately, we apply the shrinking dimer method \cite{zhang2012shrinking} (with the $\alpha$ in \cite[Equation (2.1)]{zhang2012shrinking} chosen to be 1). We let the dimer length shrink to a small but nonzero value to avoid possible numerical roundoff errors.

We now discuss the spatial discretizations. We start with the boundary conditions of $D$. On the one hand, we always use periodic boundary conditions for the gradient term in \eqref{diffuse interface version of energy} (i.e., the Laplacian term in \eqref{semi-implicit scheme}), so that we can apply the Fourier spectral method described in \cite[Section 2.2]{chen1998applications}. On the other hand, for the Green's function $G$, in this paper we will consider two types of boundary conditions: periodic boundary conditions (the periodic Green's function), and no boundary conditions (the Green's function in the free space). Therefore unless otherwise specified, by boundary conditions we refer to the boundary conditions for $G$. When we choose no boundary conditions, as long as $\eta$ almost vanishes near the boundaries of $D$, we can treat $\eta$ as a periodic function on $D$, and thus periodic boundary conditions can be used for the gradient term in \eqref{diffuse interface version of energy} without any trouble. Hence in our simulations under no boundary conditions, we need to ensure that the liquid drop stays clear of the boundaries of $D$, by choosing suitable initial values or choosing a larger $D$ if necessary.

%\subsubsection*{Periodic boundary conditions}

Under periodic boundary conditions, to obtain $\phi=(-\Delta)^{-1}\eta$, i.e., to solve $-\Delta\phi=\eta$, we use the Fourier spectral method, see, e.g., \cite[upper half of Page 1670]{xu2019energy}. Note that $\eta$ must satisfy a compatibility condition $\int_D\eta=0$ in order for $\phi$ to exist. However, for the sake of brevity we abbreviate $\phi=(-\Delta)^{-1}\big(\eta-\fint_D\eta\big)$ to $\phi=(-\Delta)^{-1}\eta$, and we require $\int_D\phi=0$. In other words, our notation $(-\Delta)^{-1}\eta$ refers to $\int_DG(\vec x,\vec y)\,\eta(\vec y)\dd{\vec y}$, where we require $\int_DG(\,\cdot\,,\vec y)=0$ (for any $\vec y\in D$) but not necessarily $\int_D\eta=0$.

The Green's function $G$ under periodic boundary conditions and the fundamental solution given in \eqref{energy-liquid-drop-model} are different by the addition of a regular part. See \cite[Equations (2.1) and (2.2)]{choksi2010small}; see also \cite[Equation (2.1)]{ren2009oval} and \cite[Equation (2.1)]{ren2008spherical}. Intuitively, near its singularity, the fundamental solution dominates the regular part. Therefore as $\omega\rightarrow0$, we expect $G$ to asymptotically capture the Coulomb potential in \eqref{energy-liquid-drop-model}. In fact, similar $\Gamma$-convergence results have been obtained in \cite[Theorems 4.3 and 6.1]{choksi2010small}. In our numerical simulations, however, $\omega$ cannot be made too small, otherwise we would need to choose a very small $c$ and a very fine grid, resulting in much higher computational costs. Thus, when using periodic boundary conditions, we choose $\omega=0.04$ in 3-D and $\omega=0.1$ in 2-D, although with such choices, $G$ might not accurately approximate the Coulomb potential in \eqref{energy-liquid-drop-model}. Nevertheless, the results under periodic boundary conditions can be compared to those under no boundary conditions and offer us a better understanding of the problem. Moreover, the Green's function under periodic boundary conditions serves as an interesting example of the generalization of \eqref{energy-liquid-drop-model} to the cases where $G$ is a general nonlocal kernel.

%\subsubsection*{No boundary conditions}

Under no boundary conditions (the free space case), the Green's function is the fundamental solution, and thus a larger $\omega$ can be used. For example, we can choose $\omega=0.09$ in 3-D, and $\omega=0.15$ in 2-D. In our numerical scheme \eqref{semi-implicit scheme}, when no boundary conditions are adopted, they are only implemented for the first $\Delta^{-1}$ operator in $g_2$, while periodic boundary conditions are used for other occurrences of $\Delta^{-1}$ for convenience. In the literature \cite{jiang2014fast} there are some involved algorithms to calculate the convolution with a singular kernel, but here we use a simple approach. In order to calculate the convolution $\phi=G*\eta$, we use a piecewise constant interpolation of $\eta$, and then we compute a discrete convolution using fast Fourier transforms. Below we present our algorithms in 2-D and 3-D. In Appendix \ref{numerical convergence of G*eta} we present a numerical test to demonstrate that our algorithm, when used to solve Poisson's equation in the free space, is of quadratic convergence. A complete error analysis of the discretization schemes used in our phase-field model is beyond the scope of this paper, and we leave it to future works.

\paragraph*{The 2-D case}

We define the following double integral as $H(x,y)$:
\begin{equation*}
\int\!\dd{y}\!\int\!\ln\big(x^2\!+\!y^2\big)\dd{x}
=
xy\Big(\ln\big(x^2\!+\!y^2\big)-3\Big)+x^2\arctan\Big(\frac{y}{x}\Big)+y^2\arctan\Big(\frac{x}{y}\Big),
\end{equation*}
where the right-hand side is obtained with the help of WOLFRAM MATHEMATICA. On a uniform grid $\{x_i, y_j\}$ with a grid spacing $h$, we define a discretized version of $G$, denoted by $G_h$, as follows
\begin{equation*}
\begin{aligned}
-4\pi G_h(x_i,y_j)&:=\int_{y_{j-\frac12}}^{y_{j+\frac12}}\int_{x_{i-\frac12}}^{x_{i+\frac12}}\ln\big(x^2\!+\!y^2\big)\dd{x}\dd{y}\\
&\;=H\big(x_{i-\frac12},y_{j-\frac12}\big)-H\big(x_{i-\frac12},y_{j+\frac12}\big)-H\big(x_{i+\frac12},y_{j-\frac12}\big)+H\big(x_{i+\frac12},y_{j+\frac12}\big).
\end{aligned}
\end{equation*}
where $x_{i+\frac12}:=(x_i+x_{i+1})/2$ (similarly for $y_{j+\frac12}$). We then use the following $\phi_h$ to approximate $\phi=(-\Delta)^{-1}\eta$:
\begin{equation*}
\phi_h(x_i,y_j) = \sum_{k,\,l}\eta(x_k,y_l)\,G_h(x_{i-k},y_{j-l}),
\end{equation*}
which is a convolution and can be efficiently computed using discrete Fourier transforms (see Appendix \ref{algorithm 2-D Poissons equation}).

\paragraph*{The 3-D case}

In the 3-D case, similarly we have
\begin{equation*}
\begin{aligned}
\int\!\dd{z}\!\int\!\dd{y}\!\int\!\frac{\dd{x}}{\sqrt{x^2\!+\!y^2\!+\!z^2}}
=&\;x z \ln (r\!+\!y)+x y \ln (r\!+\!z)+y z \ln (r\!+\!x)\\
&-\frac{x^2}{2}\arctan\Big(\frac{y z}{r x}\Big)-\frac{y^2}{2} \arctan\Big(\frac{x z}{r y}\Big)-\frac{z^2}{2} \arctan\Big(\frac{x y}{r z}\Big),
\end{aligned}
\end{equation*}
where $r=\sqrt{x^2\!+\!y^2\!+\!z^2}$. The code is given in Appendix \ref{algorithm 3-D Poissons equation}.

\subsection{Implementation details}

In all our simulations in 3-D, we let $c=2\times10^{-4}$ and $K=2\sqrt{3}\times10^5$. In most of our simulations in 3-D, unless otherwise specified, the simulation domain $D$ is chosen to be $[0,1]^3$ and discretized into $128\times128\times128$ uniform grid points, where $128$ is a power of $2$ so that the fast Fourier transform can work efficiently.

In all of our simulations in 2-D, we let $c=10^{-4}$ and $K=2\sqrt{6}\times10^5$, and the simulation domain is chosen to be $[0,1]^2$ and discretized into $256\times256$ uniform grid points.

We accelerate the computation with the GPU support of MATLAB. The longest simulation runs no more than a few days on an NVIDIA Tesla P100 GPU. We have conducted refined computations to ensure that those choices of parameters give satisfactory accuracies.

\section{Simulation results: bifurcation and fission}
\label{Simulation results}

In this section, we present our numerical simulations of the diffuse interface energy \eqref{diffuse interface version of energy}. In all the figures throughout this paper, if the vertical axis is labeled as "Energy", then it refers to $\tilde I$; if the horizontal axis is labeled as "Time", then it refers to $t$; if the horizontal axis is labeled as "Node", then it refers to the arc length in the string method. Throughout this paper, we visualize the 3-D numerical results using the MATLAB command \texttt{isosurface($\eta$,1/2)}, where \texttt{1/2} is the \texttt{isovalue}. As mentioned in \eqref{reparametrization of nonlocal coefficient}, the parameter $\tdfis$ refers to  $\tilde\gamma\omega|D|/(40\pi)$ in 3-D and $\tilde\gamma\big(\omega|D|/\pi\big)^{3/2}/12$ in 2-D. We focus on the 3-D case under no boundary conditions (with $\omega|D|=0.09$), and then briefly present analogous results in 2-D under no boundary conditions (with $\omega|D|=0.15$). For analogous results under periodic boundary conditions, see Appendices \ref{Simulations in 3-D under periodic boundary conditions} and \ref{Simulations in 2-D under periodic boundary conditions}. Roughly speaking, our simulation results are organized in descending order of $\tdfis$.

\subsection{Symmetry breaking of oblate-like equilibria}
\label{Symmetry breaking of oblate-like equilibria}

In this subsection, we numerical verify that the oblate-like equilibria on the Bohr\textendash Wheeler branch are unstable against non-axisymmetric perturbations. This claim was stated by Businaro and Gallone \cite[Page 631]{businaro1955interpretation} as well as Cohen and Swiatecki \cite[Caption of Figure 39]{cohen1962deformation}, but was not given any evidence.

We numerically study the pACOK dynamics \eqref{Allen-Cahn dynamics}, with the initial value $\eta(\,\cdot\,,0)$ resembling the indicator function of a ball. According to our simulation results, there is a critical value $\tdfis_*\approx0.997$. For any $\tdfis>\tdfis_*$, there exists a time $t>0$ such that $\eta(\,\cdot\,,t)$ resembles the indicator function of two (or even more than two when $\tdfis$ is too large) separate balls. For any $\tdfis<\tdfis_*$, we observe that $\eta(\,\cdot\,,t)$ resembles the indicator function of a ball for any $t>0$ and converges numerically to machine precision as $t\rightarrow\infty$. In other words, $\tdfis_*$ is where the ball loses stability in our simulations and becomes prone to fission. For the sharp interface problem \eqref{energy-liquid-drop-model}, the critical value is known to be $\fis_*=1$. The relative error between $\tdfis_*$ and $\fis_*$ is $-0.3\%$, which could be attributed to the finitely small $c$, finitely large $K$, and finite numerical discretizations.

For $\tdfis>\tdfis_*$, we obtained the equilibria on the Bohr\textendash Wheeler branch as shown in Figure \ref{bifurcation branch}. The equilibrium resembles an oblate when $\tdfis$ is slightly larger than $\tdfis_*$, and gradually changes to a discocyte-like shape as $\tdfis$ increases. Those equilibria are stable against axisymmetric perturbations, but unstable against non-axisymmetric perturbations. In order to overcome their instability against non-axisymmetric perturbations and numerically compute those equilibria, we use the pACOK dynamics \eqref{Allen-Cahn dynamics} with the initial value $\eta(\,\cdot\,,0)$ resembling the indicator function of an oblate, and we employ the following technique to maintain the axisymmetry of the shape that $\eta(\,\cdot\,,t)$ represents: rotate it by $45^\circ,\,90^\circ,\,135^\circ,\,\cdots,\,360^\circ$, then take the average in the sense of superposition. In the rotation step, we need to use interpolation with zero-padding, similar to rotating an RGB image made up of a matrix of pixels. Such a symmetrization procedure only needs to be invoked every 500 time steps in the pACOK dynamics in order to inhibit non-axisymmetric deformations.

To illustrate that the oblate-like equilibria are indeed unstable against non-axisymmetric perturbations, we carry out the following two numerical experiments using the pACOK dynamics \eqref{Allen-Cahn dynamics}, with the initial values being the oblate-like equilibria that we obtained above. In the first numerical experiment, we choose $\tdfis=1.116$ where the oblate-like equilibrium is noticeably flattened, and our goal is to demonstrate how symmetry breaking occurs. In the second numerical experiment, we choose $\tdfis=1.004$ where the oblate-like equilibrium is very close to a ball, and our goal is to demonstrate that it is still unstable against non-axisymmetric perturbations, even when $\tdfis$ is just slightly larger than $\tdfis_*$.

The first numerical experiment is shown in Figure \ref{symmetry breaking gamma=9}. In order to simulate the later stages (where the two fragments repel each other farther away) without the fragments reaching the boundaries of the simulation box, we carry out the simulation in a larger box $[-1/2,\,3/2]^2\times[0,1]$ instead of $[0,1]^3$. To make Figure \ref{symmetry breaking gamma=9} more compact, however, we use $[0,1]^3$ as the bounding box indicated by the straight black line segments. The first five shapes lie inside $[0,1]^3$, whereas the last four shapes do not. In particular, the last shape sits completely outside $[0,1]^3$, although this spatial relation may not be visually intuitive at first glance. When viewed from above (along the axis of revolution), the 2-D projection of such an oblate-like equilibrium tends to change from a disk to an ellipse-like shape (symmetry breaking), and eventually separate into two disconnected disk-like components.

\begin{figure}[H]
\centering
\includegraphics[width=500pt]{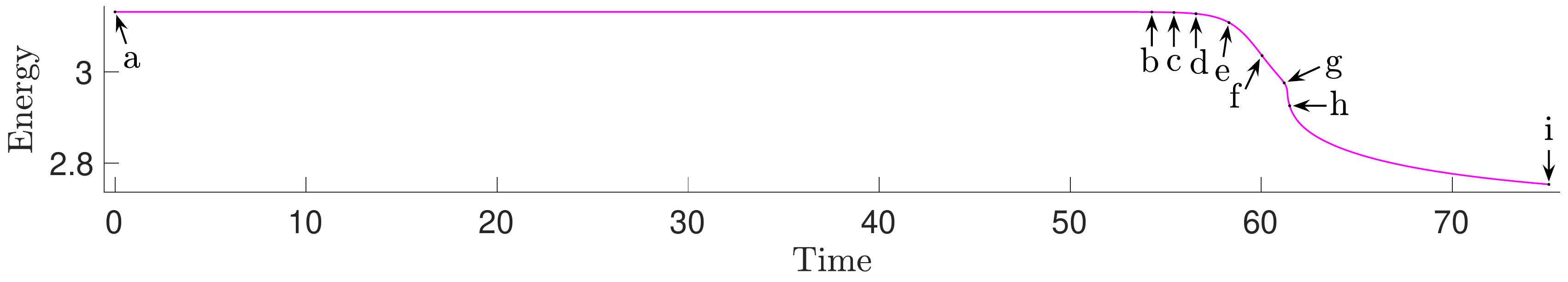}

$\overset{\includegraphics[width=58pt]{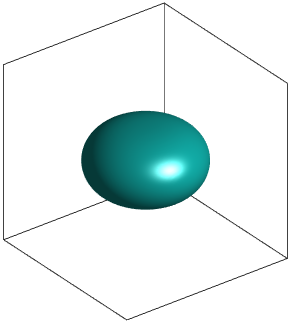}}{\begin{minipage}[t][\height][b]{24pt}\hspace{7pt}a\end{minipage}}$
$\overset{\includegraphics[width=58pt]{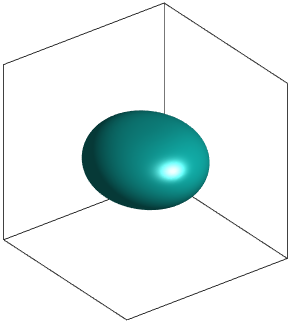}}{\begin{minipage}[t][\height][b]{24pt}\hspace{7pt}b\end{minipage}}$
$\overset{\includegraphics[width=58pt]{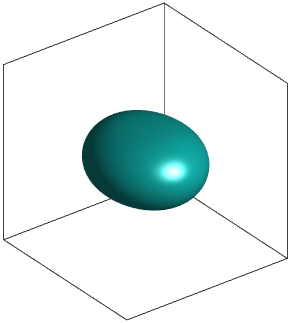}}{\begin{minipage}[t][\height][b]{24pt}\hspace{7pt}c\end{minipage}}$
$\overset{\includegraphics[width=58pt]{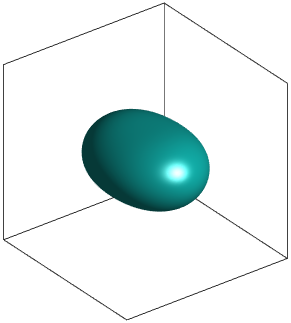}}{\begin{minipage}[t][\height][b]{24pt}\hspace{7pt}d\end{minipage}}$
$\overset{\includegraphics[width=58pt]{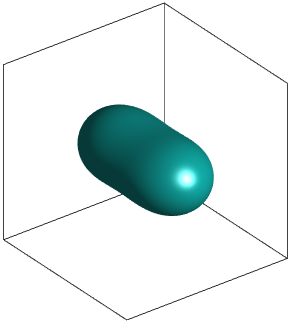}}{\begin{minipage}[t][\height][b]{24pt}\hspace{7pt}e\end{minipage}}$

$\overset{\includegraphics[width=58pt]{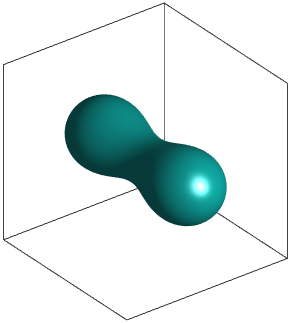}}{\begin{minipage}[t][\height][b]{24pt}\hspace{7pt}f\end{minipage}}$
$\overset{\includegraphics[width=58pt]{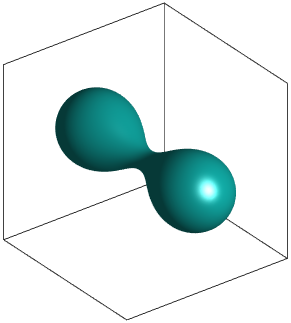}}{\begin{minipage}[t][\height][b]{24pt}\hspace{7pt}g\end{minipage}}$
$\overset{\includegraphics[width=58pt]{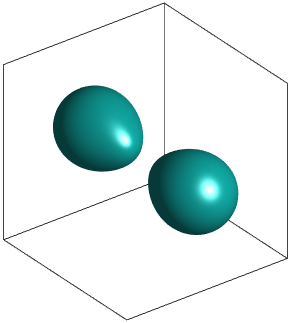}}{\begin{minipage}[t][\height][b]{24pt}\hspace{7pt}h\end{minipage}}$
$\overset{\includegraphics[width=58pt]{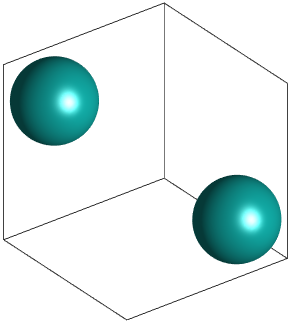}}{\begin{minipage}[t][\height][b]{24pt}\hspace{7pt}i\end{minipage}}$
\caption{pACOK dynamics for $\tdfis=1.116$.}
\label{symmetry breaking gamma=9}
\end{figure}

The second numerical experiment is shown in Figure \ref{symmetry breaking gamma=8.09}. The simulation box is $[0,1]^3$, therefore we can only simulate the stages up to f, after which the shape will reach the boundaries of the simulation box. Note that the initial value (represented by a) is an oblate-like stationary point (whose axis of revolution is a vertical line) with a very tiny non-axisymmetric perturbation. After a long time, the perturbation grows exponentially large, resulting in a dramatic deformation, thus demonstrating the instability of the oblate-like equilibrium for a $\tdfis$ that is slightly larger than $\tdfis_*$.

\begin{figure}[H]
\centering
\includegraphics[width=396.735pt]{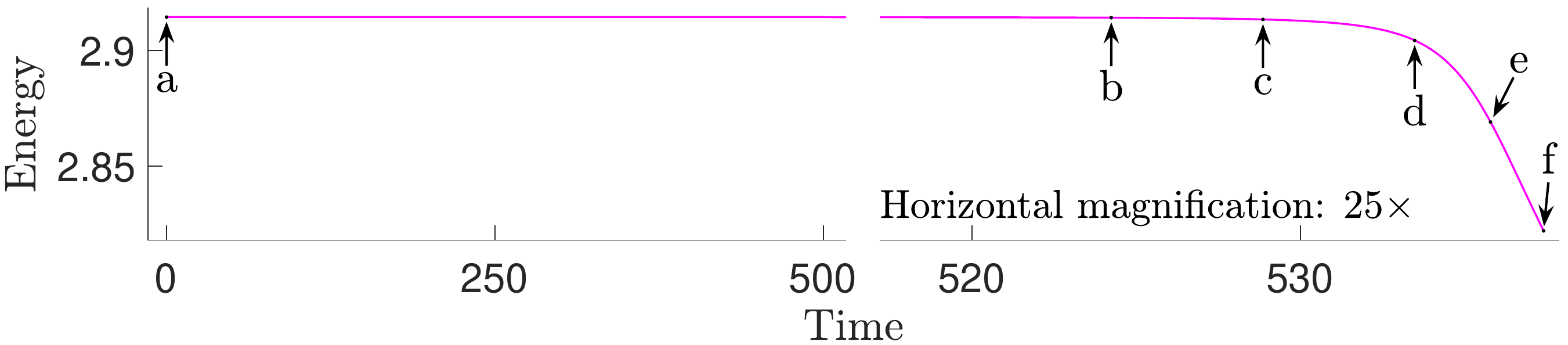}

$\overset{\includegraphics[width=58pt]{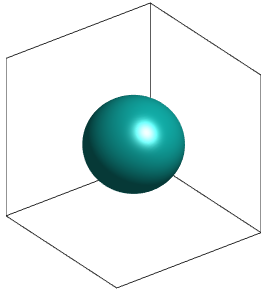}}{\begin{minipage}[t][\height][b]{24pt}\hspace{7pt}a\end{minipage}}$
$\overset{\includegraphics[width=58pt]{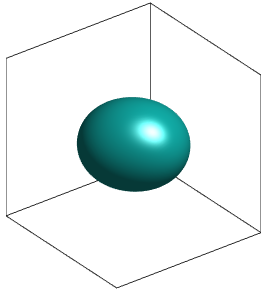}}{\begin{minipage}[t][\height][b]{24pt}\hspace{7pt}b\end{minipage}}$
$\overset{\includegraphics[width=58pt]{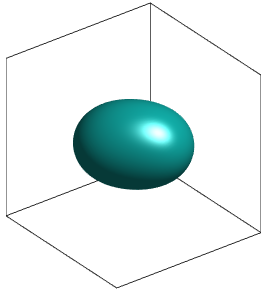}}{\begin{minipage}[t][\height][b]{24pt}\hspace{7pt}c\end{minipage}}$
$\overset{\includegraphics[width=58pt]{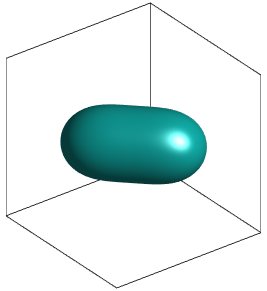}}{\begin{minipage}[t][\height][b]{24pt}\hspace{7pt}d\end{minipage}}$
$\overset{\includegraphics[width=58pt]{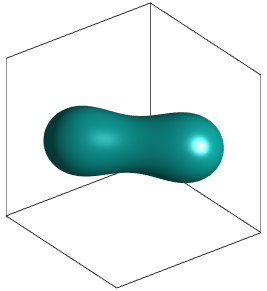}}{\begin{minipage}[t][\height][b]{24pt}\hspace{7pt}e\end{minipage}}$
$\overset{\includegraphics[width=58pt]{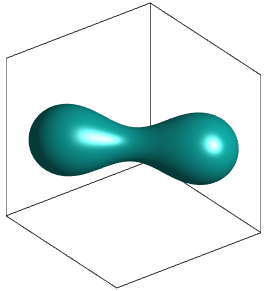}}{\begin{minipage}[t][\height][b]{24pt}\hspace{7pt}f\end{minipage}}$
\caption{pACOK dynamics for $\tdfis=1.004$.}
\label{symmetry breaking gamma=8.09}
\end{figure}

\subsection{Minimum energy paths of fission}
\label{Minimum energy paths of fission}

In this subsection, we obtain the minimum energy paths of fission for $\tdfis<\tdfis_*$, with the help of the string method. Complete fission paths have not been obtained in the literature before. Difficulties have been encountered in calculating the fission path near the scission point \cite{ivanyuk2013scission}. The reason we suspect is that topological changes cannot be easily handled by the sharp interface formulation in \cite{ivanyuk2013scission}. In contrast, topological changes can be easily handled in our phase-field formulation, and we are able to obtain the complete fission paths shown in Figure \ref{Minimum energy paths of fission for various fissility parameters}. We omit the latter part of the string where the two fragments reach the boundaries of the simulation box. The saddle points are indicated by b, c, d, e, f, g and h. The scission points are indicated by i, j, k, l, m, n and o. The details can be found in Appendix \ref{Simulations in 3-D under no boundary conditions}.
\begin{figure}[H]
\centering
\includegraphics[width=155pt]{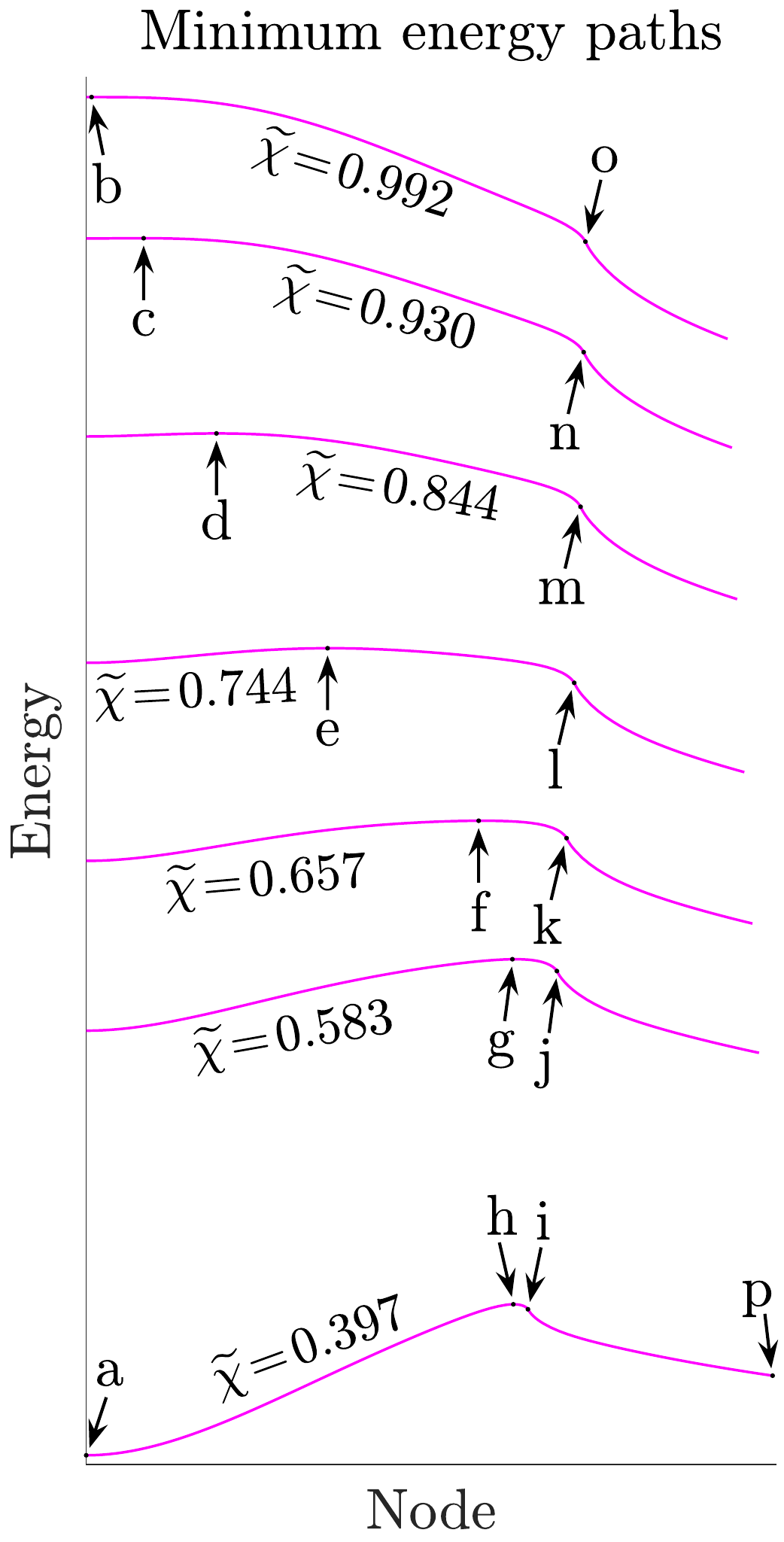}
\hspace{2pt}
\begin{minipage}[b][305pt][t]{270pt}

\hspace{-15pt}
$\overset{\includegraphics[width=58pt]{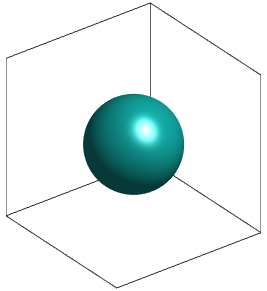}}{\begin{minipage}[t][\height][b]{24pt}\hspace{7pt}a\end{minipage}}$
$\overset{\includegraphics[width=58pt]{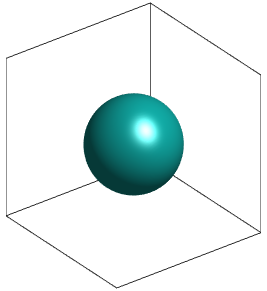}}{\begin{minipage}[t][\height][b]{24pt}\hspace{7pt}b\end{minipage}}$
$\overset{\includegraphics[width=58pt]{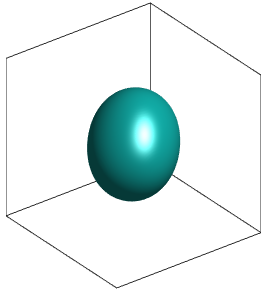}}{\begin{minipage}[t][\height][b]{24pt}\hspace{7pt}c\end{minipage}}$
$\overset{\includegraphics[width=58pt]{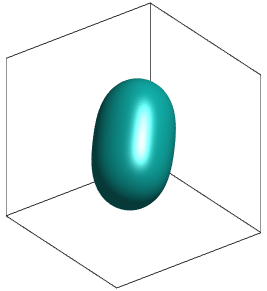}}{\begin{minipage}[t][\height][b]{24pt}\hspace{7pt}d\end{minipage}}$

\hspace{15pt}
$\overset{\includegraphics[width=58pt]{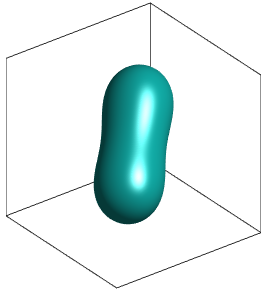}}{\begin{minipage}[t][\height][b]{24pt}\hspace{7pt}e\end{minipage}}$
$\overset{\includegraphics[width=58pt]{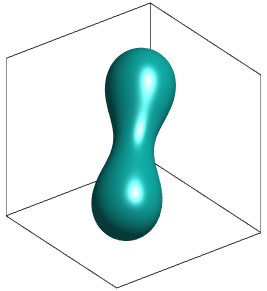}}{\begin{minipage}[t][\height][b]{24pt}\hspace{7pt}f\end{minipage}}$
$\overset{\includegraphics[width=58pt]{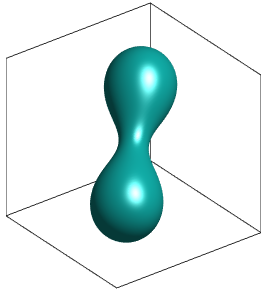}}{\begin{minipage}[t][\height][b]{24pt}\hspace{7pt}g\end{minipage}}$
$\overset{\includegraphics[width=58pt]{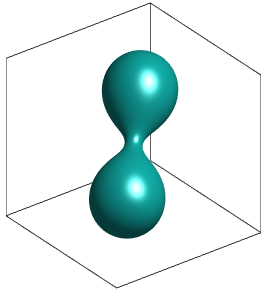}}{\begin{minipage}[t][\height][b]{24pt}\hspace{7pt}h\end{minipage}}$

\hspace{-15pt}
$\overset{\includegraphics[width=58pt]{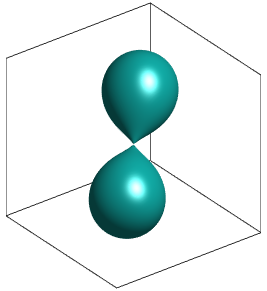}}{\begin{minipage}[t][\height][b]{24pt}\hspace{7pt}i\end{minipage}}$
$\overset{\includegraphics[width=58pt]{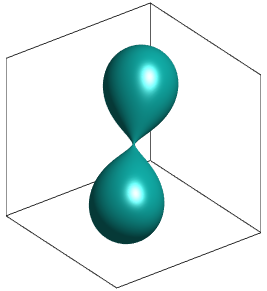}}{\begin{minipage}[t][\height][b]{24pt}\hspace{7pt}j\end{minipage}}$
$\overset{\includegraphics[width=58pt]{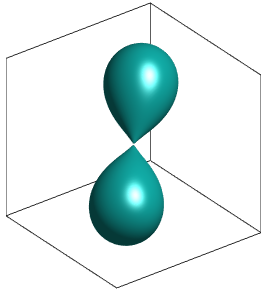}}{\begin{minipage}[t][\height][b]{24pt}\hspace{7pt}k\end{minipage}}$
$\overset{\includegraphics[width=58pt]{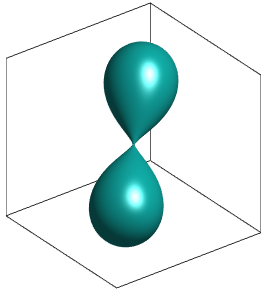}}{\begin{minipage}[t][\height][b]{24pt}\hspace{7pt}l\end{minipage}}$

\hspace{15pt}
$\overset{\includegraphics[width=58pt]{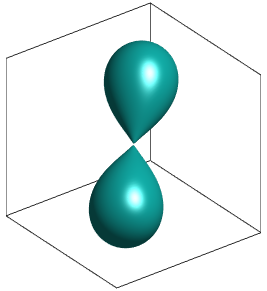}}{\begin{minipage}[t][\height][b]{24pt}\hspace{7pt}m\end{minipage}}$
$\overset{\includegraphics[width=58pt]{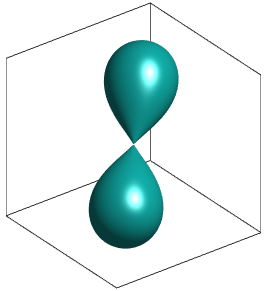}}{\begin{minipage}[t][\height][b]{24pt}\hspace{7pt}n\end{minipage}}$
$\overset{\includegraphics[width=58pt]{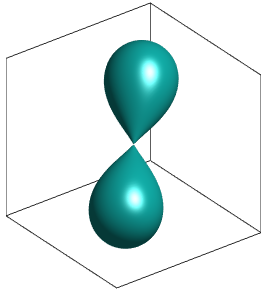}}{\begin{minipage}[t][\height][b]{24pt}\hspace{7pt}o\end{minipage}}$
$\overset{\includegraphics[width=58pt]{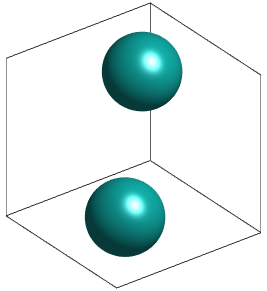}}{\begin{minipage}[t][\height][b]{24pt}\hspace{7pt}p\end{minipage}}$

\end{minipage}
\caption{Minimum energy paths of fission for various fissility parameters.}
\label{Minimum energy paths of fission for various fissility parameters}
\end{figure}
As we will discuss in Section \ref{Businaro-Gallone point}, the saddle point loses stability against mass-asymmetric perturbations when $\tdfis$ is below a threshold $\tdfis_3$. For $\fis\in(\fis_3,\fis_*)$, it is unknown in the physics literature if there is instability against mass-asymmetric perturbations during the descent from the saddle point to the scission point \cite[Middle of Page 242]{nix1969further}. Nix \cite[Bottom of Page 265]{nix1969further} suggested that there is likely no such instability. Our numerical results support that there is no such instability, because we are able to obtain a stable fission path for any $\tdfis\in(\tdfis_3,\tdfis_*)$.

\subsection{Businaro\textendash Gallone point}
\label{Businaro-Gallone point}

In this subsection, we numerically verify the known results about the Businaro\textendash Gallone point. There is a critical point $\tdfis_3$ in our simulations, such that for $\tdfis<\tdfis_3$, the saddle point on the Bohr\textendash Wheeler branch loses stability against mass-asymmetric perturbations, and thus the fission path becomes unstable. See Appendix \ref{Simulations in 3-D under periodic boundary conditions} for another possible minimum energy path called the spallation path.

We carry out the following numerical experiment to study the instability against mass-asymmetric perturbations. First we use the shrinking dimer dynamics \cite{zhang2012shrinking} to calculate the saddle point along the fission path for $\tdfis\in(\tdfis_3,\tdfis_*)$. Note that we already have a good initial guess of the saddle point and its unstable direction: in our simulation results from the string method, we can use the node of the highest energy as the initial guess of the saddle point, and for the initial guess of the unstable direction we can use the difference between the top two nodes of the highest energy.

The saddle point obtained above is a stationary point on the Bohr\textendash Wheeler branch. We can track this bifurcation branch down as $\tdfis$ decreases. For any stationary point on the Bohr\textendash Wheeler branch, we can make a very tiny mass-asymmetric perturbation to it and then use the perturbed shape as the initial value of the shrinking dimer dynamics. For $\tdfis\leqslant0.39386$, as shown in Figure \ref{Shrinking dimer dynamics for 0.372}, the shape changes very little for a long time, but then goes through a dramatic change with one end shrinking and the other growing, and eventually converges to a ball. For $\tdfis\geqslant0.39398$, the shape restores to equilibrium. So we know $\tdfis_3\in(0.39386,0.39398)$. We use $\tdfis_3/\tdfis_*\approx0.395$ to estimate $\fis_3$. Note that the shrinking dimer dynamics normally would not converge to a local minimizer, but our system is degenerate in the sense that there are zero eigenvalues whose eigenvectors correspond to translational motions.

\begin{figure}[H]
\centering
\includegraphics[width=472.577pt]{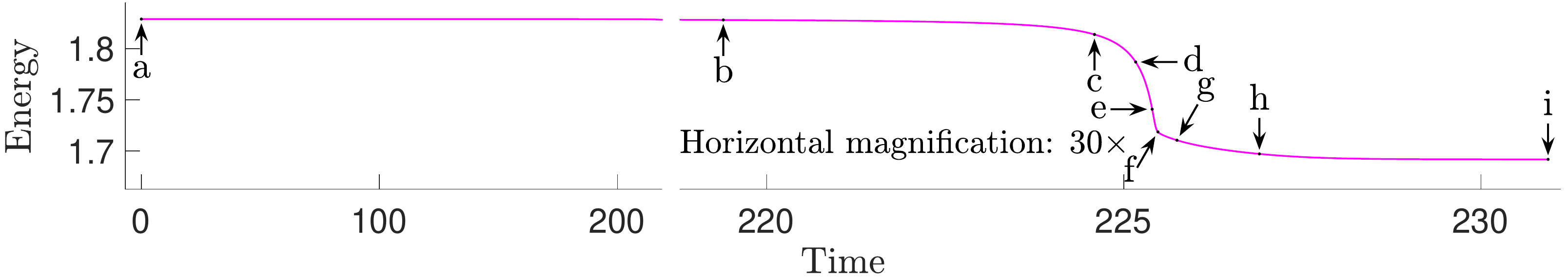}

$\overset{\includegraphics[width=58pt]{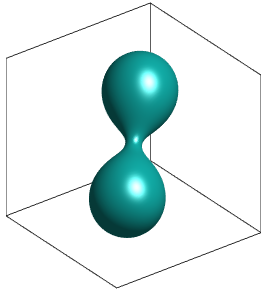}}{\begin{minipage}[t][\height][b]{24pt}\hspace{7pt}a\end{minipage}}$
$\overset{\includegraphics[width=58pt]{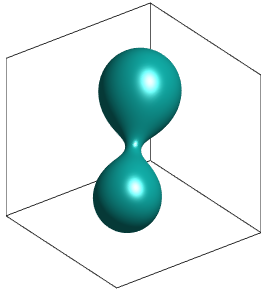}}{\begin{minipage}[t][\height][b]{24pt}\hspace{7pt}b\end{minipage}}$
$\overset{\includegraphics[width=58pt]{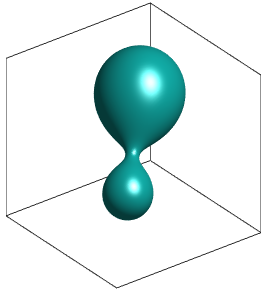}}{\begin{minipage}[t][\height][b]{24pt}\hspace{7pt}c\end{minipage}}$
$\overset{\includegraphics[width=58pt]{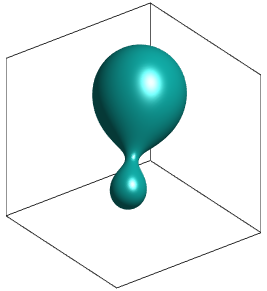}}{\begin{minipage}[t][\height][b]{24pt}\hspace{7pt}d\end{minipage}}$
$\overset{\includegraphics[width=58pt]{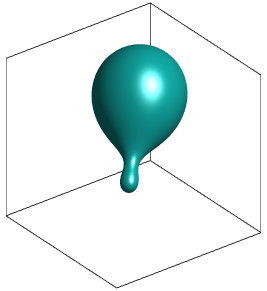}}{\begin{minipage}[t][\height][b]{24pt}\hspace{7pt}e\end{minipage}}$

$\overset{\includegraphics[width=58pt]{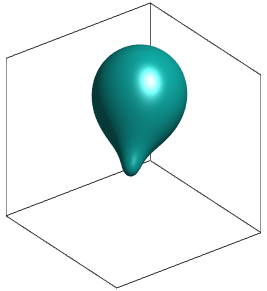}}{\begin{minipage}[t][\height][b]{24pt}\hspace{7pt}f\end{minipage}}$
$\overset{\includegraphics[width=58pt]{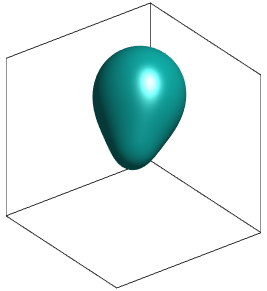}}{\begin{minipage}[t][\height][b]{24pt}\hspace{7pt}g\end{minipage}}$
$\overset{\includegraphics[width=58pt]{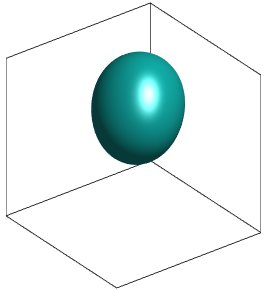}}{\begin{minipage}[t][\height][b]{24pt}\hspace{7pt}h\end{minipage}}$
$\overset{\includegraphics[width=58pt]{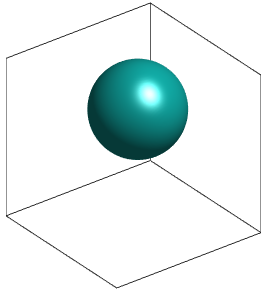}}{\begin{minipage}[t][\height][b]{24pt}\hspace{7pt}i\end{minipage}}$
\caption{Shrinking dimer dynamics for $\tdfis=0.372$.}
\label{Shrinking dimer dynamics for 0.372}
\end{figure}

For $\tdfis<\tdfis_3$, in order to calculate the equilibrium on the Bohr\textendash Wheeler branch which is an index-2 saddle point, while one can use a higher-index extension of the shrinking dimer method \cite[Bottom of Page 1919]{zhang2012shrinking}, we instead employ a simple technique to maintain its centrosymmetry: after every a few (e.g., 500) iterations, take the average of the computed shape and its point reflection with respect to the center of the simulation box.

\subsection{Analogue in 2-D}
\label{Analogue in 2-D}

In this subsection, we study the 2-D analogue under no boundary conditions. To the best of our knowledge, most of our results have not been seen in the literature. Note that in 2-D, a disk is only a stationary point of \eqref{energy-liquid-drop-model} since no global minimizer or local minimizer exists in the full space \cite[Bottom of Page 78]{julin2014isoperimetric}. However, if we restrict $\Omega$ to the interior of some bounded domain in $\mathbb R^2$, then the minimizers may come into existence \cite[Remark 2.14]{bonacini2014local}. For brevity such restricted 2-D local minimizers are simply called local minimizers in this paper.

We numerically study the pACOK dynamics \eqref{Allen-Cahn dynamics}, with the initial value $\eta(\,\cdot\,,0)$ resembling the indicator function of a disk. According to our simulation results, there is a critical value $\tdfis_*\approx0.99978$. For any $\tdfis>\tdfis_*$, there exists a time $t>0$ such that $\eta(\,\cdot\,,t)$ is approximately the indicator function of a shape resembling an eye mask. For any $\tdfis<\tdfis_*$, we observe that $\eta(\,\cdot\,,t)$ resembles the indicator function of a disk for any $t>0$ and converges numerically to machine precision as $t\rightarrow\infty$. In other words, $\tdfis_*$ is where the disk loses stability in our simulations. For the sharp interface problem \eqref{energy-liquid-drop-model}, the critical value is known to be $\fis_*=1$ \cite[Top of Page 1126]{ren2009oval}. The relative error between $\tdfis_*$ and $\fis_*$ is $-0.02\%$, which could be attributed to the finitely small $c$, finitely large $K$, and finite numerical discretizations.

As mentioned above, after losing its stability, a disk will deform and converge to an eye-mask shaped local minimizer. We use such a local minimizer as the initial value, and numerically solve the pACOK dynamics \eqref{Allen-Cahn dynamics} with a larger or smaller $\tdfis$. In this way we can track this bifurcation branch up and down.

According to our numerical simulations, as $\tdfis$ increases, the eye-mask shaped local minimizer becomes thinner and elongates. This persists even for $\tdfis$ very large. When $\tdfis$ is extremely large, the entire eye-mask shape becomes very thin especially in the middle of its neck, and its neck thickness is comparable to the diffuse interface thickness. Only then, the eye-mask shape would break up. We suspect that such a break-up is due to diffuse interface approximations, and that it would never happen in the sharp interface limit \eqref{energy-liquid-drop-model}.

As $\tdfis$ decreases, we observe that the eye-mask shaped local minimizer fattens and shortens. We can track this bifurcation branch all the way down to $\tdfis=0.905$. For $\tdfis\leqslant0.903$, the eye-mask shape will converge to a disk. Therefore we suspect that there is a $\tdfis_4\in(0.903,0.905)$ such that a saddle-node bifurcation happens at $\tdfis_4$, as shown on the top-left of Figure \ref{Minimum energy paths in 2-D various fissility parameters.}.

We use the string method to obtain the minimum energy paths of fission, as shown in Figure \ref{Minimum energy paths in 2-D various fissility parameters.} (see Appendix \ref{Simulations in 2-D under no boundary conditions} for details). We omit the latter part of the string where the two fragments reach the boundaries of the simulation box. For $\tdfis\in(\tdfis_4,\tdfis_*)$, there is an intermediate local minimizer (of an eye-mask shape, indicated by g, h, i, j and k) along the minimum energy path of fission, and there is a transition state (indicated by b, c, d, e and f) from a disk to such an intermediate local minimizer. Such a transition state resembles an ellipse for $\tdfis$ slightly less than $\tdfis_*$, and has been found by Ren and Wei in the sharp interface limit \cite{ren2009oval}. It bifurcates from a disk at $\tdfis=\tdfis_*$, and lies on a bifurcation branch which we call the Ren\textendash Wei branch. For $\tdfis<\tdfis_4$, there is no intermediate local minimizer. The scission point (indicated by l, m, n, o, p, q and r) resembles a lemniscate and is a transition state along the minimum energy path. There is a 2-D analogue of the Businaro\textendash Gallone point $\tdfis_3\approx0.428$, such that the critical point resembling a lemniscate becomes unstable against mass-asymmetric perturbations when $\tdfis<\tdfis_3$.

\begin{figure}[H]
\centering
\includegraphics[width=153.06pt]{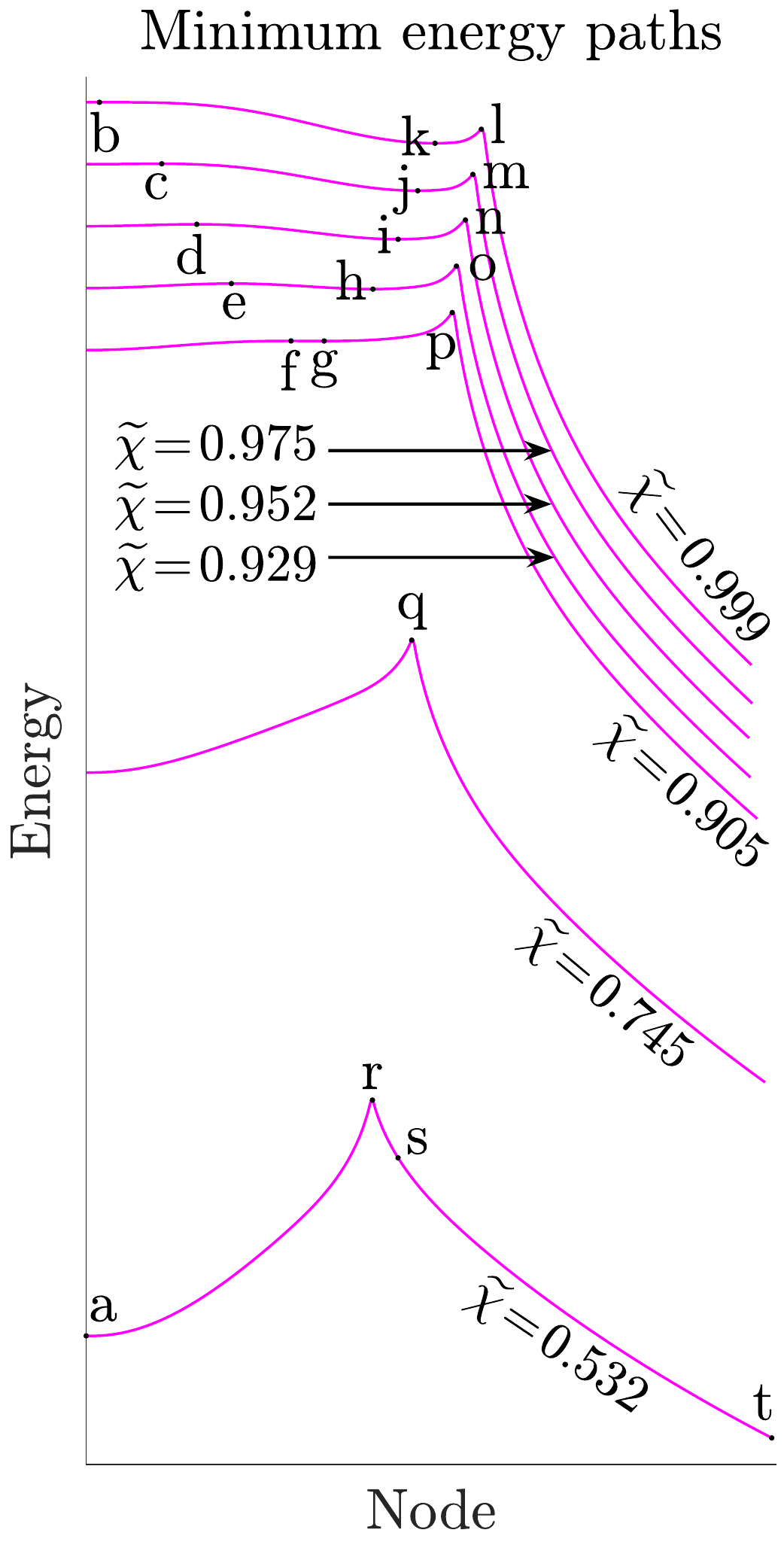}
\hspace{-2pt}
\begin{minipage}[b][293pt][t]{290pt}

$\overset{\includegraphics[width=53pt]{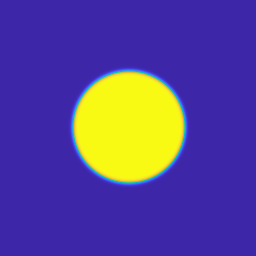}}{\begin{minipage}[t][\height][b]{24pt}\hspace{10pt}a\end{minipage}}$
$\overset{\includegraphics[width=53pt]{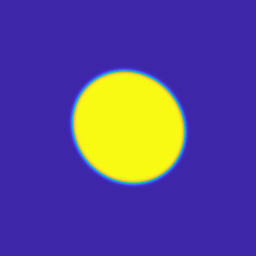}}{\begin{minipage}[t][\height][b]{24pt}\hspace{10pt}b\end{minipage}}$
$\overset{\includegraphics[width=53pt]{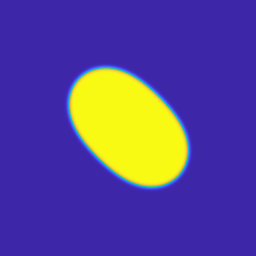}}{\begin{minipage}[t][\height][b]{24pt}\hspace{10pt}c\end{minipage}}$
$\overset{\includegraphics[width=53pt]{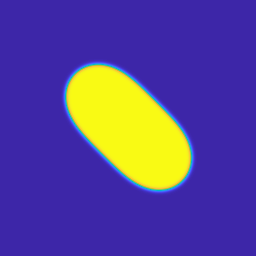}}{\begin{minipage}[t][\height][b]{24pt}\hspace{10pt}d\end{minipage}}$
$\overset{\includegraphics[width=53pt]{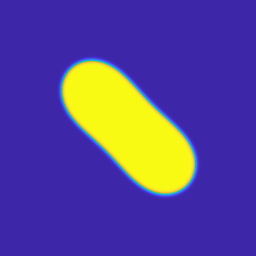}}{\begin{minipage}[t][\height][b]{24pt}\hspace{10pt}e\end{minipage}}$

\vspace{8pt}

$\overset{\includegraphics[width=53pt]{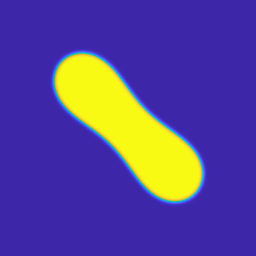}}{\begin{minipage}[t][\height][b]{24pt}\hspace{10pt}f\end{minipage}}$
$\overset{\includegraphics[width=53pt]{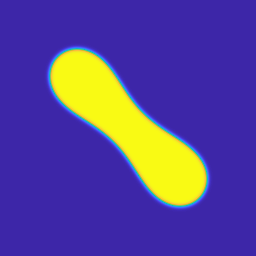}}{\begin{minipage}[t][\height][b]{24pt}\hspace{10pt}g\end{minipage}}$
$\overset{\includegraphics[width=53pt]{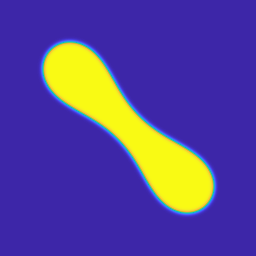}}{\begin{minipage}[t][\height][b]{24pt}\hspace{10pt}h\end{minipage}}$
$\overset{\includegraphics[width=53pt]{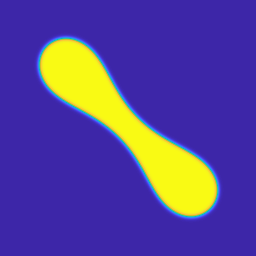}}{\begin{minipage}[t][\height][b]{24pt}\hspace{10pt}i\end{minipage}}$
$\overset{\includegraphics[width=53pt]{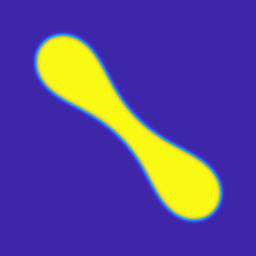}}{\begin{minipage}[t][\height][b]{24pt}\hspace{10pt}j\end{minipage}}$

\vspace{8pt}

$\overset{\includegraphics[width=53pt]{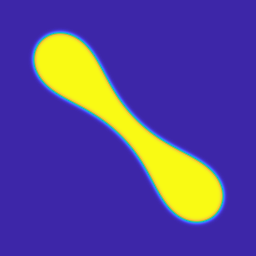}}{\begin{minipage}[t][\height][b]{24pt}\hspace{10pt}k\end{minipage}}$
$\overset{\includegraphics[width=53pt]{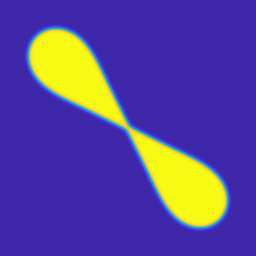}}{\begin{minipage}[t][\height][b]{24pt}\hspace{10pt}l\end{minipage}}$
$\overset{\includegraphics[width=53pt]{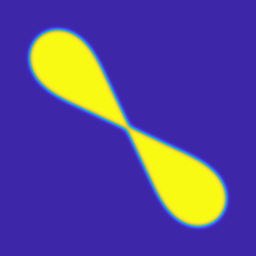}}{\begin{minipage}[t][\height][b]{24pt}\hspace{10pt}m\end{minipage}}$
$\overset{\includegraphics[width=53pt]{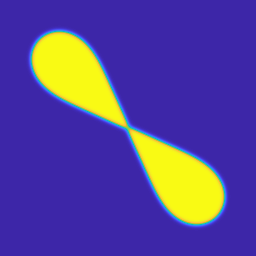}}{\begin{minipage}[t][\height][b]{24pt}\hspace{10pt}n\end{minipage}}$
$\overset{\includegraphics[width=53pt]{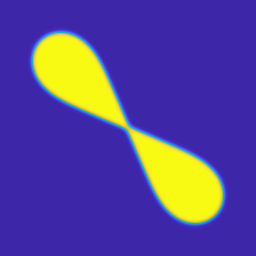}}{\begin{minipage}[t][\height][b]{24pt}\hspace{10pt}o\end{minipage}}$

\vspace{8pt}

$\overset{\includegraphics[width=53pt]{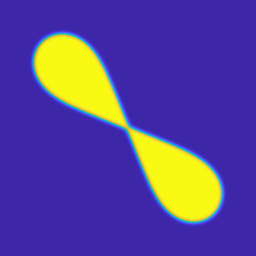}}{\begin{minipage}[t][\height][b]{24pt}\hspace{10pt}p\end{minipage}}$
$\overset{\includegraphics[width=53pt]{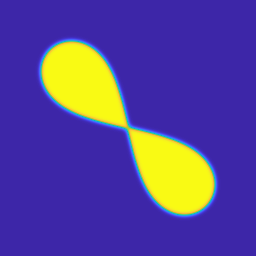}}{\begin{minipage}[t][\height][b]{24pt}\hspace{10pt}q\end{minipage}}$
$\overset{\includegraphics[width=53pt]{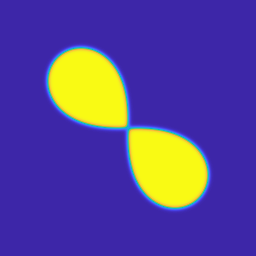}}{\begin{minipage}[t][\height][b]{24pt}\hspace{10pt}r\end{minipage}}$
$\overset{\includegraphics[width=53pt]{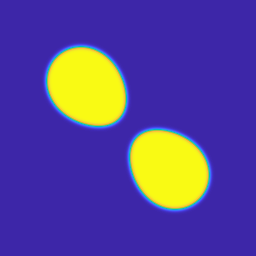}}{\begin{minipage}[t][\height][b]{24pt}\hspace{10pt}s\end{minipage}}$
$\overset{\includegraphics[width=53pt]{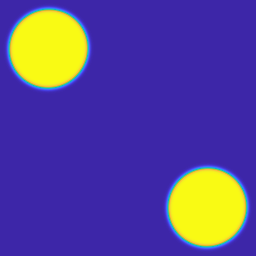}}{\begin{minipage}[t][\height][b]{24pt}\hspace{10pt}t\end{minipage}}$

\end{minipage}
\caption{Minimum energy paths in 2-D for various fissility parameters.}
\label{Minimum energy paths in 2-D various fissility parameters.}
\end{figure}

\subsection{Bifurcation diagrams}
\label{explainations3D}

\subsubsection*{The 3-D case}

Our observation of the instability of the Bohr\textendash Wheeler branch might seem surprising at first glance. The bifurcation of this branch from a ball at $\fis=\fis_*$ is transcritical, so we may think that its stability simply exchanges with the ball at $\fis=\fis_*$. In the paradigm of transcritical bifurcation, the ball is stable for $\fis<\fis_*$ and unstable for $\fis>\fis_*$, and it seems plausible that the oblate-like stationary point is stable for $\fis>\fis_*$, as shown on the left of Figure \ref{Transcritical} (see also Figure \ref{Spheroids}), where the horizontal axis qualitatively represents the aspect ratio of a spheroid, i.e., the ratio of the polar radius to the equatorial radius.

If we only consider axisymmetric deformations, the above paradigm of transcritical bifurcation is qualitatively correct, because the oblate-like stationary point is indeed stable against certain axisymmetric perturbations for $\fis>\fis_*$ \cite[Corollary 5]{frank2019non}. However, they are unstable against non-axisymmetric perturbations according to Section \ref{Symmetry breaking of oblate-like equilibria}. To explain this phenomenon, we need to use a 2-D transcritical bifurcation diagram.

We consider an ellipsoid which is close to a ball. It has three degrees of freedom, i.e., the lengths of its three axes. Due to the volume constraint, the degrees of freedom is reduced to two, and can be represented by the $xy$-plane on the right of Figure \ref{Transcritical} (where the $z$-axis represents the energy). The red dot at the center represents the ball, the three yellow dots represent three prolates whose axes of revolution are perpendicular to each other, the three blue dots represent oblates, and they are all stationary points. As $\fis$ exceeds $\fis_*$, the three yellow dots simultaneously collide with the red dot and then convert into three blue dots. Each blue dot is stable along the direction of its motion (corresponding to axisymmetric deformations), but unstable along the perpendicular direction of its motion (corresponding to non-axisymmetric deformations).
\begin{figure}[H]
\centering
\includegraphics[height=130pt]{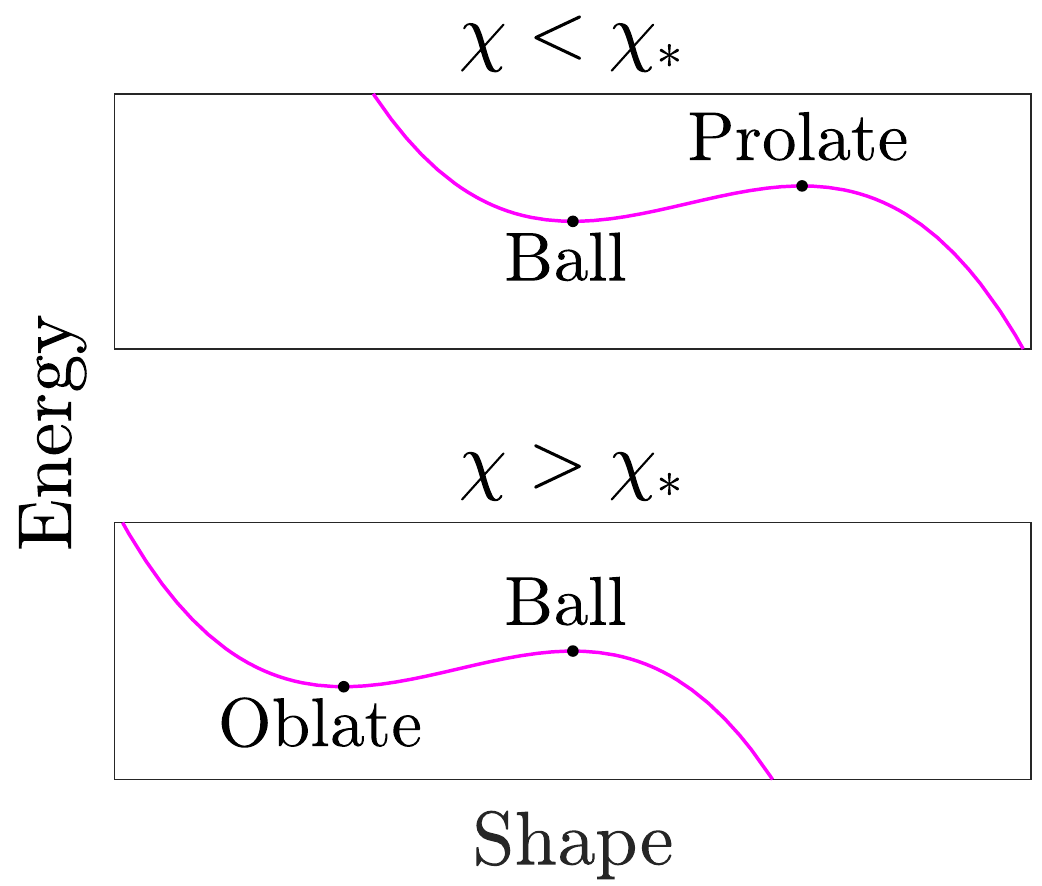}
\includegraphics[height=117pt]{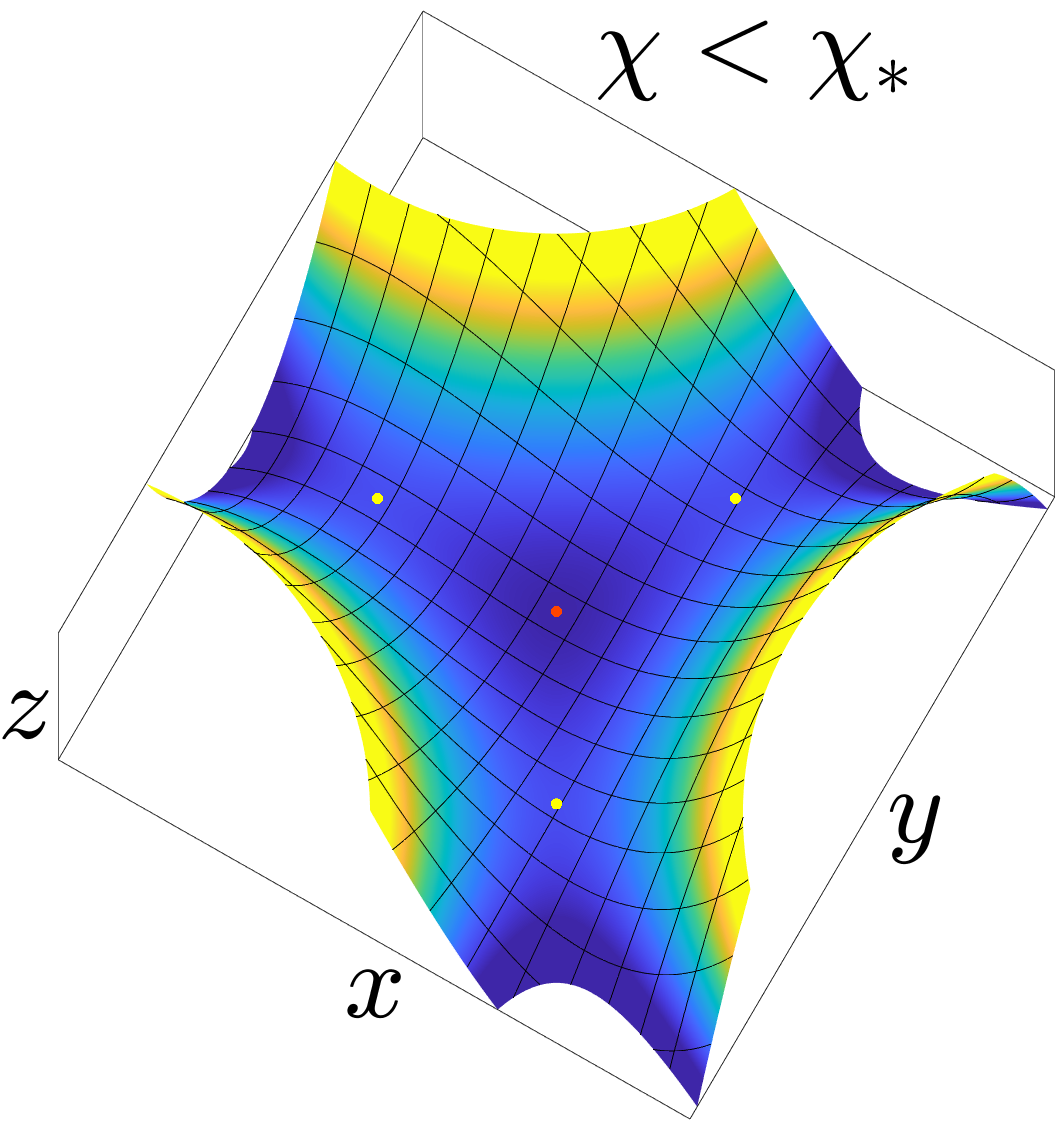}
\includegraphics[height=117pt]{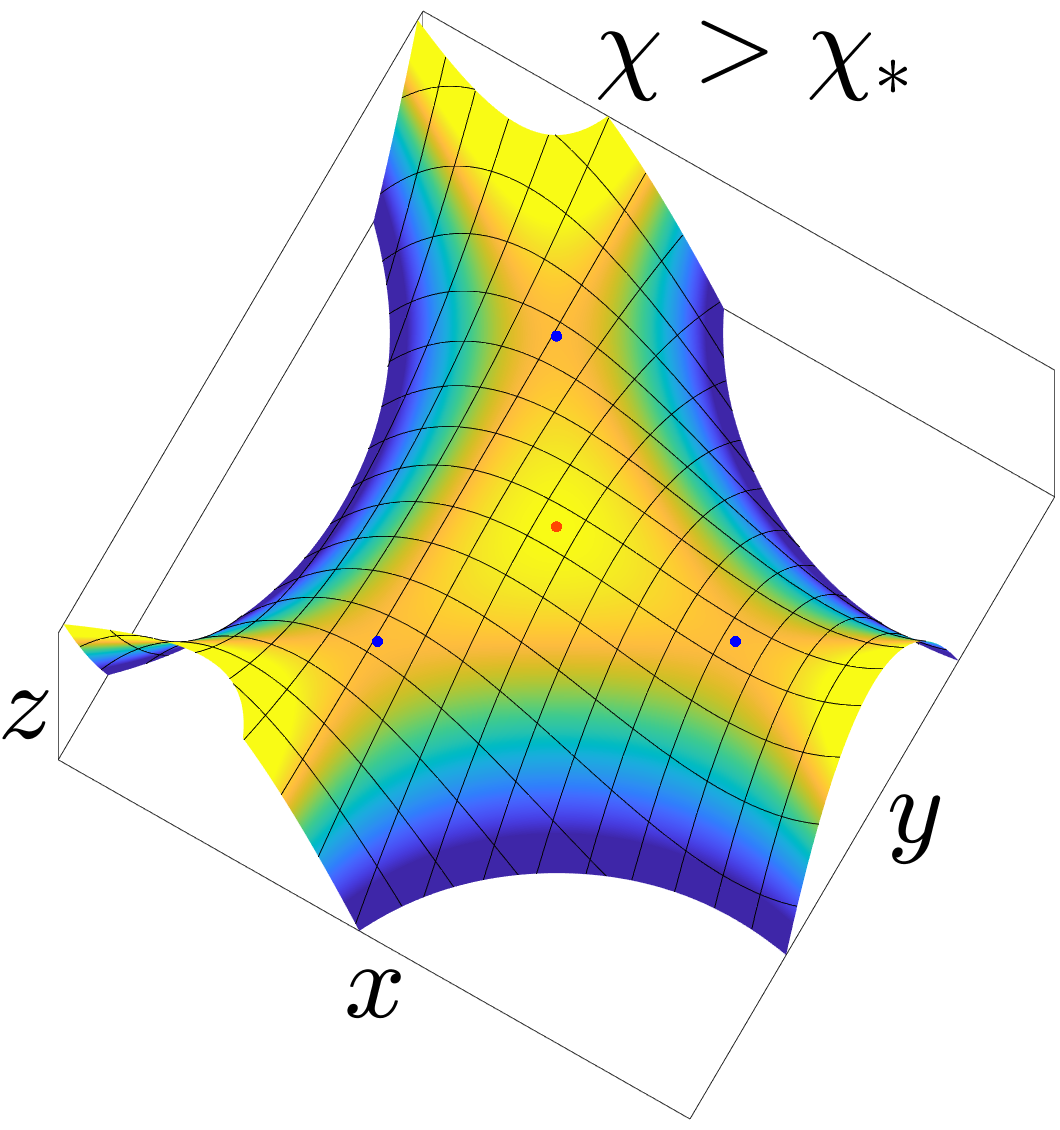}
\caption{Left: typical transcritical bifurcation in 1-D. Right: a transcritical bifurcation in 2-D.}
\label{Transcritical}
\end{figure}
Note that Figure \ref{Transcritical} is not produced from actual computational results but just for illustration. The left of Figure \ref{Transcritical} is produced using a prototype transcritical bifurcation: $y=sx^2/2-x^3/3$, where $s=1$ for $\fis<\fis_*$, and $s=-1$ for $\fis>\fis_*$. The right of Figure \ref{Transcritical} is produced by constructing three transcritical bifurcations along different directions colliding simultaneously, i.e., $z=sr^2/2+\cos(3\theta)\,r^3/3$, where $s = 0.85$ for $\fis<\fis_*$, $s = -0.85$ for $\fis>\fis_*$, and $(r,\theta)$ is $(x,y)$ in polar coordinates. It has been noticed in the literature \cite[Caption of Figure 39]{cohen1962deformation} that the instability of the oblate-like stationary points may be attributed to identically deformed but differently oriented bifurcation branches colliding simultaneously. The right of our Figure \ref{Transcritical} provides a concrete example for visual illustration.

\subsubsection*{The 2-D case}

The situation in 2-D is a little different, because a 3-D prolate and oblate have no analogues in 2-D (they both correspond to ellipses which cannot be distinguished like a prolate and oblate). As a result, the bifurcation in 2-D at $\fis=\fis_*$ is not transcritical but of a pitchfork type (in a generalized sense, see \cite[Page 1127]{ren2009oval}). It can be either supercritical or subcritical. Ren and Wei proposed the following conjecture \cite[Observation 1.1, Hypothesis 3.2, and Section 6]{ren2009oval}:
\begin{conjecture}
\label{supercritical conjecture}
The bifurcation at $\fis=\fis_*$ is of the supercritical pitchfork type. The equilibrium on this bifurcation branch (what we call the Ren\textendash Wei branch) given by \cite[Equation (3.7)]{ren2009oval}, only exists for $\fis>\fis_*$, and is stable.
\end{conjecture}
If the above conjecture is true, a disk, after losing its stability (for $\fis$ slightly larger than $\fis_*$), would deform into a local minimizer resembling an ellipse, whose eccentricity increases with $\fis$. This scenario can be qualitatively illustrated by Figure \ref{Saturation process and bifurcation}-(b), where the points $A$ and $B$ represent ellipse-like local minimizers, $O$ represents the disk, and $s$ qualitatively represents $\fis-\fis_*$. Therefore Ren and Wei suggested a saturation process depicted in Figure \ref{Saturation process and bifurcation}-(a), where the word saturation refers to the process of a disk splitting into two disks as $\fis$ increases. They believed that the local minimizer will no longer be a disk but resemble an ellipse after $\fis$ exceeds $\fis_*$, then a neck will appear for some larger $\fis$, and finally it will break into two disconnected disks for $\fis$ large enough \cite[Page 1120]{ren2009oval}.

However, for $\fis>\fis_*$ we have not numerically found a stable equilibrium resembling an ellipse. In fact, as mentioned in Section \ref{Analogue in 2-D}, a disk after becoming unstable will eventually deform into a local minimizer resembling an eye mask, instead of settling into an ellipse-like equilibrium. For $\fis<\fis_*$, we are able to numerically find unstable equilibria resembling ellipses (represented by b and c in Figure \ref{Minimum energy paths in 2-D various fissility parameters.}). Therefore our numerical results suggest the scenario illustrated by Figure \ref{Saturation process and bifurcation}-(c), where the points $B$, $D$ and $E$ represent eye-mask shaped local minimizers, $A$ and $C$ represent ellipse-like saddle points, $O$ represents the disk, and $s$ qualitatively represents $\fis_*-\fis$. In other words, we believe that Conjecture \ref{supercritical conjecture} is false, and we propose the following conjecture:
\begin{conjecture}
\label{subcritical conjecture}
The bifurcation at $\fis=\fis_*$ is of the subcritical pitchfork type. The equilibrium on the Ren\textendash Wei branch only exists for $\fis<\fis_*$, and is unstable.
\end{conjecture}
\begin{figure}[H]
\centering
\raisebox{45pt}{(a)} \includegraphics[width=300pt]{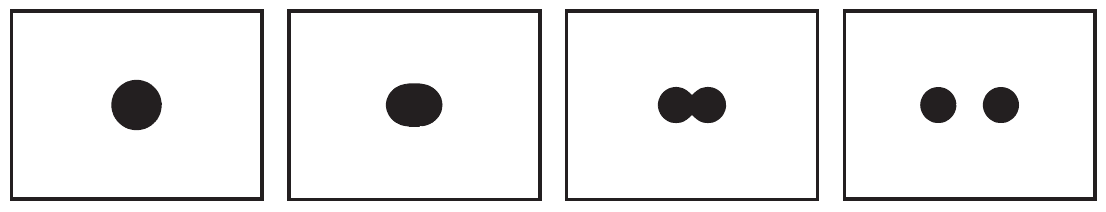}

\includegraphics[width=380pt]{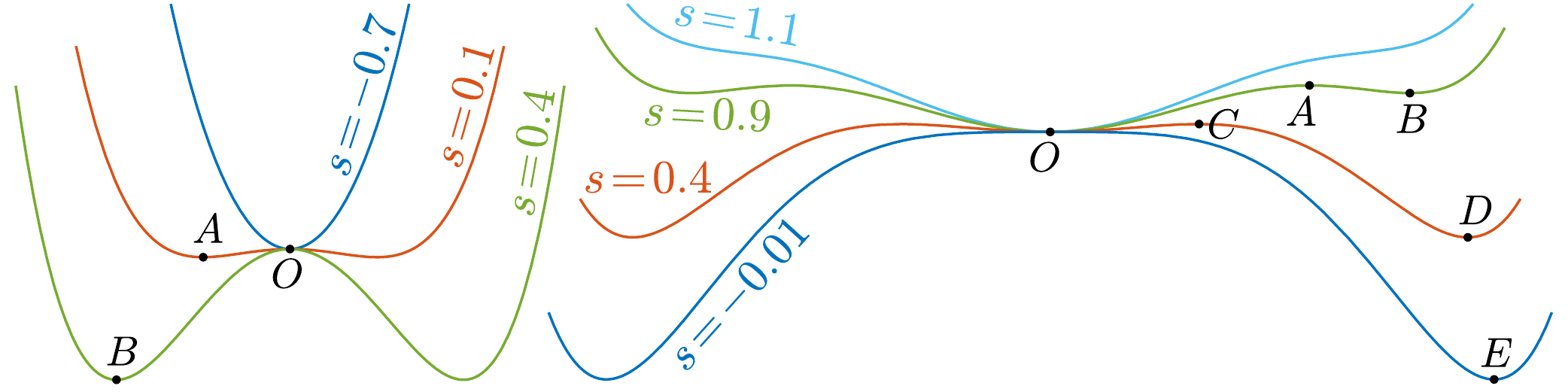}\\[-20pt]

\hspace{-57pt}(b)\hspace{170pt}(c)

\caption[Saturation process and bifurcation caption]{(a) Saturation process in the supercritical scenario (Figure from \cite[Figure 1]{ren2009oval}. Copyright \copyright2009 Society for Industrial and Applied Mathematics. Reprinted with permission. All rights reserved). (b) Supercritical scenario. (c) Subcritical scenario.}
\label{Saturation process and bifurcation}
\end{figure}
Figure \ref{Saturation process and bifurcation}-(b) is produced using a prototype supercritical pitchfork bifurcation, i.e., $y=x^4-2sx^2$. Figure \ref{Saturation process and bifurcation}-(c) is produced using $y = x^6/3 - x^4 + sx^2$. Neither of them is produced from actual computational results, and they are both for illustration only.

\section{Periodic isoperimetric problem: resistance to rupture}

\label{Resistance to rupture in the classical isoperimetric problem}

As noted in the literature \cite[Page 6]{krappe2012theory}, the underlying reason for the rupture of one nucleus into two nuclei is the Plateau\textendash Rayleigh instability, which explains the rupture of a long thin liquid cylinder into many droplets due to surface tension. This is why a thin fiber of water is rarely observed, while a thin film of water can be observed sometimes (e.g., on a bubble). According to Appendix \ref{Stability of a strip or lamella against rupture}, a lamella is always a stable equilibrium.

From the results in Section \ref{Simulation results}, we notice a difference in the fission process between 2-D and 3-D. In 2-D, the rupture requires energy inputs, because the scission point is always a transition state. In 3-D, however, the scission point is never a transition state, and the rupture is energetically spontaneous (following the energy descent direction).

To explain such a difference from the viewpoint of the Plateau\textendash Rayleigh instability, we demonstrate the roles that surface tension and dimensionality play, and focus on the cases where only the surface energy is present. In all simulations throughout this section, we choose $\tdfis=0$ and numerically calculate the minimum energy paths involving topological changes. The problem under study in this section is the periodic isoperimetric problem and its diffuse interface version, which have been considered in \cite[Equations (3.9) and (3.10)]{choksi2006periodic}.

\subsection{The 2-D case}

With a fixed volume fraction $\omega = (4\pi\!-\!3\sqrt3)/18\approx0.4095$, we use the string method to compute a minimum energy path between two local minimizers: from a strip to a disk. Note that our 2-D simulation can be regarded as a cross section of a 3-D simulation (by a trivial extension along the third dimension), corresponding to a path from a lamella to a cylinder. For disambiguation, we use the convention that a lamella is 3-D while a strip is 2-D.

As shown in Figure \ref{Minimum energy path from a strip to a disk.}, along the minimum energy path, the scission point c is the transition state. It resembles two arcs intersecting at a single point. Intuitively, the shape of the scission point is away from a disk, so the surface tension is trying to induce a rupture into left and right halves; however, the connection of the two halves at a single point is strong enough to resist the tendency to rupture. Therefore, the shape stays in a stalemate. This phenomenon is similar to what we have seen from the minimum energy paths presented in Section \ref{Analogue in 2-D} where $\tdfis>0$.
\begin{figure}[H]
\centering
\includegraphics[width=255.102pt]{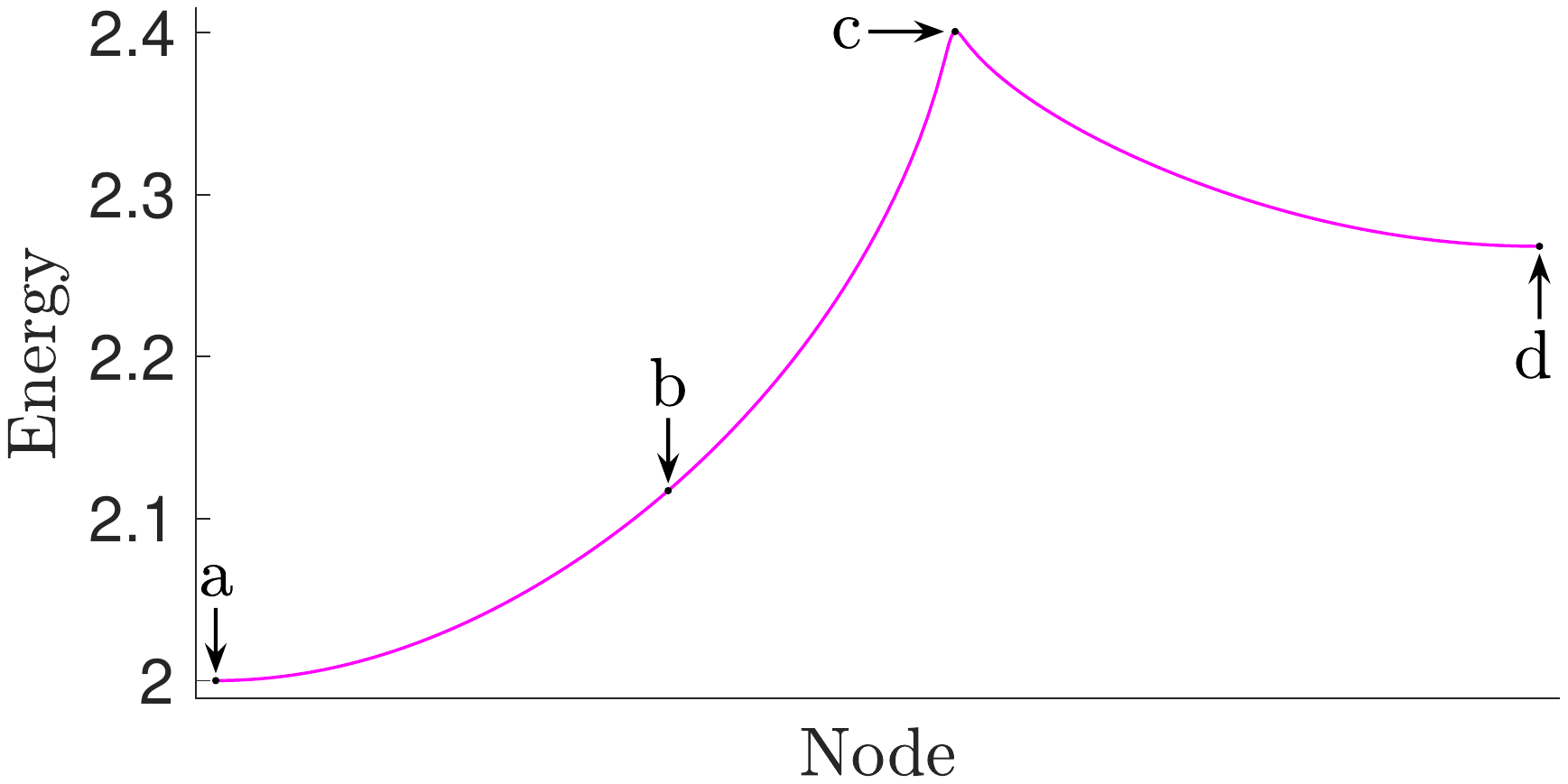}
\hspace{2pt}
\begin{minipage}[b][138pt][t]{130pt}
$\overset{\includegraphics[width=53pt]{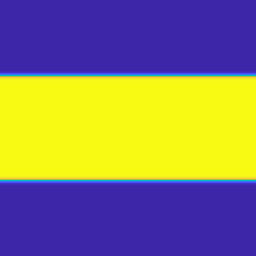}}{\begin{minipage}[t][\height][b]{24pt}\hspace{10pt}a\end{minipage}}$
\hspace{7pt}
$\overset{\includegraphics[width=53pt]{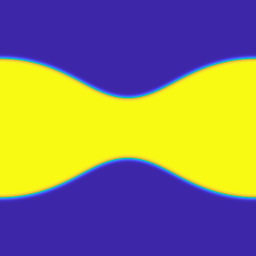}}{\begin{minipage}[t][\height][b]{24pt}\hspace{10pt}b\end{minipage}}$

\vspace{6pt}

$\overset{\includegraphics[width=53pt]{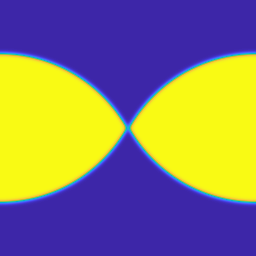}}{\begin{minipage}[t][\height][b]{24pt}\hspace{10pt}c\end{minipage}}$
\hspace{7pt}
$\overset{\includegraphics[width=53pt]{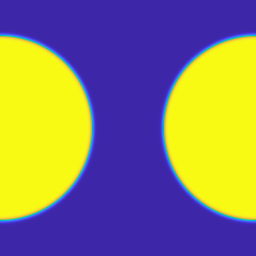}}{\begin{minipage}[t][\height][b]{24pt}\hspace{10pt}d\end{minipage}}$

\end{minipage}
\caption{Minimum energy path from a strip to a disk.}
\label{Minimum energy path from a strip to a disk.}
\end{figure}
In the sharp interface limit, the minimum energy path in Figure \ref{Minimum energy path from a strip to a disk.} should follow the trajectory of a volume preserving mean curvature flow (with the arc length corresponding to a reparametrization of the time variable). The sharp interface counterpart of the transition state c, if satisfying the Euler\textendash Lagrange equation \eqref{stationary point sharp interface} almost everywhere, must have circular boundaries. Therefore we propose the following conjecture:
\begin{conjecture}
\label{symmetric lens conjecture}
With a fixed volume fraction $\omega$, as $c\rightarrow0$ and $K\rightarrow\infty$, the transition state (represented by c in Figure \ref{Minimum energy path from a strip to a disk.}) from a strip to a disk, converges in $L^2$ to the indicator function of a symmetric lens, which is formed by two intersecting circles of equal radii $R$.
\end{conjecture}
After a simple calculation, we can obtain the precise shape of the transition state.
\begin{proposition}
\label{vesica-piscis}
If Conjecture \ref{symmetric lens conjecture} is true, as $\omega$ increases from 0 to $\pi/4$, the radius $R$ decreases from $\infty$ to $1/2$, as determined by $\omega=2R^2\arcsin\big(1/(2R)\big)-\sqrt{R^2\!-\!1/4}$. In particular, when $\omega = (4\pi\!-\!3\sqrt3)/18\approx0.4095$, we have $R=1/\sqrt3$, and thus the transition state is of the shape of the vesica piscis.
\end{proposition}

\subsection{The 3-D case}

For $\omega=0.09$, $0.15$ and $0.2$, we use the string method to compute a minimum energy path between two local minimizers: from a cylinder to a ball, as shown in Figures \ref{V=0.09}, \ref{V=0.15} and \ref{V=0.2}, respectively. We can see that the scission point, represented by c, d and d in the three Figures, respectively, is never a transition state. Intuitively, in 3-D, the connection of the top and bottom halves at a single point is not strong enough to resist the tendency to rupture.
\begin{figure}[H]
\centering
\raisebox{0pt}{\includegraphics[width=289.116pt]{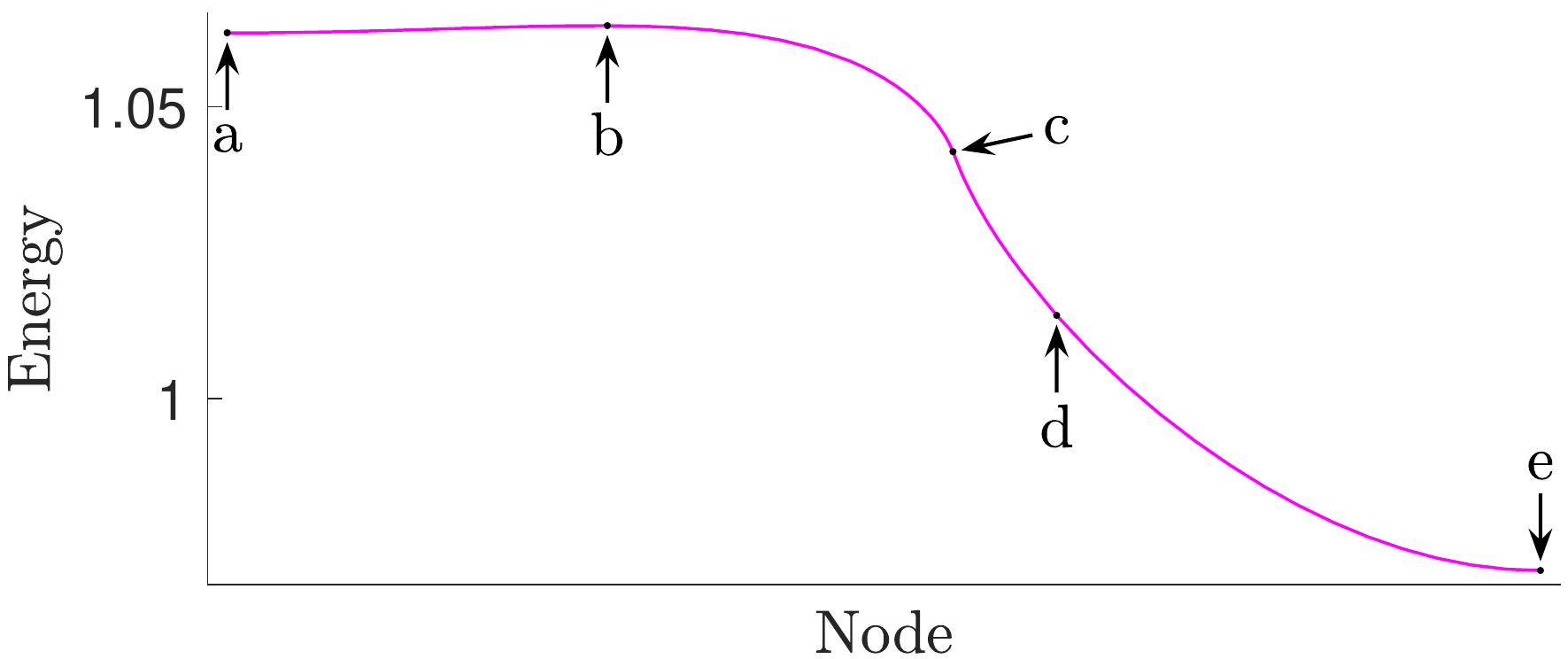}}
\hspace{-32pt}
\begin{minipage}[b][140pt][t]{190pt}
$\overset{\includegraphics[width=58pt]{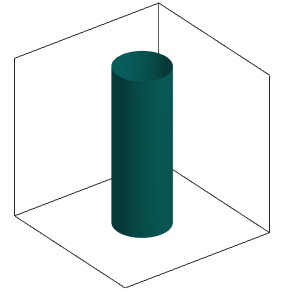}}{\begin{minipage}[t][\height][b]{24pt}\hspace{7pt}a\end{minipage}}$
$\overset{\includegraphics[width=58pt]{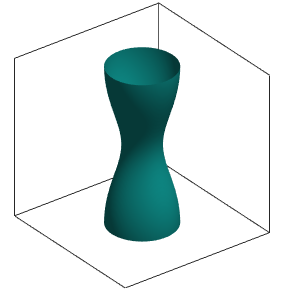}}{\begin{minipage}[t][\height][b]{24pt}\hspace{7pt}b\end{minipage}}$
$\overset{\includegraphics[width=58pt]{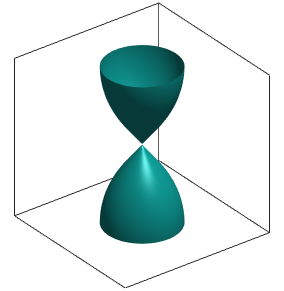}}{\begin{minipage}[t][\height][b]{24pt}\hspace{7pt}c\end{minipage}}$

\hspace{30pt}$\overset{\includegraphics[width=58pt]{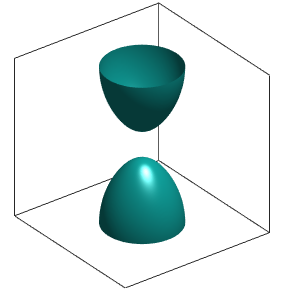}}{\begin{minipage}[t][\height][b]{24pt}\hspace{7pt}d\end{minipage}}$
$\overset{\includegraphics[width=58pt]{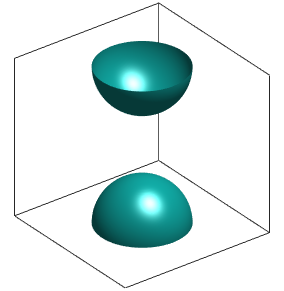}}{\begin{minipage}[t][\height][b]{24pt}\hspace{7pt}e\end{minipage}}$
\end{minipage}
\vspace{-2pt}
\caption{Minimum energy path from a cylinder to a ball for $\omega=0.09$}
\label{V=0.09}
\end{figure}

\begin{figure}[H]
\centering
\raisebox{28pt}{\includegraphics[width=289.116pt]{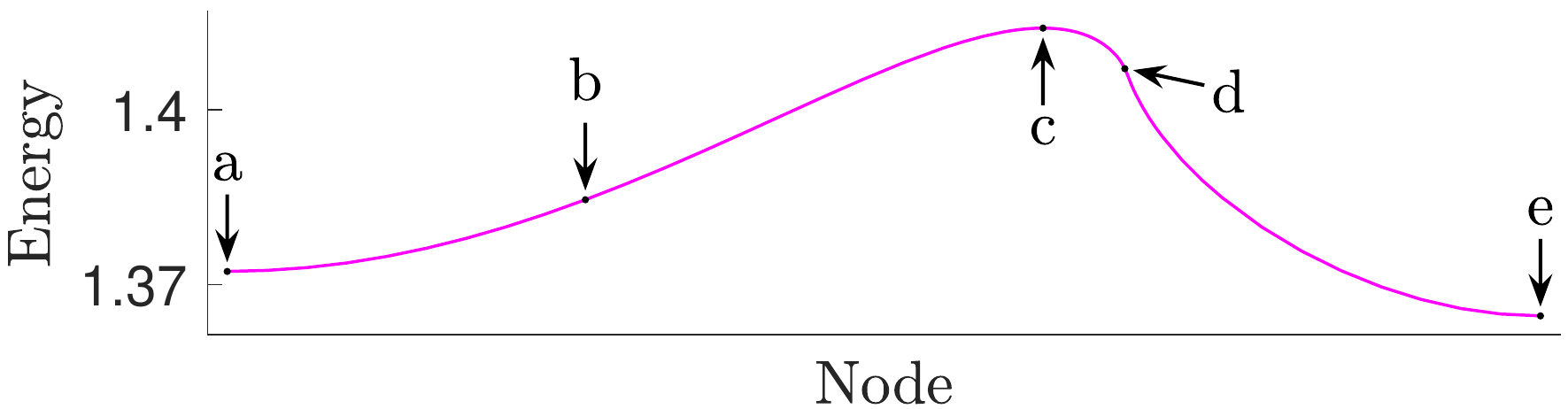}}
\hspace{-32pt}
\begin{minipage}[b][150pt][t]{190pt}
$\overset{\includegraphics[width=58pt]{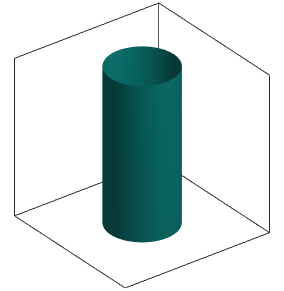}}{\begin{minipage}[t][\height][b]{24pt}\hspace{7pt}a\end{minipage}}$
$\overset{\includegraphics[width=58pt]{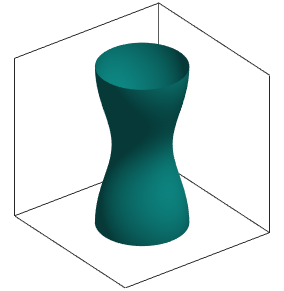}}{\begin{minipage}[t][\height][b]{24pt}\hspace{7pt}b\end{minipage}}$
$\overset{\includegraphics[width=58pt]{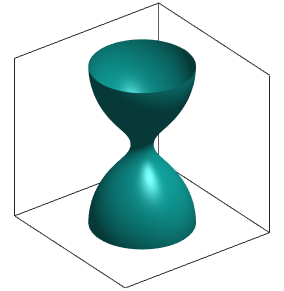}}{\begin{minipage}[t][\height][b]{24pt}\hspace{7pt}c\end{minipage}}$

\hspace{30pt}$\overset{\includegraphics[width=58pt]{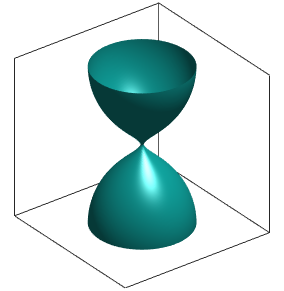}}{\begin{minipage}[t][\height][b]{24pt}\hspace{7pt}d\end{minipage}}$
$\overset{\includegraphics[width=58pt]{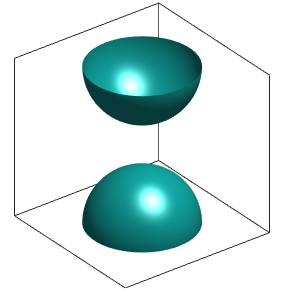}}{\begin{minipage}[t][\height][b]{24pt}\hspace{7pt}e\end{minipage}}$
\end{minipage}
\vspace{-11pt}
\caption{Minimum energy path from a cylinder to a ball for $\omega=0.15$}
\label{V=0.15}
\end{figure}

\begin{figure}[H]
\centering
\raisebox{0pt}{\includegraphics[width=289.116pt]{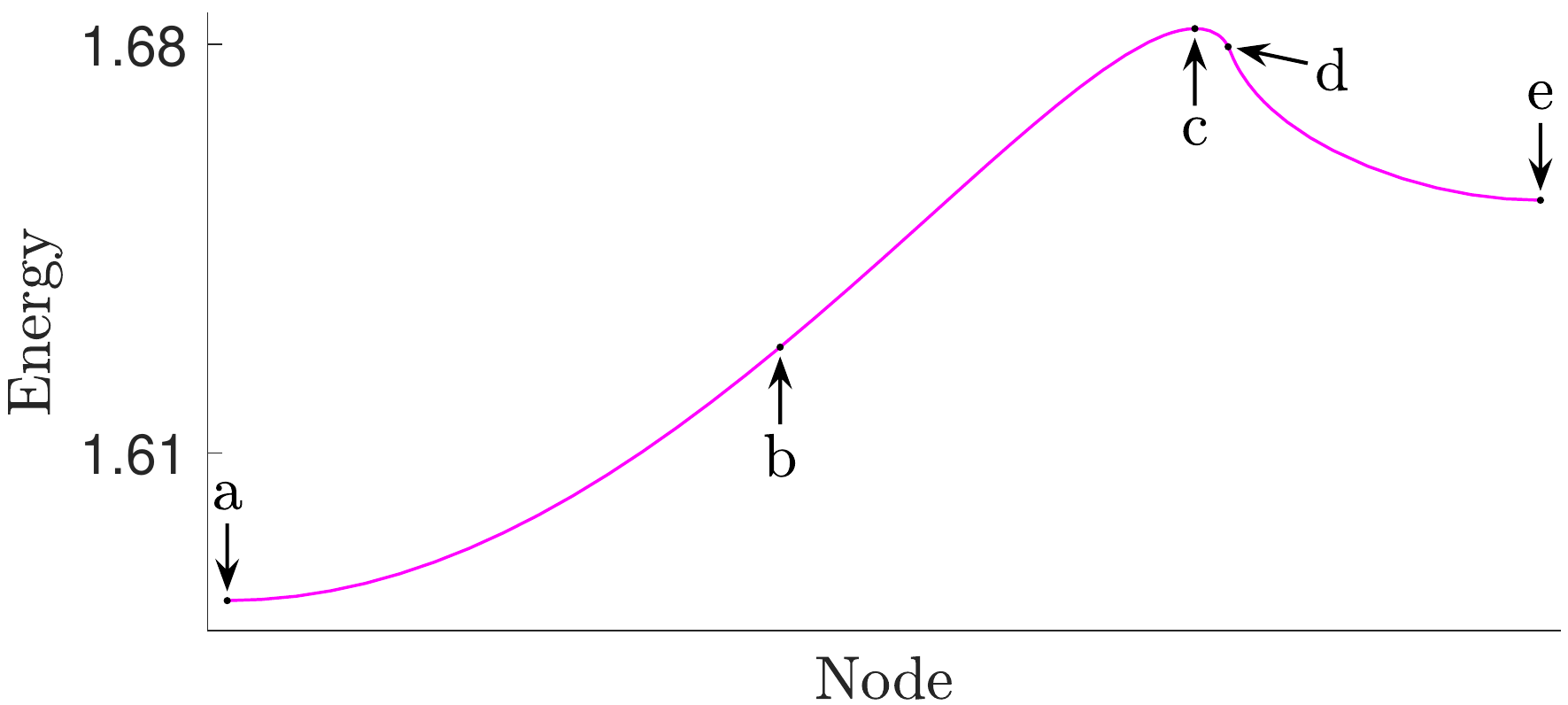}}
\hspace{-32pt}
\begin{minipage}[b][140pt][t]{190pt}
\hspace{30pt}$\overset{\includegraphics[width=58pt]{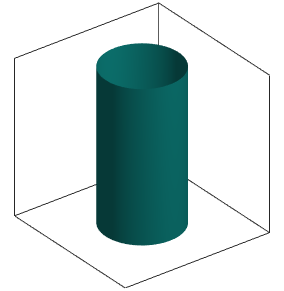}}{\begin{minipage}[t][\height][b]{24pt}\hspace{7pt}a\end{minipage}}$
$\overset{\includegraphics[width=58pt]{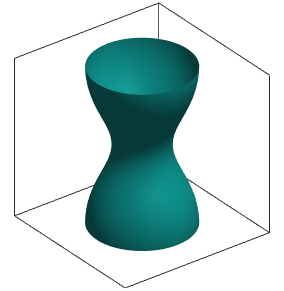}}{\begin{minipage}[t][\height][b]{24pt}\hspace{7pt}b\end{minipage}}$

$\overset{\includegraphics[width=58pt]{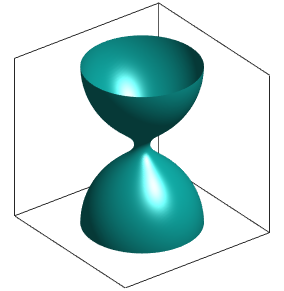}}{\begin{minipage}[t][\height][b]{24pt}\hspace{7pt}c\end{minipage}}$
$\overset{\includegraphics[width=58pt]{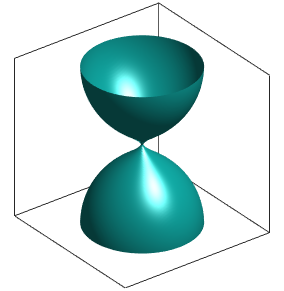}}{\begin{minipage}[t][\height][b]{24pt}\hspace{7pt}d\end{minipage}}$
$\overset{\includegraphics[width=58pt]{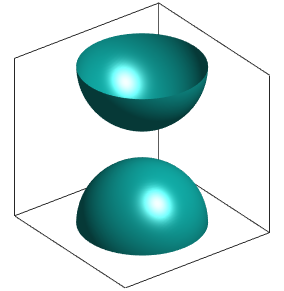}}{\begin{minipage}[t][\height][b]{24pt}\hspace{7pt}e\end{minipage}}$
\end{minipage}
\vspace{-2pt}
\caption{Minimum energy path from a cylinder to a ball for $\omega=0.2$}
\label{V=0.2}
\end{figure}
The transition state, represented by b, c and c in Figures \ref{V=0.09}, \ref{V=0.15} and \ref{V=0.2}, respectively, resembles the interior of an unduloid. Unduloids \cite{hadzhilazova2007unduloids} are a family of constant mean curvature surfaces of revolution, whose meridians are obtained by tracing a focus of an ellipse rolling without slippage along the axis of revolution. The eccentricity $e$ of an unduloid refers to the eccentricity of the rolling ellipse. As $e\rightarrow0$, the unduloid converges to a cylinder; as $e\rightarrow1$, the unduloid converges to a series of balls, with its neck circumference shrinking to 0, as shown in Figure \ref{unduloid picture}.
\begin{figure}[H]
\centering
\includegraphics[bb=0 0 942 398,scale=0.27]{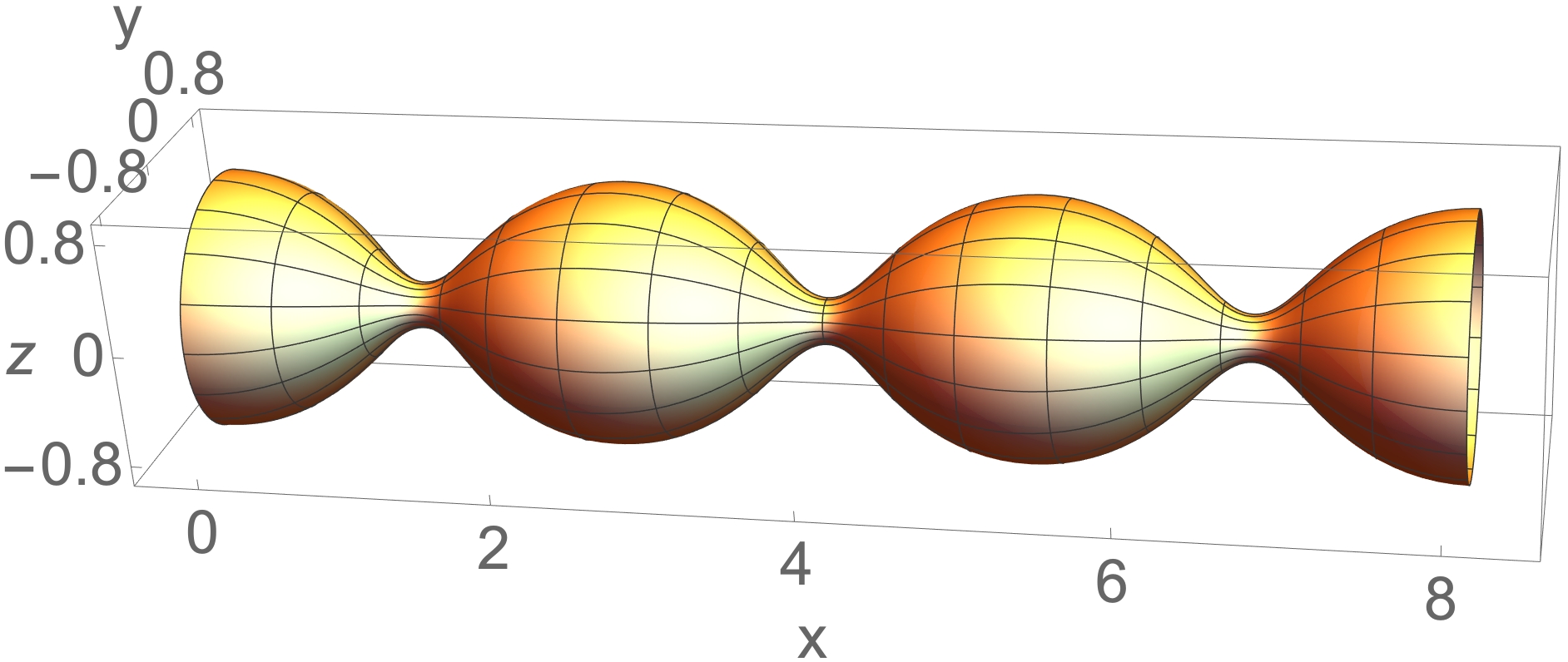}
\includegraphics[bb=0 0 730.5 389,scale=0.27]{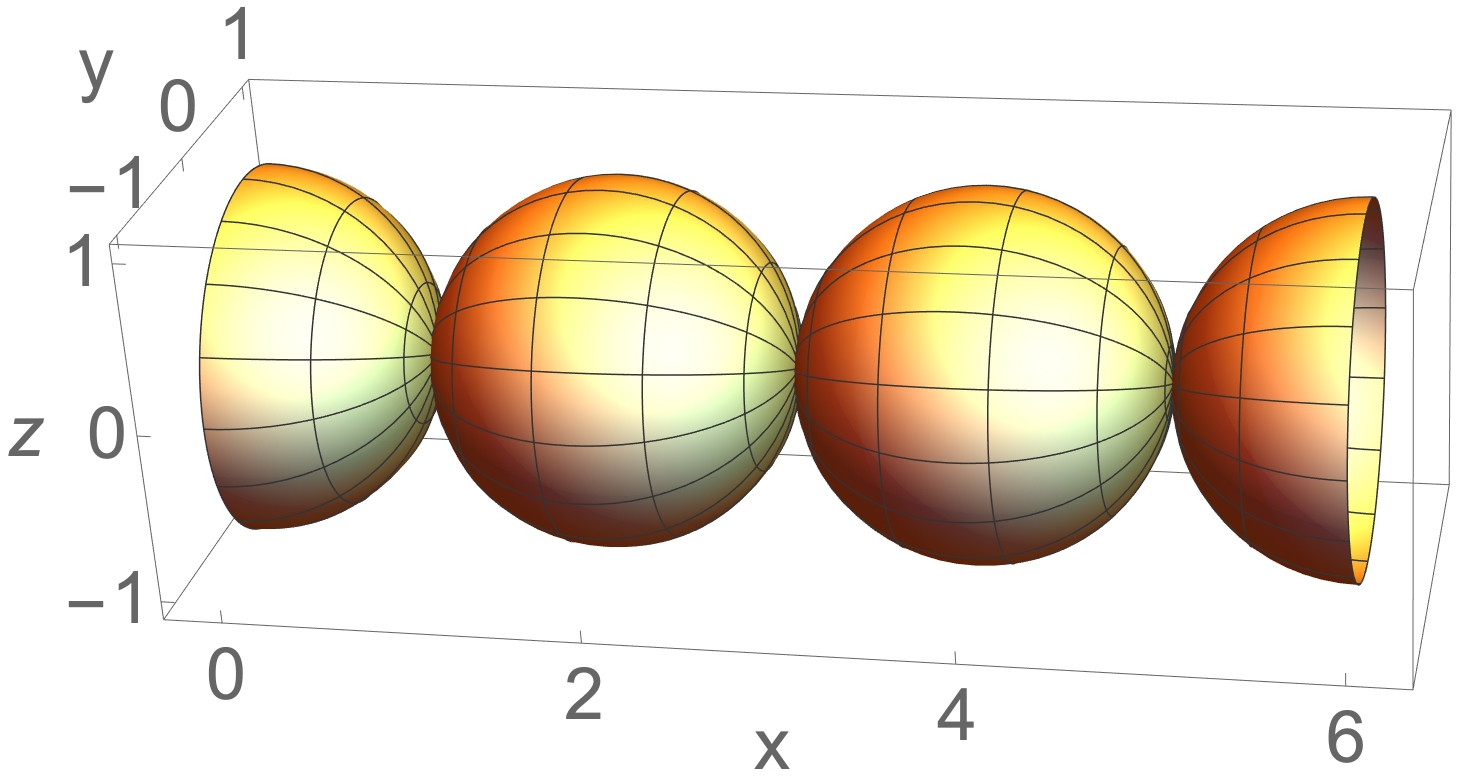}
\caption{Three periods of unduloids: $e=0.7$ (left) and $e=0.9999$ (right).}
\label{unduloid picture}
\end{figure}

\begin{conjecture}
\label{unduloid conjecture}
With a fixed volume fraction $\omega\in\big((4\pi)^{-1},\pi/6\big)$, as $c\rightarrow0$ and $K\rightarrow\infty$, the transition state (represented by b, c and c in Figures \ref{V=0.09}, \ref{V=0.15} and \ref{V=0.2}, respectively) from a cylinder to a ball, converges in $L^2$ to the indicator function of an unduloid. As $\omega$ increases from $(4\pi)^{-1}$ to $\pi/6$, the eccentricity $e$ increases from $0$ to $1$.
\end{conjecture}

\begin{remark}
When $\omega < (4\pi)^{-1}\approx0.08$, the unduloid disappears, through what we believe is a subcritical pitchfork bifurcation, and the cylinder becomes unstable (the Plateau\textendash Rayleigh instability). Note that a cylinder is stable if and only if its height is less than its circumference, see \cite[\S\,396\,--\,\S\,401]{plateau1873experimental} and \cite[Appendix I]{strutt1879vi}. According to the isoperimetric inequality, a disk and a ball are always stable. A strip and a lamella are also always stable according to Appendix \ref{Stability of a strip or lamella against rupture}.
\end{remark}

\section{Asymptotic analysis: two touching balls or disks}

\label{Two touching balls (or disks)}

As shown in Figure \ref{bifurcation branch}, Sections \ref{Minimum energy paths of fission} and \ref{Analogue in 2-D}, when $\tdfis$ is small, the transition state in 3-D seems to have a thin neck, while the transition state in 2-D seems to be non-smooth and form an angle in the middle. In Section \ref{Resistance to rupture in the classical isoperimetric problem}, we observe similar phenomena for $\tdfis=0$. To gain more insight into such a difference between 2-D and 3-D, in this section we consider the cases for $\fis\ll1$ using asymptotic analysis.

In 3-D, the following conjecture is stated in \cite[Middle of Page 071506-2]{frank2019non}.
\begin{conjecture}
\label{continuation in 3-D}
The Bohr\textendash Wheeler branch in 3-D continues to arbitrarily small $\fis$, and the saddle point on this branch converges to two touching balls in the limit of $\fis\rightarrow0$.
\end{conjecture}
As shown in Figure \ref{bifurcation branch}, our numerical results indicate that Conjecture \ref{continuation in 3-D} is correct. We perform asymptotic analysis and obtain the following result.
\begin{theorem}
\label{3D asymptotics}
If Conjecture \ref{continuation in 3-D} is true, the neck circumference of the saddle point on the Bohr\textendash Wheeler branch is $\sqrt[3]{3|\Omega|/(8\pi)}5\fis\pi/3+\sqrt[3]{|\Omega|}o(\fis)$ for $\fis\ll1$.
\end{theorem}

In 2-D, we propose the following conjectures based on our numerical results in Section \ref{Analogue in 2-D}.
\begin{conjecture}
The Ren\textendash Wei branch in 2-D does not continue to arbitrarily small $\fis$, and it disappears at $\fis=\fis_4\approx0.904\pm0.001$ through a saddle-node bifurcation. As $\fis$ decreases from $\fis_*$ to $\fis_4$, the equilibrium on this branch changes its shape from a disk to an ellipse-like shape to an eye mask.
\end{conjecture}
\begin{conjecture}
\label{continuation in 2-D}
There is a separate branch consisting of the scission points. Each scission point resembles a lemniscate and is a transition state for $\fis>\fis_3$, where $\fis_3\approx0.428$ is the 2-D analogue of the Businaro\textendash Gallone point. For $\fis<\fis_3$, the scission point becomes unstable against mass-asymmetric perturbations, but the branch continues to arbitrarily small $\fis$, and the scission point on this branch converges to two touching disks as $\fis\rightarrow0$.
\end{conjecture}

We perform asymptotic analysis and obtain the following result.
\begin{theorem}
\label{2D asymptotics}
If Conjecture \ref{continuation in 2-D} is true, the angle formed by the lemniscate-like scission point is $3\fis\pi/\sqrt8+o(\fis)$ for $\fis\ll1$.
\end{theorem}

\begin{proof} [of Theorem \ref{2D asymptotics}]
Consider the energy functional \eqref{rescaled energy-liquid-drop-model}. We study its critical point configuration that is close to two touching disks when $\gamma\ll1$. We expect such a critical point to satisfy the Euler\textendash Lagrange equation \eqref{stationary point sharp interface} almost everywhere. When $\gamma=0$, obviously \eqref{stationary point sharp interface} is satisfied almost everywhere by two perfect disks touching at a single point. For $0<\gamma\ll1$, we have the following intuition: the potential $\phi$ is higher at the center where the two disk-like components touch each other, and thus the curvature $H$ is smaller there, so an angle must be formed by the intersecting interfaces. We want to determine the angle asymptotically.

Fix $|\Omega|$ to be $2\pi$, then in the limit of $\gamma=0$, the radius of each disk is 1, as described by $\rho(\theta)=1$ for $\theta\in[0,2\pi)$ in polar coordinates (corresponding to the black circles in Figure \ref{Perturbation of two touching disks}). For $0<\gamma\ll1$, we expect the critical point configuration to be slightly perturbed, as described by $\rho(\theta)=1+\varepsilon p(\theta)$, with $0<\varepsilon\ll1$ and $p\in C^2[0,2\pi]$ (corresponding to the red curves in Figure \ref{Perturbation of two touching disks}). We assume that $p$ is symmetric, i.e., $p(\theta)=p(2\pi-\theta)$. We can also require $p(0)=0$. Note that $p$ describes the deformation of the left half, and that the right half goes through a symmetric deformation.
\begin{figure}[H]
\centering
\begin{tikzpicture}
\draw (-40*1.3pt,0*1.3pt) circle (40*1.3pt);
\draw (40*1.3pt,0*1.3pt) circle (40*1.3pt);
\draw [red] (-43*1.3pt,0*1.3pt) ellipse (43*1.3pt and 37.2*1.3pt);
\draw [red] (43*1.3pt,0*1.3pt) ellipse (43*1.3pt and 37.2*1.3pt);
\draw [-{Stealth[length=5*1.3pt]}] (-40*1.3pt,0*1.3pt) -- (-73*1.3pt,22.2*1.3pt);
\draw [red, -{Stealth[length=5*1.3pt]}] (-40*1.3pt,0*1.3pt) -- (-35*1.3pt,36.3*1.3pt);
\draw [dashed] (0*1.3pt,0*1.3pt) -- (-40*1.3pt,0*1.3pt);
\draw [red] (-36*1.3pt,0) arc (0:82:4*1.3pt);
\node [red] at (-34*1.3pt,5*1.3pt) {\fontsize{10.4pt}{0pt}\selectfont$\theta$};
\node[fill,circle,inner sep=0.8*1.3pt,label={below:\fontsize{10.4pt}{0pt}\selectfont$O$}] at (-40*1.3pt,0*1.3pt) {};
\node[fill,circle,inner sep=0.8*1.3pt,label={right:\fontsize{10.4pt}{0pt}\selectfont$T$}] at (0*1.3pt,0*1.3pt) {};
\node [red,rotate=82] at (-41.5*1.3pt,18*1.3pt) {\fontsize{10.4pt}{0pt}\selectfont$1\!\!+\!\varepsilon p(\theta)$};
\node [rotate=-35] at (-55*1.3pt,15*1.3pt) {\fontsize{10.4pt}{0pt}\selectfont$1$};
\end{tikzpicture}
\caption{Perturbation of two touching disks. Black: perfect disks. Red: perturbed disks.}
\label{Perturbation of two touching disks}
\end{figure}
Under mass-symmetric assumptions, the area of the left half is $\pi$, i.e.,
\begin{equation*}
\pi=\frac12\int_0^{2\pi}\hspace{-7pt}\rho^2(\theta)\dd{\theta}=\pi+\varepsilon\int_0^{2\pi}\hspace{-7pt}p(\theta)\dd{\theta}+o(\varepsilon),
\end{equation*}
therefore we require $\int_0^{2\pi}\hspace{-3pt}p(\theta)\dd{\theta}=0$.

The curvature $H$ is given by
\begin{equation*}
H=\frac{2\rho'^2+\rho^2-\rho\rho''}{(\rho'^2+\rho^2)^{3/2}}=1-(p+p'')\varepsilon+o(\varepsilon).
\end{equation*}

When $\varepsilon\ll1$, we expect $\Omega$ to be very close to two perfect disks. Hence, we expect
\begin{equation*}
\phi(\theta)=\frac{\ln(5\!-\!4\cos\theta)}{-4\pi}\pi+o(1),
\end{equation*}
where $\phi(\theta)$ denotes the potential $\phi$ evaluated at $\big(\rho(\theta),\theta\big)$ in polar coordinates.

According to \eqref{stationary point sharp interface}, we have
\begin{equation*}
1-(p+p'')\varepsilon+o(\varepsilon)+\gamma\frac{\ln(5\!-\!4\cos\theta)}{-4}+\gamma\,o(1)=\lambda,
\end{equation*}
therefore we know that $\varepsilon$ and $\gamma$ are of the same order, and without loss of generality, we assume $\varepsilon=\gamma$. Equating the first-order terms in the above equation, we obtain
\begin{equation}
\label{ode 2d}
p(\theta)+p''(\theta)+\frac{\ln(5\!-\!4\cos\theta)}{4}=\text{constant},
\end{equation}
where the constant is independent of $\theta$, and arises from the fact that $\lambda$ may be dependent on $\varepsilon$. Using WOLFRAM MATHEMATICA, we obtain a solution satisfying \eqref{ode 2d} as well as all the other conditions that we have imposed:
\begin{equation*}
16p(\theta)=
6\arccot\big(3\tan(\theta/2)\big)\sin\theta + (6\!+\!8\ln2)(1\!-\!\cos\theta) + 5(\theta\!-\!\pi)\sin\theta + (5\cos\theta\!-\!4)\ln(5\!-\!4\cos\theta),
\end{equation*}
with the constant in \eqref{ode 2d} adopting the value $(1\!+\!4\ln2)/8$. As shown in Figure \ref{2-D rho and mirror}, since $\partial_+p(0)=-\pi/8$ and $\partial_-p(2\pi)=\pi/8$, we know that the angle formed by the tangent lines is $\varepsilon\pi/4+o(\varepsilon)$. By \eqref{rescaled energy-liquid-drop-model} we have $\varepsilon=\gamma=12\fis(2\pi/\pi)^{-3/2}$, so the angle is $3\fis\pi/\sqrt8$ asymptotically.
\begin{figure}[H]
\centering
\includegraphics[width=234pt]{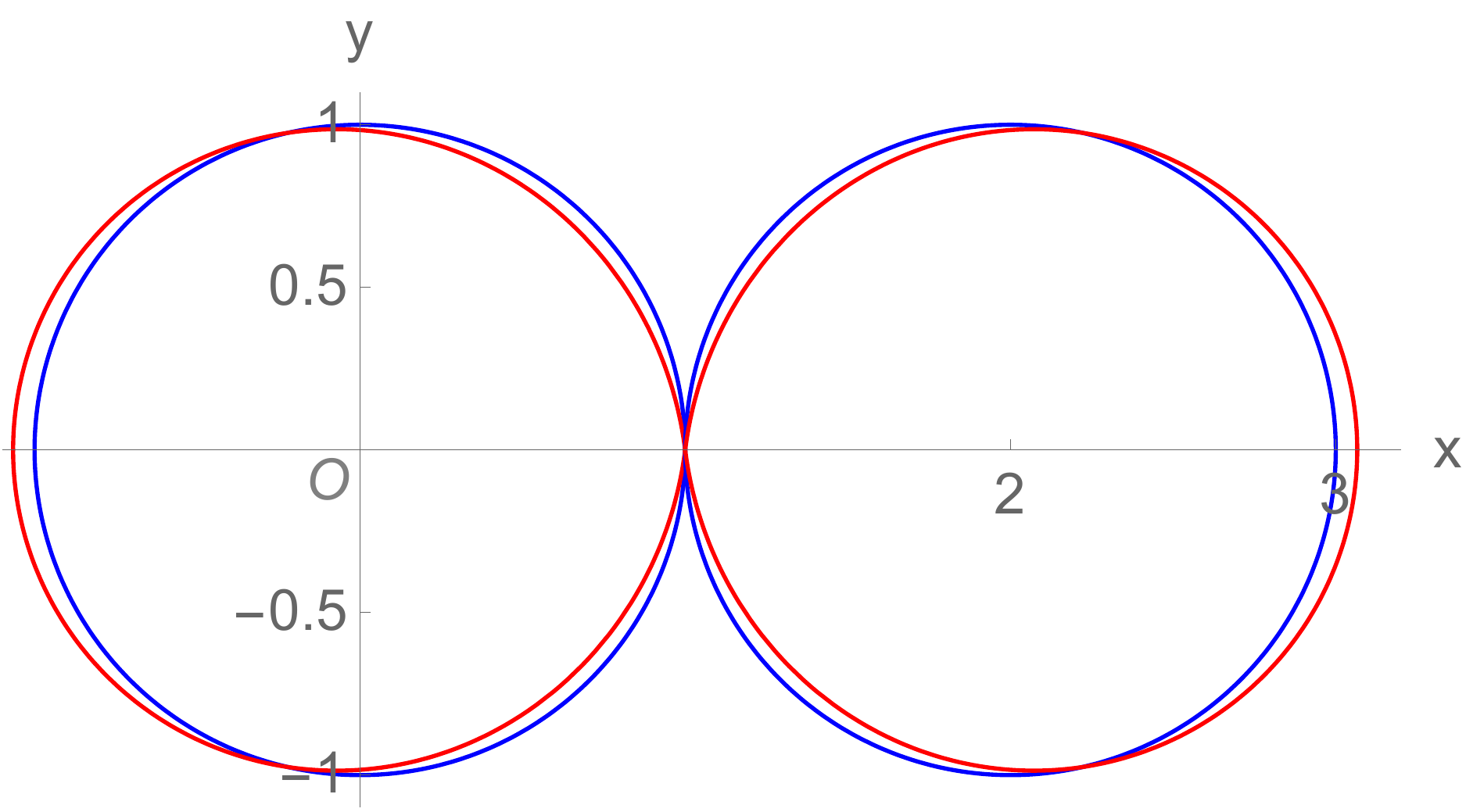}
\caption{Blue: $\rho=1$ and its mirror image. Red: $\rho=1+p/\pi$ and its mirror image.}
\label{2-D rho and mirror}
\end{figure}
\end{proof}

\begin{proof} [of Theorem \ref{3D asymptotics}]
When $\gamma=0$, the Euler\textendash Lagrange equation \eqref{stationary point sharp interface} is satisfied almost everywhere by two perfect balls touching at a single point. For $0<\gamma\ll1$, we have the following intuition: the potential $\phi$ is higher at the center where the two ball-like components are connected, and thus the mean curvature $H$ is smaller there, so naturally we expect a cone or a neck to be formed. However, there cannot be a cone-like structure connecting the two ball-like components, because the mean curvature of a cone is very large near its apex. Instead, it is reasonable to have a neck there, resembling the neck of an unduloid, because $\phi$ and thus $H$ are almost constant there (their gradients vanish due to symmetry). We want to determine the neck circumference asymptotically.

We assume that the equilibrium is axisymmetric. In Figure \ref{cross-section of balls}, we make no assumption on the existence of a neck, and still use a similar parameterization as in the Proof of Theorem \ref{2D asymptotics} (the 2-D case), i.e., we let $\rho(\theta)=1+\varepsilon p(\theta)$, but no longer require $p(0)=0$. We require $p(\pi)=0$ instead, and assume that $p$ is symmetric, i.e., $p(\theta)=p(2\pi-\theta)$.
\begin{figure}[H]
\centering
\begin{tikzpicture}
\draw (-40*1.3pt,0*1.3pt) circle (40*1.3pt);
\draw (40*1.3pt,0*1.3pt) circle (40*1.3pt);
\draw [red] (-37*1.3pt,0*1.3pt) ellipse (43*1.3pt and 37.2*1.3pt);
\draw [red] (49*1.3pt,0*1.3pt) ellipse (43*1.3pt and 37.2*1.3pt);
\draw [-{Stealth[length=5*1.3pt]}] (-40*1.3pt,0*1.3pt) -- (-4*1.3pt,-17.0*1.3pt);
\draw [red, -{Stealth[length=5*1.3pt]}] (-40*1.3pt,0*1.3pt) -- (-48*1.3pt,35.7*1.3pt);
\draw [dashed] (0*1.3pt,0*1.3pt) -- (-40*1.3pt,0*1.3pt);
\draw [red] (-36*1.3pt,0) arc (0:103:4*1.3pt);
\node [red] at (-34*1.3pt,5*1.3pt) {\fontsize{10.4pt}{0pt}\selectfont$\theta$};
\node[fill,circle,inner sep=0.8*1.3pt,label={below:\fontsize{10.4pt}{0pt}\selectfont$O$}] at (-40*1.3pt,0*1.3pt) {};
\node [red,rotate=-77] at (-48*1.3pt,15*1.3pt) {\fontsize{10.4pt}{0pt}\selectfont$1\!\!+\!\varepsilon p(\theta)$};
\node [rotate=-26] at (-25*1.3pt,-12*1.3pt) {\fontsize{10.4pt}{0pt}\selectfont$1$};
\end{tikzpicture}
\caption{Meridians of two touching balls. Black: perfect balls. Red: perturbed balls.}
\label{cross-section of balls}
\end{figure}

The mean curvature $H$ is given by \cite[Equation (15.13)]{gray2006modern}
\begin{equation*}
H = \frac{x' y''-x'' y'}{2\left(x'^2+y'^2\right)^{3/2}}-\frac{x'/y}{2\sqrt{x'^2+y'^2}},
\end{equation*}
where $x(\theta)=\rho(\theta)\cos(\theta)$ and $y(\theta)=\rho(\theta)\sin(\theta)$. For $\rho=1+\varepsilon p$ and $\varepsilon\ll1$, we have
\begin{equation*}
H=1-\frac\varepsilon2\big(p''(\theta)+p'(\theta)\cot\theta+2 p(\theta)\big)+o(\varepsilon).
\end{equation*}
Under the volume constraints and mass-symmetric assumptions, the volume of the left half is fixed to be $4\pi/3$, that is
\begin{equation*}
\frac{4\pi}3=\bigg\vert\pi\hspace{-4pt}\int_0^{\pi}\hspace{-6pt}y^2x'\dd{\theta}\bigg\vert=\frac{4\pi}3+\pi\varepsilon\hspace{-3pt}\int_0^{\pi}\hspace{-5pt}\big(3 p(\theta) \sin\theta-p'(\theta)\cos\theta\big)\sin^2\theta\dd{\theta}+o(\varepsilon).
\end{equation*}
Therefore we require $\int_0^{\pi}\hspace{-3pt}\big(3 p(\theta) \sin\theta-p'(\theta)\cos\theta\big)\sin^2\theta\dd{\theta}=0$.

For $\varepsilon\ll1$, we expect $\Omega$ to be very close to two perfect balls touching each other at a single point. Hence, we expect
\begin{equation*}
\phi(\theta)=\Big(4\pi\sqrt{5\!-\!4\cos\theta}\Big)^{-1}\frac{4\pi}3+(4\pi)^{-1}\frac{4\pi}3+o(1).
\end{equation*}
According to \eqref{stationary point sharp interface}, we have
\begin{equation*}
2-\varepsilon\big(p''(\theta)+p'(\theta)\cot\theta+2 p(\theta)\big)+o(\varepsilon)+\frac\gamma3\Big(\sqrt{5\!-\!4\cos\theta}\Big)^{-1}+\frac\gamma3+\gamma\,o(1)=\lambda.
\end{equation*}
As before, by letting $\varepsilon=\gamma$ and equating the first-order terms in the above equation, we obtain
\begin{equation*}
p''(\theta)+p'(\theta)\cot\theta+2 p(\theta)=\Big(3\sqrt{5\!-\!4\cos\theta}\Big)^{-1}+\text{constant}.
\end{equation*}
With the help of WOLFRAM MATHEMATICA, we obtain a solution, with the constant on the above right-hand side adopting the value $-5/36$,
\begin{equation}
\label{3-D p(theta)}
\begin{aligned}
72 p(\theta)=&\;10\sqrt{5\!-\!4\cos\theta}-23+\Big(\,7+6\ln(4/27)-2\ln(1\!-\!\cos\theta)\\
&+18\ln\big(\sqrt{5\!-\!4\cos\theta}\!+\!3\big)-14\ln\big(\sqrt{5\!-\!4\cos\theta}\!+\!1\big)\,\Big)\cos\theta.
\end{aligned}
\end{equation}
Figure \ref{3-D Blue: and its mirror image. Red: } visualizes the deformation. The singularity of $p$ indicates the formation of a neck near $(x,y)=(1,0)$. As mentioned earlier, near the neck, $\phi$ and thus $H$ are almost constant, so the neck should resemble the neck of an unduloid.
\begin{figure}[H]
\centering
\includegraphics[width=234pt]{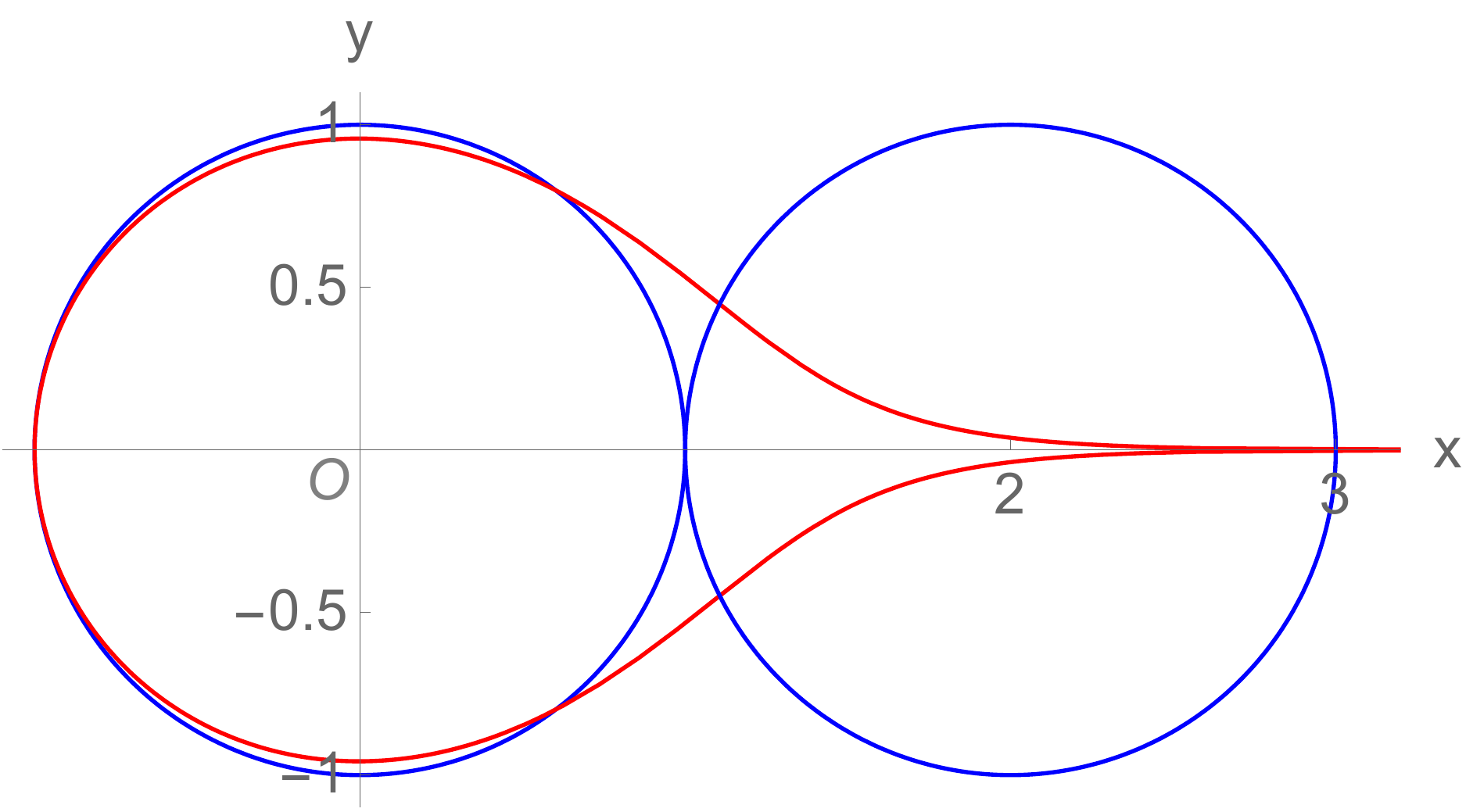}
\caption{Blue: $\rho=1$ and its mirror image. Red: $\rho=1+15p/\pi$.}
\label{3-D Blue: and its mirror image. Red: }
\end{figure}

As shown in Figure \ref{unduloid picture}, unduloids are surfaces of revolution with constant mean curvature. The meridians of unduloids (whose mean curvature is always 1) are parameterized by \cite[Theorem 3.1]{hadzhilazova2007unduloids}
\begin{equation*}
x(u)=\frac{1\!-\!e}2 F\left(\frac u2,\frac{2\sqrt e}{1\!+\!e}\right)+\frac{1\!+\!e}2 E\left(\frac u2,\frac{2\sqrt e}{1\!+\!e}\right),\quad y(u)=\sqrt{1+2 e \cos u+e^2}\big/2,\quad u\in[0,2\pi],
\end{equation*}
where $e\in(0,1)$ is the eccentricity, $F(\,\cdot\,,\,\cdot\,)$ and $E(\,\cdot\,,\,\cdot\,)$ are the incomplete elliptic integral of the first and second kinds, respectively.
\begin{figure}[H]
\centering
\includegraphics[width=180pt]{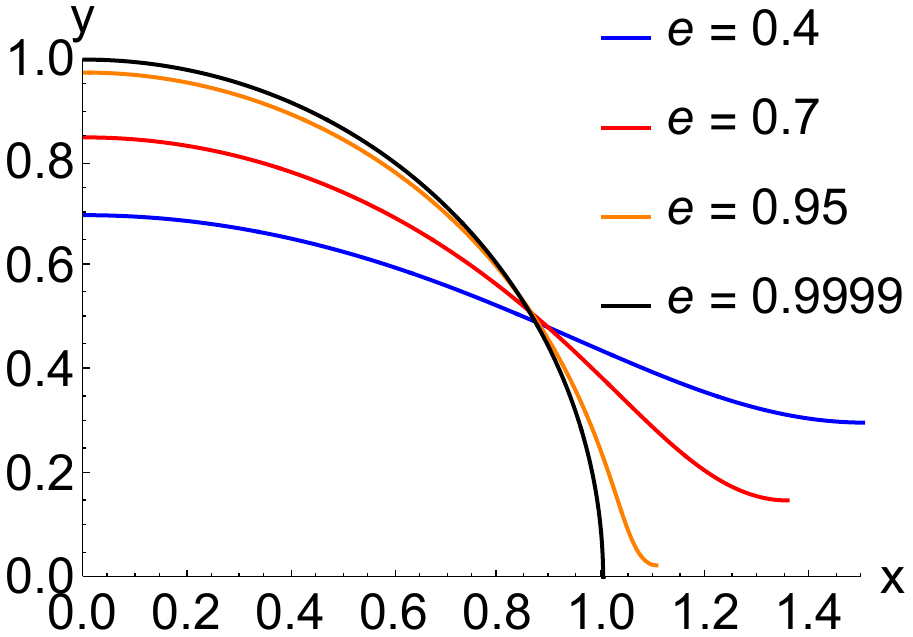}
\caption{Meridians of half periods of unduloids.}
\label{Meridians of half periods of unduloids.}
\end{figure}
As shown in Figure \ref{Meridians of half periods of unduloids.}, the meridian converges to a limiting curve $x=\sqrt{1\!-\!y^2}$ as $e\rightarrow1$. For $e\approx1$, we consider the deviation of the meridian from the limiting curve and then compare it to \eqref{3-D p(theta)}. To this end, we eliminate $u$ using $u=\arccos\big((4y^2\!-\!1\!-\!e^2)/(2e)\big)$, and then compute the Taylor expansion of $x(y)$ at $e=1$ as follows,
\begin{equation*}
x=\sqrt{1\!-\!y^2}+\frac{e\!-\!1}{2}\Big(\sec\theta-\ln\cot\frac\theta2\Big)+o(1\!-\!e),
\end{equation*}
where $\theta=\arcsin y$. Write $\rho=\sqrt{x^2+y^2}$, we have
\begin{equation}
\label{Taylor expansion of unduloids}
\rho=1+\frac{e\!-\!1}{2}\Big(1-\cos\theta\ln\cot\frac\theta2\Big)+o(1\!-\!e),
\end{equation}
where $\ln\cot(\theta/2)=-\ln\theta+\ln2+O\big(\theta^2\big)$.

In order to obtain the relation between $e$ and $\varepsilon$, we compare \eqref{Taylor expansion of unduloids} to $\rho=1+\varepsilon p(\theta)$ (with $p$ defined in \eqref{3-D p(theta)}), up to leading-order terms. Noticing that in \eqref{3-D p(theta)}, we have $\ln(1\!-\!\cos\theta)=2\ln\theta-\ln2+O\big(\theta^2\big)$. We want to match the following two leading order terms
\begin{equation*}
1+\frac{e\!-\!1}{2}\cos\theta\ln\theta\quad\text{and}\quad1+\varepsilon\frac{-2\cos\theta}{72}2\ln\theta,
\end{equation*}
therefore we obtain $\varepsilon=9(1\!-\!e)$.

Note that the minimum of $y(u)$ is attained at $y(\pi)=(1\!-\!e)/2$, so the neck circumference is $2\pi(1\!-\!e)/2=\varepsilon\pi/9$. According to \eqref{rescaled energy-liquid-drop-model}, we have
\begin{equation*}
\varepsilon=\gamma=\frac{40\pi\fis}{2(4\pi/3)}=15\fis.
\end{equation*}
Now that we know the neck circumference is asymptotically $5\fis\pi/3$ under the volume constraint $|\Omega|=2(4\pi/3)$, by rescaling we know that the neck circumference is asymptotically $\sqrt[3]{3|\Omega|/(8\pi)}5\fis\pi/3$ under an arbitrary volume constraint.
\end{proof}

\section{Discussion}

\label{Discussion}

In the literature (e.g., \cite[Figure 3]{zhao2022bifurcation}), it has been noticed that as $\fis$ increases, the energy minimizer changes from one ball to a few balls to many balls (or disks in the 2-D case). For $\fis$ sufficiently large, the minimizer resembles many periodically arranged tiny balls (or disks) of roughly equal sizes \cite{choksi2010small,knupfer2016low}. In the context of pattern formation, this is called a saturation process from a coarse structure to a fine structure, and is not well understood mathematically \cite[Page 1120]{ren2009oval}. The present work sheds light on how a ball (or disk) loses stability. We have shown that the saturation process is hysteretic, because when $\fis$ slightly exceeds the threshold $\fis_*$, a ball (or disk) will destabilize and eventually go through a dramatic change over a long period of time, after which even if $\fis$ falls back below $\fis_*$, the shape will not be restored to a ball (or disk). More specifically, our numerical results show that the oblate-like stationary points are unstable with respect to non-axisymmetric perturbations, and that the pitchfork bifurcation from a disk is subcritical instead of supercritical. But it remains open to analytically prove that the oblate-like stationary points are indeed unstable, and that the pitchfork bifurcation from a disk is indeed subcritical. As a side note, in some other systems such as a rotating liquid drop with no charge or a rotating self-gravitating celestial object \cite{brown1980shape,swiatecki1974rotating}, similar phenomena occur \cite[Figure 4.9]{heinecomputations}.

For $\fis>\fis_3$, there is a Businaro\textendash Gallone branch consisting of mass-asymmetric saddle points which are of index 2 \cite[Page 144]{cohen1962deformation}. Nix as well as Strutinski\u{\i} \cite[Middle of Page 251]{nix1969further} was able to track this branch for $\fis\in(\fis_3,0.80)$; for $\fis>0.80$, the fate of this family is unknown but of great interest. It will be interesting and promising to use our phase-field approach together with a higher-index extension of the shrinking dimer method \cite[Bottom of Page 1919]{zhang2012shrinking} to search for such index-2 saddle points for $\fis>0.80$.

In the present work, we have only considered the simplest liquid drop model. To more accurately describe atomic nuclei, the liquid drop model would need some corrections. It is likely that one could adapt our method to handle more complicated modified versions of the liquid drop model. For example, we may add a rotational energy term to handle angular momenta \cite[Page 6]{krappe2012theory}, or add a surface curvature energy term \cite[Equation (19)]{ivanyuk2009optimal}. We may also consider pairing correlations \cite[Page 18]{krappe2012theory}, as well as the quantum shell structure effects. After those modifications, the minimizer may no longer be a ball but resemble an oblate or prolate.

From a mathematical point of view, it might also be interesting to explore the problem \eqref{energy-liquid-drop-model} with the Euclidean perimeter replaced by the 1-perimeter \cite{goldman2019optimality}, a fractional perimeter \cite{dipierro2017rigidity}, or a general nonlocal perimeter \cite{cesaroni2017isoperimetric}, with the Coulomb potential replaced by a Yukawa potential \cite{fall2018periodic}, a Riesz potential \cite{bonacini2014local}, a fractional inverse Laplacian kernel \cite[Appendix]{chan2019lamellar}, or a general nonlocal kernel \cite{luo2022nonlocal}.

\section*{Acknowledgments}

The authors would like to thank Kuang Huang, Chong Wang, Rupert Frank and Michael Weinstein for helpful discussions, and thank the reviewer for helpful comments. The computing resources were provided by Columbia University's Shared Research Computing Facility project, which is supported by NIH Research Facility Improvement Grant 1G20RR030893-01, and associated funds from the New York State Empire State Development, Division of Science Technology and Innovation (NYSTAR) Contract C090171, both awarded April 15, 2010.

\section*{Data availability}

The data that support the findings of this study are openly available in Open Science Framework at \href{http://doi.org/10.17605/OSF.IO/76TDW}{DOI:10.17605/OSF.IO/76TDW}.

\bibliographystyle{apalike}
\bibliography{main}

\vfill
\begin{center}
\LARGE Appendices appear after this page
\end{center}
\vspace{-10pt}

\appendix\newpage

\section{Local minimality of two balls infinitely far apart}
\label{Local minimality of two balls}

When $\fis$ exceeds $\fis_1$, the shape consisting of two balls of equal radii infinitely far apart becomes a generalized local minimizer. By "generalized" we mean that the two balls are separated very far apart and their Coulomb interactions are very weak, so that we can consider the limit where there is no Coulomb interaction between them (see \cite[Definition 4.3]{knupfer2016low}). To make it more mathematically rigorous, we could consider a weighted background potential well which confines the shape in a bounded region, and then study the limit as the weight goes to 0 (see \cite{alama2019droplet}). Alternatively, we could use periodic boundary conditions which confines the shape inside a cell, and then let the size of the cell go to infinity (see \cite{choksi2010small}).

In this appendix, we give a simple derivation of $\fis_1$. We only consider shapes consisting of two balls infinitely far apart (not necessarily of equal sizes). The reason is that there is no Coulomb interaction between the two components which are infinitely far apart, so that we can use the following ansatz: each component is a ball that is a local minimizer of \eqref{energy-liquid-drop-model}.

For $\Omega$ being a ball of volume $V$, the energy in \eqref{energy-liquid-drop-model} is
\begin{equation*}
g(V) = \Big(1\!+\!\frac V{20\pi}\Big)\sqrt[3]{36\pi V^2}.
\end{equation*}
Now consider two balls with volumes $v$ and $V-v$, respectively, infinitely far apart so that we can neglect the interaction between them. The total energy of those two balls equals $g(v)+g(V\!-\!v)$, and are shown in Figure \ref{two balls energy} (cf. \cite[Figure 9]{frankel1947calculations}, where the two balls are touching). For $\fis>\fis_1$ (recalling $\fis=V/(40\pi)$ and $\fis_1=1/5$), the energy attains a local minimum when the two balls have equal radii. At $\fis=\fis_1$, there is a subcritical bifurcation. For $\fis<\fis_1$, it is no longer stable to have two equally large balls, one of the two balls will shrink and eventually vanish, and the other one will grow into a larger ball.
\begin{figure}[H]
\centering
\includegraphics[width=300pt]{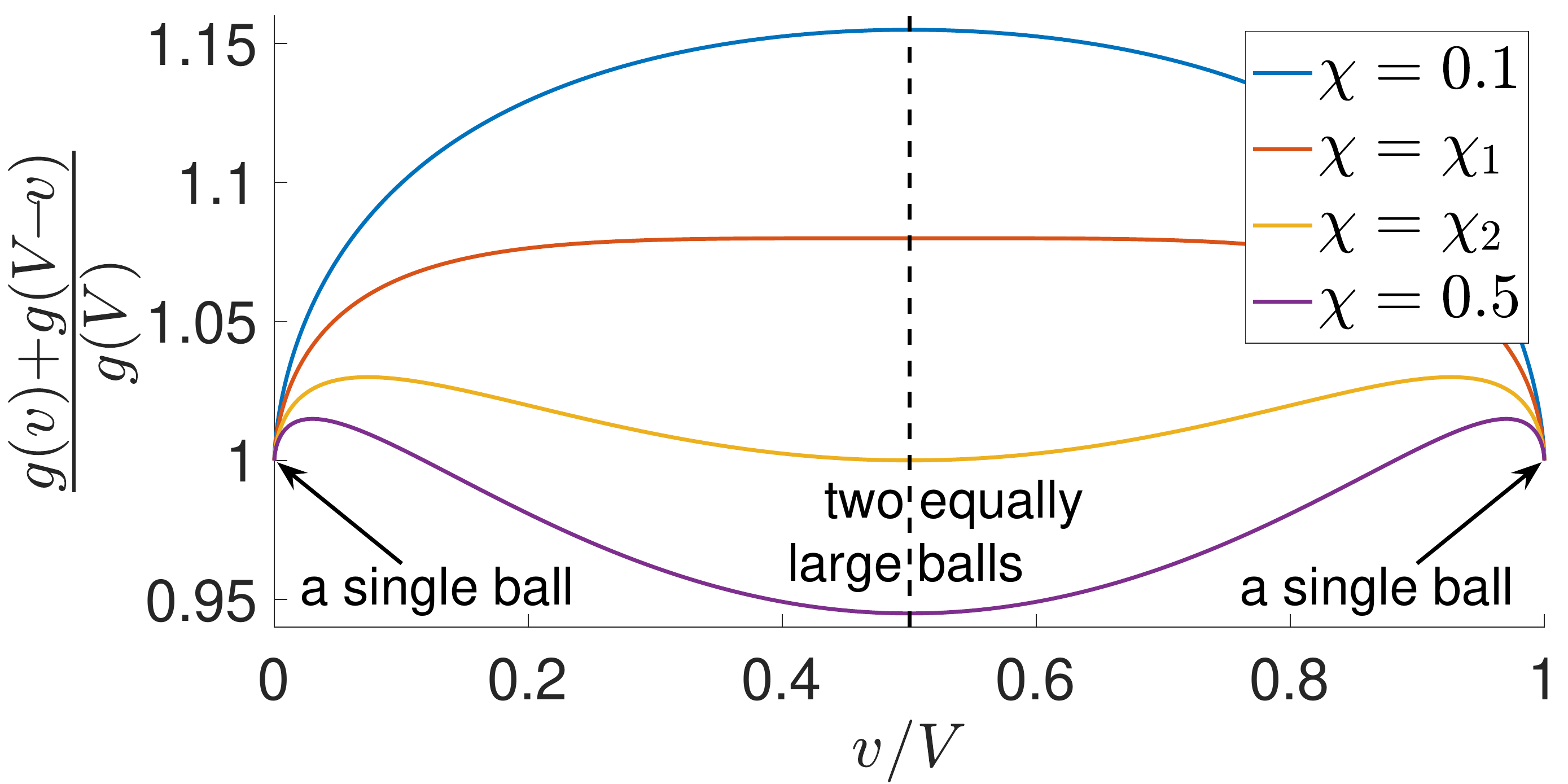}
\caption{Energy of two balls infinitely far apart. Horizontal axis: volume fraction. Vertical axis: energy ratio of two balls to one ball. Dashed vertical line: the cases where the two balls have equal radii.}
\label{two balls energy}
\end{figure}
The calculation of $\fis_1$ is straightforward by looking at the sign of the following second-order derivative:
\begin{equation*}
\frac{\dd^2}{\dd v^2}\bigg|_{v=V/2}\Big(g(v)+g(V\!-\!v)\Big)=2(V\!-\!8\pi)\big/\sqrt[3]{81\pi^2V^4}.
\end{equation*}

\section{Euler\textendash Lagrange equations for diffuse and sharp interfaces}
\label{Formal derivation of Euler-Lagrange equation from diffuse to sharp interface settings}

In this appendix, let us formally derive \eqref{stationary point sharp interface} from \eqref{stationary point diffuse interface} as $c\rightarrow0$. Let $\Gamma$ be the interface $\big\{\vec x: \eta(\vec x)=1/2\big\}$. Assuming $c\ll1$, in the vicinity of $\Gamma$, the first two terms in the left-hand side of \eqref{stationary point diffuse interface} are dominant, therefore we follow the heuristics in \cite[Equation (6)]{du2020phase} and assume that the stationary solution has the following approximation in the vicinity of $\Gamma$,
\begin{equation*}
\eta\big(\vec x\big)=\frac{1+\tanh\big(k d(\vec x)\big)}2,
\end{equation*}
where $k>0$ is an undetermined coefficient, and $d(\vec x)$ is the signed distance function between $\vec x$ and $\Gamma$. By differentiating both sides we obtain
\begin{equation*}
\begin{aligned}
2\Delta\eta&=\nabla\cdot\nabla\tanh(k d)=\nabla\cdot\big(k\tanh'(k d)\nabla d\big)\\
&=k\,\nabla\tanh'(k d)\cdot\nabla d+k\tanh'(k d)\Delta d\\
&=k^2\tanh''(k d)\nabla d\cdot\nabla d+k\tanh'(k d)\Delta d,
\end{aligned}
\end{equation*}
where
\begin{equation*}
\tanh' = 1-\tanh^2=\frac23f'\Big(\frac{1\!+\!\tanh}2\Big),\quad \tanh'' = 2\tanh^3-2\tanh=\frac43W'\Big(\frac{1\!+\!\tanh}2\Big).
\end{equation*}
According to \cite[Bottom of Page 11]{du2020phase}, we have $\nabla d\cdot\nabla d=\|\nabla d\|_2^2=1$ and $\Delta d=-(n\!-\!1)H$. Therefore
\begin{equation*}
\begin{aligned}
\sqrt{6/c}\,W'(\eta)-\sqrt{6c}\,\Delta\hspace{0.5pt}\eta&=\sqrt{\frac6c}\,W'(\eta)-\sqrt{\frac{3c}2}\big(k^2\tanh''(k d)-(n\!-\!1)Hk\tanh'(k d)\big)\\
&=\sqrt{\frac6c}\,W'(\eta)-\sqrt{\frac{3c}2}\Big(k^2\frac43W'(\eta)-(n\!-\!1)Hk\frac23f'(\eta)\Big)\\
&=\Big(\sqrt{\frac6c}\!-\!\sqrt{\frac{8c}3}k^2\Big)W'(\eta)+\sqrt{\frac{2c}3}(n\!-\!1)Hk f'(\eta),
\end{aligned}
\end{equation*}
where we choose $k=\sqrt{3/(2c)}$ so that the coefficient of $W'(\eta)$ vanishes. In this way, in the vicinity of $\Gamma$, we can rewrite \eqref{stationary point diffuse interface} as
\begin{equation*}
(n\!-\!1)Hf'(\eta)-\tilde\gamma f'(\eta)\,\Delta^{-1}f(\eta)=-Kf'(\eta)\big\langle f(\eta)\!-\!\omega,\,1\big\rangle,
\end{equation*}
or equivalently
\begin{equation*}
(n\!-\!1)H+\tilde\gamma(-\Delta)^{-1}f(\eta)=-K\big\langle f(\eta)\!-\!\omega,\,1\big\rangle,
\end{equation*}
where the right-hand side can be regarded as a Lagrange multiplier, and $f(\eta)$ converges to some indicator function $\bm1_\Omega$ as $c\rightarrow0$.

\section{Stability of a strip or a lamella against rupture}
\label{Stability of a strip or lamella against rupture}

In this appendix, we consider the periodic isoperimetric problem \cite[Equation (3.10)]{choksi2006periodic} in a unit cell. A thin strip (in 2-D) or a thin lamella (in 3-D) may have relatively large interfacial energy, but we are going to show that they are both stable against a certain class of perturbations.

\subsection{The 2-D case}

Let the periodic cell be $[0,1]^2$. Without loss of generality, we only consider shapes described by $\big\{(x,y): 0<y<f(x)\big\}$, where $f:[0,1]\to[0,1]$ is piecewise continuously differentiable and satisfies the periodic boundary condition $f(0)=f(1)$. Given $\omega\in(0,1)$, the periodic isoperimetric problem can be written as follows
\begin{equation}
\label{strip classical isoperimetric problem}
\min_{f}\text{Per}(f)\quad\text{subject to}\quad\text{Area}(f)=\omega,
\end{equation}
where
\begin{equation*}
\text{Per}(f) = \int_0^1\sqrt{f'(x)^2\!+\!1}\dd{x}\quad\text{and}\quad\text{Area}(f) = \int_0^1f(x)\dd{x}.
\end{equation*}
The following proposition shows that no matter how small $\omega$ is, a thin strip is always stable against admissible deformations.

\begin{proposition}
\label{stable strip}
For any $\omega\in(0,1)$, the second-order variation of \eqref{strip classical isoperimetric problem} at the constant function $f=\omega$ is positive among admissible perturbations.
\end{proposition}

\begin{proof}
For the constant function $f=\omega$, consider the deformation from $f$ to $f+\varepsilon g$, where $0<\varepsilon\ll1$, and $g$ is piecewise continuously differentiable and satisfies the periodic boundary condition $g(0)=g(1)$. Using Taylor expansion we obtain
\begin{equation*}
\text{Per}(f+\varepsilon g) = \text{Per}(f)+\varepsilon\!\int_0^1\frac{f'(x)g'(x)}{\sqrt{f'(x)^2\!+\!1}}\dd{x}+\frac{\varepsilon^2}2\!\int\frac{g'(x)^2\dd{x}}{\big(f'(x)^2\!+\!1\big)^{3/2}}+O\big(\varepsilon^3\big),
\end{equation*}
where the first-order term (involving $\varepsilon$) is always zero, and the second-order term (involving $\varepsilon^2$) is always positive unless $g=0$, because we have $f'=0$.
\end{proof}

\subsection{The 3-D case}

Let the periodic cell be $[0,1]^3$. Without loss of generality, we only consider shapes described by $\big\{(x,y,z): 0<z<f(x,y)\big\}$, where $f:[0,1]^2\to[0,1]$ is piecewise continuously differentiable and satisfies periodic boundary conditions. Given $\omega\in(0,1)$, the periodic isoperimetric problem can be written as follows
\begin{equation}
\label{lamella classical isoperimetric problem}
\min_{f}\text{Area}(f)\quad\text{subject to}\quad\text{Vol}(f)=\omega,
\end{equation}
where
\begin{equation*}
\text{Area}(f) = \iint_{[0,\,1]^2}\sqrt{f_x^2\!+\!f_y^2\!+\!1}\dd{x}\dd{y}\quad\text{and}\quad\text{Vol}(f) = \iint_{[0,\,1]^2}f\dd{x}\dd{y},
\end{equation*}
with $f_x$ and $f_y$ denoting the partial derivatives of $f$. The following proposition shows that no matter how small $\omega$ is, a thin lamella is always stable against admissible deformations.

\begin{proposition}
\label{stable lamella}
For any $\omega\in(0,1)$, the second-order variation of \eqref{lamella classical isoperimetric problem} at the constant function $f=\omega$ is positive among admissible perturbations.
\end{proposition}

\begin{proof}
For the constant function $f=\omega$, consider the deformation from $f$ to $f+\varepsilon g$, where $0<\varepsilon\ll1$, and $g$ is piecewise continuously differentiable and satisfies periodic boundary conditions. Using Taylor expansion we obtain
\begin{equation*}
\text{Area}(f+\varepsilon g) = \text{Area}(f)+\varepsilon\iint\frac{(f_x g_x\!+\!f_y g_y)}{\sqrt{f_x^2\!+\!f_y^2+1}}+\frac{\varepsilon^2}2\!\iint\frac{f_x^2 g_y^2\!-\!2 f_x f_y g_x g_y\!+\!f_y^2 g_x^2\!+\!g_x^2\!+\!g_y^2}{\left(f_x^2+f_y^2+1\right)^{3/2}}+O\big(\varepsilon^3\big),
\end{equation*}
where the first-order term (involving $\varepsilon$) is always zero, and the second-order term (involving $\varepsilon^2$) is always positive unless $g=0$, because we have $f_x=f_y=0$.
\end{proof}

\section{Algorithms for $(-\Delta)^{-1}$ under no boundary conditions}
In this appendix, we provide code for numerically solving Poisson's equation $-\Delta\phi=\eta$ in the free space.
\subsection{The 2-D case}
\label{algorithm 2-D Poissons equation}
As an example, we use the source term $\eta$ given by \eqref{test source term}.
\begin{lstlisting}[caption = {MATLAB code for solving Poisson's equation in 2-D},mathescape]
M = 256; N = 256; X = 1; Y = 1;
r = radius(1:N,(1:M)',M,N,X,Y);
eta = 1 - 12*min(r,1/2).^2 + 16*min(r,1/2).^3;
phi_analytical = Poisson_Analytical(r);
phi_numerical  = Poisson_Numerical(eta, X, Y);
err = max(abs(phi_analytical(:)-phi_numerical(:)));
fprintf('For [M,N] = [%d,%d], the error is %g.\n',M,N,err);

function r = radius(i,j,M,N,X,Y)
xi = X/N*(i-1/2); yj = Y/M*(j-1/2); % 1<=i<=N, 1<=j<=M
r = sqrt((xi-X/2).^2 + (yj-Y/2).^2); % calculate distance
end

function phi = Poisson_Analytical(r)
r1 = min(r,1/2); r2 = max(r,1/2);
phi = (3*r1.^4-r1.^2)/4 - 16*r1.^5/25 + 57/1600 - 3/80*log(r2);
end

function phi = Poisson_Numerical(eta, X, Y)
[M, N] = size(eta); x = X/N*((0:N)-1/2); y = Y/M*((0:M)-1/2)';
H = (y.^2-x.^2).*atan(x./y) + x.*y.*log(x.^2+y.^2);
Gh = diff(diff(H),1,2) - 3*X/N*Y/M;
Gh(1,:) = Gh(1,:) + 2*pi*(X/N)^2*(0:N-1);
Gh(1,1) = Gh(1,1) + pi/2*(X/N)^2; Gh = Gh/(-4*pi);
Gh = [Gh(end:-1:2,:);Gh;zeros(1,N)];
Gh = [Gh(:,end:-1:2) Gh zeros(2*M,1)];
convolution = ifft2(fft2(eta,2*M,2*N).*fft2(Gh));
phi = real(convolution(end-M:end-1,end-N:end-1));
end
\end{lstlisting}

\subsection{The 3-D case}
\label{algorithm 3-D Poissons equation}

\begin{lstlisting}[caption = {MATLAB code for solving Poisson's equation in 3-D},mathescape]
function phi = Poisson_Numerical(eta, X, Y, Z)
[M, N, P] = size(eta); x = X/N*((0:N)-1/2);
y = Y/M*((0:M)-1/2)'; z = reshape(Z/P*((0:P)-1/2),1,1,[]);
xy = x.*y; yz = y.*z; zx = z.*x; r = sqrt(x.^2+y.^2+z.^2);
H = yz.*log(x+r) + zx.*log(y+r) + xy.*log(z+r) -...
    x.^2/2.*atan(yz./(x.*r)) - y.^2/2.*atan(zx./(y.*r)) -...
    z.^2/2.*atan(xy./(z.*r));
Gh = diff(diff(diff(H),1,2),1,3); Gh = Gh/(4*pi);
Gh = cat(1, Gh(end:-1:2,:,:), Gh, zeros(1,    N,P));
Gh = cat(2, Gh(:,end:-1:2,:), Gh, zeros(2*M,  1,P));
Gh = cat(3, Gh(:,:,end:-1:2), Gh, zeros(2*M,2*N,1));
convolution = ifftn(fftn(eta,[2*M,2*N,2*P]).*fftn(Gh));
phi = real(convolution(end-M:end-1,end-N:end-1,end-P:end-1));
end
\end{lstlisting}

\subsection{Numerical convergence}
\label{numerical convergence of G*eta}

As a test run, we compute the Poisson potential of the following source term in 2-D:
\begin{equation}
\label{test source term}
\eta(x,y) =
\left\{
\begin{aligned}
&1\!-\!12r^2\!+\!16r^3,&&r\leqslant\frac12,\\
&0,&&r>\frac12,
\end{aligned}
\right.
\quad\text{where}\;r:=\sqrt{\Big(x\!-\!\frac12\Big)^2+\Big(y\!-\!\frac12\Big)^2}.
\end{equation}
We know that $\eta$ is axisymmetric and compactly supported within $[0,1]^2$, as shown on the left of Figure \ref{test run}. The potential $\phi$ is given by the convolution $G*\eta$, where $G(x,y)=(-4\pi)^{-1}\ln(x^2\!+\!y^2)$. We know that $\phi$ is also axisymmetric, as shown on the right of Figure \ref{test run}, and has an analytic expression
\begin{equation}
\label{analytical poisson potential}
\phi(x,y) =
\left\{
\begin{aligned}
&\frac{3 r^4\!-\!r^2}{4}-\frac{16 r^5}{25}+\frac{57}{1600}+\frac{3 \ln2}{80},&&r\leqslant\frac12,\\
&-\frac{3 \ln r}{80},&&r>\frac12.
\end{aligned}
\right.
\end{equation}
See \cite{chen2002four} for the derivation in the case of $r>1/2$. The case of $r\leqslant1/2$ can be derived using polar coordinates for Laplacian, and then adding a constant term to make $\phi$ continuous.
\begin{figure}[H]
\centering
\includegraphics[height=135pt]{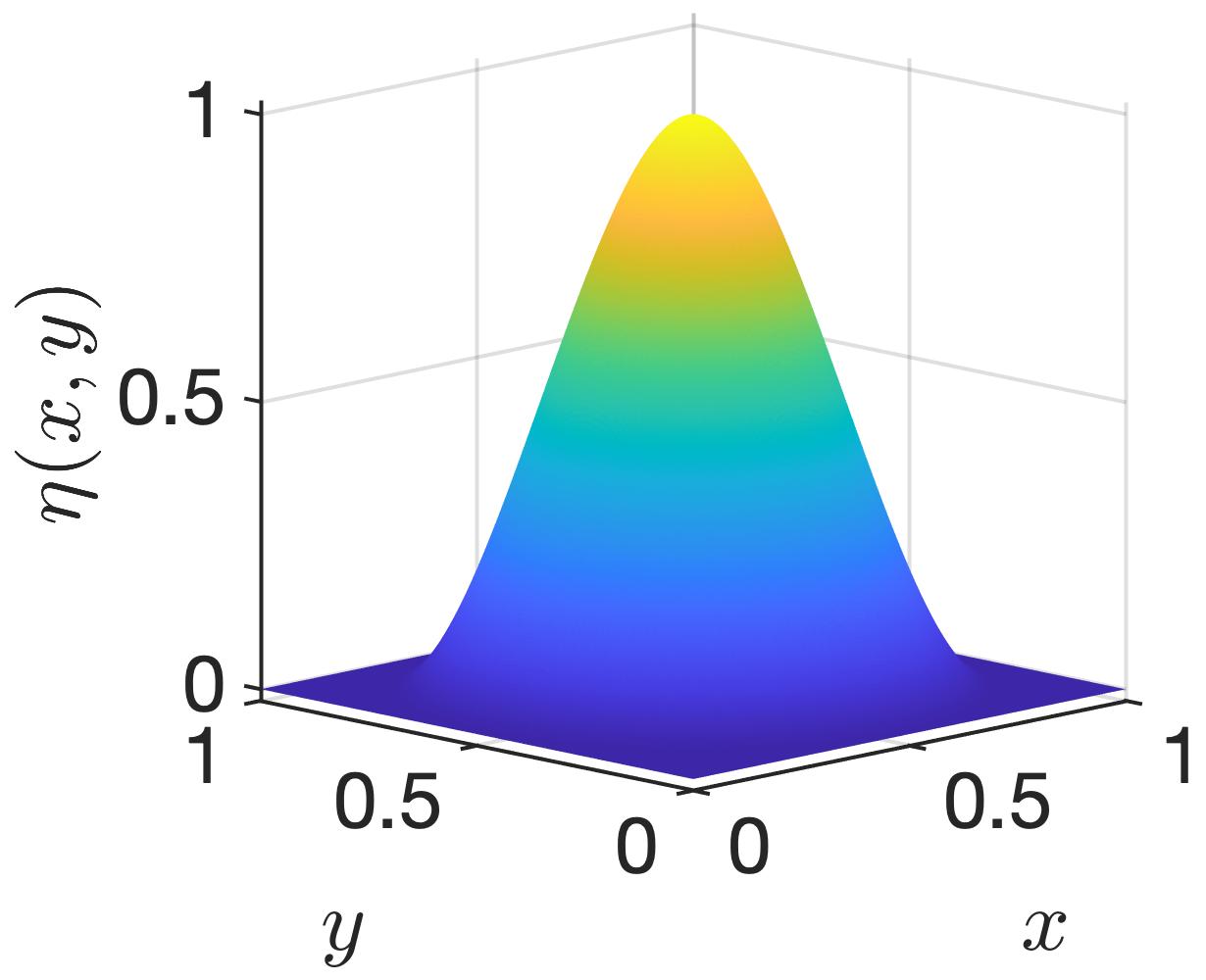}
\includegraphics[height=135pt]{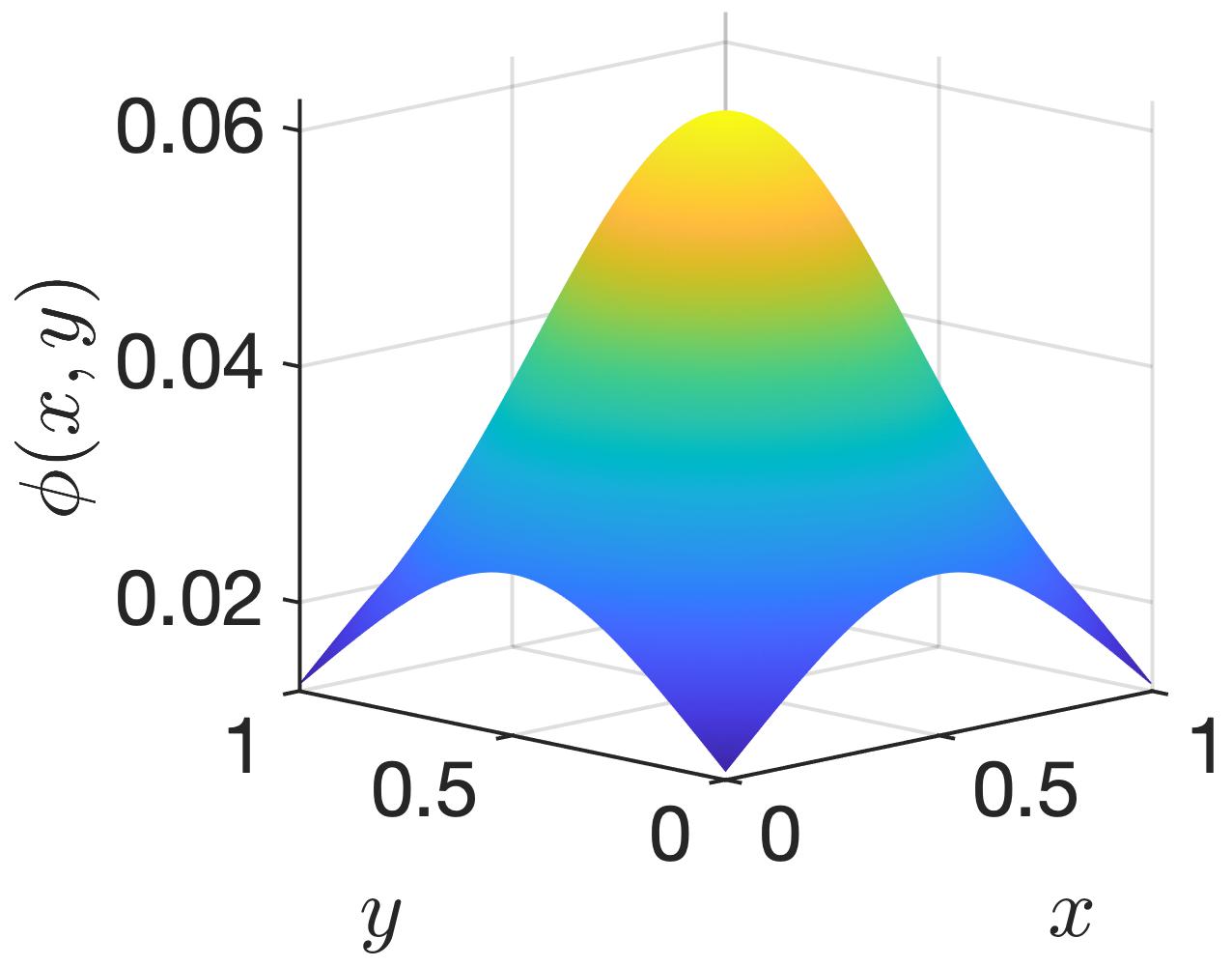}
\caption{Test run. Left: source term. Right: Poisson potential.}
\label{test run}
\end{figure}

We look at the difference between the analytic solution \eqref{analytical poisson potential} and the numerical solution obtained in Appendix \ref{algorithm 2-D Poissons equation}. Define the error to be the maximum difference on all the grid points. Our numerical results exhibit quadratic convergence, as shown in Table \ref{Convergence}.
\begin{table}[H]
  \centering
  \caption{Convergence}
    \begin{tabular}{c|c|c}
    \toprule
    $[M,N]$ & error & ratio \\
    \hline
    \multirow{2}[1]{*}{} & \multirow{2}[1]{*}{} &  \\[-10pt]
        $[64,64]$   & $1.02\times10^{-5}$ & \multirow{2}[0]{*}{\raisebox{-4pt}{4.00}} \\
    \multirow{2}[0]{*}{} & \multirow{2}[0]{*}{} &  \\
       $[128,128]$   & $2.54\times10^{-6}$ & \multirow{2}[0]{*}{\raisebox{-4pt}{4.00}} \\
    \multirow{2}[0]{*}{} & \multirow{2}[0]{*}{} &  \\
       $[256,256]$   & $6.36\times10^{-7}$ & \multirow{2}[0]{*}{\raisebox{-4pt}{4.00}} \\
    \multirow{2}[0]{*}{} & \multirow{2}[0]{*}{} &  \\
       $[512,512]$   & $1.59\times10^{-7}$ &  \\
    \bottomrule
    \end{tabular}
  \label{Convergence}
\end{table}

\section{Simulations in 3-D under no boundary conditions}
\label{Simulations in 3-D under no boundary conditions}

In this appendix we present the detailed numerical results for the minimum energy paths of fission in 3-D under no boundary conditions.

\begin{figure}[H]
\centering
\includegraphics[width=269.388pt]{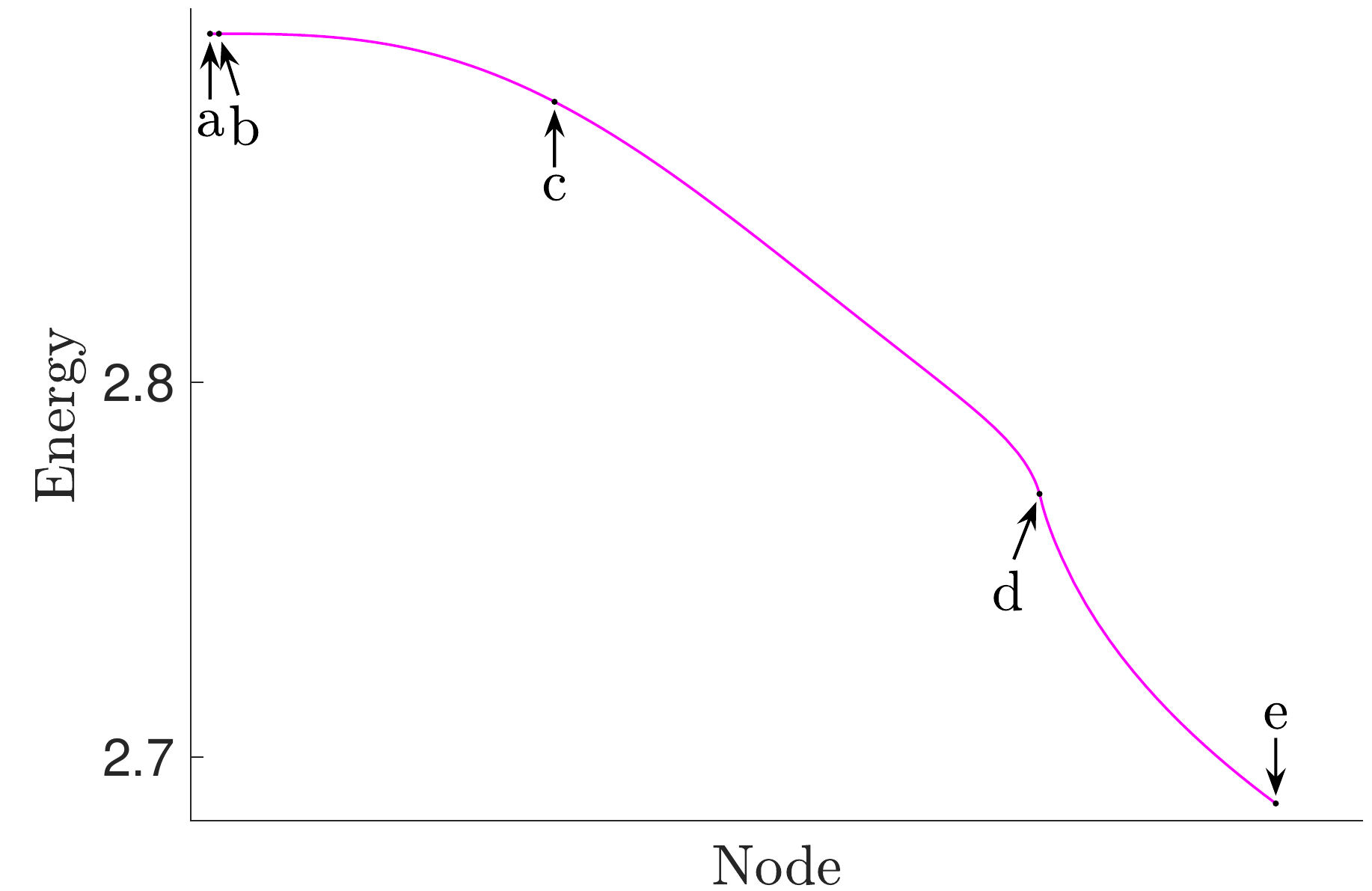}

$\overset{\includegraphics[width=58pt]{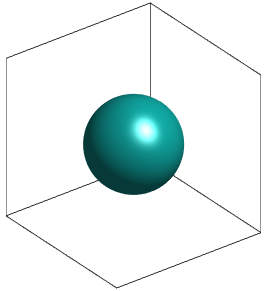}}{\begin{minipage}[t][\height][b]{24pt}\hspace{7pt}a\end{minipage}}$
$\overset{\includegraphics[width=58pt]{images/NoBoundary3D/8.0/b.png}}{\begin{minipage}[t][\height][b]{24pt}\hspace{7pt}b\end{minipage}}$
$\overset{\includegraphics[width=58pt]{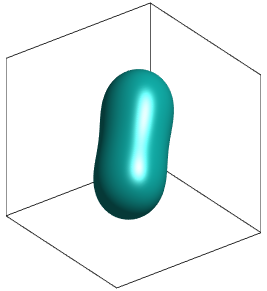}}{\begin{minipage}[t][\height][b]{24pt}\hspace{7pt}c\end{minipage}}$
$\overset{\includegraphics[width=58pt]{images/NoBoundary3D/8.0/d.png}}{\begin{minipage}[t][\height][b]{24pt}\hspace{7pt}d\end{minipage}}$
$\overset{\includegraphics[width=58pt]{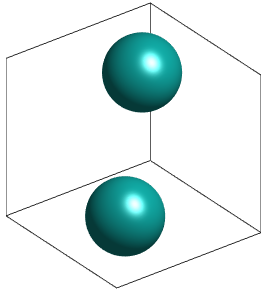}}{\begin{minipage}[t][\height][b]{24pt}\hspace{7pt}e\end{minipage}}$
\caption{Minimum energy path ($\tdfis=0.992$). Saddle point and scission point are b and d, respectively.}
\end{figure}

\begin{figure}[H]
\centering
\includegraphics[width=269.388pt]{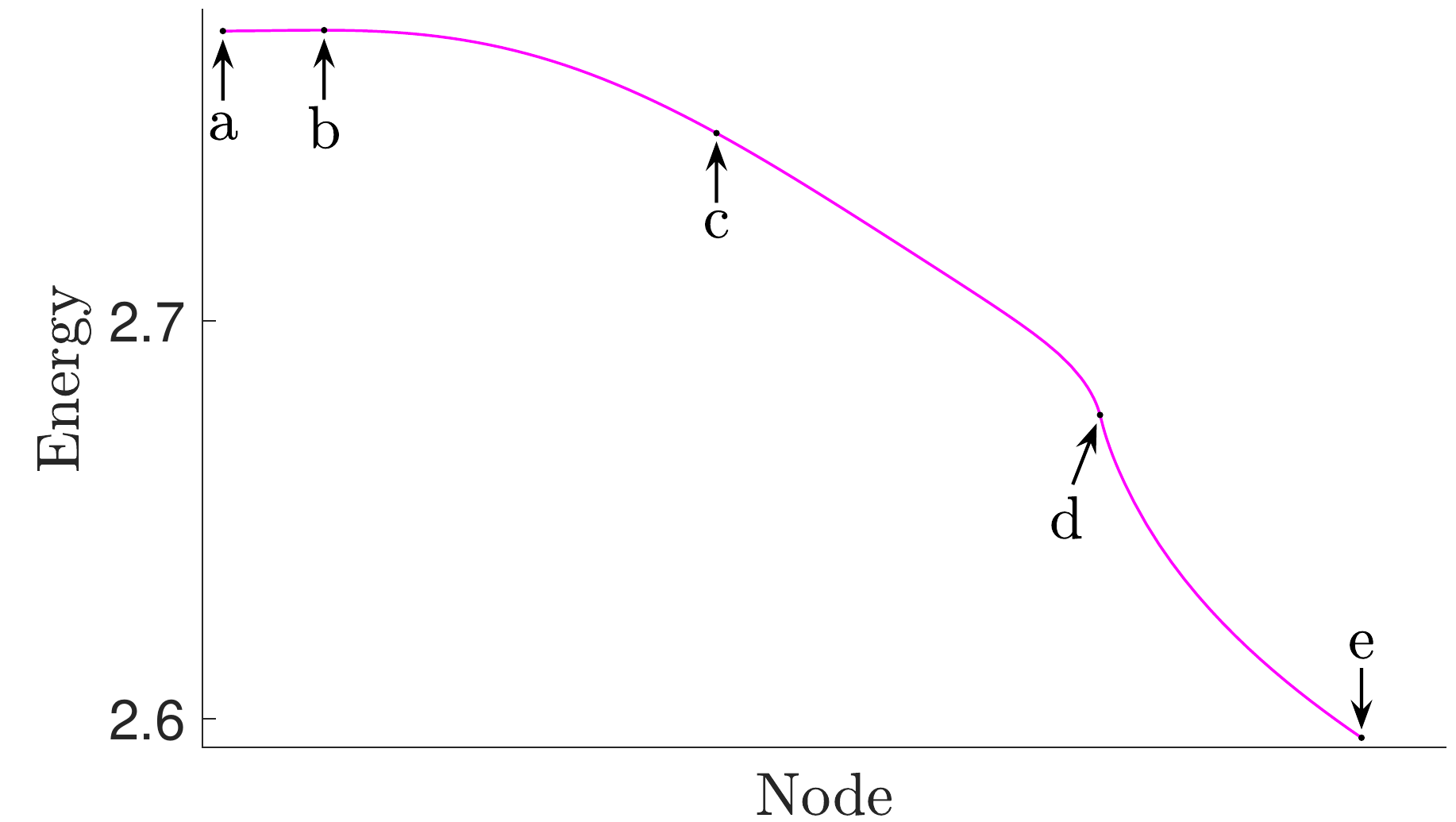}

$\overset{\includegraphics[width=58pt]{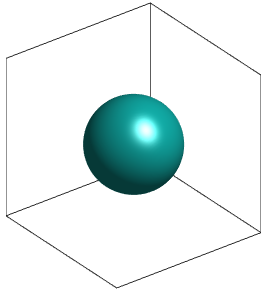}}{\begin{minipage}[t][\height][b]{24pt}\hspace{7pt}a\end{minipage}}$
$\overset{\includegraphics[width=58pt]{images/NoBoundary3D/7.5/b.png}}{\begin{minipage}[t][\height][b]{24pt}\hspace{7pt}b\end{minipage}}$
$\overset{\includegraphics[width=58pt]{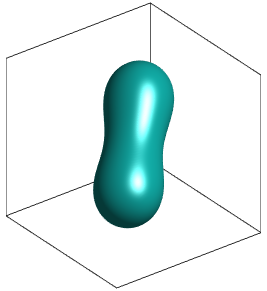}}{\begin{minipage}[t][\height][b]{24pt}\hspace{7pt}c\end{minipage}}$
$\overset{\includegraphics[width=58pt]{images/NoBoundary3D/7.5/d.png}}{\begin{minipage}[t][\height][b]{24pt}\hspace{7pt}d\end{minipage}}$
$\overset{\includegraphics[width=58pt]{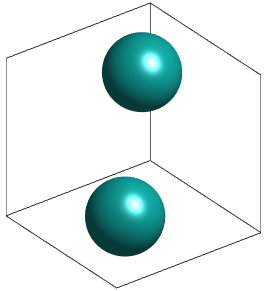}}{\begin{minipage}[t][\height][b]{24pt}\hspace{7pt}e\end{minipage}}$
\caption{Minimum energy path ($\tdfis=0.930$). Saddle point and scission point are b and d, respectively.}
\end{figure}

\begin{figure}[H]
\centering
\includegraphics[width=269.388pt]{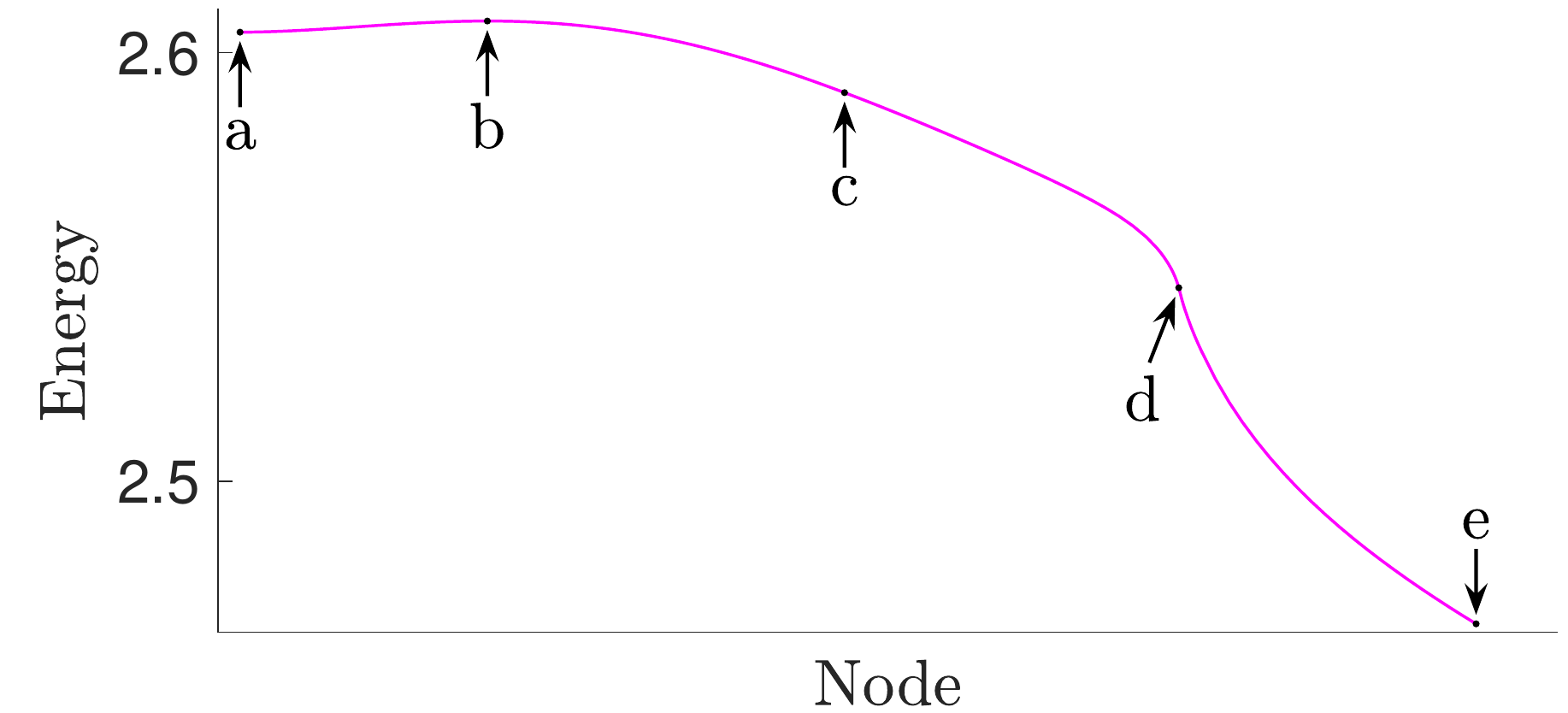}

$\overset{\includegraphics[width=58pt]{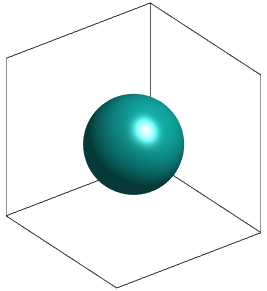}}{\begin{minipage}[t][\height][b]{24pt}\hspace{7pt}a\end{minipage}}$
$\overset{\includegraphics[width=58pt]{images/NoBoundary3D/6.8/b.png}}{\begin{minipage}[t][\height][b]{24pt}\hspace{7pt}b\end{minipage}}$
$\overset{\includegraphics[width=58pt]{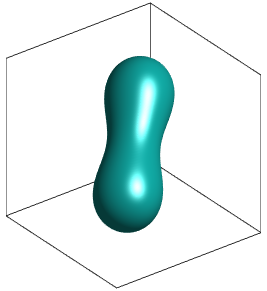}}{\begin{minipage}[t][\height][b]{24pt}\hspace{7pt}c\end{minipage}}$
$\overset{\includegraphics[width=58pt]{images/NoBoundary3D/6.8/d.png}}{\begin{minipage}[t][\height][b]{24pt}\hspace{7pt}d\end{minipage}}$
$\overset{\includegraphics[width=58pt]{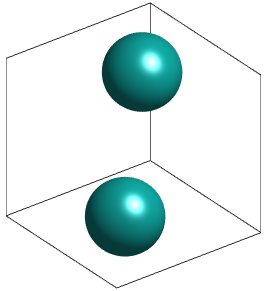}}{\begin{minipage}[t][\height][b]{24pt}\hspace{7pt}e\end{minipage}}$
\caption{Minimum energy path ($\tdfis=0.844$). Saddle point and scission point are b and d, respectively.}
\end{figure}

\begin{figure}[H]
\centering
\includegraphics[width=269.388pt]{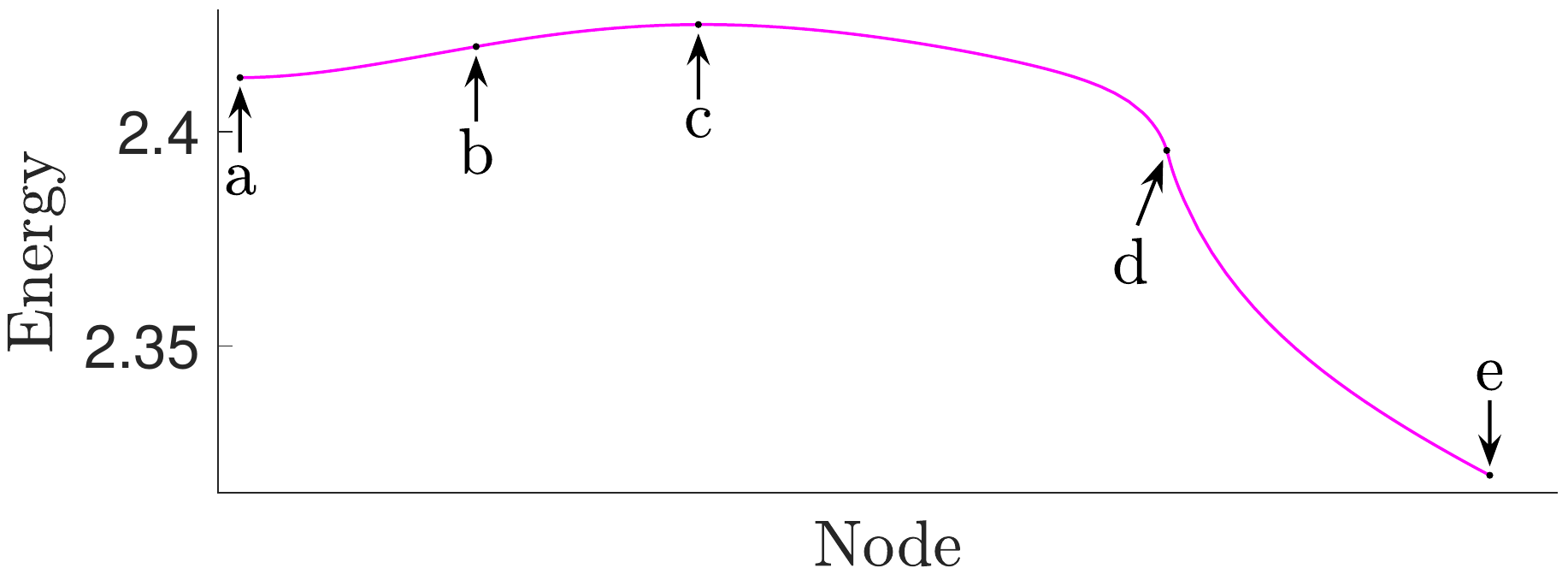}

$\overset{\includegraphics[width=58pt]{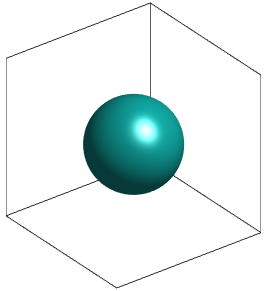}}{\begin{minipage}[t][\height][b]{24pt}\hspace{7pt}a\end{minipage}}$
$\overset{\includegraphics[width=58pt]{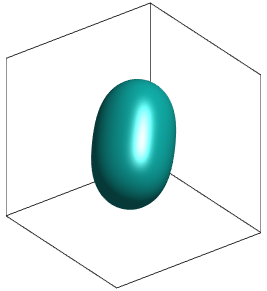}}{\begin{minipage}[t][\height][b]{24pt}\hspace{7pt}b\end{minipage}}$
$\overset{\includegraphics[width=58pt]{images/NoBoundary3D/6.0/c.png}}{\begin{minipage}[t][\height][b]{24pt}\hspace{7pt}c\end{minipage}}$
$\overset{\includegraphics[width=58pt]{images/NoBoundary3D/6.0/d.png}}{\begin{minipage}[t][\height][b]{24pt}\hspace{7pt}d\end{minipage}}$
$\overset{\includegraphics[width=58pt]{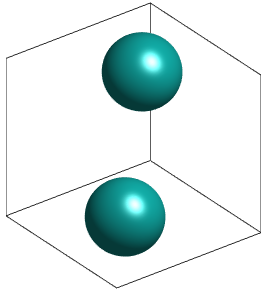}}{\begin{minipage}[t][\height][b]{24pt}\hspace{7pt}e\end{minipage}}$
\caption{Minimum energy path ($\tdfis=0.744$). Saddle point and scission point are c and d, respectively.}
\end{figure}

\begin{figure}[H]
\centering
\includegraphics[width=269.388pt]{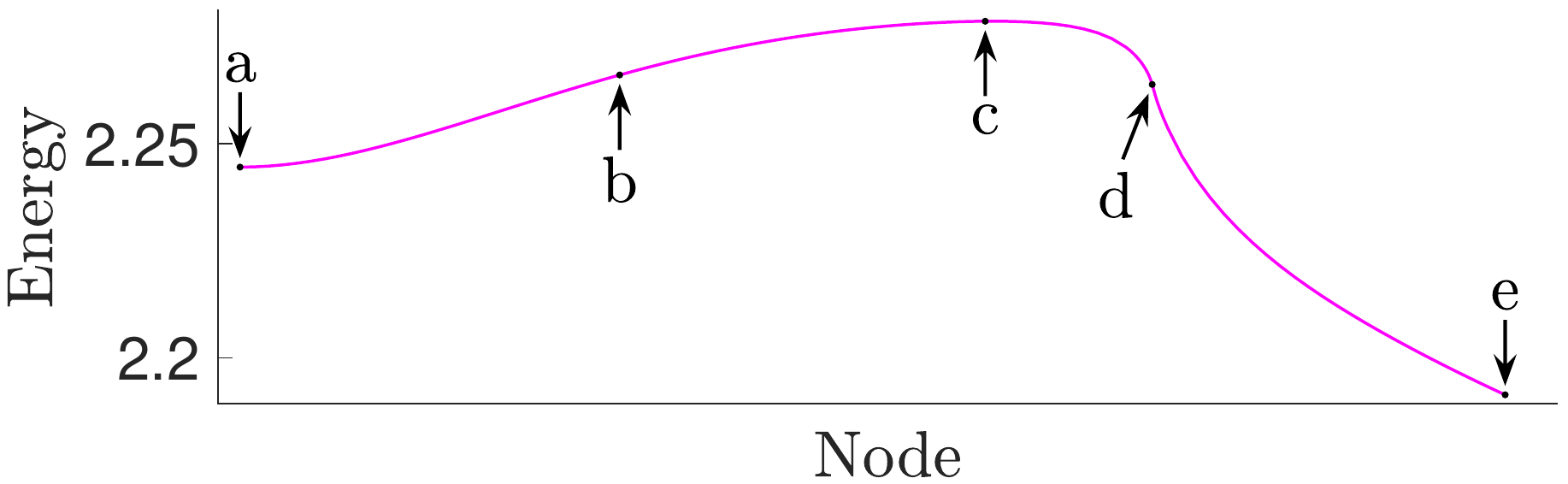}

$\overset{\includegraphics[width=58pt]{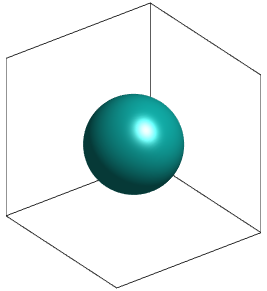}}{\begin{minipage}[t][\height][b]{24pt}\hspace{7pt}a\end{minipage}}$
$\overset{\includegraphics[width=58pt]{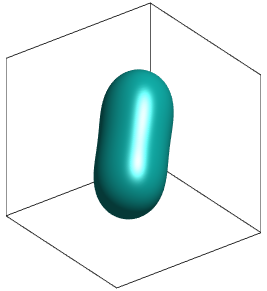}}{\begin{minipage}[t][\height][b]{24pt}\hspace{7pt}b\end{minipage}}$
$\overset{\includegraphics[width=58pt]{images/NoBoundary3D/5.3/c.png}}{\begin{minipage}[t][\height][b]{24pt}\hspace{7pt}c\end{minipage}}$
$\overset{\includegraphics[width=58pt]{images/NoBoundary3D/5.3/d.png}}{\begin{minipage}[t][\height][b]{24pt}\hspace{7pt}d\end{minipage}}$
$\overset{\includegraphics[width=58pt]{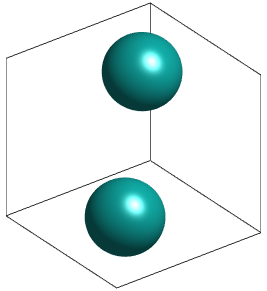}}{\begin{minipage}[t][\height][b]{24pt}\hspace{7pt}e\end{minipage}}$
\caption{Minimum energy path ($\tdfis=0.657$). Saddle point and scission point are c and d, respectively.}
\end{figure}

\begin{figure}[H]
\centering
\includegraphics[width=269.388pt]{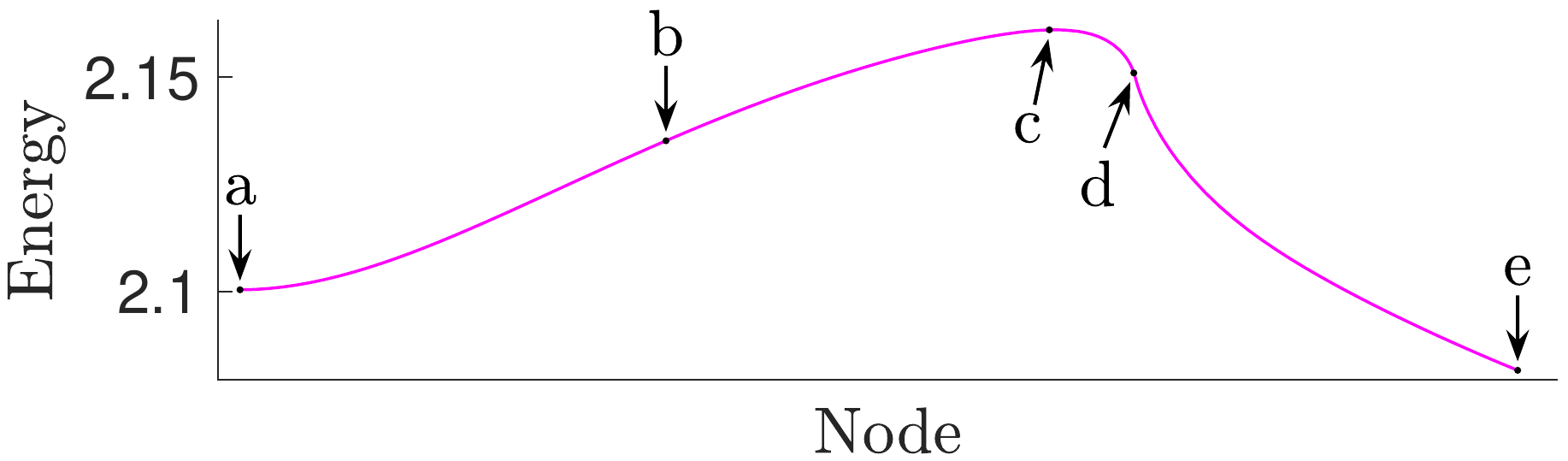}

$\overset{\includegraphics[width=58pt]{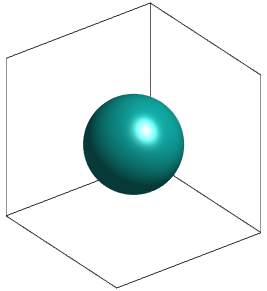}}{\begin{minipage}[t][\height][b]{24pt}\hspace{7pt}a\end{minipage}}$
$\overset{\includegraphics[width=58pt]{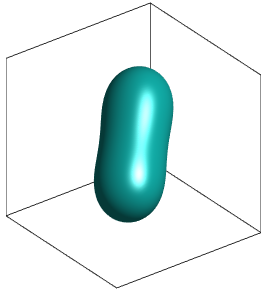}}{\begin{minipage}[t][\height][b]{24pt}\hspace{7pt}b\end{minipage}}$
$\overset{\includegraphics[width=58pt]{images/NoBoundary3D/4.7/c.png}}{\begin{minipage}[t][\height][b]{24pt}\hspace{7pt}c\end{minipage}}$
$\overset{\includegraphics[width=58pt]{images/NoBoundary3D/4.7/d.png}}{\begin{minipage}[t][\height][b]{24pt}\hspace{7pt}d\end{minipage}}$
$\overset{\includegraphics[width=58pt]{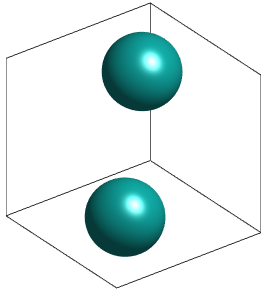}}{\begin{minipage}[t][\height][b]{24pt}\hspace{7pt}e\end{minipage}}$
\caption{Minimum energy path ($\tdfis=0.583$). Saddle point and scission point are c and d, respectively.}
\end{figure}

\begin{figure}[H]
\centering
\includegraphics[width=269.388pt]{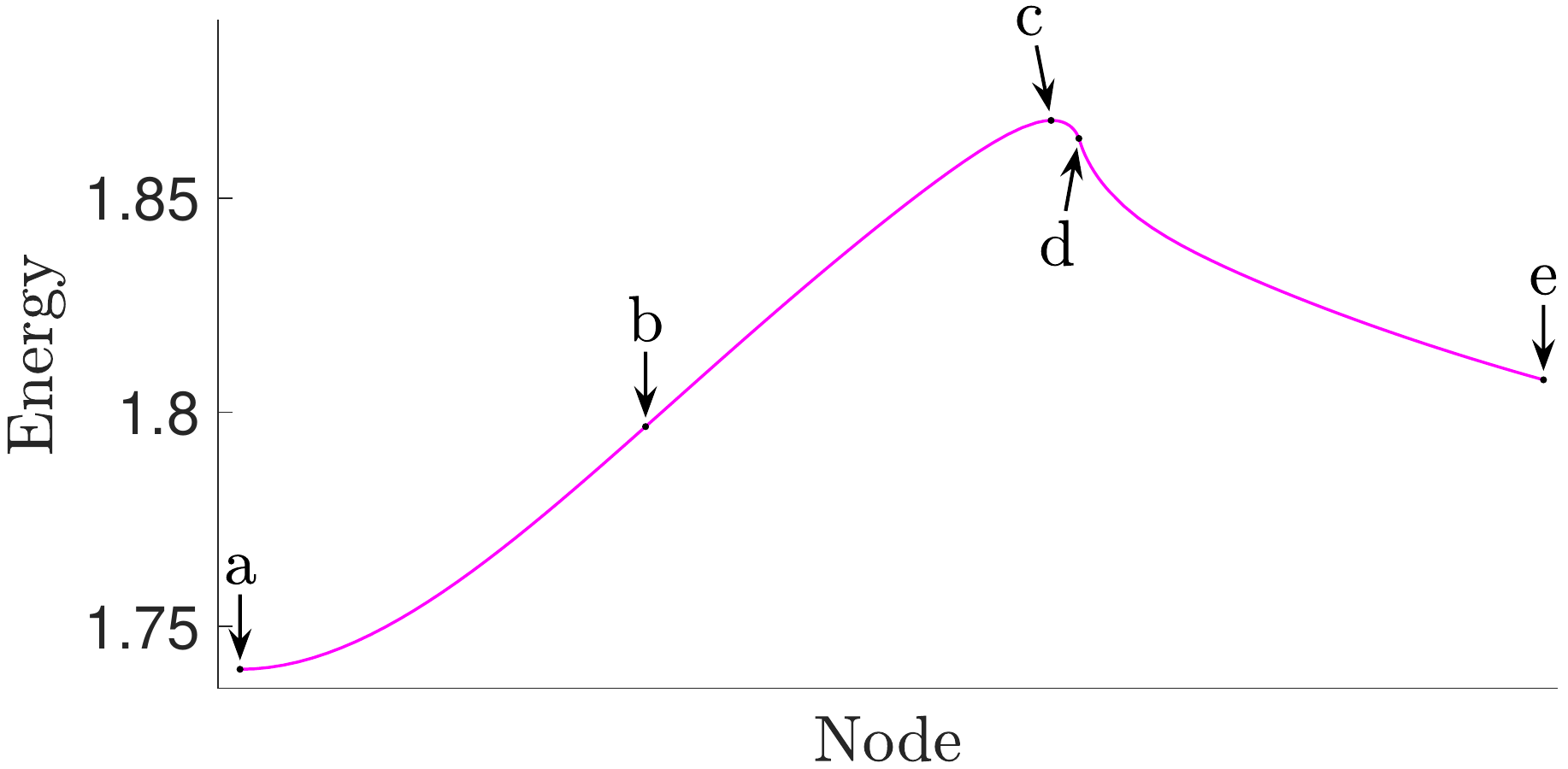}

$\overset{\includegraphics[width=58pt]{images/NoBoundary3D/3.2/a.png}}{\begin{minipage}[t][\height][b]{24pt}\hspace{7pt}a\end{minipage}}$
$\overset{\includegraphics[width=58pt]{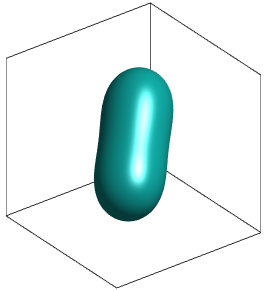}}{\begin{minipage}[t][\height][b]{24pt}\hspace{7pt}b\end{minipage}}$
$\overset{\includegraphics[width=58pt]{images/NoBoundary3D/3.2/c.png}}{\begin{minipage}[t][\height][b]{24pt}\hspace{7pt}c\end{minipage}}$
$\overset{\includegraphics[width=58pt]{images/NoBoundary3D/3.2/d.png}}{\begin{minipage}[t][\height][b]{24pt}\hspace{7pt}d\end{minipage}}$
$\overset{\includegraphics[width=58pt]{images/NoBoundary3D/3.2/e.png}}{\begin{minipage}[t][\height][b]{24pt}\hspace{7pt}e\end{minipage}}$
\caption{Minimum energy path ($\tdfis=0.397$). Saddle point and scission point are c and d, respectively.}
\end{figure}

\section{Simulations in 3-D under periodic boundary conditions}
\label{Simulations in 3-D under periodic boundary conditions}

In this appendix we present the numerical simulations of \eqref{diffuse interface version of energy} in 3-D under periodic boundary conditions. As mentioned in \eqref{reparametrization of nonlocal coefficient}, the parameter $\tdfis$ refers to $\tilde\gamma\omega|D|/(40\pi)$, where we choose $\omega|D|=0.04$. The results here are similar to those under no boundary conditions (in Section \ref{Simulation results}), and we expect the former to converge to the latter as $\omega\rightarrow0$.

We numerically study the pACOK dynamics with the initial value resembling the indicator
function of a ball. According to our simulation results, there is a critical value $\tdfis_*\in(1.086,1.092)$ such that for $\tdfis<\tdfis_*$, the ball is stable; for $\tdfis>\tdfis_*$, the ball is unstable and will eventually split into two equally large balls. Due to periodic boundary conditions, those two balls will align themselves along a body diagonal of the simulation box and keep distance from the boundaries.

We use the string method to compute the minimum energy paths. The is a critical value $\tdfis_3\in(0.463, 0.469)$ such that for $\tdfis<\tdfis_3$, the computed minimum energy path of fission is unstable against mass-asymmetric perturbations. After sufficiently many iterations in the string method, such a fission path will eventually converge to another minimum energy path, which we call the spallation path (named after the spallation process described in \cite[Page 144]{cohen1962deformation}). Along the spallation path, a tiny child ball emerges at a distance from the parent ball, and then the child ball grows while the parent ball shrinks, until they achieve equal sizes. This spallation path can be understood as the reverse process of Ostwald ripening phenomenon (or coarsening effect) \cite{niethammer1999derivation}. For $\tdfis>\tdfis_3$, we obtain two possible minimum energy paths: fission and spallation.

\begin{figure}[H]
\centering
\includegraphics[width=286.2245pt]{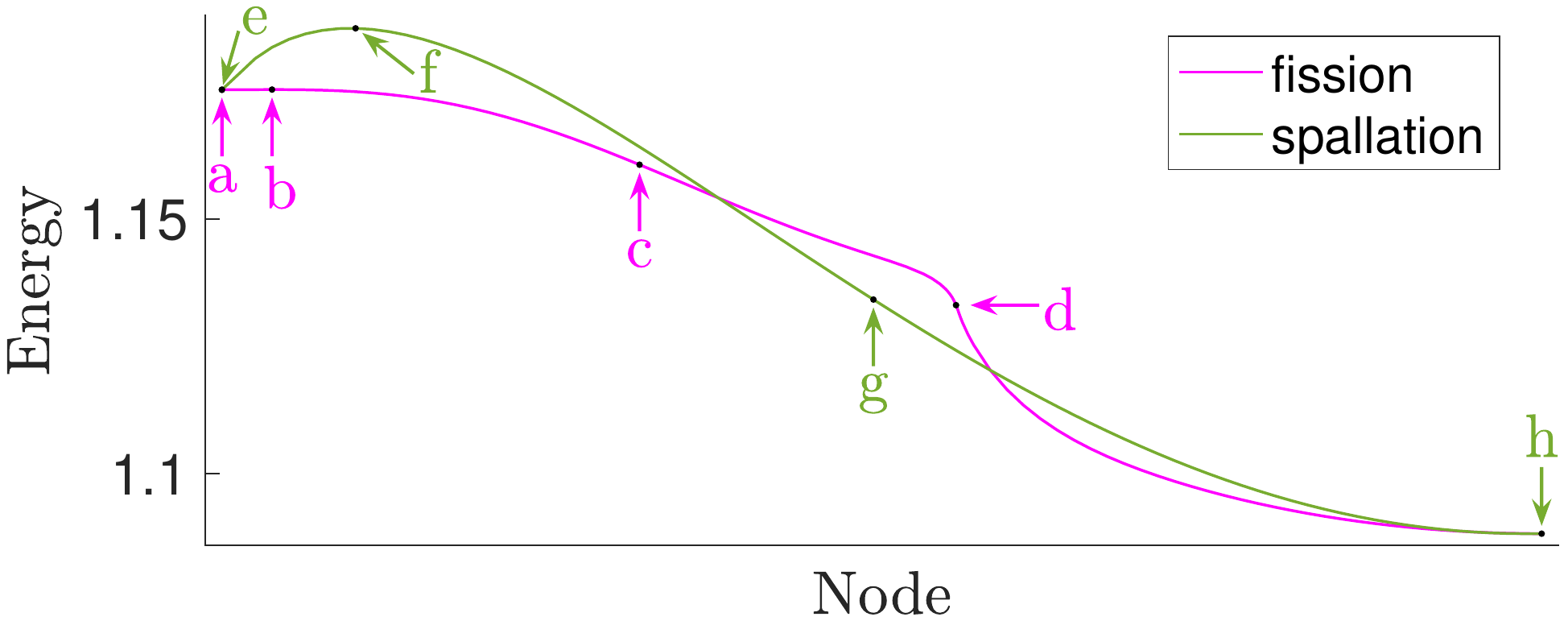}

\hspace{-30pt}
$\overset{\includegraphics[width=58pt]{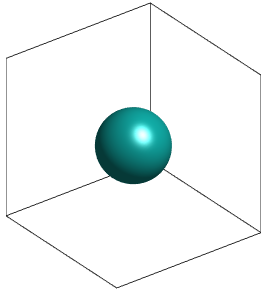}}{\begin{minipage}[t][\height][b]{24pt}\hspace{7pt}a\end{minipage}}$
$\overset{\includegraphics[width=58pt]{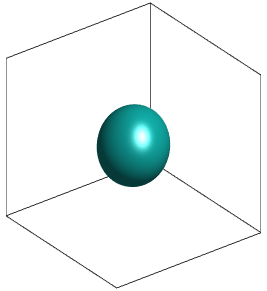}}{\begin{minipage}[t][\height][b]{24pt}\hspace{7pt}b\end{minipage}}$
$\overset{\includegraphics[width=58pt]{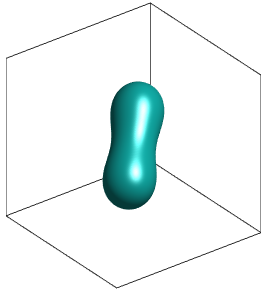}}{\begin{minipage}[t][\height][b]{24pt}\hspace{7pt}c\end{minipage}}$
$\overset{\includegraphics[width=58pt]{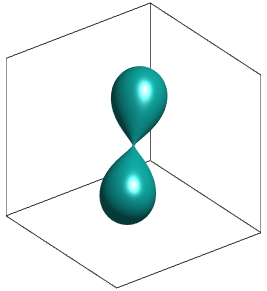}}{\begin{minipage}[t][\height][b]{24pt}\hspace{7pt}d\end{minipage}}$

\hspace{30pt}
$\overset{\includegraphics[width=58pt]{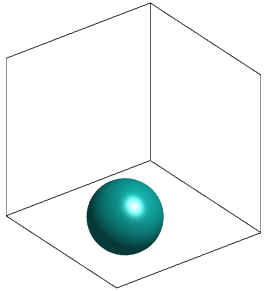}}{\begin{minipage}[t][\height][b]{24pt}\hspace{7pt}e\end{minipage}}$
$\overset{\includegraphics[width=58pt]{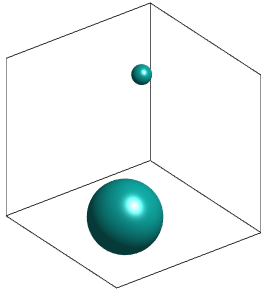}}{\begin{minipage}[t][\height][b]{24pt}\hspace{7pt}f\end{minipage}}$
$\overset{\includegraphics[width=58pt]{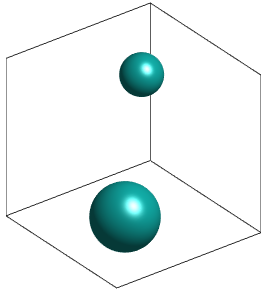}}{\begin{minipage}[t][\height][b]{24pt}\hspace{7pt}g\end{minipage}}$
$\overset{\includegraphics[width=58pt]{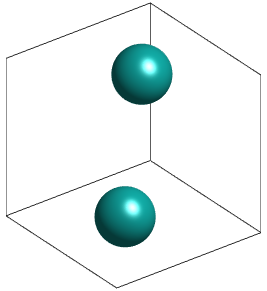}}{\begin{minipage}[t][\height][b]{24pt}\hspace{7pt}h\end{minipage}}$
\caption{Minimum energy paths ($\tdfis=1.048$). Saddle points are b and f. Scission point is d.}
\end{figure}

\begin{figure}[H]
\centering
\includegraphics[width=286.2245pt]{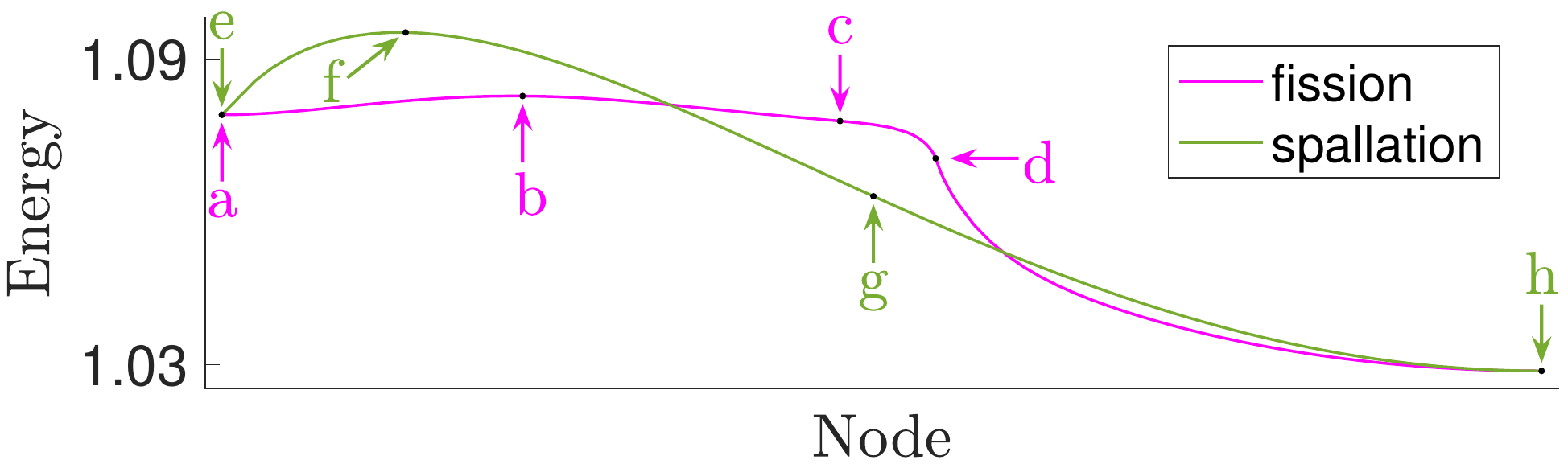}

\hspace{-30pt}
$\overset{\includegraphics[width=58pt]{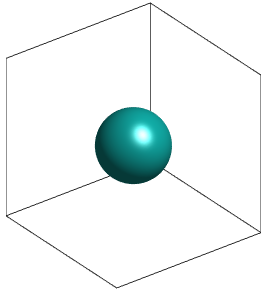}}{\begin{minipage}[t][\height][b]{24pt}\hspace{7pt}a\end{minipage}}$
$\overset{\includegraphics[width=58pt]{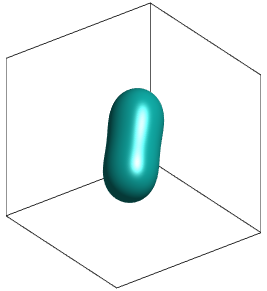}}{\begin{minipage}[t][\height][b]{24pt}\hspace{7pt}b\end{minipage}}$
$\overset{\includegraphics[width=58pt]{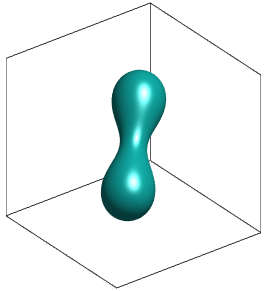}}{\begin{minipage}[t][\height][b]{24pt}\hspace{7pt}c\end{minipage}}$
$\overset{\includegraphics[width=58pt]{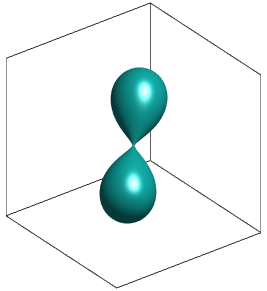}}{\begin{minipage}[t][\height][b]{24pt}\hspace{7pt}d\end{minipage}}$

\hspace{30pt}
$\overset{\includegraphics[width=58pt]{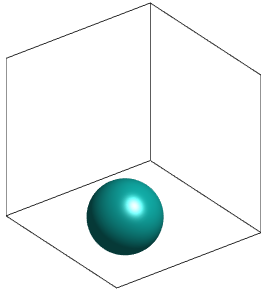}}{\begin{minipage}[t][\height][b]{24pt}\hspace{7pt}e\end{minipage}}$
$\overset{\includegraphics[width=58pt]{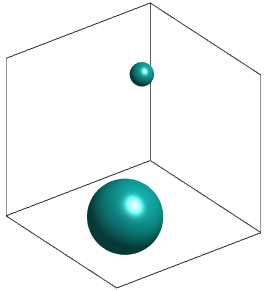}}{\begin{minipage}[t][\height][b]{24pt}\hspace{7pt}f\end{minipage}}$
$\overset{\includegraphics[width=58pt]{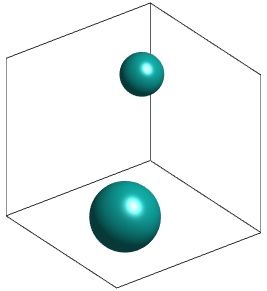}}{\begin{minipage}[t][\height][b]{24pt}\hspace{7pt}g\end{minipage}}$
$\overset{\includegraphics[width=58pt]{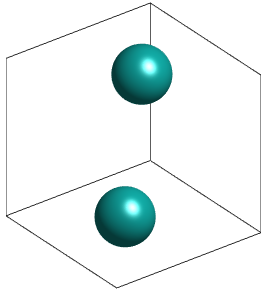}}{\begin{minipage}[t][\height][b]{24pt}\hspace{7pt}h\end{minipage}}$
\caption{Minimum energy paths ($\tdfis=0.882$). Saddle points are b and f. Scission point is d.}
\end{figure}

\begin{figure}[H]
\centering
\includegraphics[width=286.2245pt]{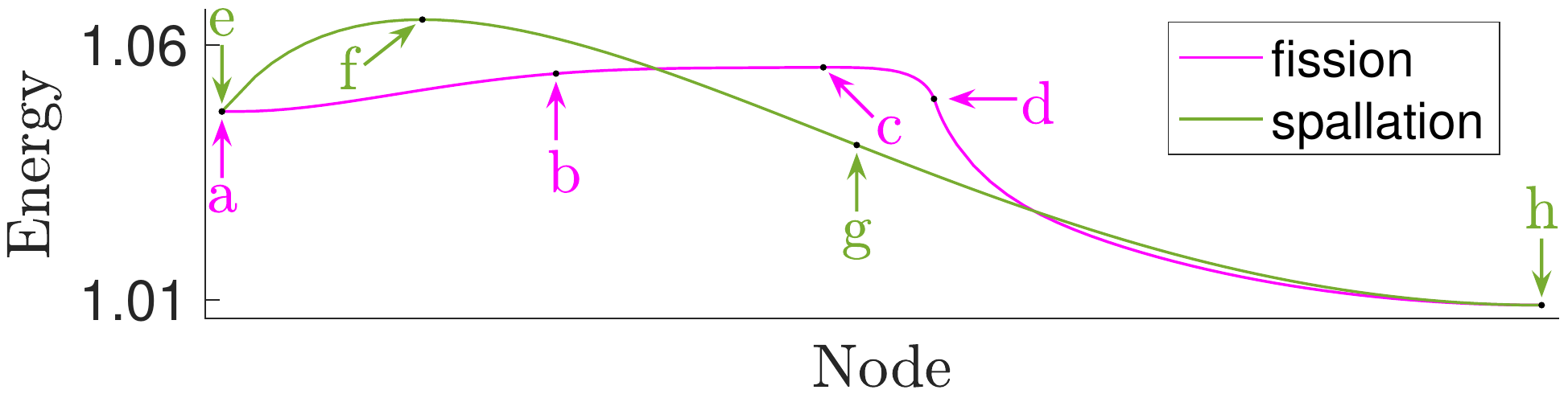}

\hspace{-30pt}
$\overset{\includegraphics[width=58pt]{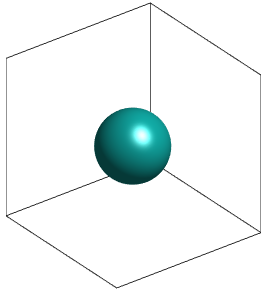}}{\begin{minipage}[t][\height][b]{24pt}\hspace{7pt}a\end{minipage}}$
$\overset{\includegraphics[width=58pt]{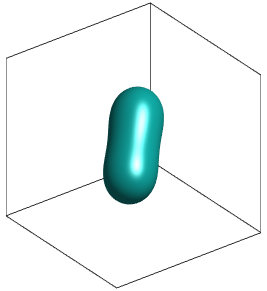}}{\begin{minipage}[t][\height][b]{24pt}\hspace{7pt}b\end{minipage}}$
$\overset{\includegraphics[width=58pt]{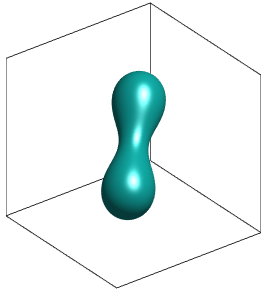}}{\begin{minipage}[t][\height][b]{24pt}\hspace{7pt}c\end{minipage}}$
$\overset{\includegraphics[width=58pt]{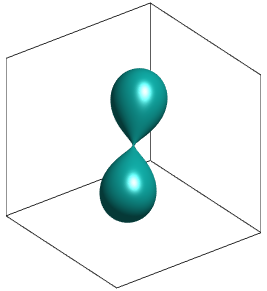}}{\begin{minipage}[t][\height][b]{24pt}\hspace{7pt}d\end{minipage}}$

\hspace{30pt}
$\overset{\includegraphics[width=58pt]{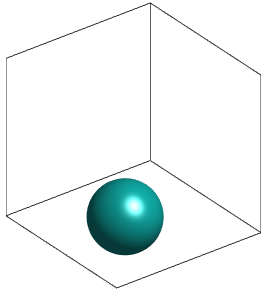}}{\begin{minipage}[t][\height][b]{24pt}\hspace{7pt}e\end{minipage}}$
$\overset{\includegraphics[width=58pt]{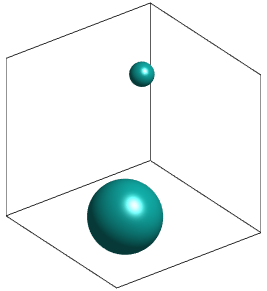}}{\begin{minipage}[t][\height][b]{24pt}\hspace{7pt}f\end{minipage}}$
$\overset{\includegraphics[width=58pt]{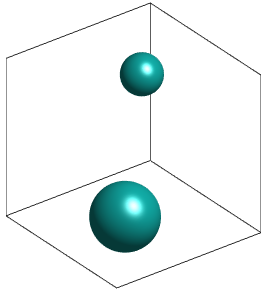}}{\begin{minipage}[t][\height][b]{24pt}\hspace{7pt}g\end{minipage}}$
$\overset{\includegraphics[width=58pt]{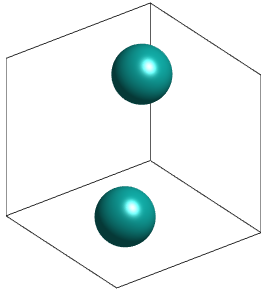}}{\begin{minipage}[t][\height][b]{24pt}\hspace{7pt}h\end{minipage}}$
\caption{Minimum energy paths ($\tdfis=0.827$). Saddle points are c and f. Scission point is d.}
\end{figure}

\begin{figure}[H]
\centering
\includegraphics[width=286.2245pt]{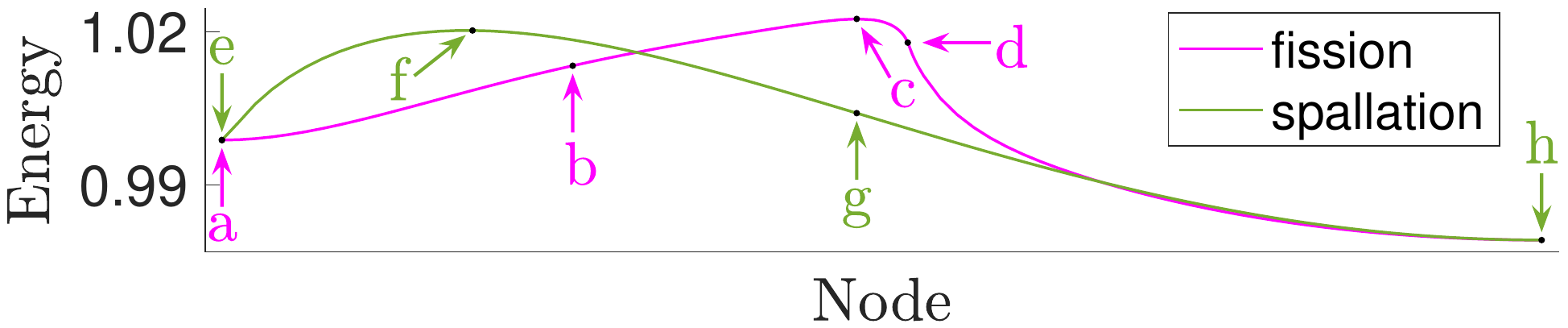}

\hspace{-30pt}
$\overset{\includegraphics[width=58pt]{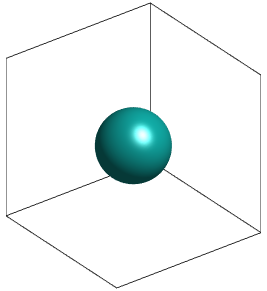}}{\begin{minipage}[t][\height][b]{24pt}\hspace{7pt}a\end{minipage}}$
$\overset{\includegraphics[width=58pt]{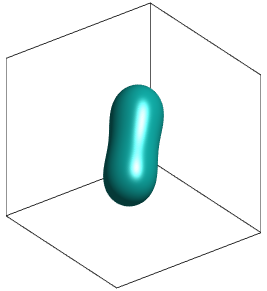}}{\begin{minipage}[t][\height][b]{24pt}\hspace{7pt}b\end{minipage}}$
$\overset{\includegraphics[width=58pt]{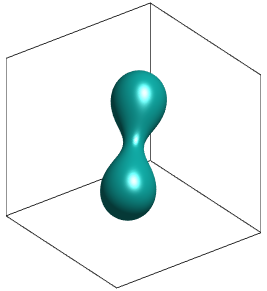}}{\begin{minipage}[t][\height][b]{24pt}\hspace{7pt}c\end{minipage}}$
$\overset{\includegraphics[width=58pt]{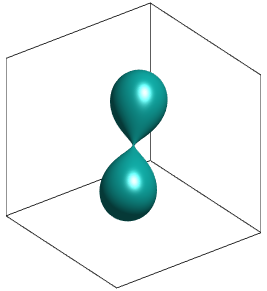}}{\begin{minipage}[t][\height][b]{24pt}\hspace{7pt}d\end{minipage}}$

\hspace{30pt}
$\overset{\includegraphics[width=58pt]{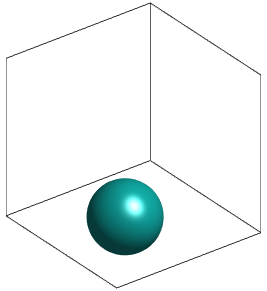}}{\begin{minipage}[t][\height][b]{24pt}\hspace{7pt}e\end{minipage}}$
$\overset{\includegraphics[width=58pt]{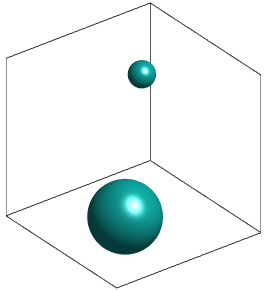}}{\begin{minipage}[t][\height][b]{24pt}\hspace{7pt}f\end{minipage}}$
$\overset{\includegraphics[width=58pt]{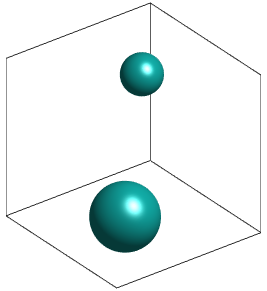}}{\begin{minipage}[t][\height][b]{24pt}\hspace{7pt}g\end{minipage}}$
$\overset{\includegraphics[width=58pt]{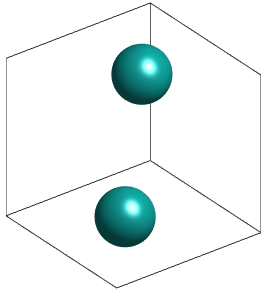}}{\begin{minipage}[t][\height][b]{24pt}\hspace{7pt}h\end{minipage}}$
\caption{Minimum energy paths ($\tdfis=0.744$). Saddle points are c and f. Scission point is d.}
\end{figure}

\begin{figure}[H]
\centering
\includegraphics[width=286.2245pt]{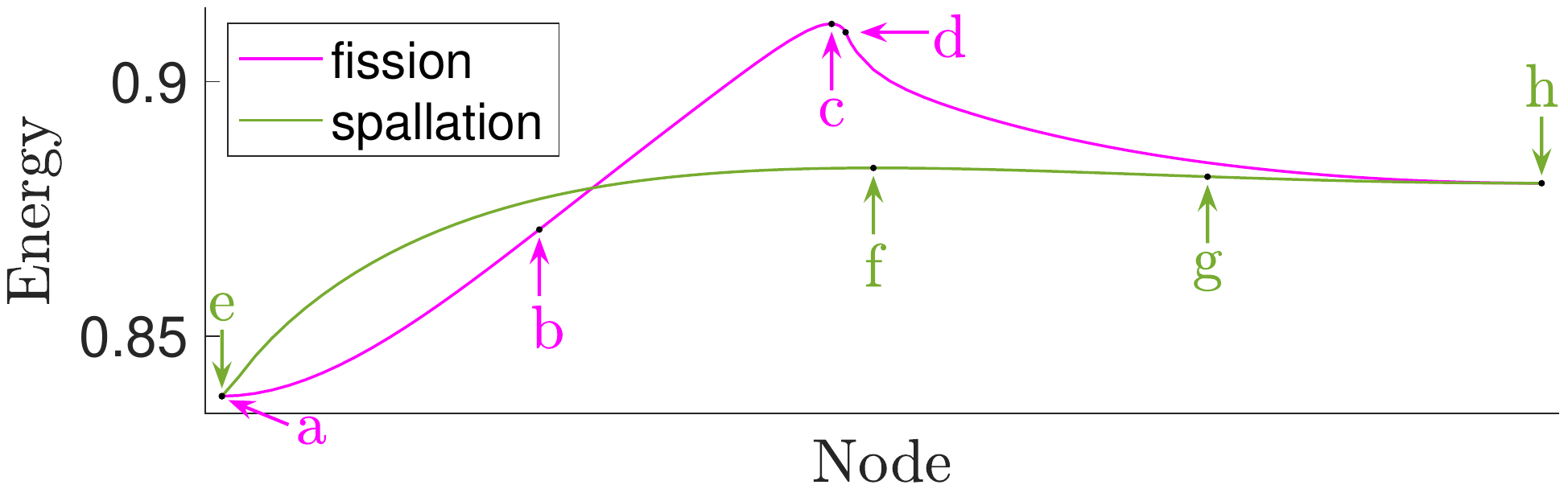}

\hspace{-30pt}
$\overset{\includegraphics[width=58pt]{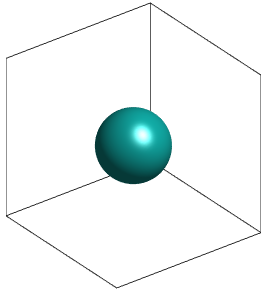}}{\begin{minipage}[t][\height][b]{24pt}\hspace{7pt}a\end{minipage}}$
$\overset{\includegraphics[width=58pt]{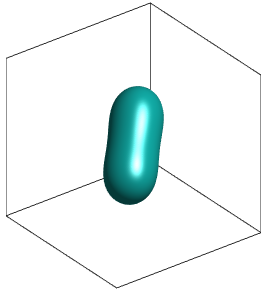}}{\begin{minipage}[t][\height][b]{24pt}\hspace{7pt}b\end{minipage}}$
$\overset{\includegraphics[width=58pt]{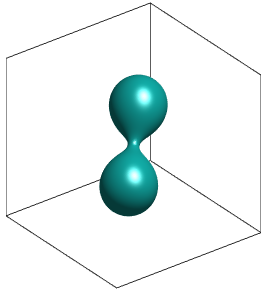}}{\begin{minipage}[t][\height][b]{24pt}\hspace{7pt}c\end{minipage}}$
$\overset{\includegraphics[width=58pt]{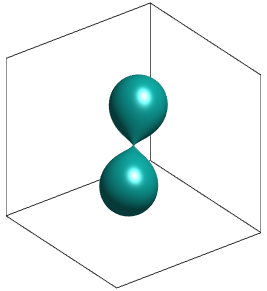}}{\begin{minipage}[t][\height][b]{24pt}\hspace{7pt}d\end{minipage}}$

\hspace{30pt}
$\overset{\includegraphics[width=58pt]{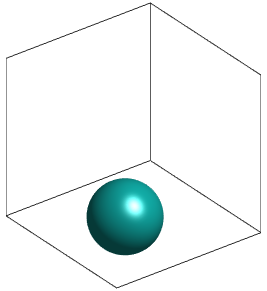}}{\begin{minipage}[t][\height][b]{24pt}\hspace{7pt}e\end{minipage}}$
$\overset{\includegraphics[width=58pt]{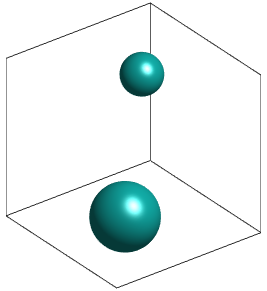}}{\begin{minipage}[t][\height][b]{24pt}\hspace{7pt}f\end{minipage}}$
$\overset{\includegraphics[width=58pt]{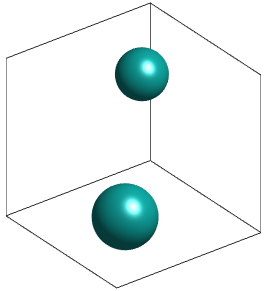}}{\begin{minipage}[t][\height][b]{24pt}\hspace{7pt}g\end{minipage}}$
$\overset{\includegraphics[width=58pt]{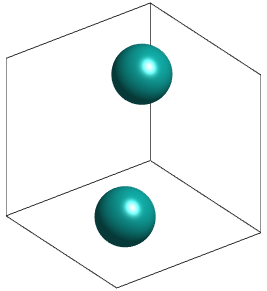}}{\begin{minipage}[t][\height][b]{24pt}\hspace{7pt}h\end{minipage}}$
\caption{Minimum energy paths ($\tdfis=0.469$). Saddle points are c and f. Scission point is d.}
\end{figure}

\section{Simulations in 2-D under no boundary conditions}
\label{Simulations in 2-D under no boundary conditions}

In this appendix we present the detailed numerical results for the minimum energy paths of fission in 2-D under no boundary conditions.

\begin{figure}[H]
\centering
\includegraphics[width=396.7347pt]{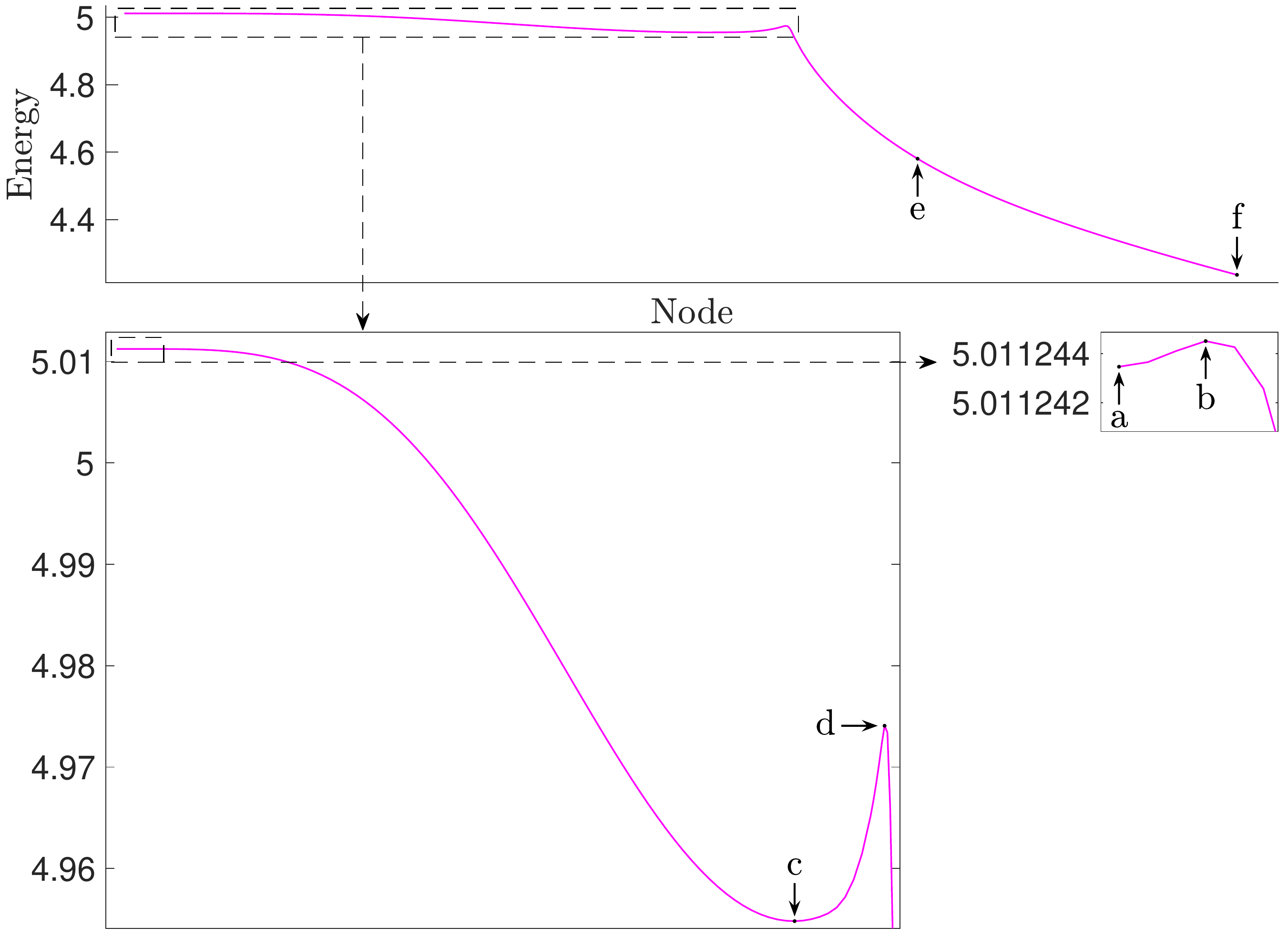}

$\overset{\includegraphics[width=53pt]{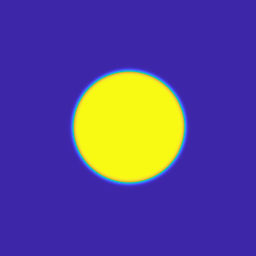}}{\begin{minipage}[t][\height][b]{24pt}\hspace{10pt}a\end{minipage}}$
$\overset{\includegraphics[width=53pt]{images/NoBoundary2D/4.69/b.png}}{\begin{minipage}[t][\height][b]{24pt}\hspace{10pt}b\end{minipage}}$
$\overset{\includegraphics[width=53pt]{images/NoBoundary2D/4.69/c.png}}{\begin{minipage}[t][\height][b]{24pt}\hspace{10pt}c\end{minipage}}$
$\overset{\includegraphics[width=53pt]{images/NoBoundary2D/4.69/d.png}}{\begin{minipage}[t][\height][b]{24pt}\hspace{10pt}d\end{minipage}}$
$\overset{\includegraphics[width=53pt]{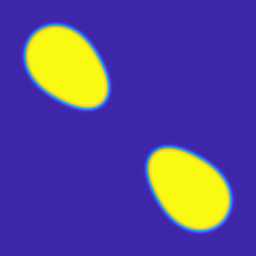}}{\begin{minipage}[t][\height][b]{24pt}\hspace{10pt}e\end{minipage}}$
$\overset{\includegraphics[width=53pt]{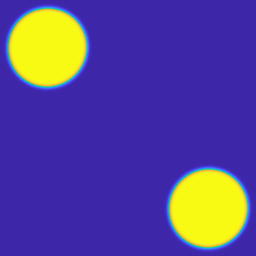}}{\begin{minipage}[t][\height][b]{24pt}\hspace{10pt}f\end{minipage}}$
\caption{Minimum energy path ($\tdfis\!=\!0.999$). Local minimizers are a and c. Transition states are b and d.}
\end{figure}

\begin{figure}[H]
\centering
\includegraphics[width=396.7347pt]{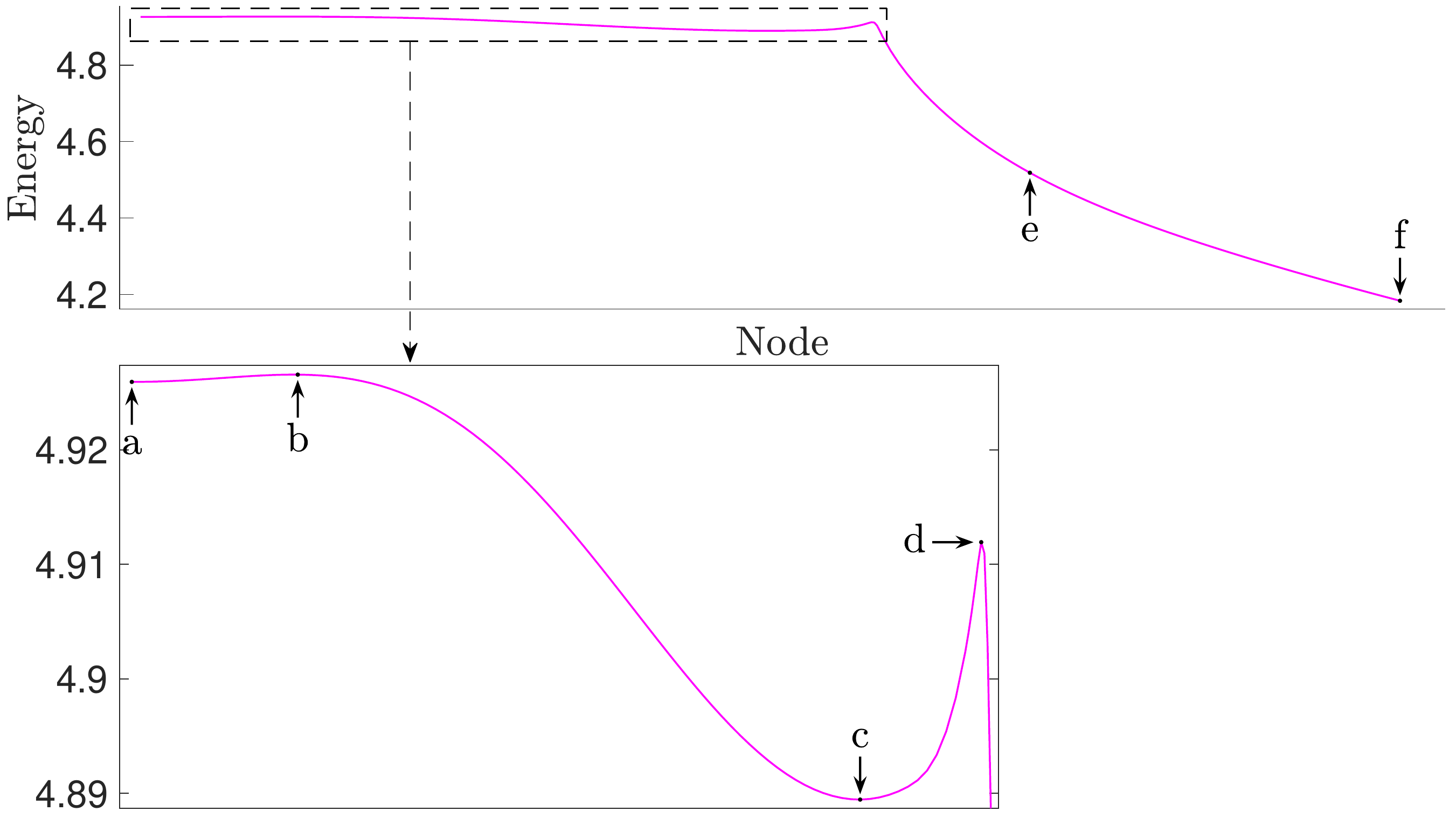}

$\overset{\includegraphics[width=53pt]{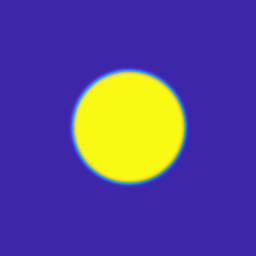}}{\begin{minipage}[t][\height][b]{24pt}\hspace{10pt}a\end{minipage}}$
$\overset{\includegraphics[width=53pt]{images/NoBoundary2D/4.58/b.png}}{\begin{minipage}[t][\height][b]{24pt}\hspace{10pt}b\end{minipage}}$
$\overset{\includegraphics[width=53pt]{images/NoBoundary2D/4.58/c.png}}{\begin{minipage}[t][\height][b]{24pt}\hspace{10pt}c\end{minipage}}$
$\overset{\includegraphics[width=53pt]{images/NoBoundary2D/4.58/d.png}}{\begin{minipage}[t][\height][b]{24pt}\hspace{10pt}d\end{minipage}}$
$\overset{\includegraphics[width=53pt]{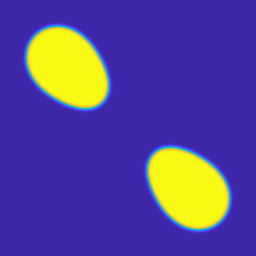}}{\begin{minipage}[t][\height][b]{24pt}\hspace{10pt}e\end{minipage}}$
$\overset{\includegraphics[width=53pt]{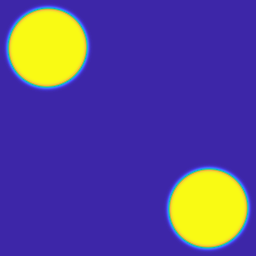}}{\begin{minipage}[t][\height][b]{24pt}\hspace{10pt}f\end{minipage}}$
\caption{Minimum energy path ($\tdfis\!=\!0.975$). Local minimizers are a and c. Transition states are b and d.}
\end{figure}

\begin{figure}[H]
\centering
\includegraphics[width=396.7347pt]{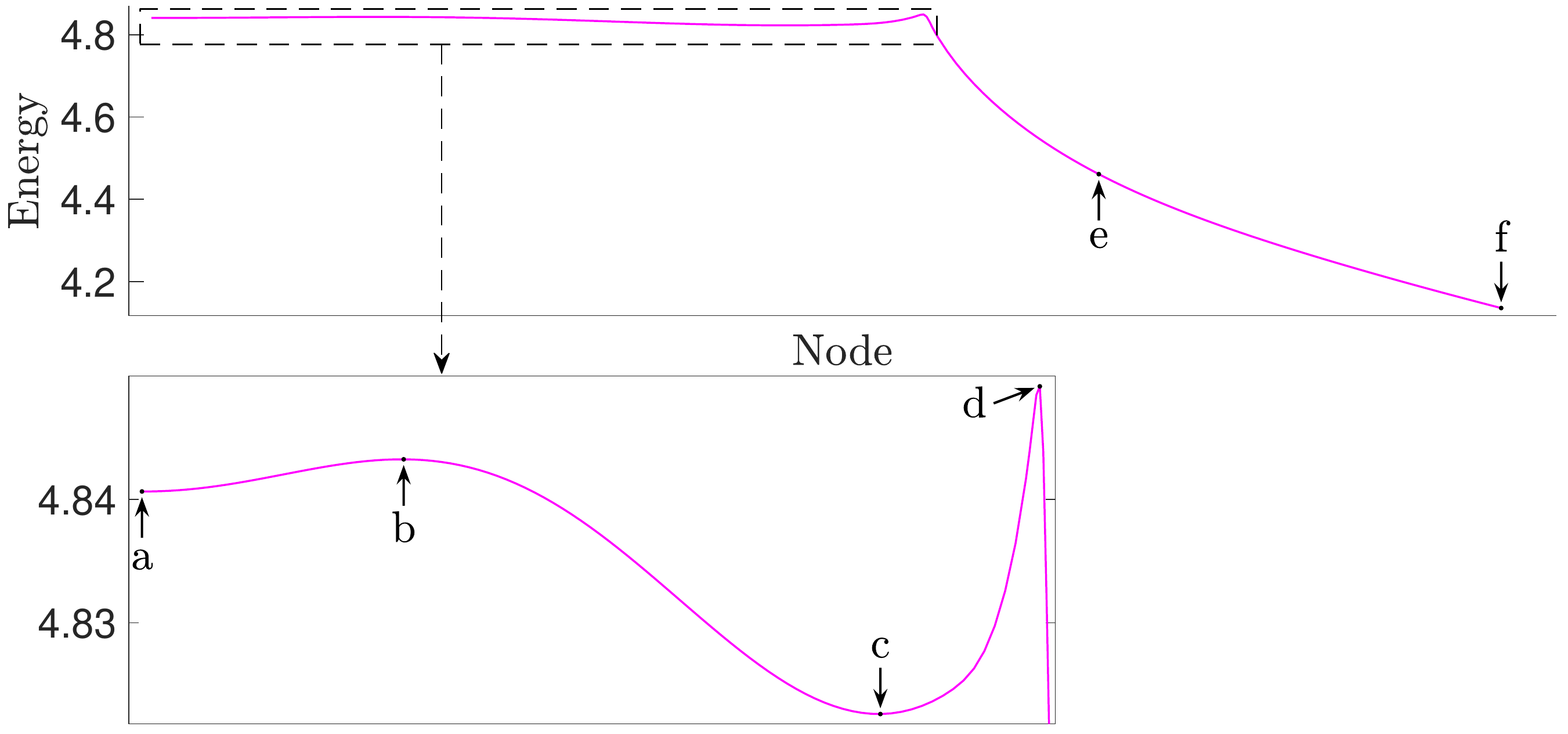}

$\overset{\includegraphics[width=53pt]{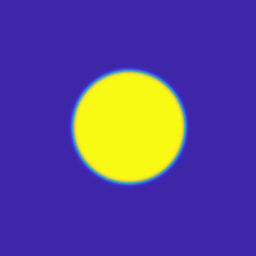}}{\begin{minipage}[t][\height][b]{24pt}\hspace{10pt}a\end{minipage}}$
$\overset{\includegraphics[width=53pt]{images/NoBoundary2D/4.47/b.png}}{\begin{minipage}[t][\height][b]{24pt}\hspace{10pt}b\end{minipage}}$
$\overset{\includegraphics[width=53pt]{images/NoBoundary2D/4.47/c.png}}{\begin{minipage}[t][\height][b]{24pt}\hspace{10pt}c\end{minipage}}$
$\overset{\includegraphics[width=53pt]{images/NoBoundary2D/4.47/d.png}}{\begin{minipage}[t][\height][b]{24pt}\hspace{10pt}d\end{minipage}}$
$\overset{\includegraphics[width=53pt]{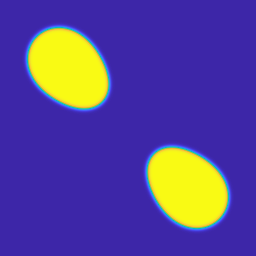}}{\begin{minipage}[t][\height][b]{24pt}\hspace{10pt}e\end{minipage}}$
$\overset{\includegraphics[width=53pt]{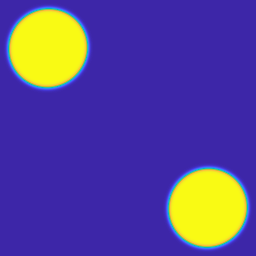}}{\begin{minipage}[t][\height][b]{24pt}\hspace{10pt}f\end{minipage}}$
\caption{Minimum energy path ($\tdfis\!=\!0.952$). Local minimizers are a and c. Transition states are b and d.}
\end{figure}

\begin{figure}[H]
\centering
\includegraphics[width=396.7347pt]{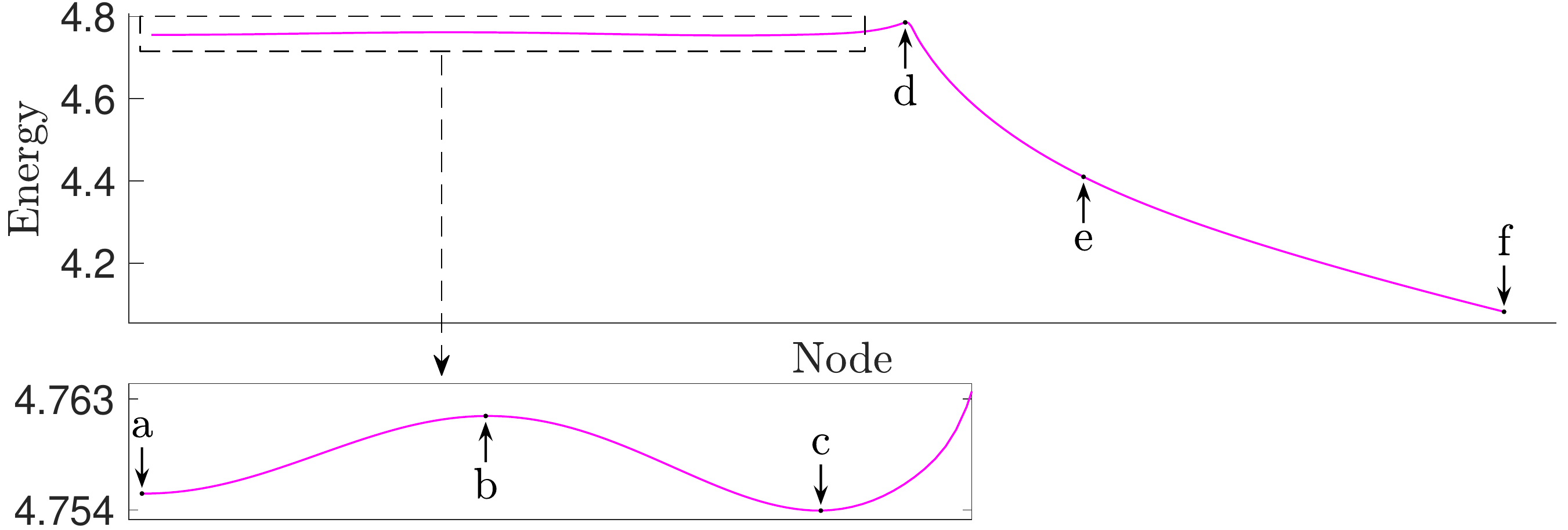}

$\overset{\includegraphics[width=53pt]{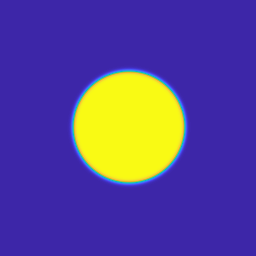}}{\begin{minipage}[t][\height][b]{24pt}\hspace{10pt}a\end{minipage}}$
$\overset{\includegraphics[width=53pt]{images/NoBoundary2D/4.36/b.png}}{\begin{minipage}[t][\height][b]{24pt}\hspace{10pt}b\end{minipage}}$
$\overset{\includegraphics[width=53pt]{images/NoBoundary2D/4.36/c.png}}{\begin{minipage}[t][\height][b]{24pt}\hspace{10pt}c\end{minipage}}$
$\overset{\includegraphics[width=53pt]{images/NoBoundary2D/4.36/d.png}}{\begin{minipage}[t][\height][b]{24pt}\hspace{10pt}d\end{minipage}}$
$\overset{\includegraphics[width=53pt]{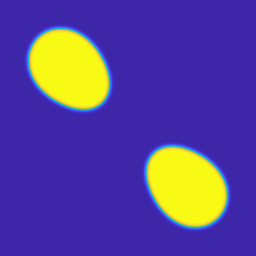}}{\begin{minipage}[t][\height][b]{24pt}\hspace{10pt}e\end{minipage}}$
$\overset{\includegraphics[width=53pt]{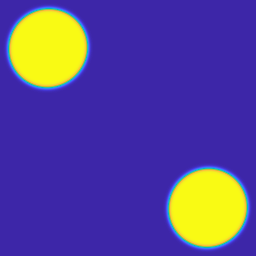}}{\begin{minipage}[t][\height][b]{24pt}\hspace{10pt}f\end{minipage}}$
\caption{Minimum energy path ($\tdfis\!=\!0.929$). Local minimizers are a and c. Transition states are b and d.}
\end{figure}

\begin{figure}[H]
\centering
\includegraphics[width=396.7347pt]{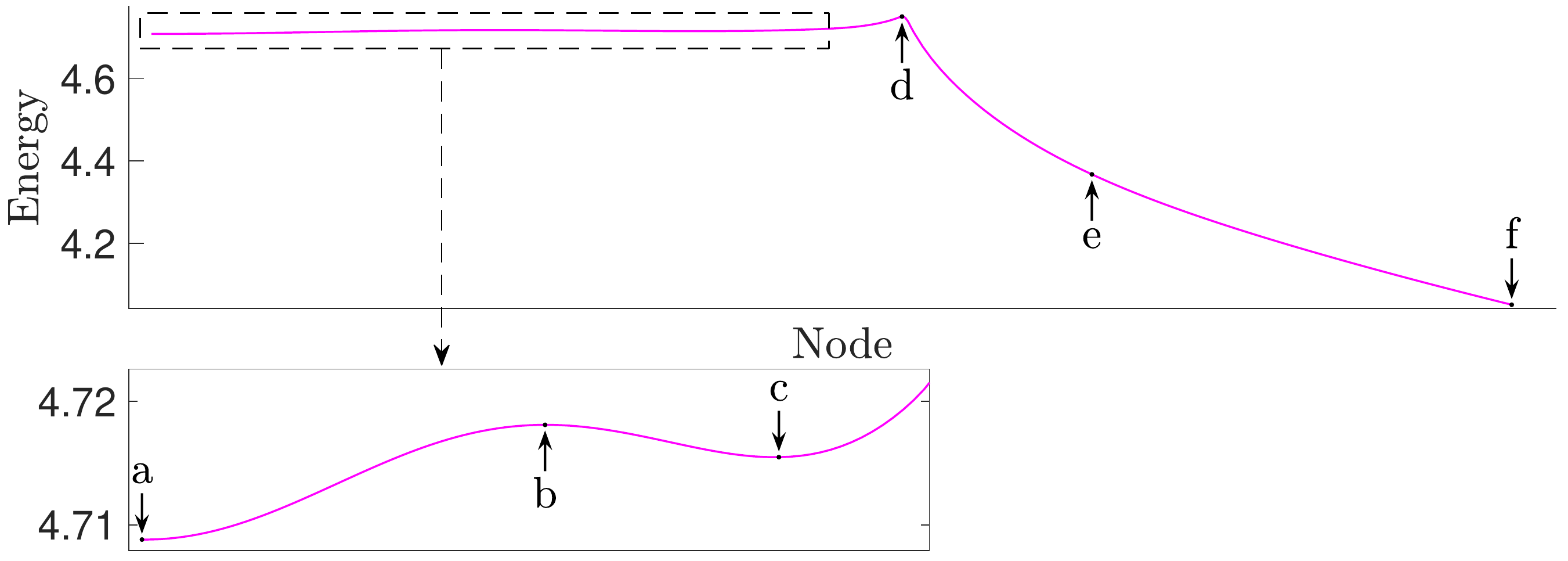}

$\overset{\includegraphics[width=53pt]{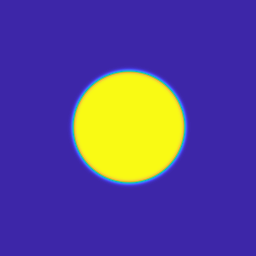}}{\begin{minipage}[t][\height][b]{24pt}\hspace{10pt}a\end{minipage}}$
$\overset{\includegraphics[width=53pt]{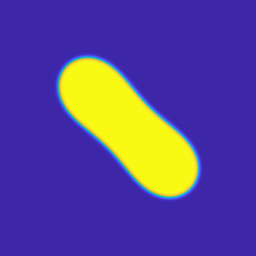}}{\begin{minipage}[t][\height][b]{24pt}\hspace{10pt}b\end{minipage}}$
$\overset{\includegraphics[width=53pt]{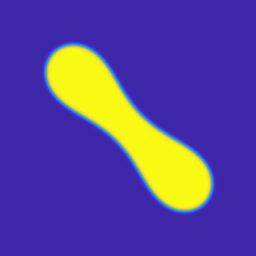}}{\begin{minipage}[t][\height][b]{24pt}\hspace{10pt}c\end{minipage}}$
$\overset{\includegraphics[width=53pt]{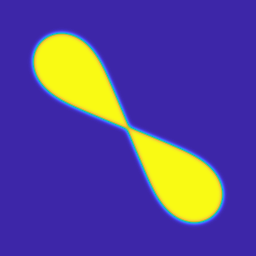}}{\begin{minipage}[t][\height][b]{24pt}\hspace{10pt}d\end{minipage}}$
$\overset{\includegraphics[width=53pt]{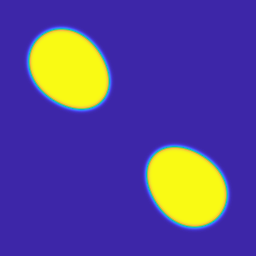}}{\begin{minipage}[t][\height][b]{24pt}\hspace{10pt}e\end{minipage}}$
$\overset{\includegraphics[width=53pt]{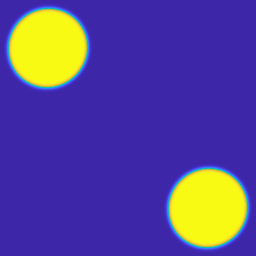}}{\begin{minipage}[t][\height][b]{24pt}\hspace{10pt}f\end{minipage}}$
\caption{Minimum energy path ($\tdfis\!=\!0.916$). Local minimizers are a and c. Transition states are b and d.}
\end{figure}

\begin{figure}[H]
\centering
\includegraphics[width=396.7347pt]{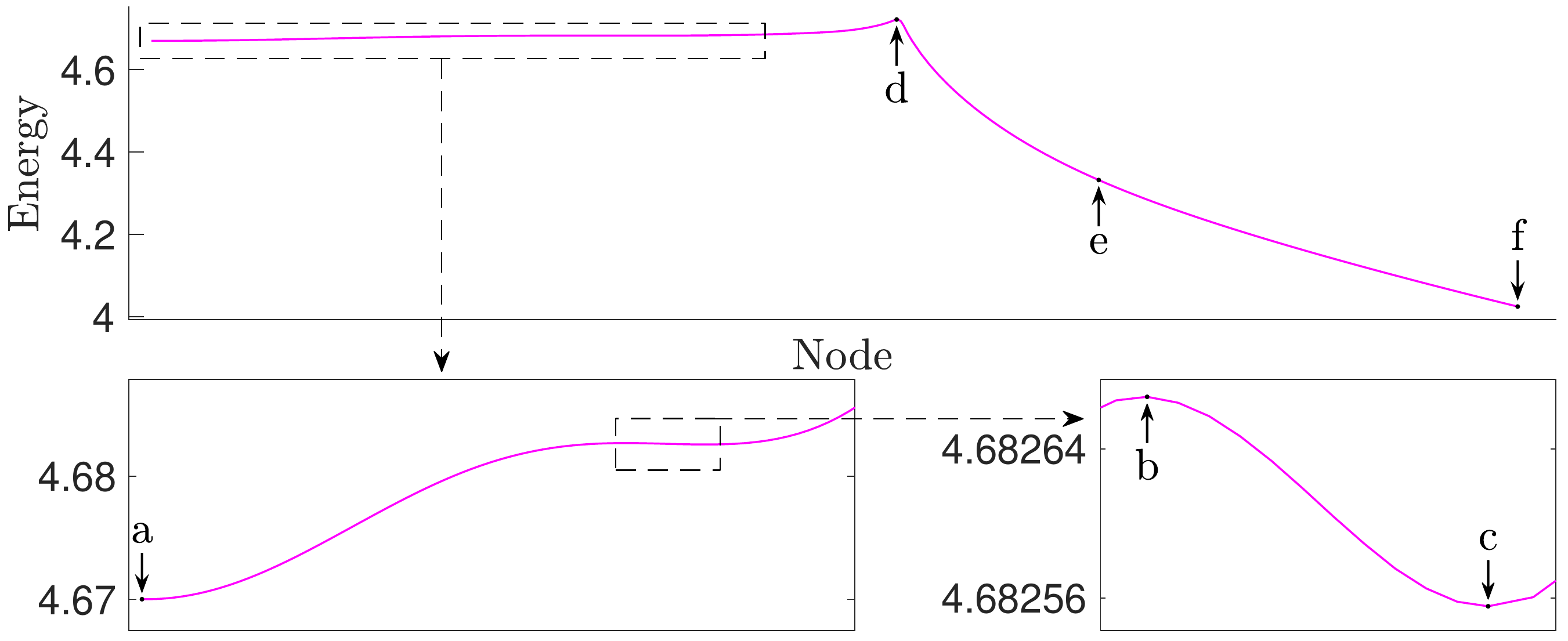}

$\overset{\includegraphics[width=53pt]{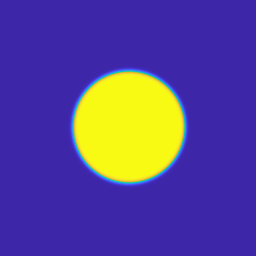}}{\begin{minipage}[t][\height][b]{24pt}\hspace{10pt}a\end{minipage}}$
$\overset{\includegraphics[width=53pt]{images/NoBoundary2D/4.25/b.png}}{\begin{minipage}[t][\height][b]{24pt}\hspace{10pt}b\end{minipage}}$
$\overset{\includegraphics[width=53pt]{images/NoBoundary2D/4.25/c.png}}{\begin{minipage}[t][\height][b]{24pt}\hspace{10pt}c\end{minipage}}$
$\overset{\includegraphics[width=53pt]{images/NoBoundary2D/4.25/d.png}}{\begin{minipage}[t][\height][b]{24pt}\hspace{10pt}d\end{minipage}}$
$\overset{\includegraphics[width=53pt]{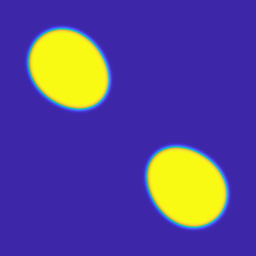}}{\begin{minipage}[t][\height][b]{24pt}\hspace{10pt}e\end{minipage}}$
$\overset{\includegraphics[width=53pt]{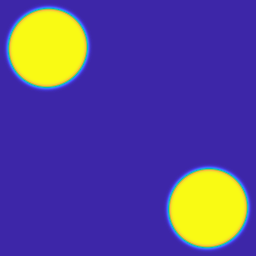}}{\begin{minipage}[t][\height][b]{24pt}\hspace{10pt}f\end{minipage}}$
\caption{Minimum energy path ($\tdfis\!=\!0.905$). Local minimizers are a and c. Transition states are b and d.}
\end{figure}

\begin{figure}[H]
\centering
\includegraphics[width=396.7347pt]{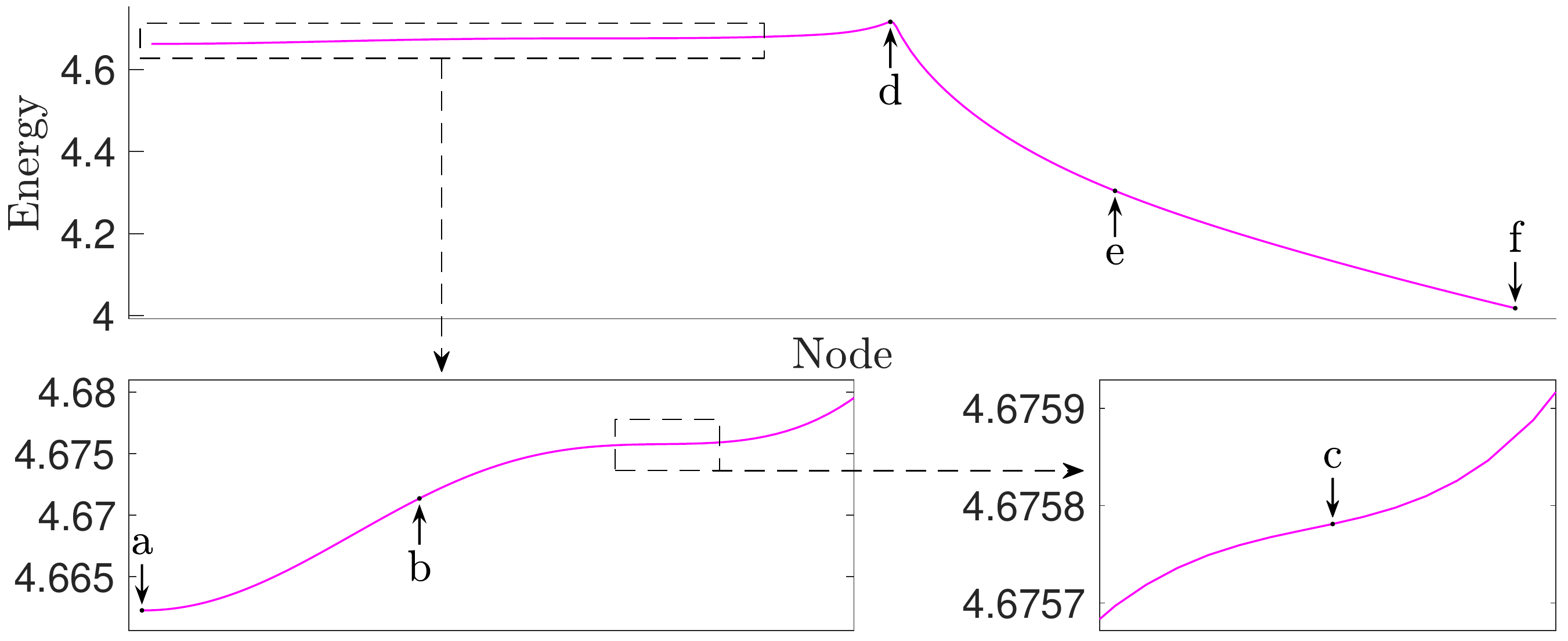}

$\overset{\includegraphics[width=53pt]{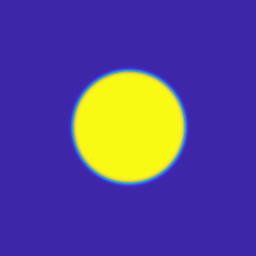}}{\begin{minipage}[t][\height][b]{24pt}\hspace{10pt}a\end{minipage}}$
$\overset{\includegraphics[width=53pt]{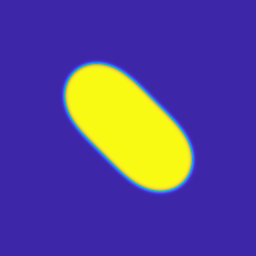}}{\begin{minipage}[t][\height][b]{24pt}\hspace{10pt}b\end{minipage}}$
$\overset{\includegraphics[width=53pt]{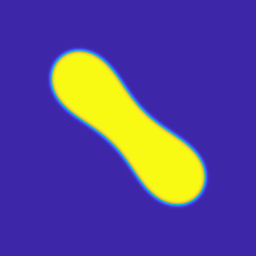}}{\begin{minipage}[t][\height][b]{24pt}\hspace{10pt}c\end{minipage}}$
$\overset{\includegraphics[width=53pt]{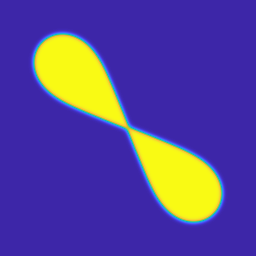}}{\begin{minipage}[t][\height][b]{24pt}\hspace{10pt}d\end{minipage}}$
$\overset{\includegraphics[width=53pt]{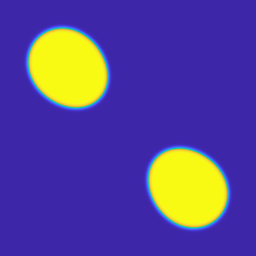}}{\begin{minipage}[t][\height][b]{24pt}\hspace{10pt}e\end{minipage}}$
$\overset{\includegraphics[width=53pt]{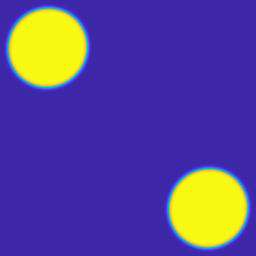}}{\begin{minipage}[t][\height][b]{24pt}\hspace{10pt}f\end{minipage}}$
\caption{Minimum energy path ($\tdfis\!=\!0.903$). Transition state is d.}
\end{figure}

\begin{figure}[H]
\centering
\includegraphics[width=396.7347pt]{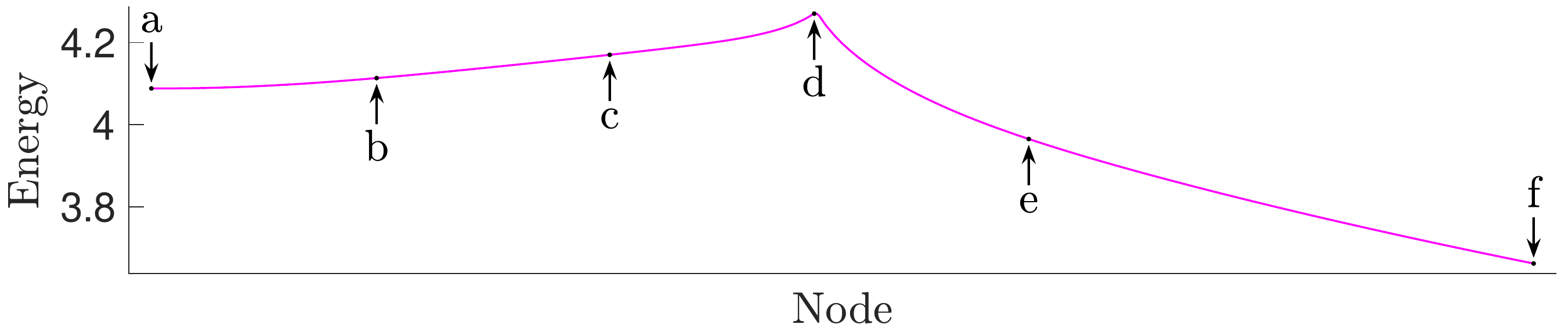}

$\overset{\includegraphics[width=53pt]{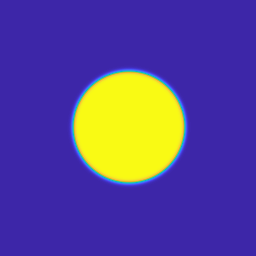}}{\begin{minipage}[t][\height][b]{24pt}\hspace{10pt}a\end{minipage}}$
$\overset{\includegraphics[width=53pt]{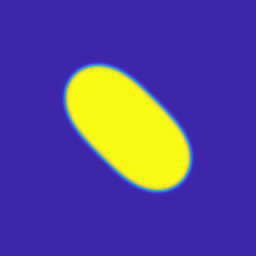}}{\begin{minipage}[t][\height][b]{24pt}\hspace{10pt}b\end{minipage}}$
$\overset{\includegraphics[width=53pt]{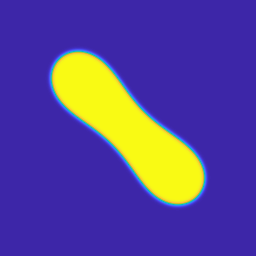}}{\begin{minipage}[t][\height][b]{24pt}\hspace{10pt}c\end{minipage}}$
$\overset{\includegraphics[width=53pt]{images/NoBoundary2D/3.50/d.png}}{\begin{minipage}[t][\height][b]{24pt}\hspace{10pt}d\end{minipage}}$
$\overset{\includegraphics[width=53pt]{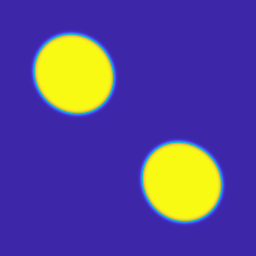}}{\begin{minipage}[t][\height][b]{24pt}\hspace{10pt}e\end{minipage}}$
$\overset{\includegraphics[width=53pt]{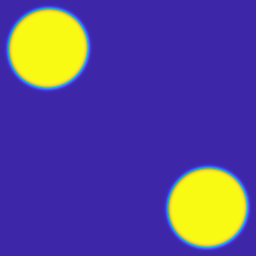}}{\begin{minipage}[t][\height][b]{24pt}\hspace{10pt}f\end{minipage}}$
\caption{Minimum energy path ($\tdfis\!=\!0.745$). Transition state is d.}
\end{figure}

\begin{figure}[H]
\centering
\includegraphics[width=396.7347pt]{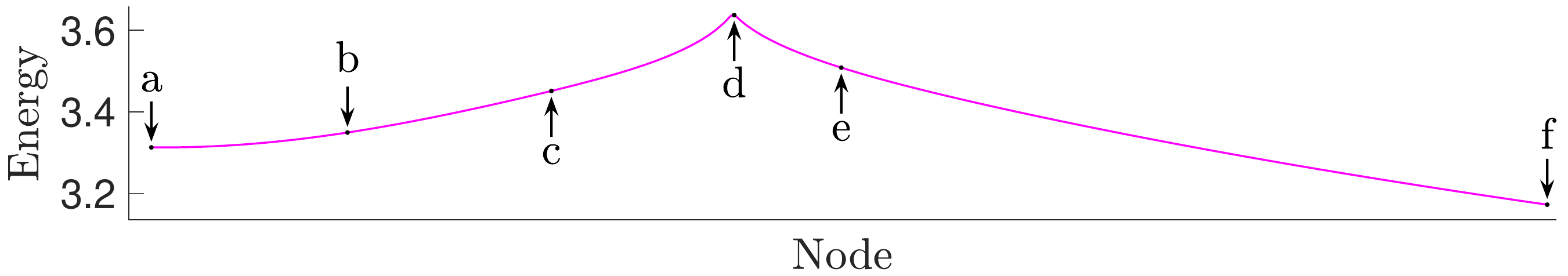}

$\overset{\includegraphics[width=53pt]{images/NoBoundary2D/2.50/a.png}}{\begin{minipage}[t][\height][b]{24pt}\hspace{10pt}a\end{minipage}}$
$\overset{\includegraphics[width=53pt]{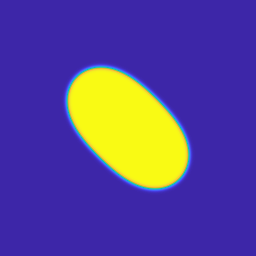}}{\begin{minipage}[t][\height][b]{24pt}\hspace{10pt}b\end{minipage}}$
$\overset{\includegraphics[width=53pt]{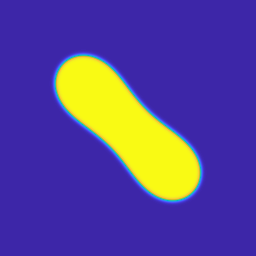}}{\begin{minipage}[t][\height][b]{24pt}\hspace{10pt}c\end{minipage}}$
$\overset{\includegraphics[width=53pt]{images/NoBoundary2D/2.50/d.png}}{\begin{minipage}[t][\height][b]{24pt}\hspace{10pt}d\end{minipage}}$
$\overset{\includegraphics[width=53pt]{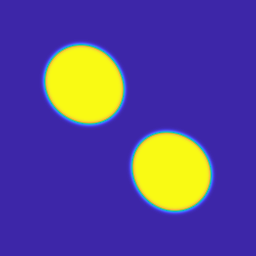}}{\begin{minipage}[t][\height][b]{24pt}\hspace{10pt}e\end{minipage}}$
$\overset{\includegraphics[width=53pt]{images/NoBoundary2D/2.50/f.png}}{\begin{minipage}[t][\height][b]{24pt}\hspace{10pt}f\end{minipage}}$
\caption{Minimum energy path ($\tdfis\!=\!0.532$). Transition state is d.}
\end{figure}

\section{Simulations in 2-D under periodic boundary conditions}
\label{Simulations in 2-D under periodic boundary conditions}

In this appendix we present the numerical simulations of \eqref{diffuse interface version of energy} in 2-D under periodic boundary conditions. As mentioned in \eqref{reparametrization of nonlocal coefficient}, the parameter $\tdfis$ refers to $\tilde\gamma\big(\omega|D|/\pi\big)^{3/2}/12$, where we choose $\omega|D|=0.1$. The results here are similar to those under no boundary conditions (in Section \ref{Analogue in 2-D}), and we expect the former to converge to the latter as $\omega\rightarrow0$.

We numerically study the pACOK dynamics with the initial value resembling the indicator
function of a disk. According to our simulation results, there is a critical value $\tdfis_*\in(1.225, 1.226)$ such that for $\tdfis<\tdfis_*$, the disk is stable; for $\tdfis>\tdfis_*$, the disk is unstable and will eventually deform into an eye-mask shaped local minimizer. We use such an eye-mask shaped local minimizer as the initial value, and numerically solve the pACOK
dynamics with a larger or smaller $\tdfis$. In this way we can track this bifurcation branch up and down. Our numerical results show that this bifurcation branch disappears at $\tdfis=\tdfis_4\in(1.182, 1.184)$ through a saddle-node bifurcation.

We use the string method to compute the minimum energy paths. The is a critical value $\tdfis_3\in(0.545, 0.556)$ such that for $\tdfis<\tdfis_3$, the computed minimum energy path of fission is unstable against mass-asymmetric perturbations. After sufficiently many iterations in the string method, such a fission path will eventually converge to another minimum energy path called the spallation path.

\begin{figure}[H]
\centering
\includegraphics[width=391.02pt]{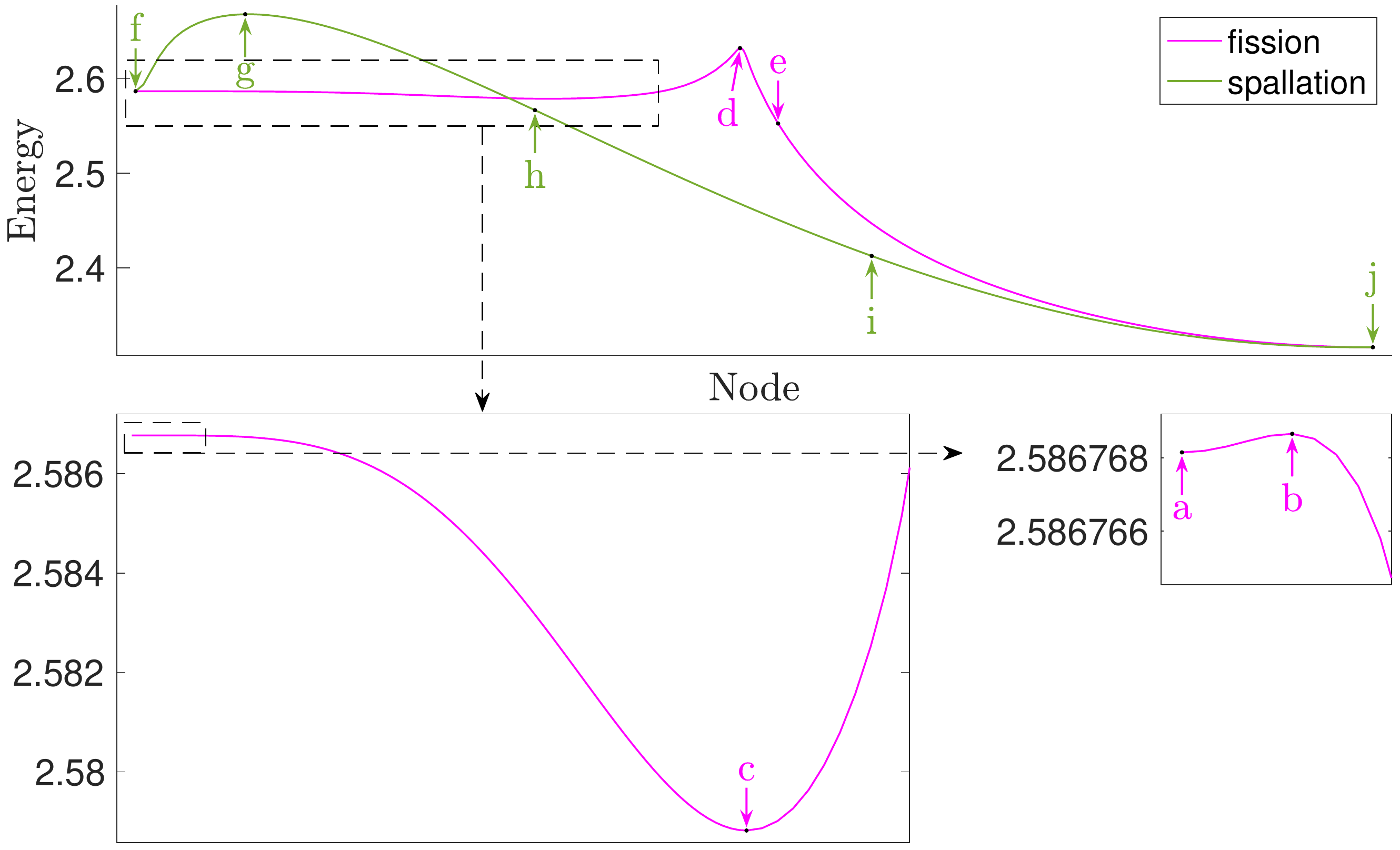}

$\overset{\includegraphics[width=53pt]{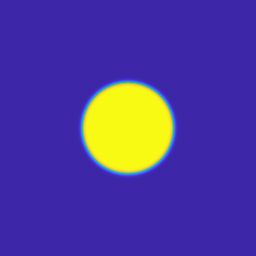}}{\begin{minipage}[t][\height][b]{24pt}\hspace{10pt}a\end{minipage}}$
$\overset{\includegraphics[width=53pt]{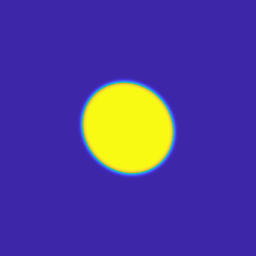}}{\begin{minipage}[t][\height][b]{24pt}\hspace{10pt}b\end{minipage}}$
$\overset{\includegraphics[width=53pt]{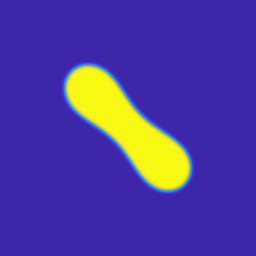}}{\begin{minipage}[t][\height][b]{24pt}\hspace{10pt}c\end{minipage}}$
$\overset{\includegraphics[width=53pt]{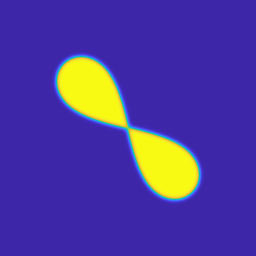}}{\begin{minipage}[t][\height][b]{24pt}\hspace{10pt}d\end{minipage}}$
$\overset{\includegraphics[width=53pt]{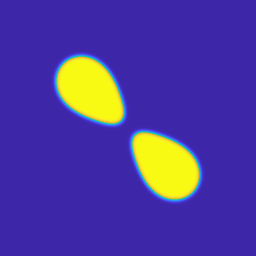}}{\begin{minipage}[t][\height][b]{24pt}\hspace{10pt}e\end{minipage}}$

\vspace{6pt}

$\overset{\includegraphics[width=53pt]{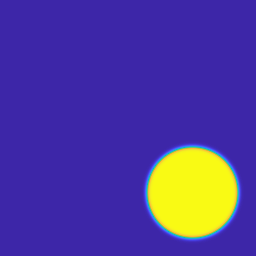}}{\begin{minipage}[t][\height][b]{24pt}\hspace{10pt}f\end{minipage}}$
$\overset{\includegraphics[width=53pt]{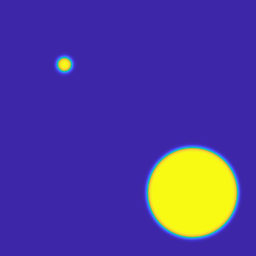}}{\begin{minipage}[t][\height][b]{24pt}\hspace{10pt}g\end{minipage}}$
$\overset{\includegraphics[width=53pt]{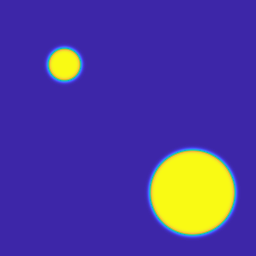}}{\begin{minipage}[t][\height][b]{24pt}\hspace{10pt}h\end{minipage}}$
$\overset{\includegraphics[width=53pt]{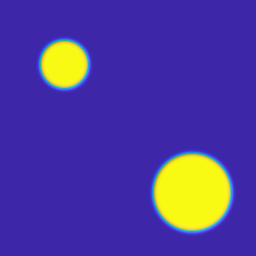}}{\begin{minipage}[t][\height][b]{24pt}\hspace{10pt}i\end{minipage}}$
$\overset{\includegraphics[width=53pt]{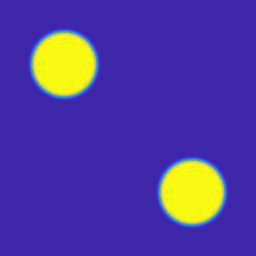}}{\begin{minipage}[t][\height][b]{24pt}\hspace{10pt}j\end{minipage}}$
\caption{Minimum energy paths ($\tdfis\!=\!1.225$). Transition states: b, d and g. Local minimizers: a, c and j.}
\end{figure}

\begin{figure}[H]
\centering
\includegraphics[width=391.02pt]{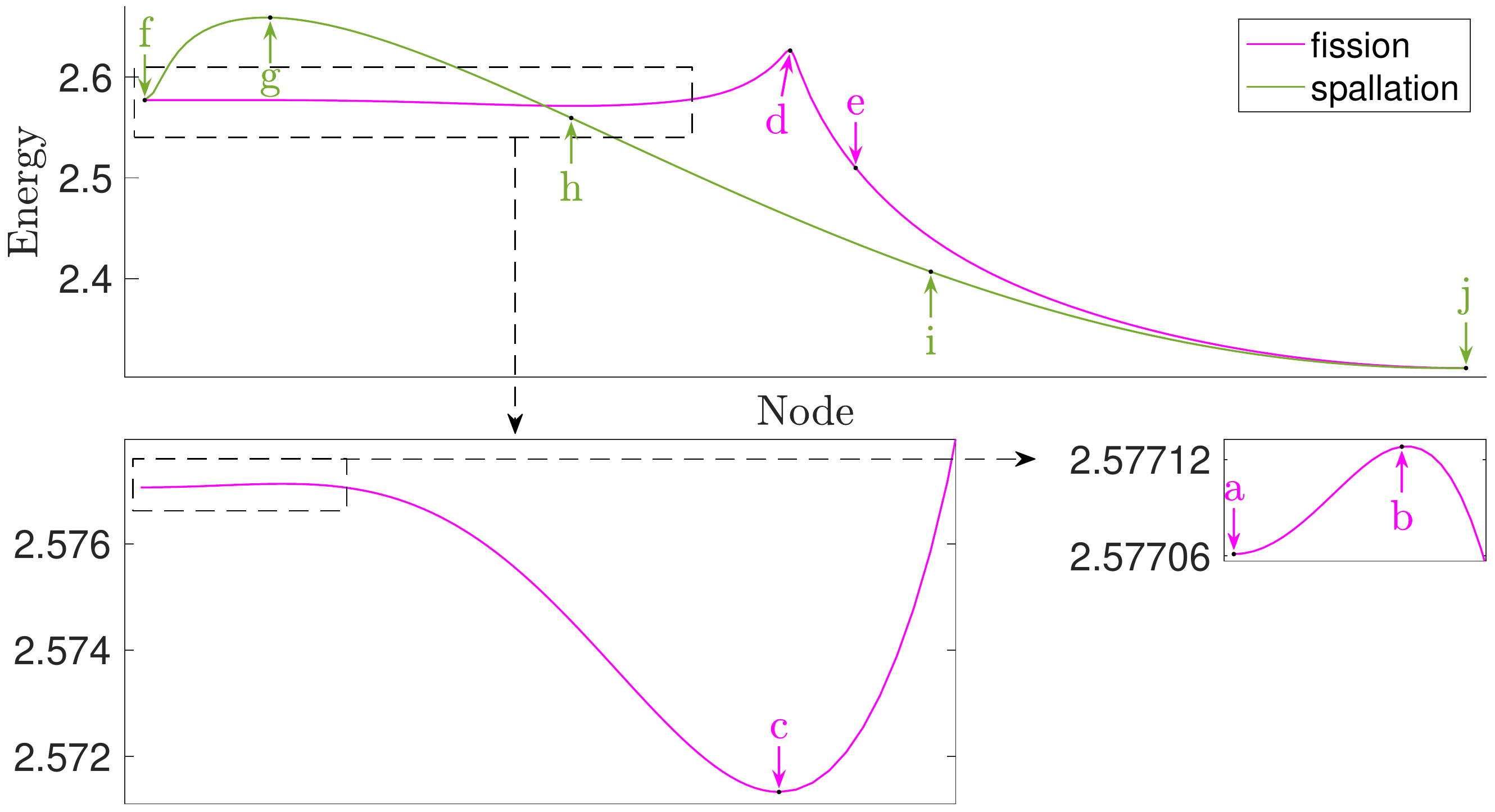}

$\overset{\includegraphics[width=53pt]{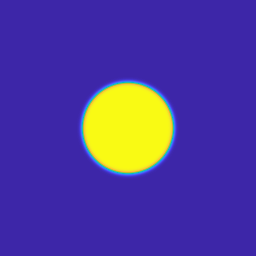}}{\begin{minipage}[t][\height][b]{24pt}\hspace{10pt}a\end{minipage}}$
$\overset{\includegraphics[width=53pt]{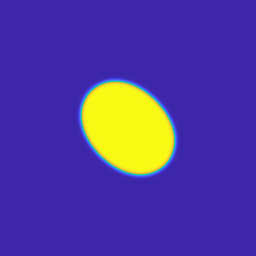}}{\begin{minipage}[t][\height][b]{24pt}\hspace{10pt}b\end{minipage}}$
$\overset{\includegraphics[width=53pt]{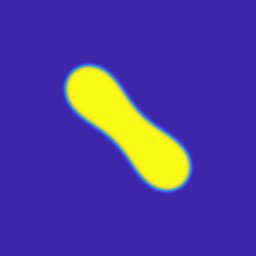}}{\begin{minipage}[t][\height][b]{24pt}\hspace{10pt}c\end{minipage}}$
$\overset{\includegraphics[width=53pt]{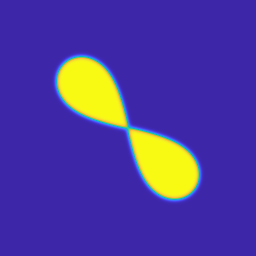}}{\begin{minipage}[t][\height][b]{24pt}\hspace{10pt}d\end{minipage}}$
$\overset{\includegraphics[width=53pt]{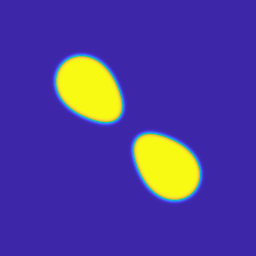}}{\begin{minipage}[t][\height][b]{24pt}\hspace{10pt}e\end{minipage}}$

\vspace{6pt}

$\overset{\includegraphics[width=53pt]{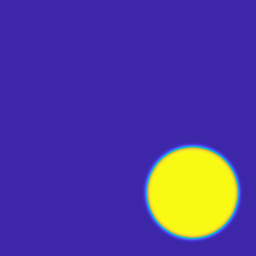}}{\begin{minipage}[t][\height][b]{24pt}\hspace{10pt}f\end{minipage}}$
$\overset{\includegraphics[width=53pt]{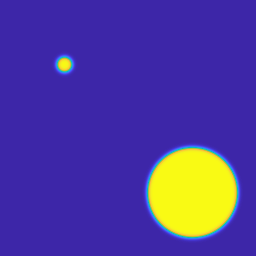}}{\begin{minipage}[t][\height][b]{24pt}\hspace{10pt}g\end{minipage}}$
$\overset{\includegraphics[width=53pt]{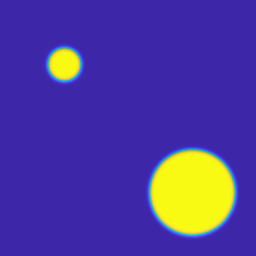}}{\begin{minipage}[t][\height][b]{24pt}\hspace{10pt}h\end{minipage}}$
$\overset{\includegraphics[width=53pt]{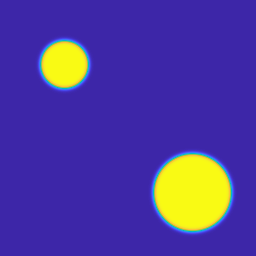}}{\begin{minipage}[t][\height][b]{24pt}\hspace{10pt}i\end{minipage}}$
$\overset{\includegraphics[width=53pt]{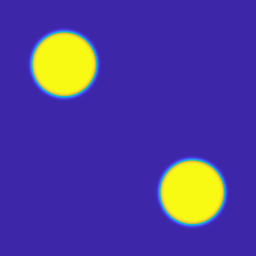}}{\begin{minipage}[t][\height][b]{24pt}\hspace{10pt}j\end{minipage}}$
\caption{Minimum energy paths ($\tdfis\!=\!1.217$). Transition states: b, d and g. Local minimizers: a, c and j.}
\end{figure}

\begin{figure}[H]
\centering
\includegraphics[width=391.02pt]{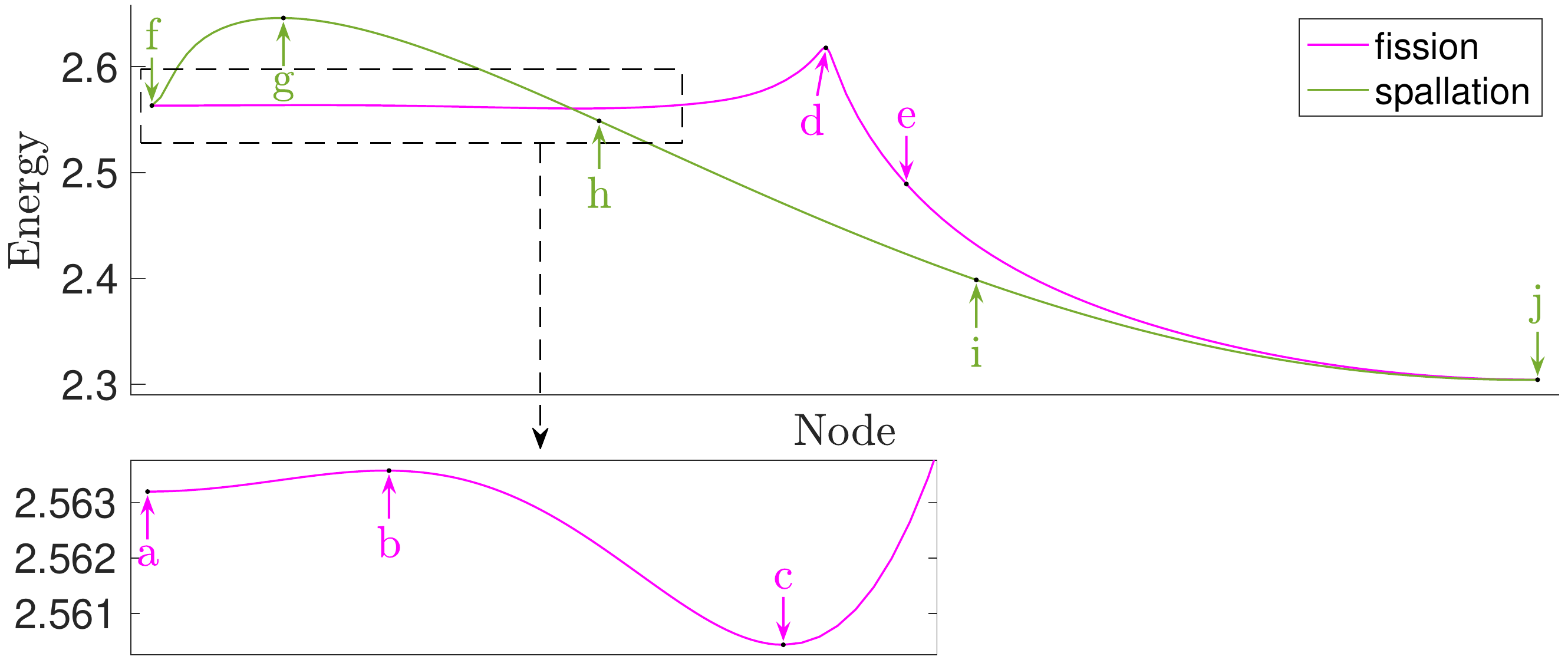}

$\overset{\includegraphics[width=53pt]{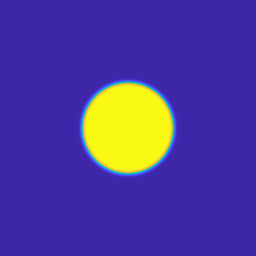}}{\begin{minipage}[t][\height][b]{24pt}\hspace{10pt}a\end{minipage}}$
$\overset{\includegraphics[width=53pt]{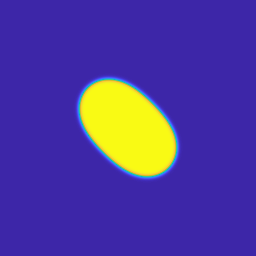}}{\begin{minipage}[t][\height][b]{24pt}\hspace{10pt}b\end{minipage}}$
$\overset{\includegraphics[width=53pt]{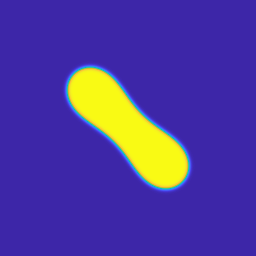}}{\begin{minipage}[t][\height][b]{24pt}\hspace{10pt}c\end{minipage}}$
$\overset{\includegraphics[width=53pt]{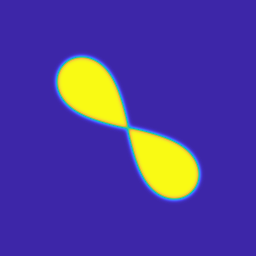}}{\begin{minipage}[t][\height][b]{24pt}\hspace{10pt}d\end{minipage}}$
$\overset{\includegraphics[width=53pt]{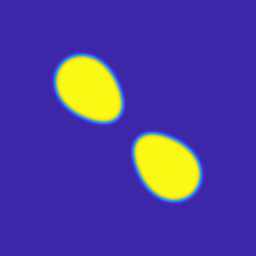}}{\begin{minipage}[t][\height][b]{24pt}\hspace{10pt}e\end{minipage}}$

\vspace{6pt}

$\overset{\includegraphics[width=53pt]{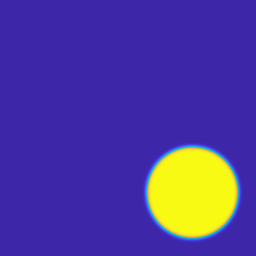}}{\begin{minipage}[t][\height][b]{24pt}\hspace{10pt}f\end{minipage}}$
$\overset{\includegraphics[width=53pt]{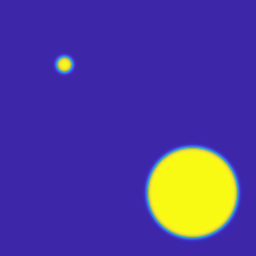}}{\begin{minipage}[t][\height][b]{24pt}\hspace{10pt}g\end{minipage}}$
$\overset{\includegraphics[width=53pt]{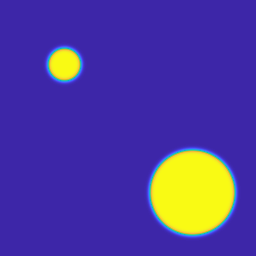}}{\begin{minipage}[t][\height][b]{24pt}\hspace{10pt}h\end{minipage}}$
$\overset{\includegraphics[width=53pt]{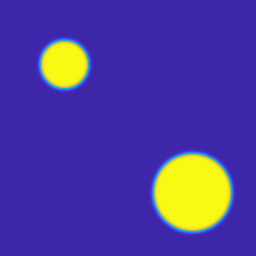}}{\begin{minipage}[t][\height][b]{24pt}\hspace{10pt}i\end{minipage}}$
$\overset{\includegraphics[width=53pt]{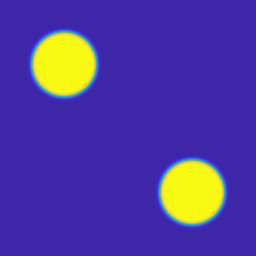}}{\begin{minipage}[t][\height][b]{24pt}\hspace{10pt}j\end{minipage}}$
\caption{Minimum energy paths ($\tdfis\!=\!1.206$). Transition states: b, d and g. Local minimizers: a, c and j.}
\end{figure}

\begin{figure}[H]
\centering
\includegraphics[width=391.02pt]{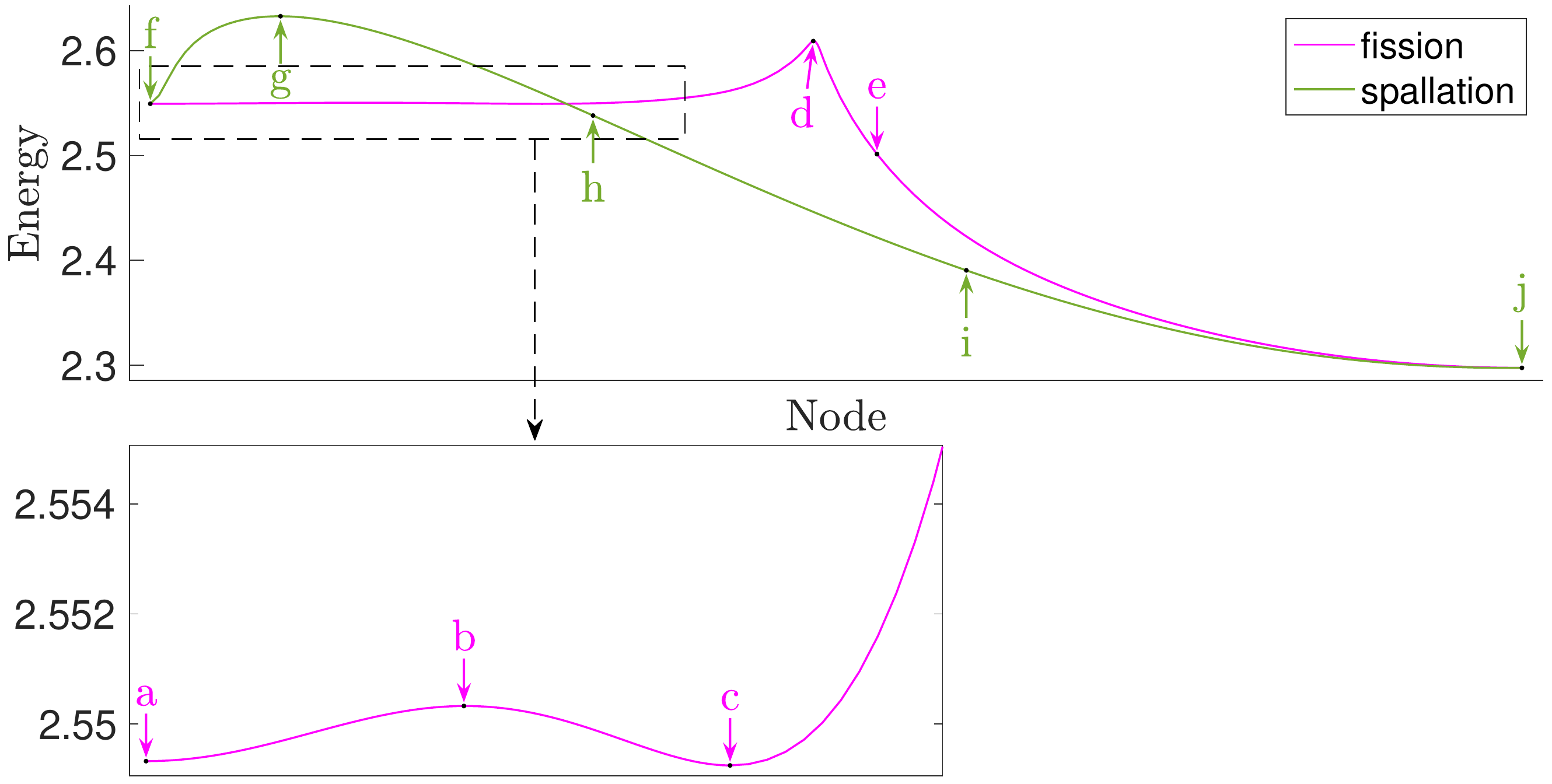}

$\overset{\includegraphics[width=53pt]{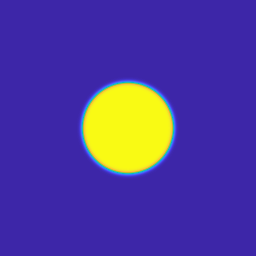}}{\begin{minipage}[t][\height][b]{24pt}\hspace{10pt}a\end{minipage}}$
$\overset{\includegraphics[width=53pt]{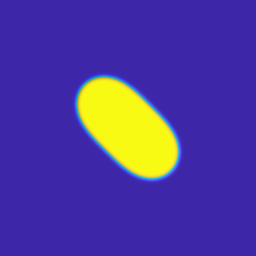}}{\begin{minipage}[t][\height][b]{24pt}\hspace{10pt}b\end{minipage}}$
$\overset{\includegraphics[width=53pt]{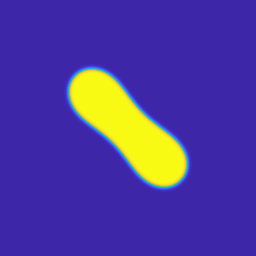}}{\begin{minipage}[t][\height][b]{24pt}\hspace{10pt}c\end{minipage}}$
$\overset{\includegraphics[width=53pt]{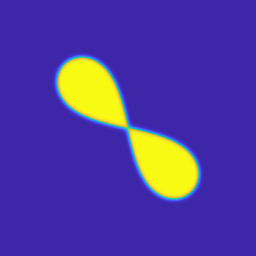}}{\begin{minipage}[t][\height][b]{24pt}\hspace{10pt}d\end{minipage}}$
$\overset{\includegraphics[width=53pt]{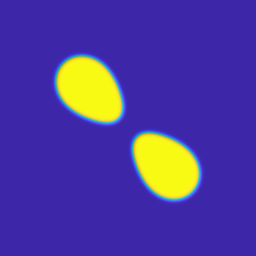}}{\begin{minipage}[t][\height][b]{24pt}\hspace{10pt}e\end{minipage}}$

\vspace{6pt}

$\overset{\includegraphics[width=53pt]{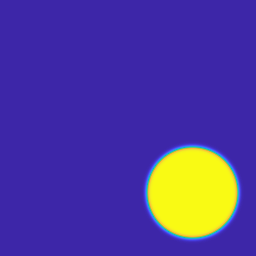}}{\begin{minipage}[t][\height][b]{24pt}\hspace{10pt}f\end{minipage}}$
$\overset{\includegraphics[width=53pt]{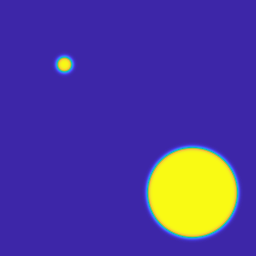}}{\begin{minipage}[t][\height][b]{24pt}\hspace{10pt}g\end{minipage}}$
$\overset{\includegraphics[width=53pt]{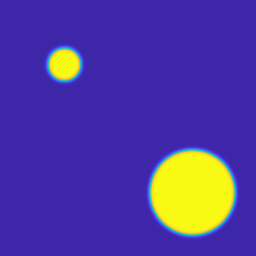}}{\begin{minipage}[t][\height][b]{24pt}\hspace{10pt}h\end{minipage}}$
$\overset{\includegraphics[width=53pt]{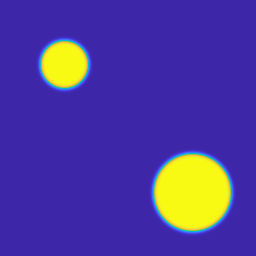}}{\begin{minipage}[t][\height][b]{24pt}\hspace{10pt}i\end{minipage}}$
$\overset{\includegraphics[width=53pt]{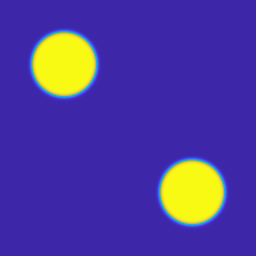}}{\begin{minipage}[t][\height][b]{24pt}\hspace{10pt}j\end{minipage}}$
\caption{Minimum energy paths ($\tdfis\!=\!1.194$). Transition states: b, d and g. Local minimizers: a, c and j.}
\end{figure}

\begin{figure}[H]
\centering
\includegraphics[width=391.02pt]{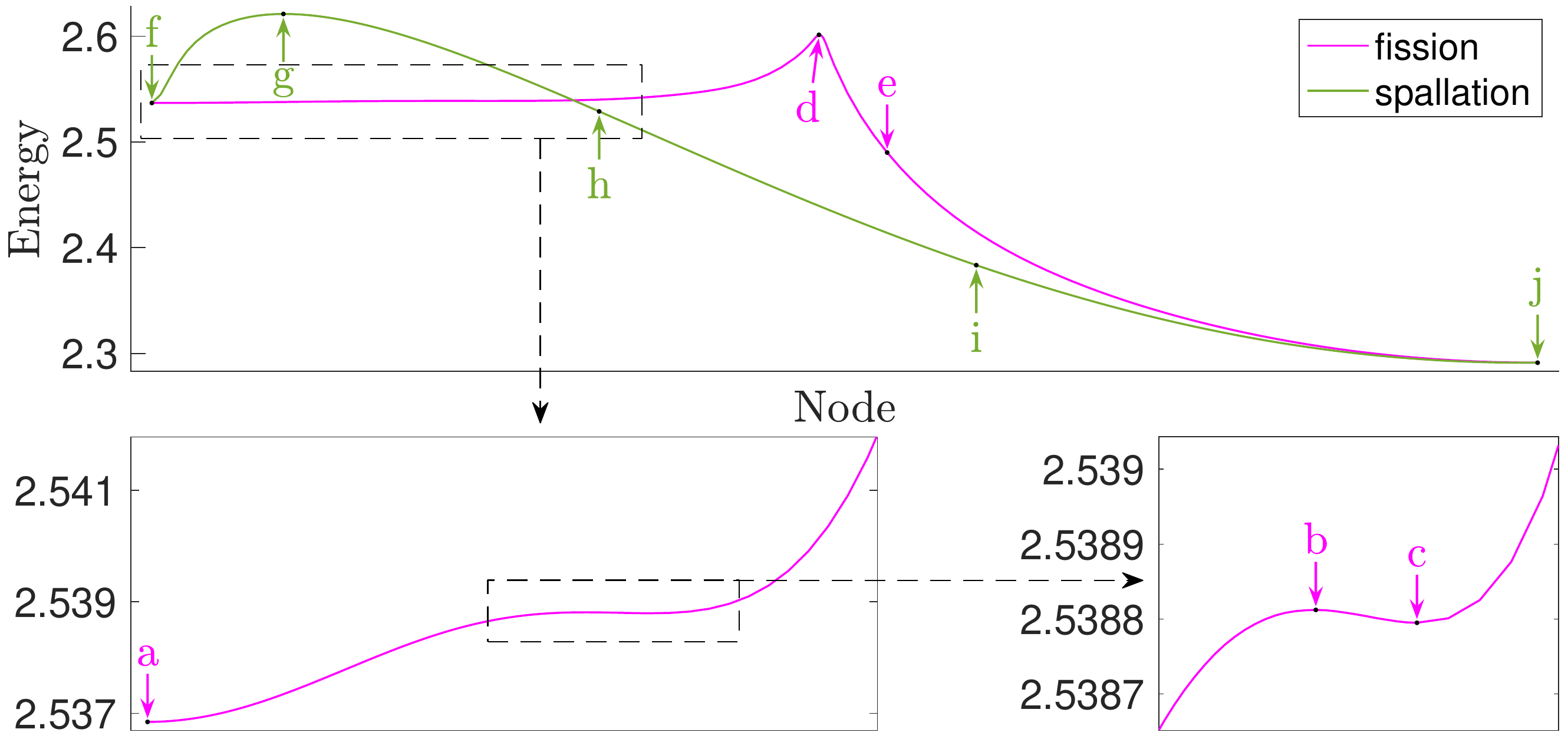}

$\overset{\includegraphics[width=53pt]{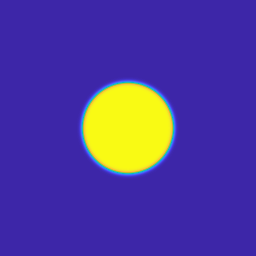}}{\begin{minipage}[t][\height][b]{24pt}\hspace{10pt}a\end{minipage}}$
$\overset{\includegraphics[width=53pt]{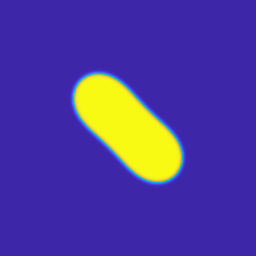}}{\begin{minipage}[t][\height][b]{24pt}\hspace{10pt}b\end{minipage}}$
$\overset{\includegraphics[width=53pt]{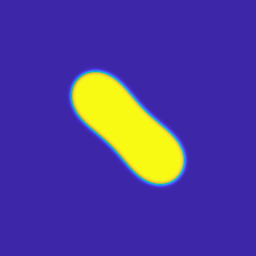}}{\begin{minipage}[t][\height][b]{24pt}\hspace{10pt}c\end{minipage}}$
$\overset{\includegraphics[width=53pt]{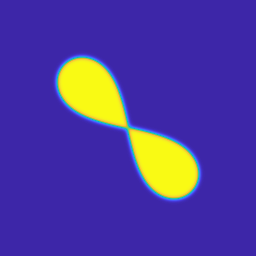}}{\begin{minipage}[t][\height][b]{24pt}\hspace{10pt}d\end{minipage}}$
$\overset{\includegraphics[width=53pt]{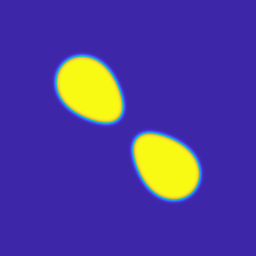}}{\begin{minipage}[t][\height][b]{24pt}\hspace{10pt}e\end{minipage}}$

\vspace{6pt}

$\overset{\includegraphics[width=53pt]{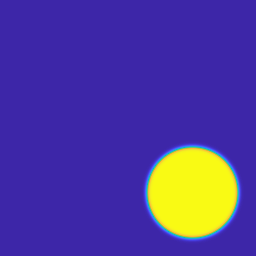}}{\begin{minipage}[t][\height][b]{24pt}\hspace{10pt}f\end{minipage}}$
$\overset{\includegraphics[width=53pt]{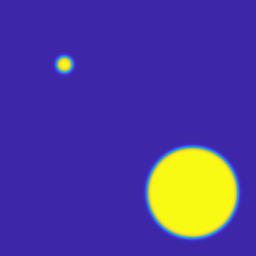}}{\begin{minipage}[t][\height][b]{24pt}\hspace{10pt}g\end{minipage}}$
$\overset{\includegraphics[width=53pt]{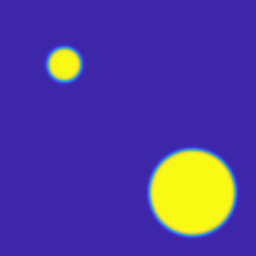}}{\begin{minipage}[t][\height][b]{24pt}\hspace{10pt}h\end{minipage}}$
$\overset{\includegraphics[width=53pt]{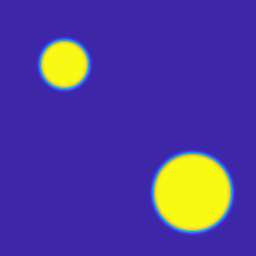}}{\begin{minipage}[t][\height][b]{24pt}\hspace{10pt}i\end{minipage}}$
$\overset{\includegraphics[width=53pt]{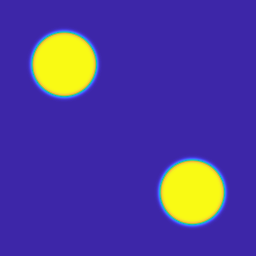}}{\begin{minipage}[t][\height][b]{24pt}\hspace{10pt}j\end{minipage}}$
\caption{Minimum energy paths ($\tdfis\!=\!1.184$). Transition states: b, d and g. Local minimizers: a, c and j.}
\end{figure}

\begin{figure}[H]
\centering
\includegraphics[width=391.02pt]{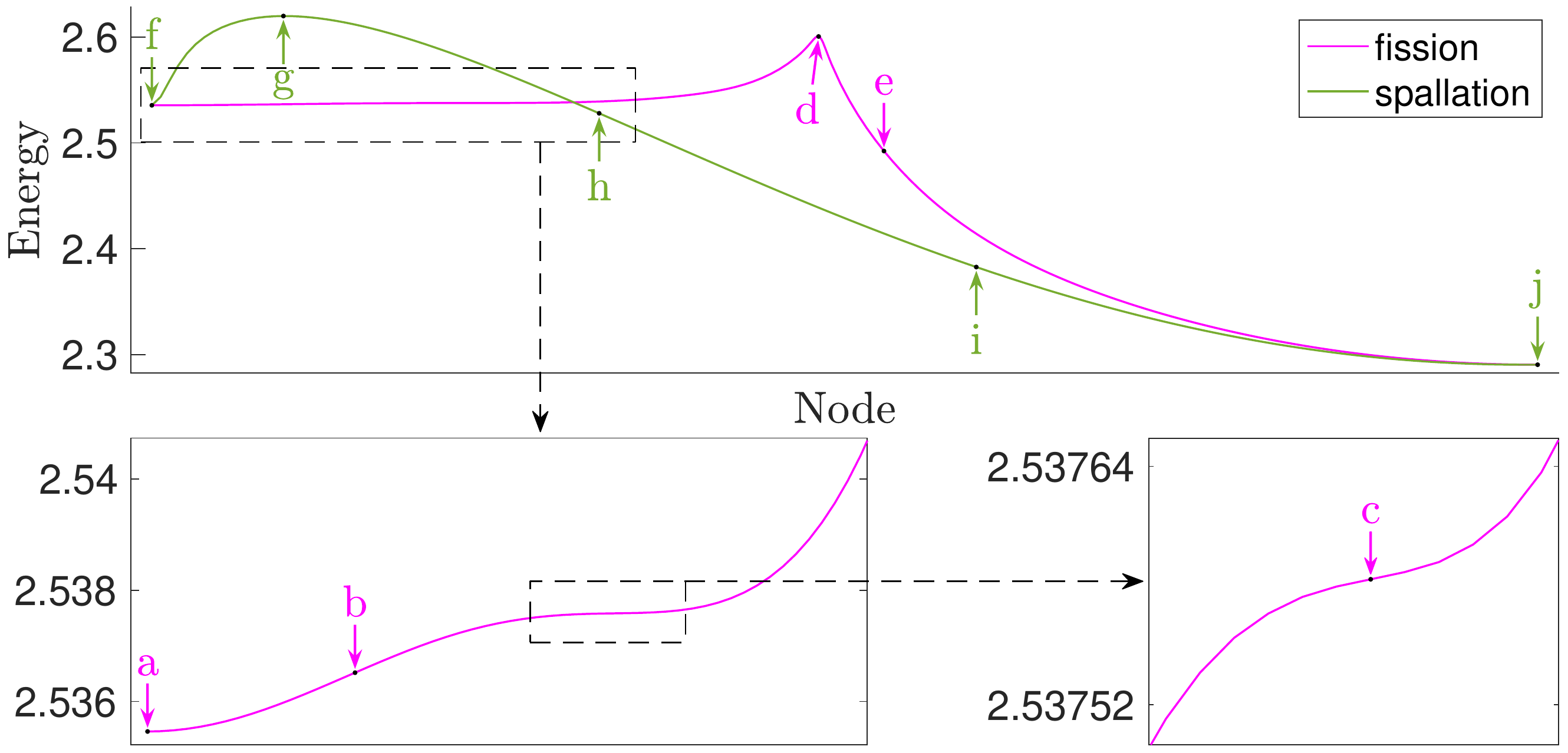}

$\overset{\includegraphics[width=53pt]{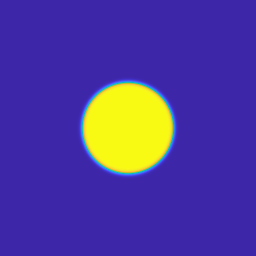}}{\begin{minipage}[t][\height][b]{24pt}\hspace{10pt}a\end{minipage}}$
$\overset{\includegraphics[width=53pt]{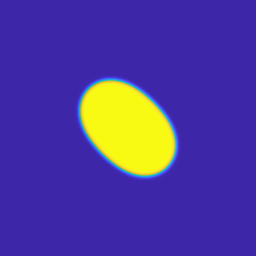}}{\begin{minipage}[t][\height][b]{24pt}\hspace{10pt}b\end{minipage}}$
$\overset{\includegraphics[width=53pt]{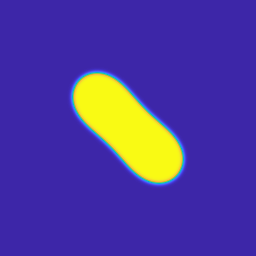}}{\begin{minipage}[t][\height][b]{24pt}\hspace{10pt}c\end{minipage}}$
$\overset{\includegraphics[width=53pt]{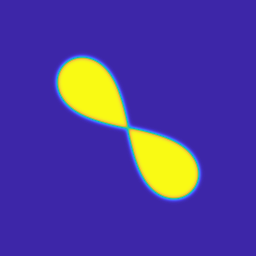}}{\begin{minipage}[t][\height][b]{24pt}\hspace{10pt}d\end{minipage}}$
$\overset{\includegraphics[width=53pt]{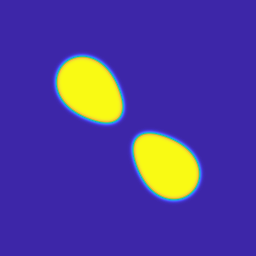}}{\begin{minipage}[t][\height][b]{24pt}\hspace{10pt}e\end{minipage}}$

\vspace{6pt}

$\overset{\includegraphics[width=53pt]{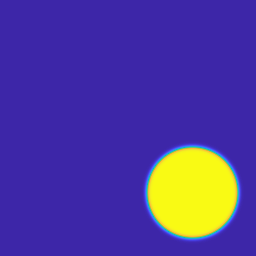}}{\begin{minipage}[t][\height][b]{24pt}\hspace{10pt}f\end{minipage}}$
$\overset{\includegraphics[width=53pt]{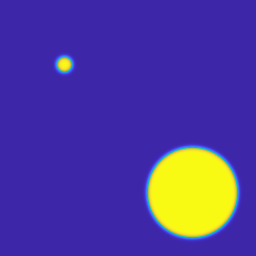}}{\begin{minipage}[t][\height][b]{24pt}\hspace{10pt}g\end{minipage}}$
$\overset{\includegraphics[width=53pt]{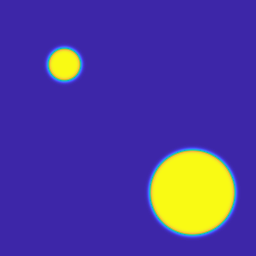}}{\begin{minipage}[t][\height][b]{24pt}\hspace{10pt}h\end{minipage}}$
$\overset{\includegraphics[width=53pt]{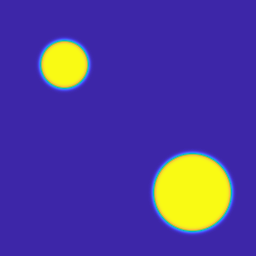}}{\begin{minipage}[t][\height][b]{24pt}\hspace{10pt}i\end{minipage}}$
$\overset{\includegraphics[width=53pt]{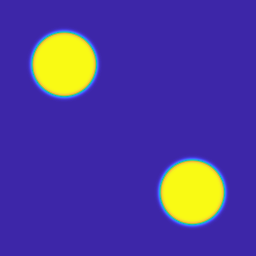}}{\begin{minipage}[t][\height][b]{24pt}\hspace{10pt}j\end{minipage}}$
\caption{Minimum energy paths ($\tdfis\!=\!1.182$). Transition states: d and g.}
\end{figure}

\begin{figure}[H]
\centering
\includegraphics[width=391.02pt]{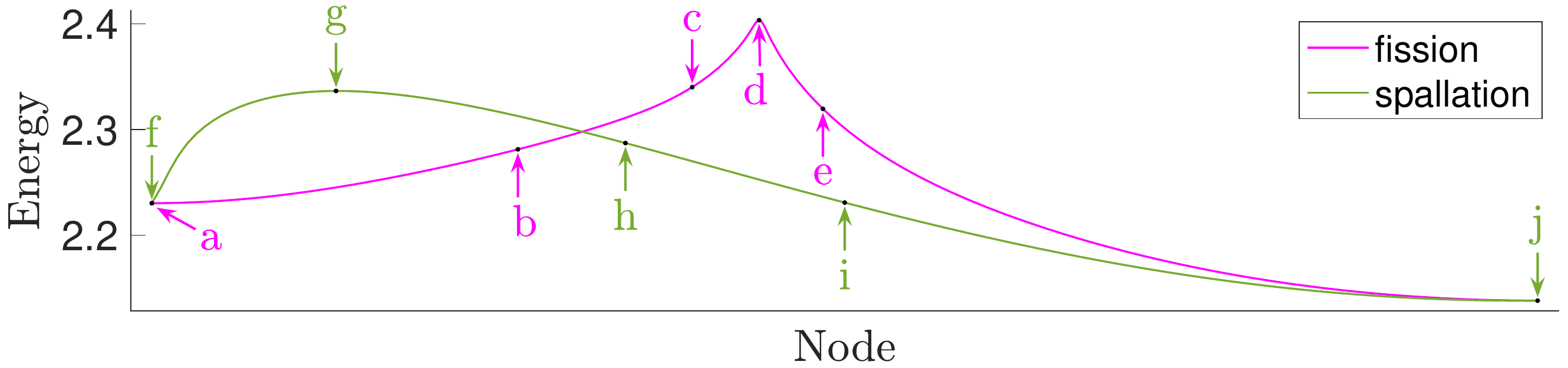}

$\overset{\includegraphics[width=53pt]{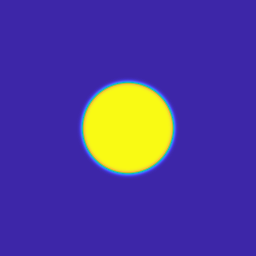}}{\begin{minipage}[t][\height][b]{24pt}\hspace{10pt}a\end{minipage}}$
$\overset{\includegraphics[width=53pt]{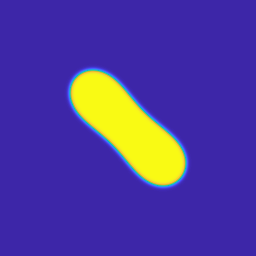}}{\begin{minipage}[t][\height][b]{24pt}\hspace{10pt}b\end{minipage}}$
$\overset{\includegraphics[width=53pt]{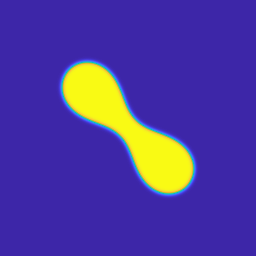}}{\begin{minipage}[t][\height][b]{24pt}\hspace{10pt}c\end{minipage}}$
$\overset{\includegraphics[width=53pt]{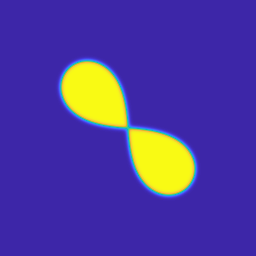}}{\begin{minipage}[t][\height][b]{24pt}\hspace{10pt}d\end{minipage}}$
$\overset{\includegraphics[width=53pt]{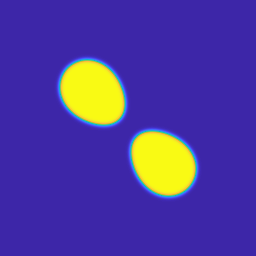}}{\begin{minipage}[t][\height][b]{24pt}\hspace{10pt}e\end{minipage}}$

\vspace{6pt}

$\overset{\includegraphics[width=53pt]{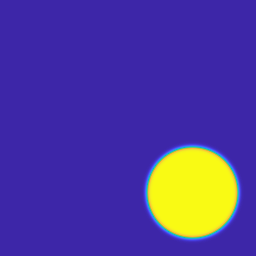}}{\begin{minipage}[t][\height][b]{24pt}\hspace{10pt}f\end{minipage}}$
$\overset{\includegraphics[width=53pt]{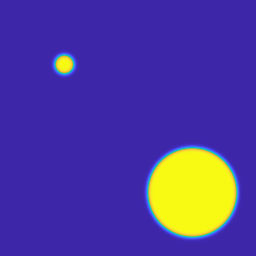}}{\begin{minipage}[t][\height][b]{24pt}\hspace{10pt}g\end{minipage}}$
$\overset{\includegraphics[width=53pt]{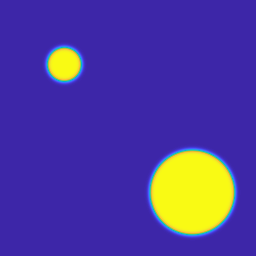}}{\begin{minipage}[t][\height][b]{24pt}\hspace{10pt}h\end{minipage}}$
$\overset{\includegraphics[width=53pt]{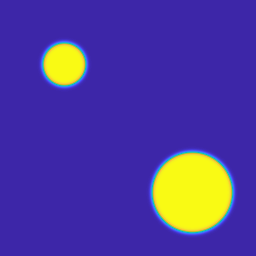}}{\begin{minipage}[t][\height][b]{24pt}\hspace{10pt}i\end{minipage}}$
$\overset{\includegraphics[width=53pt]{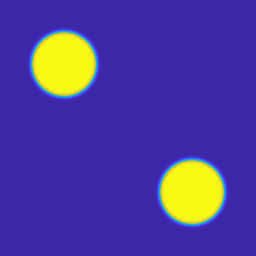}}{\begin{minipage}[t][\height][b]{24pt}\hspace{10pt}j\end{minipage}}$
\caption{Minimum energy paths ($\tdfis\!=\!0.927$). Transition states: d and g.}
\end{figure}

\begin{figure}[H]
\centering
\includegraphics[width=391.02pt]{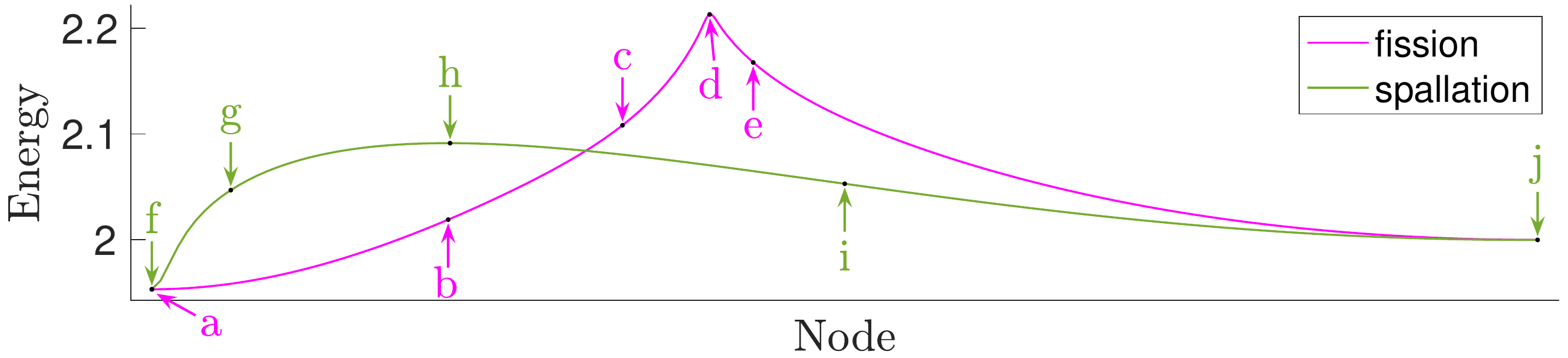}

$\overset{\includegraphics[width=53pt]{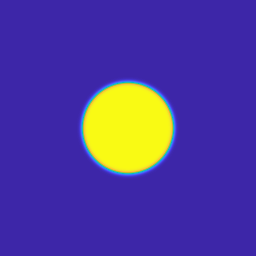}}{\begin{minipage}[t][\height][b]{24pt}\hspace{10pt}a\end{minipage}}$
$\overset{\includegraphics[width=53pt]{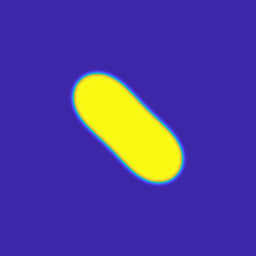}}{\begin{minipage}[t][\height][b]{24pt}\hspace{10pt}b\end{minipage}}$
$\overset{\includegraphics[width=53pt]{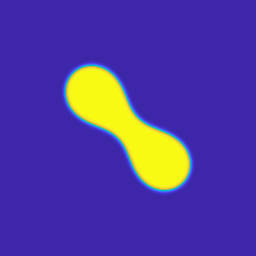}}{\begin{minipage}[t][\height][b]{24pt}\hspace{10pt}c\end{minipage}}$
$\overset{\includegraphics[width=53pt]{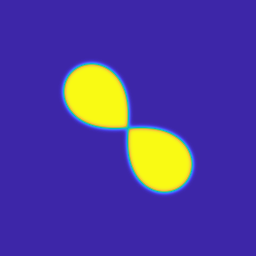}}{\begin{minipage}[t][\height][b]{24pt}\hspace{10pt}d\end{minipage}}$
$\overset{\includegraphics[width=53pt]{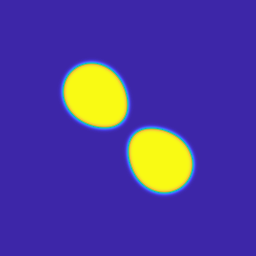}}{\begin{minipage}[t][\height][b]{24pt}\hspace{10pt}e\end{minipage}}$

\vspace{6pt}

$\overset{\includegraphics[width=53pt]{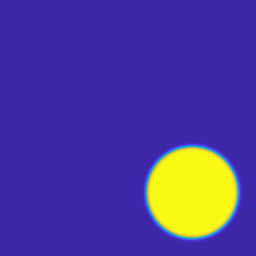}}{\begin{minipage}[t][\height][b]{24pt}\hspace{10pt}f\end{minipage}}$
$\overset{\includegraphics[width=53pt]{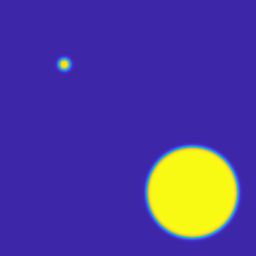}}{\begin{minipage}[t][\height][b]{24pt}\hspace{10pt}g\end{minipage}}$
$\overset{\includegraphics[width=53pt]{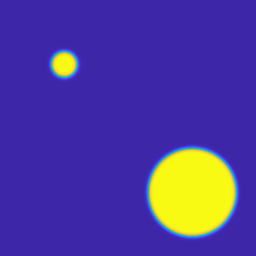}}{\begin{minipage}[t][\height][b]{24pt}\hspace{10pt}h\end{minipage}}$
$\overset{\includegraphics[width=53pt]{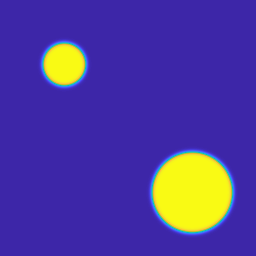}}{\begin{minipage}[t][\height][b]{24pt}\hspace{10pt}i\end{minipage}}$
$\overset{\includegraphics[width=53pt]{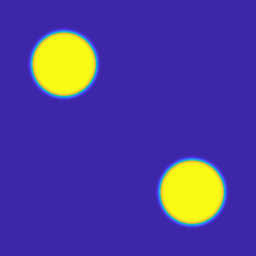}}{\begin{minipage}[t][\height][b]{24pt}\hspace{10pt}j\end{minipage}}$
\caption{Minimum energy paths ($\tdfis\!=\!0.696$). Transition states: d and h.}
\end{figure}

\begin{figure}[H]
\centering
\includegraphics[width=391.02pt]{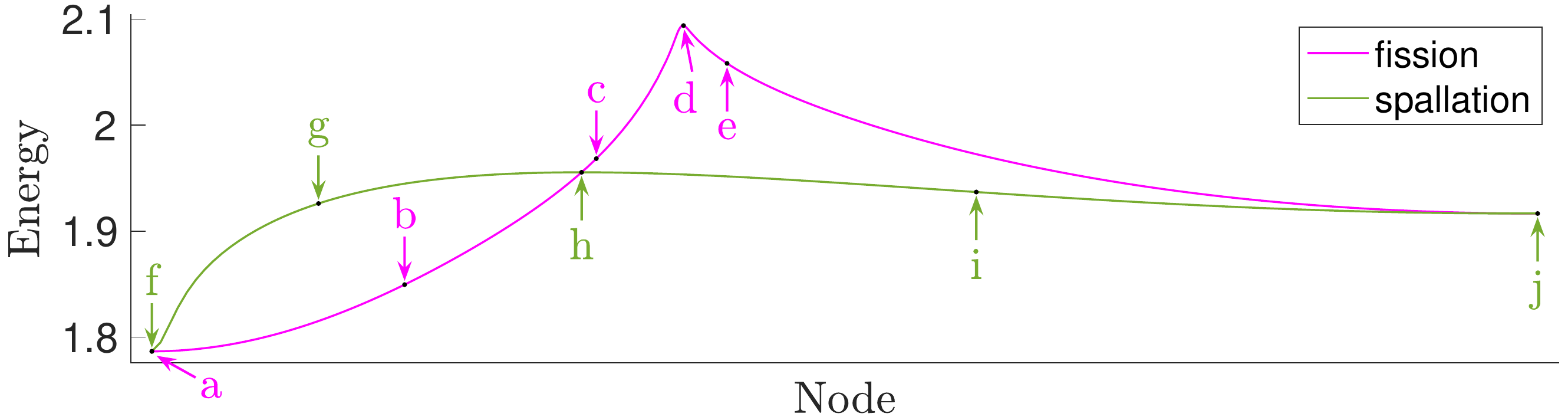}

$\overset{\includegraphics[width=53pt]{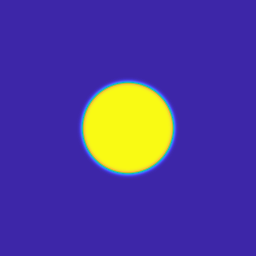}}{\begin{minipage}[t][\height][b]{24pt}\hspace{10pt}a\end{minipage}}$
$\overset{\includegraphics[width=53pt]{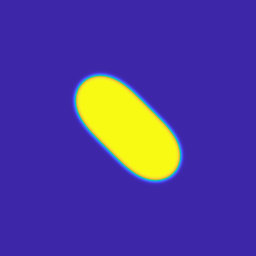}}{\begin{minipage}[t][\height][b]{24pt}\hspace{10pt}b\end{minipage}}$
$\overset{\includegraphics[width=53pt]{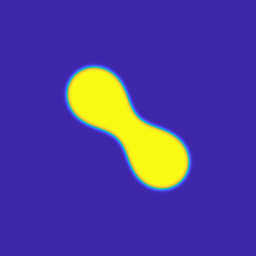}}{\begin{minipage}[t][\height][b]{24pt}\hspace{10pt}c\end{minipage}}$
$\overset{\includegraphics[width=53pt]{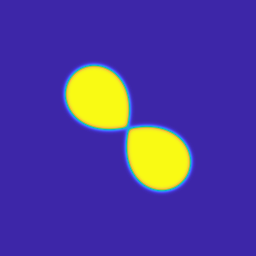}}{\begin{minipage}[t][\height][b]{24pt}\hspace{10pt}d\end{minipage}}$
$\overset{\includegraphics[width=53pt]{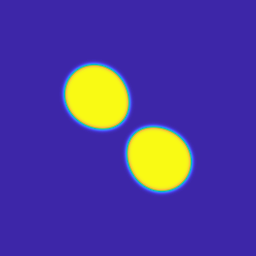}}{\begin{minipage}[t][\height][b]{24pt}\hspace{10pt}e\end{minipage}}$

\vspace{6pt}

$\overset{\includegraphics[width=53pt]{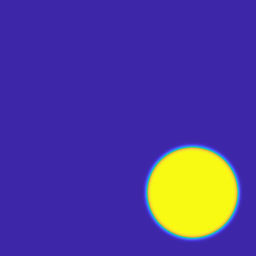}}{\begin{minipage}[t][\height][b]{24pt}\hspace{10pt}f\end{minipage}}$
$\overset{\includegraphics[width=53pt]{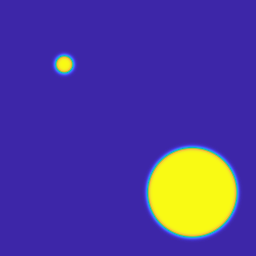}}{\begin{minipage}[t][\height][b]{24pt}\hspace{10pt}g\end{minipage}}$
$\overset{\includegraphics[width=53pt]{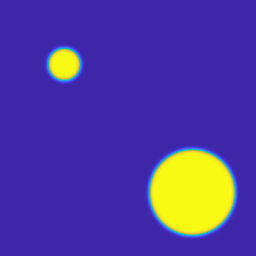}}{\begin{minipage}[t][\height][b]{24pt}\hspace{10pt}h\end{minipage}}$
$\overset{\includegraphics[width=53pt]{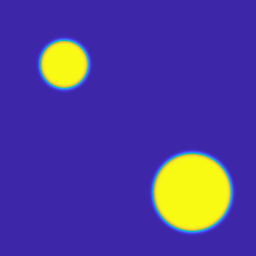}}{\begin{minipage}[t][\height][b]{24pt}\hspace{10pt}i\end{minipage}}$
$\overset{\includegraphics[width=53pt]{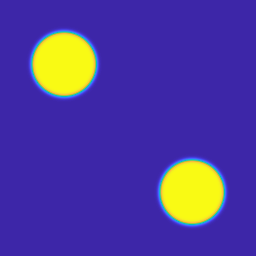}}{\begin{minipage}[t][\height][b]{24pt}\hspace{10pt}j\end{minipage}}$
\caption{Minimum energy paths ($\tdfis\!=\!0.556$). Transition states: d and h.}
\end{figure}

\section{Simulations in 3-D of torus-like equilibria}
\label{torus appendix}

In this appendix, we present torus-like equilibria in 3-D under no boundary conditions. A portion (where $\fis$ is large enough) of this bifurcation branch has already been found in \cite{ren2011toroidal}. As mentioned in Section \ref{Resistance to rupture in the classical isoperimetric problem}, a cylinder with a height larger than its circumference is unstable under surface tension. Therefore we expect the torus-like equilibria to be unstable (see also \cite{ren2017spectrum}). (In contrast, a 2-D strip is stable under surface tension, so an annulus might be stable in 2-D \cite[Theorem 2.2]{kang2009ring}. Similarly, we expect a spherical shell to be stable in 3-D. However, the spherical shell is reported to be unstable \cite[Page 981]{ren2009shell}, which may seem surprising at first, but such an instability can be understood as the consequence of the Coulomb repulsive term.)

In order to overcome the instability of the torus-like equilibrium, we employ the following technique to maintain the axisymmetry of the shape: after every a few (e.g., 500) iterations in the pACOK dynamics, rotate it by $30^\circ,\,60^\circ,\,90^\circ,\,\cdots,\,360^\circ$, then take the average in the sense of superposition. In our simulations, there exists a $\tdfis_4\in(0.968, 0.969)$, such that for $\tdfis<\tdfis_4$, the torus will eventually converge to a ball, so we suspect that the torus-like equilibrium disappears through a saddle-node bifurcation at $\fis=\fis_4$. We use $\tdfis_4/\tdfis_*\approx0.972$ to estimate $\fis_4$. For $\tdfis>\tdfis_4$, our numerical results are shown in Figure \ref{torus equilibria}.

\begin{figure}[H]
\centering
\hspace{-40pt}$\overset{\includegraphics[width=103.6pt,bb=0 0 569 344]{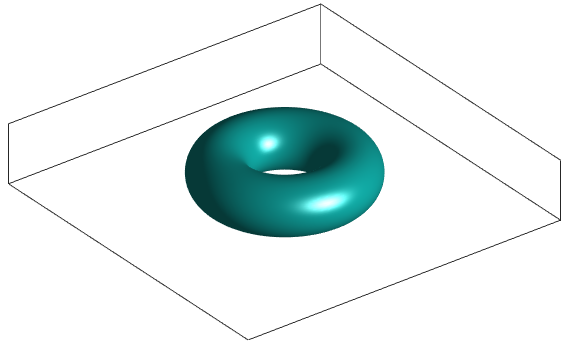}}{0.969}$
\hspace{-12pt}
$\overset{\includegraphics[width=103.6pt,bb=0 0 569 344]{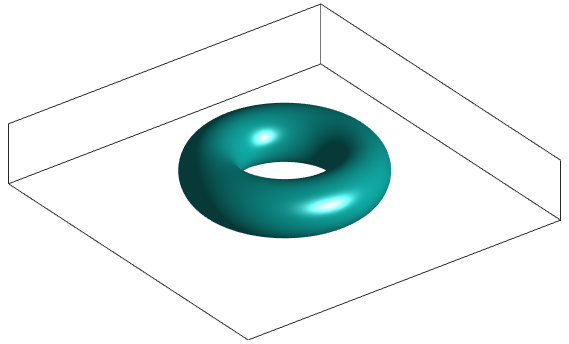}}{0.992}$
\hspace{-12pt}
$\overset{\includegraphics[width=103.6pt,bb=0 0 569 344]{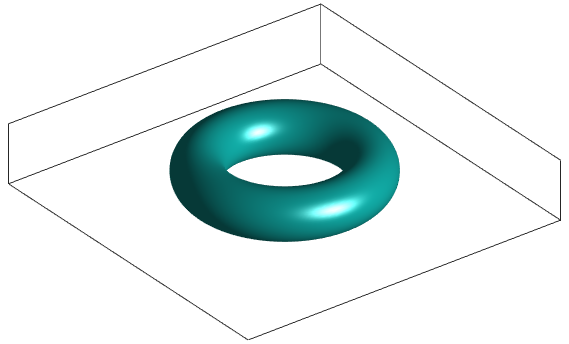}}{1.054}$
\hspace{-12pt}
$\overset{\includegraphics[width=103.6pt,bb=0 0 569 344]{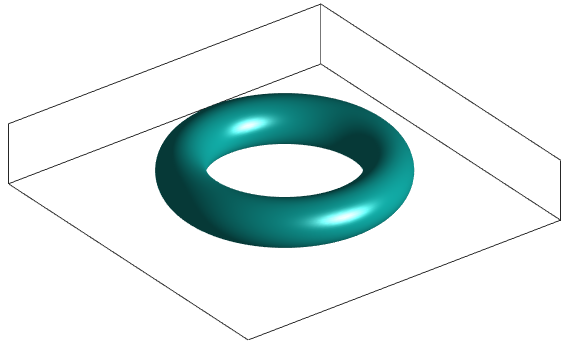}}{1.116}$
\\
\hspace{40pt}$\overset{\includegraphics[width=103.6pt,bb=0 0 569 344]{images/BifurcationBranch/t9.5.png}}{1.178}$
\hspace{-12pt}
$\overset{\includegraphics[width=103.6pt,bb=0 0 569 344]{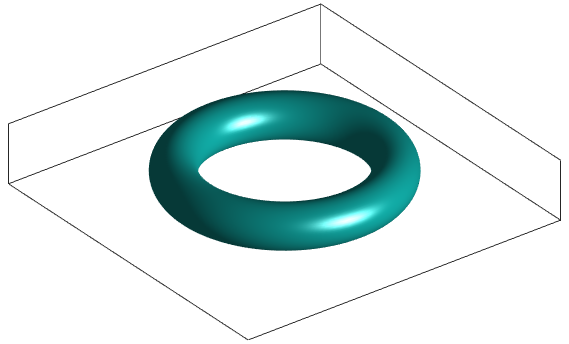}}{1.240}$
\hspace{-12pt}
$\overset{\includegraphics[width=103.6pt,bb=0 0 569 344]{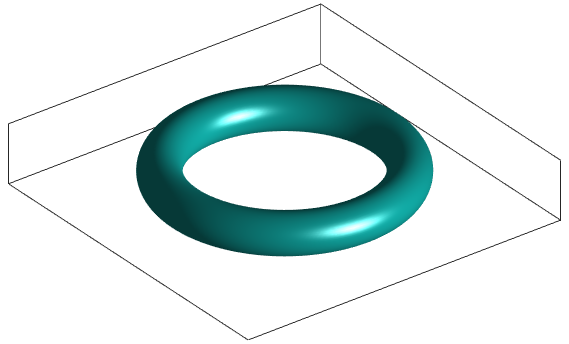}}{1.365}$
\hspace{-12pt}
$\overset{\includegraphics[width=103.6pt,bb=0 0 569 344]{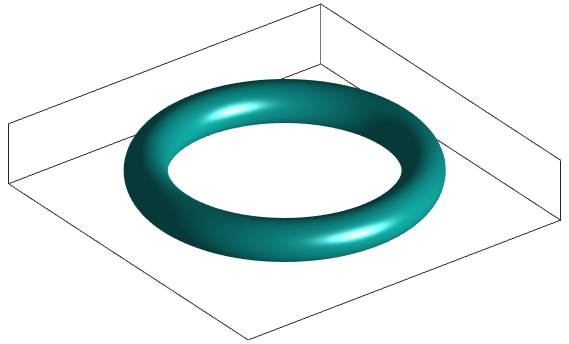}}{1.489}$
\\
\hspace{-40pt}$\overset{\includegraphics[width=103.6pt,bb=0 0 569 344]{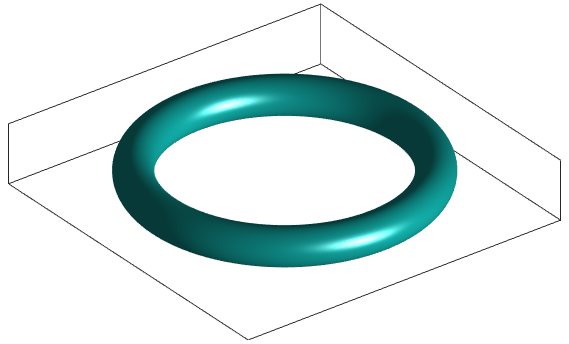}}{1.613}$
\hspace{-12pt}
$\overset{\includegraphics[width=103.6pt,bb=0 0 569 344]{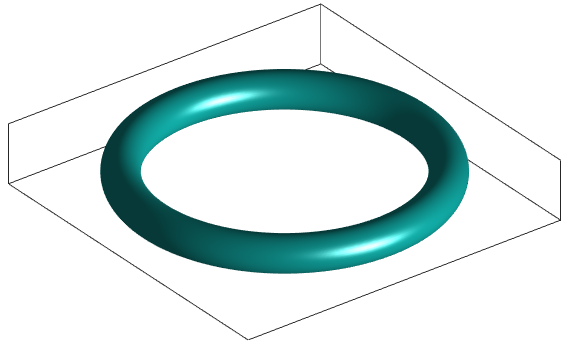}}{1.737}$
\hspace{-12pt}
$\overset{\includegraphics[width=103.6pt,bb=0 0 569 344]{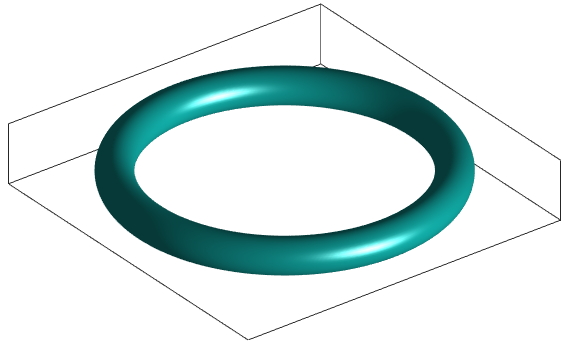}}{1.799}$
\hspace{-12pt}
$\overset{\includegraphics[width=103.6pt,bb=0 0 569 344]{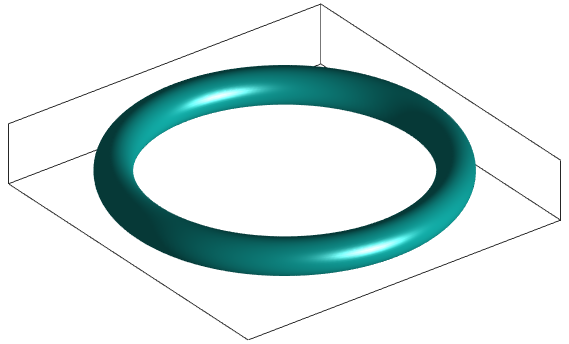}}{1.811}$
\caption{Torus-like equilibria for various $\tdfis$. Size of bounding box: 1.6$\times$1.6$\times$0.28125.}
\label{torus equilibria}
\end{figure}

%\section*{Future works}

%\begin{enumerate}[label=\protect\CircleAroundChar{\arabic*}]
%\item 
%\end{enumerate}

\end{document}